\documentclass[11pt,table]{article}

\usepackage{fullpage}
\usepackage{graphicx}
\usepackage{xspace}
\usepackage{url}
\usepackage{amssymb}
\usepackage{latexsym}
\usepackage{amsthm}
\usepackage{amsmath}
\usepackage{amstext}
\usepackage{amstext}
\usepackage{calc}
\usepackage{amsopn}
\usepackage[noend]{algorithmic}
\usepackage[boxed]{algorithm}
\usepackage{eucal}
\usepackage{latexsym}
\usepackage{subfig}
\usepackage{tabularx}
\usepackage{todonotes}
\usepackage{multirow}
\usepackage{vmargin}
 \setmarginsrb{1in}{1in}{1in}{1in}{0pt}{0pt}{0pt}{7mm}
\usepackage{enumerate}
\usepackage{skull}
\usepackage{tikz}
\usepackage{marvosym}
\usepackage{comment}
\usepackage{caption}
\usepackage{subfig}
\usepackage{multicol}
\usepackage{enumitem}
\usepackage{lscape}
\usepackage{colortbl}
\usepackage{multirow}
\usepackage{hhline}
\usepackage{afterpage}
\graphicspath{{/}}

\newcommand{\executeiffilenewer}[3]{%
\ifnum\pdfstrcmp{\pdffilemoddate{#1}}%
{\pdffilemoddate{#2}}>0%
{\immediate\write18{#3}}\fi%
} 
\floatname{algorithm}{Algorithm}

\newcommand{\defproblemu}[3]{
  \vspace{1mm}
\noindent\fbox{
  \begin{minipage}{\textwidth}
  #1 \\
  {\bf{Input:}} #2  \\
  {\bf{Question:}} #3
  \end{minipage}
  }
}
\newcommand{\defparproblemu}[4]{
  \vspace{1mm}
\noindent\fbox{
  \begin{minipage}{\textwidth}
  \begin{tabular*}{\textwidth}{@{\extracolsep{\fill}}lr} #1 &
 \\ \end{tabular*}
  {\bf{Input:}} #2  \\
  {\bf{Question:}} #4
  \end{minipage}
  }
}

\newcommand{\cK}{{\mathcal{K}}}

\newtheorem{theorem}{Theorem}[section]
\newtheorem{lemma}[theorem]{Lemma}
\newtheorem{claim}[theorem]{Claim}

\newtheorem{proposition}[theorem]{Proposition}
\theoremstyle{definition}
\newtheorem{definition}[theorem]{Definition}

\newtheorem{example}[theorem]{Example}

\newcommand{\cqed}{\renewcommand{\qedsymbol}{$\lrcorner$}}

\newtheorem{innerptheorem}{Theorem}
\newenvironment{ptheorem}
{\innerptheorem}
{\endinnerptheorem}
\newtheorem{innerntheorem}{Theorem}
\newenvironment{ntheorem}
{\innerntheorem}
{\endinnerptheorem}

\newcommand{\subiso}{{\sc{Subgraph Isomorphism}}\xspace}

\newcommand{\gridtiling}{{\sc{Grid Tiling}}\xspace}
\newcommand{\planararcsupply}{{\sc{Planar Arc Supply}}\xspace}
\newcommand{\expas}{{\sc{Exact Planar Arc Supply}}\xspace}
\newcommand{\tiling}{\tau}
\newcommand{\supply}{\sigma}
\newcommand{\interface}{\iota}
\renewcommand{\ast}{\star}

\newcommand{\hm}{\eta}

\newcommand{\genus}{\mathbf{genus}}
\newcommand{\minor}{\mathbf{hadw}}
\newcommand{\topmin}{\mathbf{hadw}_\textup{T}}
\newcommand{\tw}{\mathbf{tw}}
\newcommand{\cw}{\mathbf{cw}}
\newcommand{\pw}{\mathbf{pw}}

\newcommand{\fvs}{\mathbf{fvs}}
\newcommand{\maxdeg}{\Delta}
\newcommand{\ccn}{\mathbf{cc}}

\newcommand{\sil}[3]{
{\begin{array}{|l|l|l|}
\hline
\textup{\tiny multiplier}&
\textup{\tiny exponent}&
\textup{\tiny constraint}\\
#1 &  #2 & #3\\
\hline
\end{array}}
}

\newcommand{\SetRowColor}[1]{\noalign{\gdef\RowColorName{#1}}\rowcolor{\RowColorName}}
\newcommand{\mymulticolumn}[3]{\multicolumn{#1}{>{\columncolor{\RowColorName}}#2}{#3}}

\newcolumntype{L}[1]{>{\hsize=#1\hsize\raggedright\arraybackslash}X}%
\newcolumntype{R}[1]{>{\hsize=#1\hsize\raggedleft\arraybackslash}X}%
\newcolumntype{C}[1]{>{\hsize=#1\hsize\centering\arraybackslash}X}%

\definecolor{lightorange}{rgb}{1.0,0.9,0.5}
\definecolor{lightyellow}{rgb}{1.0,1.0,0.8}
\definecolor{negativeodd}{rgb}{1.0,0.8,0.8}
\definecolor{negativeeven}{rgb}{1.0,0.9,0.9}
\definecolor{positiveodd}{rgb}{0.8,0.8,1.0}
\definecolor{positiveeven}{rgb}{0.9,0.9,1.0}

\begin{document}

  \date{}

  \author{
  D\'{a}niel Marx
  \thanks{
	Computer and Automation Research Institute, Hungarian Academy of Sciences (MTA SZTAKI), 
	\texttt{dmarx@cs.bme.hu}, supported by the ERC grant ``PARAMTIGHT: Parameterized complexity and the search for tight complexity results'', no.~280152.
  }
  \and
  Micha\l{} Pilipczuk
  \thanks{
    Department of Informatics, University of Bergen, Norway, 
    \texttt{michal.pilipczuk@ii.uib.no}, supported by the ERC grant ``Rigorous Theory of Preprocessing'', no.~267959.
  }
  }

  \title{Everything you always wanted to know about the parameterized complexity of {\textsc{Subgraph Isomorphism}}\\[0.3cm] \large{(but were afraid to ask)}}

\begin{titlepage}
\def\thepage{}
\thispagestyle{empty}
\maketitle
\begin{abstract}
  Given two graphs $H$ and $G$, the \subiso problem asks if $H$ is
  isomorphic to a subgraph of $G$. While NP-hard in general,
  algorithms exist for various parameterized versions of the problem:
  for example, the problem can be solved (1) in time
  $2^{O(|V(H)|)}\cdot n^{O(\tw(H))}$ using the color-coding technique
  of Alon, Yuster, and Zwick \cite{DBLP:journals/jacm/AlonYZ95}; (2) in time
  $f(|V(H)|,\tw(G))\cdot n$ using Courcelle's Theorem; (3) in time
  $f(|V(H)|,\genus(G))\cdot n$ using a result on first-order model
  checking by Frick and Grohe \cite{DBLP:journals/jacm/FrickG01}; or (4) in time $f(\maxdeg(H))\cdot
  n^{O(\tw(G))}$ for connected $H$ using the algorithm of Matou\v{s}ek and
  Thomas \cite{DBLP:journals/dm/MatousekT92}. Already this small sample of results shows that the way an
  algorithm can depend on the parameters is highly nontrivial and
  subtle. However, the literature contains very few negative results ruling out
  that certain combination of parameters cannot be exploited
  algorithmically.  Our goal is to systematically investigate the
  possible parameterized algorithms that can exist for \subiso.

  We develop a framework involving 10 relevant parameters
  for each of $H$ and $G$ (such as treewidth, pathwidth, genus,
  maximum degree, number of vertices, number of components, etc.), and
  ask if an algorithm with running time
\[
f_1(p_1,p_2,\dots, p_\ell)\cdot n^{f_2(p_{\ell+1},\dots, p_k)}
\]
exists, where each of $p_1,\dots, p_k$ is one of the 10 parameters
depending only on $H$ or $G$. We show that {\em all} the questions arising
in this framework are answered by a set of 11 maximal positive results
(algorithms) and a set of 17 maximal negative results (hardness
proofs); some of these results already appear in the literature, while
others are new in this paper.

On the algorithmic side, our study reveals for example that an unexpected
combination of bounded degree, genus, and feedback vertex set number
of $G$ gives rise to a highly nontrivial algorithm for \subiso. On the
hardness side, we present W[1]-hardness proofs under extremely
restricted conditions, such as when $H$ is a bounded-degree tree of constant pathwidth and
$G$ is a planar graph of bounded pathwidth.

\end{abstract}
\end{titlepage}

\tableofcontents
\clearpage

\section{Introduction}

\subiso is one of the most fundamental graph-theoretic problems: given
two graphs $H$ and $G$, the question is whether $H$ is isomorphic to a
subgraph of $G$. It can be easily seen that finding a $k$-clique, a $k$-path, a
Hamiltonian cycle, a perfect matching, or a partition of the vertices
into triangles are all special cases of \subiso. Therefore, the
problem is clearly NP-complete in general. There are well-known
polynomial-time solvable special cases of the problem, for example,
the special case of trees:
\begin{theorem}[\cite{Matula197891}]\label{th:matula}
\subiso is polynomial-time solvable if $G$ and $H$ are trees.
\end{theorem}
Theorem~\ref{th:matula} suggest that one should try to look at special
cases of \subiso involving ``tree like'' graphs. The notion of
treewidth measures, in some sense, how close a graph is being a
tree \cite{BodlaenderK08}. Treewidth has very important combinatorial and algorithmic
applications; in particular, many algorithmic problems become easier on
bounded-treewidth graphs. However, \subiso is NP-hard even if both $H$
and $G$ have treewidth at most 2 \cite{DBLP:journals/dm/MatousekT92}.

Parameterized algorithms try to cope with NP-hardness by allowing
exponential dependence of the running time on certain well-defined
parameters of the input, but otherwise the running time depends only
polynomially on the input size. We say that a problem is {\em
  fixed-parameter tractable} with a parameter $k$ if it can be solved
in time $f(k)\cdot n^{O(1)}$ for some computable function $f$
depending only on $k$ \cite{downey-fellows:book,grohe:book}. The definition can be easily extended to
multiple parameters $k_1$, $k_2$, $\dots$, $k_\ell$. The NP-hardness
of \subiso on graphs of treewidth at most 2 shows that the problem is
not fixed-parameter tractable parameterized by treewidth (under
standard complexity assumptions). However, there are tractability
results that involve other parameters besides treewidth.  For example,
the following theorem, which can be easily proved using, e.g.,
Courcelle's Theorem \cite{DBLP:journals/iandc/Courcelle90}, shows the
fixed-parameter tractability of \subiso, jointly parameterized by the
size of $H$ and the treewidth of $G$:
\begin{theorem}[cf.~\cite{grohe:book}]\label{th:twusingcourcelle}
  \subiso can be solved in time $f(|V(H)|, \tw(G))\cdot n$ for some
  computable function $f$.
\end{theorem}
Some of the results in the literature can be stated as algorithms
where certain parameters do appear in the exponent of the running
time, but others influence only the multiplicative
factor. The classical color-coding algorithm of Alon, Yuster, and
Zwick \cite{DBLP:journals/jacm/AlonYZ95} is one such result:
\begin{theorem}[\cite{DBLP:journals/jacm/AlonYZ95}]\label{th:subiso-colorcoding}
\subiso can be solved in time $2^{O(|V(H)|)}\cdot n^{O(\tw(H))}$.
\end{theorem}
Intuitively, one can interpret Theorem~\ref{th:subiso-colorcoding} as
saying that if the treewidth of $H$ is bounded by any fixed constant,
then the problem becomes fixed-parameter tractable when parameterized by $|V(H)|$.  Notice
that treewidth appears in very different ways in Theorems
\ref{th:twusingcourcelle} and \ref{th:subiso-colorcoding}: in the
first result, the treewidth of $G$ appears in the multiplicative
factor, while in the second result, it is the treewidth of $H$ that is
relevant and it appears in the exponent. Yet another algorithm for
\subiso on bounded-treewidth graphs is due to Matou\v{s}ek and Thomas \cite{DBLP:journals/dm/MatousekT92}:
\begin{theorem}[\cite{DBLP:journals/dm/MatousekT92}]\label{th:matousek-thomas}
  For connected $H$, \subiso can be solved in time $f(\Delta(H))\cdot
  n^{O(\tw(G))}$ for some computable function $f$.
\end{theorem}
Again, the dependence on treewidth takes a different form here: this
time it is the treewidth of $G$ that appears in the exponent. Note
that the connectivity requirement cannot be omitted here:
there is an easy reduction from the NP-hard problem \textsc{Bin
  Packing} with unary sizes to the special case of \subiso where $H$
and $G$ both consist of a set of disjoint paths, that is, have maximum
degree 2 and treewidth 1. Therefore, as
Theorem~\ref{th:matousek-thomas} shows, the complexity of the problem
depends nontrivially on the number of connected components of the
graphs as well.

As the examples above show, even the apparently simple question of how
treewidth influences the complexity of \subiso does not have a
clear-cut answer: the treewidth of $H$ and $G$ influences the
complexity in different ways, they can appear in the running time
either as an exponent or as a multiplier, and the influence of
treewidth can be interpreted only in combination with other parameters
(such as the number of vertices or maximum degree of $H$). The
situation becomes even more complex if we consider further parameters
of the graphs as well. Cliquewidth, introduced by Courcelle and
Olariu~\cite{DBLP:journals/dam/CourcelleO00}, is a graph measure that
can be always bounded by a function of treewidth, but treewidth can be
arbitrary large even for graphs of bounded cliquewidth (e.g., for
cliques). Therefore, algorithms for graphs of bounded cliquewidth are
strictly more general than those for graphs of bounded treewidth. By
the results of Courcelle et
al.~\cite{DBLP:journals/mst/CourcelleMR00},
Theorem~\ref{th:twusingcourcelle} can be generalized by replacing
treewidth with cliquewidth. However, no such generalization is
possible for Theorem~\ref{th:subiso-colorcoding}: cliques have
cliquewidth 2, thus replacing treewidth with cliquewidth in
Theorem~\ref{th:subiso-colorcoding} would imply that \textsc{Clique}
(parameterized by the size of the clique to be found) is
fixed-parameter tractable, contrary to widely accepted complexity
assumptions. In the case of Theorem~\ref{th:matousek-thomas}, it is
not at all clear if treewidth can be replaced by cliquewidth: we are
not aware of any result in the literature on whether \subiso is
fixed-parameter tractable parameterized by the maximum degree of $H$
if $G$ is a connected graph whose cliquewidth is bounded by a fixed
constant.

Theorem~\ref{th:twusingcourcelle} can be generalized into a different
direction using the concept of {\em bounded local treewidth.} Model
checking with a fixed first-order formula is known to be linear-time
solvable on graphs of bounded local treewidth \cite{DBLP:journals/jacm/FrickG01}, which implies that
\subiso can be solved in time $f(|V(H)|)\cdot n$ if $G$ is planar, or
more generally, in time $f(|V(H)|,\genus(G))\cdot n$ for arbitrary
$G$. Having an algorithm for bounded-genus graphs, one can try to
further generalize the results to graphs excluding a fixed minor or to graphs
not containing the subdivision of a fixed graph (that is, to graphs
not containing a fixed graph as a topological minor). Such a
generalization is possible: a result of Dvo\v{r}ak et
al.~\cite{DBLP:conf/focs/DvorakKT10} states that model checking with a
fixed first-order formula is linear-time solvable on graphs of bounded
expansion, and it follows that \subiso can be solved in time
$f(|V(H)|,\minor(G))\cdot n$ or $f(|V(H)|,\topmin(G))\cdot n$, where
$\minor(G)$ (resp.,~$\topmin(G)$) is the maximum size of a clique that
is a minor (resp.,~topological minor) of $G$. These
generalizations of Theorem~\ref{th:twusingcourcelle} show that
planarity, and more generally, topological restrictions on $G$ can be
helpful in solving \subiso, and therefore the study of
parameterizations of \subiso should include these parameters as well.

Our goal is to perform a systematic study of the influence of the
parameters: for all possible combination of parameters in the exponent
and in the multiplicative factor of the running time, we would like to
determine if there is an algorithm whose running time is of this
form. The main thesis of the paper is the following:
\begin{itemize}
\item[(1)] as the influence of the parameters on the complexity is
  highly nontrivial and subtle, even small changes in the choice of
  parameters can have substantial and counterintuitive consequences, and

\item[(2)] the current literature gives very little guidance on
  whether an algorithm with a particular combination of parameters
  exist.
\end{itemize}

\smallskip

\noindent\textbf{Our framework.} We present a framework in which the questions raised above can be systematically
treated and completely answer every question arising in the
framework. Our setting is the following. First, we define the
following 10 graph parameters (we give a brief justification for each
parameter why it is relevant for the study of \subiso):
\begin{itemize}
\item {\em Number of vertices $|V(\cdot)|$.}  As Theorems~\ref{th:twusingcourcelle}
  and \ref{th:subiso-colorcoding} show, $|V(H)|$ is a highly relevant
  parameter for the problem. Note that, however, if $|V(G)|$ appears
  in the running time (either as a multiplier or in the exponent) or $|V(H)|$ appears in the exponent, then
  the problem becomes trivial.
\item {\em Number of connected components $\ccn(\cdot)$.} As
  Theorem~\ref{th:matousek-thomas} and the simple reduction from
  \textsc{Bin Packing} (see above) show, it can make a difference if
  we restrict the problem to connected graphs (or, more generally, if
  we allow the running time to depend on the number of components).
\item {\em Maximum degree $\maxdeg(\cdot)$.} The maximum degree of $H$ plays an important
  role in Theorem~\ref{th:matousek-thomas}, thus exploring the effect
  of this parameter is clearly motivated. In general, many
  parameterized problems become easier on bounded-degree graphs,
  mainly because then the distance-$d$ neighborhood of each vertex has bounded
  size for bounded $d$. 
\item
{\em  Treewidth $\tw(\cdot)$.} Theorems~\ref{th:twusingcourcelle}--\ref{th:matousek-thomas}
  give classical algorithms where treewidth appears in different ways;
  understanding how exactly treewidth can influence complexity is one of
  the most important concrete goals of the paper.
\item {\em Pathwidth $\pw(\cdot)$.} As pathwidth is always at least treewidth, but can be
  strictly larger, algorithms parameterized by pathwidth can exists
  even if no algorithms parameterized by treewidth are possible. Given
  the importance of treewidth and bounded-treewidth graphs, it is
  natural to explore the possibility of algorithms in the more
  restricted setting of bounded-pathwidth graphs.
\item {\em Feedback vertex set number $\fvs(\cdot)$.} A feedback vertex set is a set
  of vertices whose deletion makes the graph a forest; the feedback
  vertex set number is the size of the smallest such set. Similarly to
  graphs of bounded pathwidth, graphs of bounded feedback vertex set
  number can be thought of as a subclass of bounded-treewidth graphs,
  hence it is natural to explore what algorithms we can obtain with
  this parameterization. Note that \textsc{Graph Isomorphism} (not
  subgraph!) is fixed-parameter tractable parameterized by feedback
  vertex set number \cite{DBLP:conf/swat/KratschS10}, while only
  $n^{O(\tw(G))}$ time algorithms are known parameterized by treewidth
  \cite{DBLP:journals/jal/Bodlaender90,pon88}. This is evidence that
  feedback vertex set can be a useful parameter for problems involving
  isomorphisms.
\item {\em Cliquewidth $\cw(\cdot)$.} As cliquewidth is always bounded by a function of
  treewidth, parameterization by cliquewidth leads to more general
  algorithms than parameterization by treewidth. We have seen that
  treewidth can be replaced by cliquewidth in
  Theorem~\ref{th:twusingcourcelle}, but not in
  Theorem~\ref{th:subiso-colorcoding}. Therefore, understanding the
  role of cliquewidth is a nontrivial and quite interesting challenge.
\item {\em Genus $\genus(\cdot)$.} Understanding the complexity of \subiso on planar graphs
  (and more generally, on bounded-genus graphs) is a natural goal,
  especially in light of the positive results that arise from the
  generalizations of Theorem~\ref{th:twusingcourcelle}.
\item {\em Hadwiger number $\minor(\cdot)$.} That is, the size of the largest clique that is the minor of the graph.  A graph containing a $K_k$-minor
  needs to have genus $\Omega(k^2)$; therefore, algorithms for graphs
  excluding a fixed clique as a minor generalize algorithms for
  bounded-genus graphs. In many cases, such a generalization is
  possible, thanks to structure theorems and algorithmic advances for
  $H$-minor free graphs
  \cite{Decomposition_FOCS2005,DBLP:conf/stoc/DemaineHK11,grokawree13,kawwol11,gm16}.
\item {\em Topological Hadwiger number $\topmin(\cdot)$.}
That is, the size of the largest clique whose subdivision is a subgraph of the graph.
 A graph containing
  the subdivision of a $K_k$ contains $K_k$ as a minor. Therefore,
  algorithms for graphs excluding a fixed topological clique minor
  generalize algorithms for graphs excluding a fixed clique minor
  (which in turn generalize algorithms for bounded-genus graphs). Recent
  work show that some algorithmic results for graphs excluding a
  fixed minor can be generalized to excluded topological minors
  \cite{DBLP:conf/stacs/FominLST13,DBLP:conf/stoc/GroheKMW11,DBLP:conf/stoc/GroheM12,DBLP:journals/corr/abs-1210-0260,DBLP:journals/corr/abs-1201-2780}. In
  particular, the structure theorem of Grohe and
  Marx~\cite{DBLP:conf/stoc/GroheM12} states, in a precise technical
  sense, that graphs excluding a fixed topological minor are composed
  from parts that are either ``almost bounded-degree'' or exclude a
  fixed minor. Therefore, it is interesting to investigate in our
  setting how this parameter interacts with the parameters smallest
  excluded clique minor and maximum degree.
\end{itemize}

Given this list of 10 parameters, we would like to understand if an algorithm with running time of the form
\[
f_1(p_1,p_2,\dots, p_\ell)\cdot n^{f_2(p_{\ell+1},\dots, p_k)}
\]
exists, where each $p_i$ is one of these 10 parameters applied on
either $H$ and $G$, and $f_1,f_2$ are arbitrary computable functions
of these parameters. We call such a sequence of parameters a {\em
  description,} and we say that an algorithm is {\em compatible} with
the description if its running time is of this form. Observe that
Theorems~\ref{th:twusingcourcelle} and \ref{th:subiso-colorcoding} can
be stated as the existence of algorithms compatible with particular
descriptions. However, Theorem~\ref{th:matousek-thomas} has the extra
condition that $H$ and $G$ are connected (or in other words, the number of
connected components of both $H$ and $G$ is 1) and therefore it does not seem to
fit into this framework.  In order to include such statements into our
investigations, we extend the definition of descriptions with some
number of {\em constraints} that restrict the value of certain
parameters to particular constants. Specifically, we consider the
following 5 constraints on $H$ and $G$, each of which corresponds to
a particularly motivated special case of the problem:
\begin{itemize}
\item {\em Genus is 0.} That is, the graph is planar. Any positive
  result on planar graphs is clearly of interest, even if it does not
  generalize to arbitrary fixed genus. Conversely, whenever possible,
  we would like to state hardness results for planar graphs, rather
  than for bounded-genus with an unspecified bound on the genus.
\item {\em Number of components is 1.} Any positive result under this
  restriction is quite motivated, and as the examples above show, the
  problem can become simpler on connected graphs.
\item {\em Treewidth is at most 1.} That is, the graph is a
  forest. Trees can behave very differently than bounded-treewidth
  graphs (compare Theorem~\ref{th:matula} with the fact the the
  problem is NP-hard on graphs of treewidth 2), thus investigating the
  special case of forests might turn up additional algorithmic results.
\item {\em Maximum degree is at most 2.} That is, the graph consists
  of disjoint paths and cycles. Clearly, this is a very restricted
  class of graphs, but as the NP-hardness of \textsc{Hamiltonian
    Cycle} shows, this property of $H$ does not guarantee tractability
  without further assumptions.
\item {\em Maximum degree is at most 3.} In order to provide sharp
  contrast with the case of maximum degree at most 2, we would like to
  state negative results on bounded-degree graphs with degree bound 3
  (whenever possible). Adding this constraint into the framework
  allows us to express such statements.
\end{itemize}
We restrict our attention to these 5 specific constraints. For
example, we do not specifically investigate possible algorithms that work
on, say, graphs of feedback vertex set size 1 or of pathwidth 2: we
can argue that such algorithms are interesting only if they can be
generalized to every fixed bound on the feedback vertex set size or on
pathwidth (whereas an algorithm for planar graphs is interesting even
if it does not generalize to higher genera).

\smallskip

\noindent\textbf{Results.} Our formulation of the general framework includes an enormous number
of concrete research questions. Even without considering the 5
specific constraints, we have 19 parameters (10 for $H$ and 9 for $G$)
and each parameter can be either in the exponent of the running time,
in the multiplier of the running time, or does not appear at all in
the running time. Therefore, there are at least $3^{19}\approx
10^{9}$ descriptions and corresponding complexity questions
in this framework. The present paper answers all these questions (under standard complexity assumptions).

In order to reduce the number of questions that need to be answered,
we observe that there are some clear implications between the
questions. Clearly, the $f_1(|V(H)|)\cdot n^{f_2(\tw(H))}$ time
algorithm of Theorem~\ref{th:subiso-colorcoding} implies the existence
of, say, an $f_1(|V(H)|,\genus(G))\cdot n^{f_2(\pw(H),\maxdeg(G))}$
time algorithm: $\pw(H)$ is always at least $\tw(H)$ and the fact that
the latter running time can depend on $\genus(G)$ and $\maxdeg(G)$ can
be ignored.

The main claim of the paper is that every question arising in the
framework can be answered by a set of 11 positive results
and a set of 17 negative results:
\begin{quote}
The positive and negative results presented in
Table~\ref{fig:table} imply a positive or negative answer to every
question arising in this framework. \hfill (*)
\end{quote}
That is, either there is a
positive result for a more restrictive description, or a negative
result for a less restrictive restriction. The following two examples show how one can deduce the answer to specific questions from 
Table~\ref{fig:table}.
\afterpage{
\begin{landscape}
\begin{figure}
\begin{center}

\begin{tiny}

\renewcommand\arraystretch{1.8} \setlength\minrowclearance{2pt}
\makeatletter
\makeatother

\begin{tabular}{|c|c||c|c|c|c|c|c|c|c|c|c||c|c|c|c|c|c|c|c|c|}
\hline
\SetRowColor{lightorange} {\bf{Short Description}} & {\bf{Thm}} &\mymulticolumn{10}{c||}{{$H$}} &
\mymulticolumn{9}{c|}{$G$}\\ \hline
\rowcolor{lightyellow} \cellcolor{white} & \cellcolor{white}  &  $|V(\cdot)|$ & $\ccn$ & $\maxdeg$  & $\fvs$ & $\pw$ & $\tw$ & $\cw$ & $\genus$ & $\minor$ & $\topmin$ & $\ccn$ & $\maxdeg$  & $\fvs$ & $\pw$ & $\tw$ & $\cw$ & $\genus$ & $\minor$ & $\topmin$ \\ \hline
\rowcolor{positiveeven} \cellcolor{white} \multirow{2}{*}{} &  Thm \ref{thm:FOcw} (page \pageref{thm:FOcw}) & {\bf{M}} &  &  &  &  &  &  &  &  &  &  &  &  &  &  & {\bf{M}} &  &  &  \\ 
\rowcolor{positiveodd} \cellcolor{white} \multirow{-2}{*}{FO model checking} & Thm \ref{thm:FOtopmin} (page \pageref{thm:FOtopmin}) & {\bf{M}} &  &  &  &  &  &  &  &  &  &  &  &  &  &  &  &  &  & {\bf{M}} \\ \hline 
\rowcolor{positiveeven} \cellcolor{white} Color coding & Thm \ref{thm:colour-coding} (page \pageref{thm:colour-coding}) & {\bf{M}} &  &  &  &  & {\bf{E}} &  &  &  &  &  &  &  &  &  &  &  &  &  \\ \hline 
\rowcolor{positiveodd} \cellcolor{white} Matou\v{s}ek-Thomas & Thm \ref{thm:strong-MT} (page \pageref{thm:strong-MT}) &  & {\bf{M}} & {\bf{M}} &  &  &  &  &  &  &  &  &  &  &  & {\bf{E}} &  &  &  &  \\ \hline 
\rowcolor{positiveeven} \cellcolor{white} Paths\&Cycles $\to$ Paths\&Cycles & Thm \ref{thm:vcG} (page \pageref{thm:vcG}) &  &  &  &  &  &  &  &  &  &  & {\bf{E}} & {\bf{2}} &  &  &  &  &  &  &  \\ \hline 
\rowcolor{positiveodd} \cellcolor{white} \multirow{3}{*}{} & Thm \ref{thm:dpCyclesParTw} (page \pageref{thm:dpCyclesParTw}) &  & {\bf{E}} & {\bf{2}} &  &  &  &  &  &  &  &  &  &  &  & {\bf{M}} &  &  &  &  \\ 
\rowcolor{positiveeven} \cellcolor{white}  & Thm \ref{thm:cwdp} (page \pageref{thm:cwdp}) &  & {\bf{E}} & {\bf{2}} &  &  &  &  &  &  &  &  &  &  &  &  & {\bf{E}} &  &  &  \\ 
\rowcolor{positiveodd} \cellcolor{white} \multirow{-3}{*}{Dynamic Programming} & Thm \ref{thm:treesintotrees}$^{\ast}$ (page \pageref{thm:treesintotrees}) &  & {\bf{M}} &  &  &  &  &  &  &  &  &  &  &  &  & {\bf{1}} &  &  &  &  \\ \hline
\rowcolor{positiveeven}  \cellcolor{white} \multirow{3}{*}{} & Thm \ref{th:cyclefvs}$^{\ast}$ (page \pageref{th:cyclefvs}) &  & {\bf{M}} & {\bf{2}} &  &  &  &  &  &  &  &  & {\bf{M}} & {\bf{M}} &  &  &  &  &  &  \\ 
\rowcolor{positiveodd}  \cellcolor{white}  & Thm \ref{th:bigplanar}$^{\ast}$ (page \pageref{th:bigplanar}) &  & {\bf{E}} &  &  &  &  &  &  &  &  &  & {\bf{M}} & {\bf{M}} &  &  &  & {\bf{E}} &  &  \\  
\rowcolor{positiveeven}  \cellcolor{white} \multirow{-3}{*}{FVS and CSPs} & Thm \ref{th:bigplanarminor}$^{\ast}$ (page \pageref{th:bigplanarminor}) &  & {\bf{E}} & {\bf{E}} &  &  &  &  &  &  &  &  & {\bf{M}} & {\bf{M}} &  &  &  &  & {\bf{E}} &  \\ \hline 
\hline
\rowcolor{negativeeven} \cellcolor{white} \multirow{3}{*}{} & Thm \ref{thm:bin-packing} (page \pageref{thm:bin-packing}) &  &  &  &  &  &  &  &  &  &  & {\bf{M}} & {\bf{2}} &  &  & {\bf{1}} &  &  &  &  \\ 
\rowcolor{negativeodd} \cellcolor{white}  & Thm \ref{thm:bin-packing-univ} (page \pageref{thm:bin-packing-univ}) &  & {\bf{1}} &  &  &  & {\bf{1}} &  &  &  &  &  &  & {\bf{E}} & {\bf{E}} &  &  & {\bf{0}} &  &  \\  
\rowcolor{negativeeven} \cellcolor{white} \multirow{-3}{*}{Bin Packing} & Thm \ref{thm:bin-packing-path} (page \pageref{thm:bin-packing-path}) &  &  & {\bf{2}} &  &  &  &  &  &  &  & {\bf{1}} & {\bf{3}} &  & {\bf{E}} & {\bf{1}} &  &  &  &  \\ \hline 
\rowcolor{negativeodd} \cellcolor{white} Planar cubic HamPath & Thm \ref{thm:planar-cubic} (page \pageref{thm:planar-cubic}) &  & {\bf{1}} & {\bf{2}} &  &  & {\bf{1}} &  &  &  &  &  & {\bf{3}} &  &  &  &  & {\bf{0}} &  &  \\ \hline 
\rowcolor{negativeeven} \cellcolor{white} Clique & Thm \ref{thm:clique} (page \pageref{thm:clique}) & {\bf{M}} & {\bf{1}} &  &  &  &  & {\bf{E}} &  &  &  &  &  &  &  &  &  &  &  &  \\ \hline 
\rowcolor{negativeodd} \cellcolor{white} HamPath in bounded {\bf{cw}} &  Thm \ref{thm:hampath-cliquewidth} (page \pageref{thm:hampath-cliquewidth}) &  & {\bf{1}} & {\bf{2}} &  &  & {\bf{1}} &  &  &  &  &  &  &  &  &  & {\bf{M}} &  &  &  \\ \hline 
\rowcolor{negativeeven} \cellcolor{white} \multirow{4}{*}{} & Thm \ref{thm:grid-manycomp}$^{\ast}$ (page \pageref{thm:grid-manycomp}) &  & {\bf{M}} &  &  & {\bf{E}} & {\bf{1}} &  &  &  &  & {\bf{1}} & {\bf{3}} & {\bf{M}} & {\bf{M}} &  &  & {\bf{0}} &  &  \\ 
\rowcolor{negativeodd} \cellcolor{white} &  Thm \ref{thm:many-comp-minor}$^{\ast}$ (page \pageref{thm:many-comp-minor}) &  & {\bf{1}} &  &  & {\bf{E}} & {\bf{1}} &  &  &  &  &  & {\bf{M}} & {\bf{M}} & {\bf{M}} &  &  & {\bf{M}} & {\bf{E}} &  \\ 
\rowcolor{negativeeven} \cellcolor{white} & Thm \ref{thm:many-comp-genus}$^{\ast}$ (page \pageref{thm:many-comp-genus}) &  & {\bf{1}} &  &  & {\bf{E}} & {\bf{1}} &  &  &  &  &  & {\bf{3}} & {\bf{M}} & {\bf{M}} &  &  & {\bf{M}} &  &  \\ 
\rowcolor{negativeodd} \cellcolor{white} \multirow{-4}{*}{{\sc{Grid Tiling}}, 1-in-n gadgets} & Thm \ref{thm:grid-manycomp-conn-many-biclique}$^{\ast}$ (page \pageref{thm:grid-manycomp-conn-many-biclique}) &  & {\bf{1}} & {\bf{3}} &  & {\bf{E}} & {\bf{1}} &  &  &  &  &  & {\bf{M}} & {\bf{M}} & {\bf{M}} &  & {\bf{E}} & {\bf{M}} &  &  \\ \hline 
\rowcolor{negativeeven} \cellcolor{white} \multirow{2}{*}{} & Thm \ref{thm:grid-connfvs}$^{\ast}$ (page \pageref{thm:grid-connfvs}) &  & {\bf{1}} & {\bf{3}} &  & {\bf{E}} & {\bf{1}} &  &  &  &  &  &  & {\bf{M}} & {\bf{M}} &  &  & {\bf{0}} &  &  \\ 
\rowcolor{negativeodd} \cellcolor{white} \multirow{-2}{*}{{\sc{Grid Tiling}}, moustache gadgets} & Thm \ref{thm:grid-conndegree}$^{\ast}$ (page \pageref{thm:grid-conndegree}) &  & {\bf{1}} &  &  & {\bf{E}} & {\bf{1}} &  &  &  &  &  & {\bf{3}} &  & {\bf{M}} &  &  & {\bf{0}} &  &  \\ \hline 
\rowcolor{negativeeven} \cellcolor{white} Small planar graph & Thm \ref{thm:multicolored-grid}$^{\ast}$ (page \pageref{thm:multicolored-grid}) & {\bf{M}} & {\bf{1}} & {\bf{3}} &  &  &  &  & {\bf{0}} &  &  &  &  &  &  &  &  &  &  &  \\ \hline 
\rowcolor{negativeodd} \cellcolor{white} \multirow{4}{*}{} & Thm \ref{thm:expas-fvs}$^{\ast}$ (page \pageref{thm:expas-fvs}) &  & {\bf{M}} & {\bf{2}} &  &  & {\bf{1}} &  &  &  &  & {\bf{1}} &  & {\bf{M}} & {\bf{M}} &  &  & {\bf{0}} &  &  \\ 
\rowcolor{negativeeven} \cellcolor{white} &  Thm \ref{thm:expas-degree}$^{\ast}$ (page \pageref{thm:expas-degree}) &  & {\bf{M}} & {\bf{2}} &  &  & {\bf{1}} &  &  &  &  & {\bf{1}} & {\bf{3}} &  & {\bf{M}} &  &  & {\bf{0}} &  &  \\ 
\rowcolor{negativeodd} \cellcolor{white} & Thm \ref{thm:expas-fvs-biclique}$^{\ast}$ (page \pageref{thm:expas-fvs-biclique}) &  & {\bf{M}} & {\bf{2}} &  &  & {\bf{1}} &  &  &  &  & {\bf{1}} &  & {\bf{M}} & {\bf{M}} &  & {\bf{E}} & {\bf{M}} &  &  \\
\rowcolor{negativeeven} \cellcolor{white} \multirow{-4}{*}{{\sc{Exact Planar Arc Supply}}} & Thm \ref{thm:expas-degree-biclique}$^{\ast}$ (page \pageref{thm:expas-degree-biclique}) &  & {\bf{M}} & {\bf{2}} &  &  & {\bf{1}} &  &  &  &  & {\bf{1}} & {\bf{M}} &  & {\bf{M}} &  & {\bf{E}} & {\bf{M}} &  &  \\ \hline 
\end{tabular}

\end{tiny}
\caption{Positive and negative results in the paper. Results marked with $^{\ast}$ are new findings that were not known before.}\label{fig:table}
\end{center}
\end{figure}
\end{landscape}
}

\begin{example}
  Is there an algorithm for \subiso with running time $n^{f(\fvs(G))}$
  when $G$ is a planar graph of maximum degree 3 and $H$ is connected?
  Looking at Table~\ref{fig:table}, the line of
  Theorem~\ref{th:bigplanar} shows the existence of an algorithm with
  running time $f_1(\fvs(G),\maxdeg(G))\cdot
  n^{f_2(\genus(G),\ccn(H))}$. When restricted to the case when $G$ is
  a planar graph (i.e., $\genus(G)=0$) with $\maxdeg(G)\le 3$ and $H$
  is connected (i.e., $\ccn(H)=1$), then running time of this
  algorithm can be expressed as $f(\fvs(G))\cdot n^{O(1)}$. This is in
  fact better than the running time $n^{f(\fvs(G))}$ we asked for,
  hence the answer is positive.
\end{example}
\begin{example}
  Is there an algorithm for \subiso with running time $f(\tw(G))\cdot
  n^{g(\maxdeg(G))}$ when $G$ is a connected planar graph?  Looking at
  Table~\ref{fig:table}, the line of Theorem~\ref{thm:grid-manycomp}
  gives a negative result for algorithms with running time
  $f_1(\ccn(H),\pw(G),\fvs(G))\cdot n^{f_2(\pw(H))}$ when restricted
  to instances where $H$ is a forest and $G$ is a connected planar
  graph of maximum degree 3. Note that $\tw(G)\le \pw(G)$, an
  $f(\tw(G))\cdot n^{g(\maxdeg(G))}$ time algorithm for connected planar
  graphs would give an $f(\pw(G))\cdot n^{O(1)}$ time algorithm for
  connected planar graphs of maximum degree 3, which is a better
  running time then the one ruled out by
  Theorem~\ref{thm:grid-manycomp}. Therefore, the answer is negative.
\end{example}

To make the claim (*) more formal and verifiable, we define an
ordering relation between descriptions in a way that guarantees that
if description $D_1$ is {\em stronger} than $D_2$, then an algorithm
compatible with $D_1$ implies the existence of an algorithm compatible
with $D_2$. Roughly speaking, the definition of this ordering relation
takes into account three immediate implications:
\begin{itemize}
\item Removing a parameter makes the description stronger.
\item Moving a parameter from the exponent to the multiplier of the running time makes the algorithm stronger.
\item We consider a list of combinatorial relations between the
  parameters and their implications on the descriptions: for example,
  $\tw(H)\le \pw(H)$ implies that replacing $\pw(H)$ with $\tw(H)$
  makes the description stronger. Our list of relations include some
  more complicated and less obvious connections, such as $\tw(H)$ can
  be bounded by a function of $\cw(H)$ and $\maxdeg(H)$, thus
  replacing $\cw(H)$ and $\maxdeg(H)$ with $\tw(H)$ makes the
  description stronger.
\end{itemize}
The exact definition of the ordering of the descriptions appears in
Section~\ref{sec:stronger}.  Given this ordering of the descriptions,
we need to present the positive results only for the maximally strong
descriptions and the negative results for the minimally strong
descriptions. Our main result is that every question arising in the
framework can be explained by a set of 11 maximally strong
positive results and a set of 17 minimally strong negative
results listed in Table~\ref{fig:table}.

\begin{theorem}\label{th:table-covers}
For every description $D$, either
\begin{itemize}
\item Table~\ref{fig:table} contains a positive result for a description $D'$ such that $D'$ is stronger than $D$, or
\item Table~\ref{fig:table} contains a negative result for a description $D'$ such that $D$ is stronger than $D'$.
\end{itemize}
\end{theorem}

At this point, the reader might wonder how it is possible to prove Theorem~\ref{th:table-covers}, that is to verify
that the positive and negative results on Table~\ref{fig:table} indeed
cover every possible description. Interestingly, formulating the task
of checking whether a set of positive and negative results on an
unbounded set of parameters explains every possible description leads
to an NP-hard problem (we omit the details), even if we simplify the
problem by ignoring the constraints, the combinatorial relations
between the parameters, and the fact that the parameters can appear
either in the exponent or in the multiplier. Therefore, we have
implemented a simple backtracking algorithm that checks if every
description is explained by the set of positive and negative results
given in the input. We did not make a particular effort to optimize
the program, as it was sufficiently fast for our purposes on
contemporary desktop computers. The program indeed verifies that our
set of positive and negative results is complete. We have used this
program extensively during our research to find descriptions that are
not yet explained by our current set of results. By focusing on one
concrete unexplained description, we could always either find a
corresponding algorithm or prove a hardness result, which we could add
to our set of results. By iterating this process, we have eventually
arrived at a set of results that is complete. The program and the data
files are available as electronic supplementary material of this
paper.

\smallskip

\noindent\textbf{Algorithms.}  Let us highlight some of the new algorithmic
results discovered by the exhaustive analysis of our framework. While
the negative results suggest that the treewidth of $G$ appearing in
the multiplicative factor of the running time helps very little if the
size of $H$ can be large, we show that the more relaxed parameter
feedback vertex set is useful on bounded-degree planar
graphs. Specifically, we prove the following result:
\begin{theorem}\label{th:planarfvs-intro}
\subiso can be solved in time $f(\maxdeg(G),\fvs(G))\cdot n^{O(1)}$ if $H$ is connected and $G$ is planar.
\end{theorem}
The proof of Theorem~\ref{th:planarfvs-intro} turns the \subiso
problem into a Constraint Satisfaction Problem (CSP) whose primal
graph is planar. We observe that this CSP instance has a special
variable $v$ that we call a {\em projection sink:} roughly speaking,
it has the property that $v$ can be reached from every other variable
via a sequence of constraints that are projections. We prove the
somewhat unexpected result that a planar CSP instance having a
projection sink is polynomial-time solvable, which allows us to solve the
\subiso instance within the claimed time bound. This new property of
having a projection sink and the corresponding polynomial-time
algorithm for CSPs with this property can be interesting on its own
and possibly useful in other contexts.

We generalize the result from planar graphs to bounded-genus graphs
and to graphs excluding a fixed minor in the following way:
\begin{theorem}\label{th:genusfvs-intro}
\subiso can be solved in time
\begin{itemize}
\item[(1)] $f_1(\maxdeg(G),\fvs(G))\cdot n^{f_2(\genus(G),\ccn(H))}$, and
\item[(2)] $f_1(\maxdeg(G),\fvs(G))\cdot n^{f_2(\minor(G),\maxdeg(H),\ccn(H))}$.
\end{itemize}
\end{theorem}
For (1), we need only well-known diameter-treewidth relations for
bounded-genus graphs \cite{eppstein}, but (2) needs a nontrivial
application of structure theorems for graphs excluding a fixed minor
and handling vortices in almost-embeddable graphs.  Note that these
two results are incomparable: in (2), the exponent contains
$\maxdeg(H)$ as well, thus it does not generalize (1). Intuitively, 
the reason for this is that when lifting the algorithm from the bounded-genus 
case to the minor-free case, high-degree apices turn out to be problematic. 
On the other hand, Theorem~\ref{thm:many-comp-minor} shows that incorporating 
other parameters is (probably) unavoidable when moving to the more general
minor-free setting. We find it interesting that our study revealed that 
the bounded-genus case and the minor-free case are provably different when 
the parameterized complexity of \subiso is concerned.

The reader might find it unmotivated to present algorithms that depend on so many
parameters in strange ways, but let us emphasize that these results
are maximally strong results in our framework. That is, no weakening
of the description can lead to an algorithm (under standard complexity
assumptions): for example, $\genus(H)$ or $\ccn(H)$ cannot be moved
from the exponent to the multiplier, or $\maxdeg(H)$ cannot be omitted
from the exponent in (2). Therefore, these result show, in a
well-defined sense, the limits of what can be achieved. Finding such
maximal results is precisely the goal of developing and analyzing our
framework: it seems unlikely that one would come up with results of
the form of Theorem~\ref{th:genusfvs-intro} without an exhaustive
investigation of all the possible combinations of parameters.

On the other hand, we generalize Theorem~\ref{th:matula} from trees to forests,
parameterized by the number of connected components of $H$.  This
seemingly easy task turns out to be surprisingly challenging. The
dynamic programming algorithm of Theorem~\ref{th:matula} relies on a
step that involves computing maximum matching in a bipartite
graph. The complications arising from the existence of multiple
components of $H$ makes this matching step more constrained and
significantly harder. In fact, the only way we can solve these
matching problems is by the randomized algebraic matching algorithm of
Mulmuley et al.~\cite{MulmuleyVV87}. Therefore, our result is a randomized algorithm for this problem:
\begin{theorem}\label{th:forestinforestmain}
  \subiso can be solved in randomized time $f(\ccn(H))\cdot n^{O(1)}$ with false negatives, if $H$ and
  $G$ are forests.
\end{theorem}
Again, we find it a success of our framework that it directed
attention to this particularly interesting special case of the
problem.

\smallskip

\noindent\textbf{Hardness proofs.} Two different technologies are needed for
proving negative results about algorithms satisfying certain
descriptions: NP-hardness and W[1]-hardness. Recall that a W[1]-hard
problem is unlikely to be fixed-parameter tractable and one can show that
a problem is W[1]-hard by presenting a parameterized reduction from a
known W[1]-hard problem (such as \textsc{Clique}) to it. The most
important property of a parameterized reduction is that the parameter
value of the constructed instance can be bounded by a function of the
parameter of the source instance; see
\cite{downey-fellows:book,grohe:book} for more details.
\begin{itemize}
\item To give evidence that no $n^{f(p_1,\dots,p_k)}$ time algorithm
    for \subiso exists, one would like to show that \subiso remains NP-hard on
    instances where the value of the parameters $p_1$, $\dots$, $p_k$
    are bounded by some universal constant. Clearly, an algorithm with
    such running time and the NP-hardness result together would imply
    $P=NP$.
  \item To give evidence that no $f_1(p_1,p_2,\dots, p_\ell)\cdot
    n^{f_2(p_{\ell+1},\dots, p_k)}$ time algorithm for \subiso exists,
    one would like to show that \subiso is W[1]-hard parameterized by $p_1$,
    $\dots$, $p_\ell$ on instances where the values of $p_{\ell+1}$,
    $\dots$, $p_k$ are bounded by some universal constant. That is,
    what is needed is a parameterized reduction from a known W[1]-hard
    problem to \subiso in such a way that parameters $p_1$, $\dots$,
    $p_\ell$ of the constructed instance are bounded by a function of
    the parameters of the source instance, while the values of $p_{\ell+1}$, $\dots$, $p_k$
    are bounded by some universal constant. Then an algorithm for
    \subiso with the stated running time and this reduction together
    would imply that a W[1]-hard problem is fixed-parameter tractable.
\end{itemize}
Additionally, the reductions need to take into account the extra constraints
(planarity, treewidth 1, etc.) appearing in the description. The
nontrivial results of this paper are of the second type:
we prove the W[1]-hardness of \subiso with certain parameters, under
the assumption that certain other parameters are bounded by a
universal constant. Intuitively, a substantial difference between
NP-hardness proofs and W[1]-hardness proofs is that in a typical
NP-hardness proof from, say, \textsc{3-SAT}, one replaces each
variable and clause with a small gadget having a {\em constant} number
of states, whereas in a typical W[1]-hardness proof from, say,
\textsc{Clique}, one creates a bounded number of large gadgets having
an {\em unbounded} number of states, e.g., the states correspond to
the vertices of the original graph. Therefore, usually the first goal in
W[1]-hardness proofs is to construct gadgets that are able to express
a large number of states.

Most of our W[1]-hardness results are for planar graphs or for graphs
close to planar graphs. As many parameterized problems become
fixed-parameter tractable on planar graphs, there are only a handful
of planar W[1]-hardness proofs in the literature
\cite{DBLP:conf/iwpec/BodlaenderLP09,DBLP:journals/mst/CaiFJR07,DBLP:conf/iwpec/EncisoFGKRS09,DBLP:conf/icalp/Marx12}.
These hardness proofs need to construct gadgets that are planar and
are able to express a large number of states, which can be a
challenging task. A canonical planar problems that
can serve as useful starting points for W[1]-hardness proofs for
planar graphs is \gridtiling
\cite{marx-focs2007-ptas,DBLP:conf/icalp/Marx12}, see Section~\ref{sec:hardness-results}  for
details. In our W[1]-hardness proofs, we use mostly \gridtiling as the source problem of our reductions. In Section~\ref{sec:embedding-paths-into}, we prove the hardness of a new planar problem, \expas,  which is very similar to the problem \planararcsupply defined and used by Bodlaender et al.~\cite{DBLP:conf/iwpec/BodlaenderLP09}.

Besides planarity (or near-planarity), our hardness proofs need to
overcome other challenges as well: we bound combinations of maximum degree (of
$H$ or $G$), pathwidth, cliquewidth etc. The following theorem
demonstrates the type of restricted results we are able to get. Note
that the more parameter appears in the running time and the more
restrictions $H$ and $G$ have, the stronger the hardness result is.
\begin{theorem}
Assuming $FPT\neq W[1]$, there is no algorithm for \subiso with running time
\begin{itemize}
\item $f_1(\pw(G))\cdot n^{f_2(\pw(H))}$, even if both $H$ and $G$ are connected planar graphs of maximum degree 3 and $H$ is a tree, or
\item $f_1(\maxdeg(G),\pw(G),\fvs(G),\genus(G))\cdot n^{f_2(\pw(H),\cw(G))}$, even if both $H$ and $G$ are connected and $H$ is a tree of maximum degree 3.
\end{itemize}
\end{theorem}

\smallskip

\noindent \textbf{Organization.}  In Section~\ref{sec:preliminaries}, we discuss
preliminary notions, including the formal definition of
descriptions. Section~\ref{sec:positive-results} presents many of our
positive (i.e., algorithmic) results, but the results on the interplay
on feedback vertex set number, degree, and planarity
(Theorems~\ref{th:planarfvs-intro} and \ref{th:genusfvs-intro}) are
presented independently in Section~\ref{sec:big-planar-algorithm}. The
negative results (i.e., hardness proofs) presented in
Section~\ref{sec:easy-class-negat} are either known or follow in an
easy way from known results. We develop our hardness proofs in
Section~\ref{sec:hardness-results}. Section~\ref{sec:conclusions}
concludes the paper with a general discussion of our framework and
methodology.


\section{Preliminaries}
\label{sec:preliminaries}The purpose of this section is to formally define the framework of our
investigations: we define the problem we are studying, introduce the
notation for the descriptions used in the statement of the positive
and negative results, and define the notion of ordering for the
descriptions that is needed to argue about maximal positive/negative
results.
\subsection{Definitions of problems}

All the graphs considered in this paper are simple, i.e. not containing loops or parallel edges, unless explicitely stated. Given two simple graphs $H,G$, a {\emph{homomorphism}} from $H$ into $G$ is a mapping $\hm: V(H)\to V(G)$ such that whenever $uv\in E(H)$, then we have also that $\hm(u)\hm(v)\in E(G)$. In the \subiso problem we ask whether there is an injective homomorphism from $H$ into $G$. In other words, we ask whether there is a subgraph $H'$ of $G$, such that $H$ and $H'$ are isomorphic with isomorphism $\hm$.

We often talk also about {\emph{partial subgraph isomorphisms}}. A partial subgraph isomorphism $\hm$ from $H$ to $G$ is simply a subgraph isomorphism of some induced subgraph $H'$ of $H$ into $G$. Thus, all the vertices of $V(H')$ have {\emph{defined}} images in $\hm$, while the images of vertices of $V(H)\setminus V(H')$ are {\em{undefined}}. As the definition of partial subgraph isomorphism is very relaxed, we also define the notion of {\em{boundary}}. Given a graph $H$ and a graph $G$ with a distinguished subset of vertices $B\subseteq V(G)$, called the {\emph{boundary}} or the {\emph{interface}}, we say that a partial subgraph isomorphism $\hm$ from $H$ to $G$ {\emph{respects boundary $B$}} if the following property holds: for every vertex $v\in V(H)$ with $\hm(v)$ defined and not contained in $B$, all the neighbours of $v$ in $H$ have defined images in $\hm$. Intuitively, the notion of boundary will be helpful in the following situation. Assume that for the sake of proving a negative result we construct a gadget $H_0$ in $H$ and its counterpart $G_0$ in $G$, such that the gadget $H_0$ does not fit completely into $G_0$. In order to talk about a partial mappings of $H_0$ into $G_0$, we use the language of partial subgraph isomorphisms from $H_0$ to $G_0$ respecting the border defined as all the vertices of $G_0$ that are incident to some edges outside $G_0$.

\subsection{Notation for problems}

In this paper, we study the complexity of the \subiso problem, when imposing different constraints on the input, and measuring the complexity using different measures. Hence, we are given graphs $H$ and $G$, and we ask for existence of an algorithm with running time $f_1(\overline{\mathbf{x}})\cdot n^{f_2(\overline{\mathbf{y}})}$ for some computable functions $f_1,f_2$ that works assuming some structural properties $\overline{\mathbf{z}}$ of $G$ and $H$. Here, $n$ is the total size of the input (note that we may assume $n\geq 2$), $\overline{\mathbf{x}}$ and $\overline{\mathbf{y}}$ are vectors of structural parameters of $G$ and $H$, while $\overline{\mathbf{z}}$ is a vector of inequalities bounding structural parameters of $G$ and $H$ by some fixed values. Values in $\overline{\mathbf{x}}$ will be called {\emph{parameters of the multiplier}}, values in $\overline{\mathbf{y}}$ will be called {\emph{parameters of the exponent}}, while inequalities in $\overline{\mathbf{z}}$ will be called {\emph{constraints}} or {\emph{constraints on the input}}.

The \subiso problem considered with parameters $\overline{\mathbf{x}}$ in the multiplier, parameters $\overline{\mathbf{y}}$ in the exponent, and constraints $\overline{\mathbf{z}}$ will be described as 
\begin{eqnarray*}
\sil{\overline{\mathbf{x}}}{\overline{\mathbf{y}}}{\overline{\mathbf{z}}}
\end{eqnarray*}
We say that an algorithm is {\emph{compatible}} with this description if it runs in $f_1(\overline{\mathbf{x}})\cdot n^{f_2(\overline{\mathbf{y}})}$ time on instances constrained as in $\overline{\mathbf{z}}$, for some computable functions $f_1,f_2$. In this notation, the classical algorithm of Matou\v{s}ek and Thomas (Theorem~\ref{th:matousek-thomas}) is compatible with description 
\begin{eqnarray*}
\sil{\maxdeg(H)}{\tw(G)}{\ccn(H)\leq 1}
\end{eqnarray*}

Whenever we use pathwidth, treewidth, or cliquewidth of a graph, for simplicity we assume that a corresponding decomposition is given. As for each of these parameters an $f(OPT)$-approximation can be found in FPT time~\cite{Bodlaender96,BodlaenderK96,Oum09}, this assumption does not change the complexity in any of our lower or upper bounds. Similarly, whenever we speak of the size of a minimum feedback vertex set, we assume that such a feedback vertex set is explicitely given since it can be computed in FPT time~\cite{CaoCL10}. Moreover, when genus is concerned, we assume that an embedding into an appropriate surface is given.

\subsection{Ordering of descriptions}\label{sec:stronger}
The following technical lemma will be convenient when arguing about
algorithm with different running times. Essentially, it states that
the multivariate functions appearing in the in the running can be
assumed to be multiplicative. 
\begin{lemma}\label{lem:multiplicative}
Let $f:\mathbb{N}^r\to \mathbb{N}$ be a computable function. Then there is a non-decreasing function $\bar f:\mathbb{N}\to \mathbb{N}$ such 
that $\bar f(x)\geq 2$ for every $x\in \mathbb{N}$, and $\bar f(x_1)\cdot \ldots \cdot \bar f(x_r)\ge f(x_1,\dots,x_r)$ for every $x_1,\dots,x_r\in \mathbb{N}$.
\end{lemma}
\begin{proof}
Let 
\[
\bar f(x)= 2+\max_{0\le x_1,\dots,x_r\le x}f(x_1,\dots,x_r).
\]
Consider an $r$-tuple $(x_1,\dots,x_r)$ of nonnegative integers, and
suppose that $x_i$ has maximal value among $x_1,\dots, x_r$. Then $\bar f(x_i)\ge
f(x_1,\dots,x_r)$ by definition. As $\bar f(x_j)\ge 2$ for every $1\le j
\le r$, we have
\[
\bar f(x_1)\cdot \ldots \cdot \bar f(x_r) \ge \bar f(x_i) \ge f(x_1,\dots,x_r).
\]
\end{proof}
First we show that certain modifications on a description has the effect that an algorithm compatible with the modified description satisfies the original definition as well.
\begin{lemma}\label{lem:stronger}
Let $D_1$ and $D_2$ be two descriptions such that $D_2$ can be obtained from $D_1$ by either
\begin{enumerate}[label=(\alph*),noitemsep]
\item \label{item-d:remove} removing a parameter,
\item \label{item-d:exptomult} moving a parameter from the exponent to the multiplier,
\item \label{item-d:removeconstr} removing a constraint,
\item \label{item-d:addconstrained} adding a parameter to the exponent or to the multiplier whose value is bounded by a constraint already present, or
\item \label{item-d:addbounded} adding a parameter to the multiplier  (resp., exponent) whose value can be bounded by a computable function of the parameters already in the 
multiplier  (resp., exponent) on instances where all the constraints in the description hold.
\end{enumerate}
Then an algorithm compatible with description $D_2$ implies the existence of an algorithm compatible with description $D_1$.
\end{lemma}
\begin{proof}
  Suppose that an algorithm $\mathbb{A}$ compatible with description $D_2$ exists;
  by Lemma~\ref{lem:multiplicative}, we may assume that its running
  time is of the form 
\[
f_1(p_1,p_2,\dots, p_\ell)\cdot n^{f_2(p_{\ell+1},\dots, p_k)}
=\bar f_1(p_1) \cdot \ldots \cdot \bar f_1(p_\ell)
  \cdot n^{\bar f_2(p_{\ell+1})\cdot \ldots \cdot \bar f_2(p_k)}
\]
for  some parameter $p_1$, $\dots$, $p_k$.

To prove \ref{item-d:remove}, suppose that $D_1$ includes parameter $p'$ in the multiplier. Then algorithm $\mathbb{A}$ is clearly compatible with description $D_1$: we can bound its running time using the function $f'_1(p_1,\dots,p_\ell,p')=f_1(p_1,\dots,p_\ell)$. The argument is the same for removing a parameter from the exponent.

To prove \ref{item-d:exptomult}, suppose that $p_{\ell+1}$ appears in the multiplier in $D_1$. To show that $\mathbb{A}$ is compatible with description $D_1$, observe that its running time can be bound as
\begin{multline*}
\bar f_1(p_1) \cdot \ldots \cdot \bar f_1(p_\ell)
  \cdot n^{\bar f_2(p_{\ell+1})\cdot \ldots \cdot \bar f_2(p_k)}
\ge 
\bar f_1(p_1) \cdot \ldots \cdot \bar f_1(p_\ell)
  \cdot n^{\bar f_2(p_{\ell+1})+\bar f_2(p_{\ell+2})\cdot \ldots \cdot \bar f_2(p_k)}\\
=
\bar f_1(p_1) \cdot \ldots \cdot \bar f_1(p_\ell) \cdot n^{\bar f_2(p_{\ell+1})}
  \cdot n^{\bar f_2(p_{\ell+2})\cdot \ldots \cdot \bar f_2(p_k)}\ge
\bar f_1(p_1) \cdot \ldots \cdot \bar f_1(p_\ell) \cdot \bar f_2(p_{\ell+1})
  \cdot n^{\bar f_2(p_{\ell+2})\cdot \ldots \cdot \bar f_2(p_k)}\\
=f'_1(p_1,\dots,p_{\ell},p_{\ell+1})\cdot n^{f'_2(p_{\ell+2},\dots,p_k)}
\end{multline*}
for some computable functions $f'_1$ and $f'_2$. Note here that the first inequality holds due to the fact that $\bar f_2(x)\geq 2$ for every $x\in \mathbb{N}$, while the second holds due to the assumption the $n\geq 2$.

To prove \ref{item-d:removeconstr}, observe that the running time of $\mathbb{A}$ is compatible with $D_1$ (as it is the same as in $D_2$) and $\mathbb{A}$ works on  larger set of instances than those required by the specification of $D_1$.

To prove \ref{item-d:addconstrained}, suppose that the a constraint bounds
the value of $p_\ell$ by some constant $c$.  Then $\bar f_1(p_1) \cdot
\ldots \cdot \bar f_1(p_\ell) \le \bar f_1(p_1) \cdot \ldots \cdot \bar
f_1(p_{\ell-1})\cdot \bar f_1(c)=f'_1(p_1,\dots, p_{\ell-1})$ for some
computable function $f'_1(p_1,\dots,p_{\ell-1})$, hence $\mathbb{A}$
is compatible with description $D_1$. The argument is the same if the bounded parameter appears in the exponent.

To prove \ref{item-d:addbounded}, suppose that $p_{\ell}\le g(p_1,\dots,p_{\ell-1})$ for some computable function $g$. 
Then $\bar f_1(p_1) \cdot
\ldots \cdot \bar f_1(p_\ell) \le \bar f_1(p_1) \cdot \ldots \cdot \bar
f_1(p_{\ell-1})\cdot g(p_1,\dots,p_{\ell-1})=f'_1(p_1,\dots, p_{\ell-1})$ for some
computable function $f'_1(p_1,\dots,p_{\ell-1})$. The argument is the same if the bounded parameter appears in the exponent.
\end{proof}
The following definition formalizes when we call a description $D_2$
stronger than some other description $D_1$: if $D_2$ can be obtained
from $D_1$ by a sequence of allowed operations listed in the
definition.
\begin{definition}\label{def:stronger}
Let $D_1$ and $D_2$ be two descriptions. We say that $D_2$ is {\em stronger} $D_1$ if $D_2$ can be obtained from $D_1$ by a sequence of the following operations:
\begin{enumerate}[label=(\alph*),noitemsep]
\item \label{item-d2:remove} removing a parameter,
\item \label{item-d2:exptomult} moving a parameter from the exponent to the multiplier,
\item \label{item-d2:removeconstr} removing a constraint,
\item \label{item-d2:addconstrained} adding a parameter to the exponent or to the multiplier whose value is bounded by a constraint already present, or
\item \label{item-d2:addbounded} adding a parameter to the multiplier  (resp., exponent) following these rules:
\begin{enumerate}[label=(\roman*),noitemsep]
\item \label{item-d2:GtoH} for any parameter except $\cw(\cdot)$ and $\ccn(\cdot)$, we may add the parameter for $H$ to the multiplier  (resp., exponent), if the same parameter for $G$ already appears in the multiplier  (resp., exponent),
\item \label{item-d2:GtoHconst} for any constraint except connectivity ($\ccn(\cdot)\leq 1$), we may add the constraint to $H$, if the same constraint for $G$ is already present,
\item \label{item-d2:HtoGconstcc} if constraint $\ccn(H)\leq 1$ is present, we may add the constraint $\ccn(G)\leq 1$ as well,
\item \label{item-d2:maxdeg2} if we have the constraint $\maxdeg(\cdot)\le 2$, then we may add any of the parameters $\pw(\cdot)$, $\tw(\cdot)$, $\cw(\cdot)$, $\genus(\cdot)$, $\minor(\cdot)$, $\topmin(\cdot)$ to either the exponent or the multiplier, or add any of the constraints $\maxdeg(\cdot)\le 3$, $\genus(\cdot)\le 0$,
\item \label{item-d2:tw1} if we have the constraint $\tw(\cdot) \le 1$, then we may add any of the parameters $\fvs(\cdot)$, $\tw(\cdot)$, $\cw(\cdot)$, $\genus(\cdot)$, $\minor(\cdot)$, $\topmin(\cdot)$ to either the exponent or the multiplier, or add the constraint $\genus(\cdot)\le 0$,
\item \label{item-d2:implicit} we may add a new parameter using the following rules, where $X \to Y$ means that any of the parameters $Y$ can be added if all the parameters in $X$ are already present:
\begin{align*}
|V(\cdot)| &\to  \maxdeg(\cdot),\ccn(\cdot),\fvs(\cdot),\pw(\cdot), \tw(\cdot), \cw(\cdot), \genus(\cdot), \minor(\cdot), \topmin(\cdot)\\
\pw(\cdot) &\to \tw(\cdot)\\
\fvs(\cdot) &\to \tw(\cdot)\\
\tw(\cdot)&\to \cw(\cdot),\minor(\cdot)\\
\genus(\cdot)&\to \minor(\cdot)\\
\minor(\cdot) &\to \topmin(\cdot)\\
\maxdeg(\cdot) &\to \topmin(\cdot)\\
\topmin(\cdot),\cw(\cdot)&\to \tw(\cdot).
\end{align*}
\end{enumerate}
\end{enumerate}
Clearly, the relation ``$D_2$ is stronger than $D_1$'' is reflexive
and transitive (thus by ``stronger,'' we really mean ``stronger or
equal''). As an example, observe that replacing $\pw(G)$ by $\tw(G)$
makes the description stronger: the replacement can be expressed as
adding $\tw(G)$ using rule
\ref{item-d2:addbounded}\ref{item-d2:implicit} and then removing
$\pw(G)$ using rule \ref{item-d2:remove}.

 The following lemma justifies the name ``stronger'': an
algorithm compatible with a description stronger than $D$ implies an
algorithm compatible with $D$.
\end{definition}
\begin{lemma}\label{lem:strongerimplies}
  Let $D_1$ and $D_2$ be two descriptions. If $D_2$ is stronger $D_1$,
  then an algorithm compatible with $D_2$ implies the existence of an
  algorithm compatible with $D_1$.
\end{lemma}
\begin{proof}
  By induction, it is sufficient to prove the statement in the case
  when $D_2$ is obtained using only one of the operations listed in
  Definition~\ref{def:stronger}. For the operations \ref{item-d2:remove}--\ref{item-d2:addconstrained}, this follows from Lemma~\ref{lem:stronger}. 

  For \ref{item-d2:addbounded}\ref{item-d2:GtoH}, observe that all the parameters in our framework, with the exception of cliquewidth and the number of components, are monotone under taking subgraphs: if $H$ is a subgraph of $G$, then the value of the parameter is not greater on $H$ than on $G$. Therefore, the algorithm may assume that the values of these parameters are not greater on $H$ than on $G$, otherwise the algorithm can immediately stop with a ``no'' answer. Therefore, by Lemma~\ref{lem:stronger}\ref{item-d:addbounded}, we may add a parameter for $H$ if the same parameter for $G$ already present. As no constraint involves $\cw(\cdot)$, the same argument proves the validity of \ref{item-d2:addbounded}{item-d2:GtoHconst}.

  For \ref{item-d2:addbounded}\ref{item-d2:HtoGconstcc}, observe that if $H$ is connected then any subgraph isomorphism from $H$ to $G$ must map $H$ into one of the connected components of $G$. Hence, having an algorithm $\mathbb{A}$ working with the assumption that $G$ is connected as well, we may apply this algorithm to $H$ and each connected component of $G$ separately, thus adding $n$ overhead in the running time.

  For \ref{item-d2:addbounded}\ref{item-d2:maxdeg2}, note that if
  the maximum degree is 2, then the graph is the disjoint union of
  paths and cycles, thus pathwidth, treewidth, cliquewidth, genus, and
  the size of the largest (topological) minor are all bounded by
  constants. Thus we may add these parameters by
  Lemma~\ref{lem:stronger}\ref{item-d:addbounded}.

  For \ref{item-d2:addbounded}\ref{item-d2:tw1}, observe that if
  treewidth is 1, then the graph is a forest, thus feedback vertex set
  number, treewidth, cliquewidth, genus, and the size of
  the largest (topological) minor are all bounded by constants.  Thus
  we may add these parameters by
  Lemma~\ref{lem:stronger}\ref{item-d:addbounded}.

  Finally, \ref{item-d2:addbounded}\ref{item-d2:implicit} also uses
  Lemma~\ref{lem:stronger}\ref{item-d:addbounded} and the fact that in
  every rule $X\to Y$ listed above, every parameter in $Y$ can be
  bounded by a function of the parameters in $X$. This is clear when
  $X$ is $|V(\cdot)|$. The rule $\topmin(\cdot),\cw(\cdot)\to \tw(\cdot)$ follows
  from a result of Gurski and Wanke \cite{DBLP:conf/wg/GurskiW00}
  stating that if a graph $G$ contains no $K_{r,r}$ complete bipartite
  subgraph, then $\tw(G)\le 3\cw(G)(r-1)-1$. The maximum size of a
  complete bipartite subgraph can be clearly bounded by the largest subdivision of a clique appearing in the
  graph, hence the validity of the rules follow from
  Lemma~\ref{lem:stronger}\ref{item-d:addbounded}. The rest of the
  rules follow from well-known relations between the parameters such
  as the facts $\fvs(\cdot),\pw(\cdot)\le \tw(\cdot)$ and the fact
  that cliquewidth can be bounded by a function of treewidth.
\end{proof}

Theorem~\ref{th:table-covers} claims that for every possible
description, either there is a stronger positive result
appearing in Table~\ref{fig:table}, or it is stronger than a negative
result appearing in Table~\ref{fig:table}. By
Lemma~\ref{lem:strongerimplies}, this indeed shows that for every
description, either a positive or a negative result is implied by the
results in Table~\ref{fig:table}. To prove
Theorem~\ref{th:table-covers}, one needs to check the statement for
every description, that is, a finite number of statements have to
verified. While this is a tedious task due to the enormous number of
description, in principle it could be done by hand. However, the
program available as the electronic supplementary material of this paper performs this case analysis quickly.

Finally, let us clarify a subtle point. We have defined the relation
``stronger'' for descriptions, and proved that it works as expected:
an algorithm for the stronger description implies an algorithm for the
weaker one. Let us note that we did not say anything about
descriptions not related by this relation: in principle, it is
possible that for some pairs of descriptions, similar implications
hold due to some obvious or not so obvious reasons. However, for our
purposes, the relation as defined in Definition~\ref{def:stronger} is
sufficient: in particular, Theorem~\ref{th:table-covers} holds with
this notion.

\subsection{Tree decompositions}\label{sec:tree-decompositions}
Most readers will be familiar with the definition of a tree
decomposition definition, but it will be convenient for us to view
tree decompositions from a slightly different perspective here
(following \cite{gro10+a,DBLP:conf/stoc/GroheM12}).  A \emph{tree
  decomposition} of a graph is a pair $(T,\beta)$, where $T$ is a
rooted tree and $\beta:V(T)\to 2^{V(G)}$, such that for all nodes
$v\in V(G)$ the set $\{t\in V(G)\mid v\in\beta(t)\}$ is nonempty and
connected in the undirected tree underlying $T$, and for all edges
$e\in E(G)$ there is a $t\in V(T)$ such that $e\subseteq\beta(t)$. It
will be namely convenient for us to view the tree in a tree
decomposition as being directed: we direct all the edges in $T$ from
parents to children.  The {\em width} of the decomposition is
$\max_{t\in V(T)} |\beta(t)|-1$; the treewidth $\tw(G)$ of $G$ is the
minimum possible width of a decomposition of $G$.  Path decompositions
and pathwidth $\pw(G)$ is defined similarly, but we require $T$ to be
a (rooted) path.

If $(T,\beta)$ is a tree decomposition of a graph $G$, then we define mappings
$\sigma,\gamma,\alpha:V(T)\to 2^{V(G)}$ by letting for all $t\in V(T)$
\begin{align}
  \label{eq:bagsep}
  \sigma(t)&:=
  \begin{cases}
    \emptyset&\text{if $t$ is the root of $T$},\\
    \beta(t)\cap\beta(s)&\text{if $s$ is the parent of $t$ in $T$},
  \end{cases}\\
  \label{eq:bagcone}
  \gamma(t)&:=\bigcup_{\text{$u$ is a descendant of $t$}}\beta(u),\\
  \label{eq:bagcomp}
  \alpha(t)&:=\gamma(t)\setminus\sigma(t).
\end{align}
We call 
$\beta(t),\sigma(t),\gamma(t),\alpha(t)$ the \emph{bag} at $t$,
\emph{separator} at $t$,
\emph{cone} at $t$, \emph{component} at $t$, respectively. Note that we follow the convention that every vertex is an descendant of itself.

\subsection{Bounding pathwidth}
\label{sec:bounding-pathwidth}

We now prove a simple technical result that will be useful for bounding pathwidth of graphs obtained in hardness reductions. In the following, if we are given a rooted tree $T$ with root $r$ and a graph $G$ with a prescribed vertex $v$, then by {\emph{attaching $T$ at $v$}} we mean the following operation: take the disjoint union of $T$ and $G$, and identify $r$ with $v$.

\begin{lemma}\label{lem:pathw}
Let $D_1,D_2,\ldots,D_r$ be a family of trees such that $\pw(D_i)\leq c$ for some constant $c$. Let $D$ be a graph obtained by taking a path $P$ on at least $r$ vertices, and attaching trees $D_i$ at different vertices of $P$. Then $\pw(D)\leq \max(1,c+1)$.
\end{lemma}
\begin{proof}
Let $\mathcal{P}_i$ be a path decomposition of $D_i$ of width at most $c$. Construct a path decomposition $\mathcal{P}'$ of $P$ of width $1$ by taking as bags pairs of consecutive vertices on $P$. Now modify $\mathcal{P}'$ to a path decomposition $\mathcal{P}$ of $D$ as follows: for every vertex $v$ at which a tree $D_i$ is attached, with $v',v''$ being the predecessor and successor of $v$ on $P$, insert between bags $\{v',v\}$ and $\{v,v''\}$ the whole decomposition $\mathcal{P}_i$ with $v$ added to every bag. It is easy to see that $\mathcal{P}$ is indeed a path decomposition of $D$, and that its width is at most $\max(1,c+1)$.
\end{proof}

\subsection{Fixing images of a prescribed set of vertices}
\label{sec:fixing-imag-prescr}

\newcommand{\eG}{\overline{G}}
\newcommand{\eH}{\overline{H}}
\newcommand{\ehm}{\overline{\hm}}

We present a simple reduction using which we can fix the images of a
prescribed set of vertices. We use this reduction both in algorithms
and hardness proofs.

\begin{lemma}\label{lem:setting-images}
Assume we are given an instance $(H,G)$ of \subiso, and we are given a sequence of vertices $v_1,v_2,\ldots,v_p\in V(H)$ and images $w_1,w_2,\ldots,w_p\in V(G)$. Then it is possible to construct in polynomial time a new instance $(\eH,\eG)$ of \subiso, such that the following hold:
\begin{itemize}
\item[(i)] There are subgraphs $H_0$ and $G_0$ of $\eH$ and $\eG$, isomorphic to $H$ and $G$, respectively (from now on we identify $H$ and $G$ with $H_0$ and $G_0$, respectively, using these isomorphisms).
\item[(ii)] $\eH$ and $\eG$ can be constructed from $H_0$ and $G_0$ by attaching a tree of maximum degree 3, constant pathwidth, and of size polynomial in the size $G$ to every vertex $v_1,v_2,\ldots,v_p$ and $w_1,w_2,\ldots,w_p$. Each of the attached trees has exactly one edge incident to the corresponding vertex $v_i$ or $w_i$. Moreover, if $G$ does not contain any path on at least $L$ vertices passing only through vertices $w_1,w_2,\ldots,w_p$ or other vertices of degree at least $3$, then we may assume that each of the attached trees has size $O(pL)$.
\item[(iii)] For every subgraph isomorphism $\hm$ from $H$ to $G$ such that $\hm(v_i)=w_i$ for $i=1,2,\ldots,p$, there exists a subgraph isomorphism $\ehm$ from $\eH$ to $\eG$ that restricted to $H_0$ and $G_0$ is equal to $\hm$.
\item[(iv)] For every subgraph isomorphism $\ehm$ from $\eH$ to $\eG$, we have that $\ehm$ restricted to $H_0$ is a subgraph isomorphism from $H_0$ to $G_0$, and moreover $\ehm(v_i)=w_i$ for every $i=1,2,\ldots,p$.
\end{itemize}
\end{lemma}
\begin{proof}
To obtain $\eH$ and $\eG$, we perform the following construction. We begin with $H_0$ and $G_0$ isomorphic to $H$ and $G$, respectively, and for every $i$ we add the following gadget $P_i$ to $v_i$ and $w_i$. We attach a path of length $(i+1)\cdot L+1$ to the considered vertex (called the {\emph{root}} of the gadget). Furthermore, we attach a pendant edge (i.e, a new degree-1 neighbor) to every internal node of this path apart from the first $L$ ones (closest to the root). This concludes the construction. Properties (i) and (ii) follow directly from the construction. To see that (iii) is also satisfied, observe that given any subgraph isomorphism $\hm$ from $H_0$ to $G_0$ such that $\hm(v_i)=w_i$ for every $i=1,2,\ldots,p$, we may easily extend it to gadgets $P_i$ in $\eH$ and $\eG$ using the fact that roots of these gadgets are mapped appropriately. We are left with proving that property (iv) holds as well.

Assume that $\ehm$ is a subgraph isomorphism from $\eH$ to $\eG$. Consider the internal vertices of $P_p$ in $\eH$ that have pendant edges attached. There are $pL$ of them, they induce a connected subgraph of $\eH$, and each of them has degree $3$. It follows that their images in $\eG$ also must have degrees at least $3$, and moreover they must induce a connected subgraph in $\eG$ admitting a Hamiltonian path. By the choice of $L$, it is easy to observe that the only subset of vertices of $\eH$ that has these properties and has size $pL$ consists of the internal vertices of the gadget $P_p$ in $\eG$ that have pendant edges. We infer that these internal vertices of $P_p$ in $\eH$ must be mapped to these internal vertices of $P_p$ in $\eG$, so $\ehm$ must map $P_p$ in $\eH$ into $P_p$ in $\eG$ using the natural isomorphism between them. In particular, $\hm(v_p)=w_p$.

Now, as $P_p$ in $\eG$ is already assigned some preimages, we may perform the same reasoning for gadget $P_{p-1}$: the only place where the internal vertices of $P_{p-1}$ in $\eH$ that have pendant edges can be mapped, are the corresponding vertices of $P_{p-1}$ in $\eG$. We may proceed in the same manner with $P_{p-2},P_{p-3},\ldots,P_1$, thus concluding that $\ehm$ maps each $P_i$ in $\eH$ to $P_i$ in $\eG$ and $\ehm(v_i)=w_i$ for every $i=1,2,\ldots,p$. We infer that the remaining vertices of $\eH$ must be mapped to the vertices of $G_0$, hence $\ehm$ restricted to $H_0$ is a subgraph isomorphism from $H_0$ to $G_0$.
\end{proof}

In most cases, it is sufficient to know that the sizes of the trees attached in Lemma~\ref{lem:setting-images} are polynomially bounded; however, in some situations it will be useful to use a smaller $L$ to bound the number of vertices introduced in the construction. Note that application of Lemma~\ref{lem:setting-images} introduces at most $O(p^2L)$ vertices to both $H$ and $G$. 


\section{Positive results}
\label{sec:positive-results}We present a set of positive results in this section. In
Section~\ref{sec:known-easy-results}, we review first results that
follow from known results and techniques in an easy way. In
Section~\ref{sec:packing-long-cycles}, we present results on finding
paths and cycles of specified lengths in bounded-treewidth and
bounded-cliquewidth graphs, which can be thought of as a
generalization of standard algorithms for finding Hamiltonian
cycles. Finally, in Section~\ref{sec:packing-forest-into}, we
generalize Theorem~\ref{th:matula} from packing a tree in a tree to
packing a forest into a forest. This generalization turns out to be
surprisingly challenging: solving the problem apparently hinges on a
variant of matching, for which the only known technique is an
algebraic randomized algorithm. Therefore, our algorithm is
randomized; note that this is the only positive result in the paper
for which we give a randomized algorithm. 

Further positive results involving graphs of bounded degree and
bounded feedback vertex set number follow in
Section~\ref{sec:big-planar-algorithm}. As these results are
significantly more involved and require more background, we present
them in their own section.

\subsection{Known and easy results}\label{sec:known-easy-results}

The following results follow directly from the fact that, given a
graph $H$, we can express that $G$ contains $H$ as a subgraph by a
formula of first order logic of length bounded polynomially in
$|H|$. Hence, we may apply the results on model checking first-order
logic on graphs of bounded cliquewidth, by Courcelle et
al.~\cite{DBLP:journals/mst/CourcelleMR00}, and on graphs of bounded
expansion, by Dvo\v{r}ak et al.~\cite{DBLP:conf/focs/DvorakKT10}.
Note that every graph excluding the subdivision of a fixed graph has
bounded expansion, hence an algorithm parameterized by expansion
implies an algorithm parameterized by $\topmin$ (see
\cite{DBLP:conf/focs/DvorakKT10}).

\begin{ptheorem}\label{thm:FOcw}
There exists an algorithm compatible with the description
\begin{eqnarray*}
\sil{|V(H)|,\cw(G)}{}{}
\end{eqnarray*}
\end{ptheorem}

\begin{ptheorem}\label{thm:FOtopmin}
There exists an algorithm compatible with the description
\begin{eqnarray*}
\sil{|V(H)|,\topmin(G)}{}{}
\end{eqnarray*}
\end{ptheorem}

Next, the following result was observed by Alon et al.~\cite{DBLP:journals/jacm/AlonYZ95} as an easy corollary of their technique of colour-coding.

\begin{ptheorem}\label{thm:colour-coding}
There exists an algorithm compatible with the description
\begin{eqnarray*}
\sil{|V(H)|}{\tw(H)}{}
\end{eqnarray*}
\end{ptheorem}

We now give a generalization of Theorem~\ref{th:matousek-thomas}, the
algorithm due to Matou\v{s}ek and
Thomas~\cite{DBLP:journals/dm/MatousekT92}, that can handle the case
when $H$ is disconnected. More precisely, the number of connected
components of $H$ can be treated as a parameter. The proof is an easy
reduction to the original algorithm of Matou\v{s}ek and Thomas \cite{DBLP:journals/dm/MatousekT92} that
works only when $H$ is connected.

\begin{ptheorem}\label{thm:strong-MT}
There exists an algorithm compatible with the description
\begin{eqnarray*}
\sil{\maxdeg(H),\ccn(H)}{\tw(G)}{}
\end{eqnarray*}
\end{ptheorem}
\begin{proof}
Let $H$ and $G$ be the input graphs. Let us construct graphs $\overline{H}$ and $\overline{G}$ as follows. We add one vertex $u^*$ to $H$ and make it adjacent to exactly one vertex of each connected component of $H$. We add also one vertex $v^*$ to $G$ and make it adjacent to every other vertex of $G$, thus making it universal in $G$. Let $\hat{H}$ and $\hat{G}$ be the resulting graphs. In order to obtain $\overline{H}$ and $\overline{G}$ we apply Lemma~\ref{lem:setting-images} to graphs $\hat{H}$ and $\hat{G}$ ensuring that $u^*$ needs to be mapped to $v^*$. Observe that application of Lemma~\ref{lem:setting-images} attaches a tree of constant pathwidth and maximum degree $3$ to $u^*$ and to $v^*$.

Clearly, if there is a subgraph isomorphism $\hm$ from $H$ to $G$, then we may extend this subgraph isomorphism to a subgraph isomorphism from $\hat{H}$ to $\hat{G}$ by putting $\hm(u^*)=v^*$, and then to a subgraph isomorphism from $\overline{H}$ to $\overline{G}$ using Lemma~\ref{lem:setting-images}. On the other hand, by Lemma~\ref{lem:setting-images} every subgraph isomorphism $\hm'$ from $\overline{H}$ to $\overline{G}$ is also a subgraph isomorphism from $\hat{H}$ to $\hat{G}$ when restricted to $\hat{H}$, and moreover $\hm'(u^*)=v^*$. It follows that $\hm'$ restricted to $H$ must be a subgraph isomorphism from $H$ to $G$. Consequently, instances $(H,G)$ and $(\overline{H},\overline{G})$ of \subiso are equivalent.

By the construction of $\overline{H}$, we have that $\overline{H}$ is connected and $\maxdeg(\overline{H})\leq \max\{\maxdeg(H)+1,\ccn(H)+1,3\}$, since the original vertices of $H$ have degrees at most one larger in $\overline{H}$, the new vertex $u^*$ has degree $\ccn(H)+1$, and vertices introduced by application of Lemma~\ref{lem:setting-images} have degrees at most $3$. On the other hand we have that $\tw(\overline{G})\leq \tw(G)+1$, since we may create a tree decomposition ${\overline{\cal{T}}}$ of $\overline{G}$ from any tree decomposition $\cal{T}$ of $G$ by adding $v^*$ to every bag, and then attaching to any node of ${\cal{T}}$ the tree decomposition of width $1$ of the tree attached to $v^*$ by the application of Lemma~\ref{lem:setting-images}. Finally, we can apply the algorithm of Matou\v{s}ek and Thomas to the instance $(\overline{H},\overline{G})$. Since the sizes of $\overline{H}$ and $\overline{G}$ are polynomial in the size of the input instance, we obtain running time $f_1(\max\{\maxdeg(H)+1,\ccn(H)+1,3\})\cdot n^{f_2(\tw(G)+1)}$ for some functions $f_1,f_2$, which is $f_1'(\maxdeg(H),\ccn(H))\cdot n^{f_2'(\tw(G))}$ for some functions $f_1',f_2'$.
\end{proof}

Finally, we give a positive result for the trivial case when $G$ consists of cycles and paths only.

\begin{ptheorem}\label{thm:vcG}
There exists an algorithm compatible with the description
\begin{eqnarray*}
\sil{}{\ccn(G)}{\maxdeg(G)\leq 2}
\end{eqnarray*}
\end{ptheorem}
\begin{proof}
As $\maxdeg(G)\leq 2$, we may safely assume that also $\maxdeg(H)\leq 2$, as otherwise we may provide a negative answer immediately. Then both $H$ and $G$ are graphs consisting of disjoint paths and cycles. Observe that if a subgraph isomorphism exists, then for any cycle in $H$ there must be cycle of the same length in $G$; given any such cycle in $G$, without loss of generality we may map the cycle in $H$ to the cycle in $G$. Hence, from now on we can assume that $H$ contains only paths.

Let $a_1,a_2,\ldots,a_p$ and $b_1,b_2,\ldots,b_r$ be numbers of vertices of connected components of $H$ and $G$, respectively (paths for $H$, paths or cycles for $G$), where $r=\ccn(G)$. Observe that a subset of paths of $H$ can be mapped into a connected component $C$ in $G$ if and only if the total number of vertices on these paths does not exceed the number of vertices in $C$. Hence, we are left with an instance of bin packing: $a_1,a_2,\ldots,a_p$ are sizes of items and $b_1,b_2,\ldots,b_r$ are sizes of bins. This instance can be solved by a standard dynamic programming routine working in time $O(p\cdot r\cdot \prod_{i=1}^r (b_i+1))$: states of the dynamic program corresponds to the space of partial fillings of the bins. 
\end{proof}

\subsection{Packing long paths and cycles}
\label{sec:packing-long-cycles}

We now proceed to dynamic programming algorithms that check whether a family of disjoint cycles and paths of prescribed lengths can be found in a graph of small treewidth or cliquewidth. In case of treewidth, the number of cycles and paths must stand in the exponent while treewidth may be in the multiplier, while for cliquewidth both the number of components and cliquewidth must appear in the exponent. The approach is standard dynamic programming on graphs of bounded treewidth or cliquewidth.

\begin{ptheorem}\label{thm:dpCyclesParTw}
There exists an algorithm compatible with the description
\begin{eqnarray*}
\sil{\tw(G)}{\ccn(H)}{\maxdeg(H)\leq 2}
\end{eqnarray*}
\end{ptheorem}
\begin{proof}
Since $\tw(H)\leq 1$ and $\maxdeg(H)\leq 2$, we have that $H$ is a family of $c$ paths or cycles $C_1,C_2,\ldots,C_c$. Let $n_i=|V(C_i)|$ for $n_i=1,2,\ldots,c$ and let $n=|V(G)|$. We run the algorithm of Bodlaender~\cite{Bodlaender96} to compute an optimal tree decomposition $\mathcal{T}=(T,\beta)$ of $G$ of width $\tw(G)$, in $f(\tw(G))\cdot n$ time for some function $f$ (see Section~\ref{sec:tree-decompositions}). Let $r$ be the root of $T$. Using standard methods, we transform $\mathcal{T}$ into a {\em{nice}} tree decomposition, i.e., a decomposition such that $\beta(r)=\emptyset$, and where each node is of four possible types:
\begin{itemize}
\item {\bf{Leaf node $t$.}} A node with no children and $\beta(t)=\emptyset$.
\item {\bf{Introduce node $t$.}} A node with one child $t'$ such that $\beta(t)=\beta(t')\cup \{v\}$ for some $v\notin B_{t'}$.
\item {\bf{Forget node $t$.}} A node with one child $t'$ such that $\beta(t)=\beta(t')\setminus \{v\}$ for some $v\in B_{t'}$.
\item {\bf{Join node $t$.}} A node with two children $t_1,t_2$ such that $\beta(t)=\beta(t_1)=\beta(t_2)$.
\end{itemize}
We denote $G_t=G[\gamma(t)]$.

We proceed to defining the encoding of intersection of an embedding of components $C_1,C_2,\ldots,C_c$ with a bag of $\mathcal{T}$. Examine a node $t$ with bag $\beta(t)$. We say that a triple $S=(\mathcal{X},\mathcal{M},\mathcal{N})$ is a {\em{signature}} on node $t$, when:
\begin{itemize}
\item $\mathcal{X}=(X_0,(X_1^0,X_1^1,X_1^2),(X_1^0,X_1^1,X_1^2),\ldots,(X_c^0,X_c^1,X_c^2))$ is such that sets $X_0$ and $X_i^j$ for $i\in [c]$ and $j\in \{0,1,2\}$ form a partition of $\beta(t)$;
\item $\mathcal{M}=(M_1,M_2,\ldots,M_c)$ is such that each $M_i$ is a partition of $X_i^1$ into pairs if $C_i$ is a cycle, or $M_i$ is a partition of $X_i^1$ into pairs and at most two singletons if $C_i$ is a path;
\item $\mathcal{N}$ is a function from $[c]$ to integers, such that $0\leq \mathcal{N}(i)\leq n_i$ for each $i\in [c]$.
\end{itemize}
We denote by $\mathcal{S}_t$ the set of possible signatures on $t$. It is easy to observe that $|\mathcal{S}_t|\leq 2^{O(|\beta(t)|\log (c|\beta(t)|))}\cdot (n+1)^c$, since the number of choices of $\mathcal{X}$ is bounded by $2^{O(|\beta(t)|\log c))}$, the number of choices of $\mathcal{M}$ is bounded by $2^{O(|\beta(t)|\log |\beta(t)|))}$, and the number of choices of $\mathcal{N}$ is bounded by $(n+1)^c$.

A vector $\mathcal{E}=(\mathcal{P}_1,\mathcal{P}_2,\ldots,\mathcal{P}_c)$ is called a {\em{partial embedding}} if each $\mathcal{P}_i$ is either
\begin{itemize}
\item in case $C_i$ is a cycle, a cycle in $G_t$ of length $n_i$;
\item in case $C_i$ is a path, a path in $G_t$ with both endpoints in $\alpha(t)$ and having $n_i$ vertices; or
\item a (possibly empty) family of paths (of possibly zero length) that have both endpoints in $\beta(t)$, apart from at most two that can have only one endpoint in $\beta(t)$ in case when $C_i$ is a path. 
\end{itemize}
We also require all the paths and cycles in all the families $\mathcal{P}_i$ to be pairwise vertex disjoint. We say that a partial embedding $\mathcal{E}$ is {\emph{compatible}} with signature $S=(\mathcal{X},\mathcal{M},\mathcal{N})$, if the following conditions are satisfied for each $i\in [c]$:
\begin{itemize}
\item $V(\mathcal{P}_i)\cap \beta(t)=X_i^0\cup X_i^1\cup X_i^2$, and vertices of $X_i^j$ are exactly vertices of $\beta(t)$ having degree $j$ in $\mathcal{P}_i$, for $j\in \{0,1,2\}$.
\item $M_i$ is the set of equivalence classes of the equivalence relation on $X_i^1$ defined as follows: two vertices are equivalent if they are in the same connected component of $\mathcal{P}_i$.
\item $|V(\mathcal{P}_i)|=\mathcal{N}(i)$.
\end{itemize}
We finally define the function that will be computed in the dynamic programming algorithm. For a node $t$ and a signature $S$ on $t$, let $\phi(t,S)$ be a boolean value denoting whether there exists a partial embedding $\mathcal{E}$ compatible with $S$. If $t$ is a leaf node, then $\phi(t,S)$ is true for exactly one signature $S=(\emptyset,\emptyset,\lambda i.0)$. The answer to the problem is exactly the value $\phi(r,S)$ for $S=(\emptyset,\emptyset,\lambda i.n_i)$. However, it is easy to derive recurrential equations on value $\phi(t,S)$ depending on values of $\phi$ for sons of $t$ for each node type: introduce, forget, and join; see for example the classical dynamic programming routine for {\sc{Hamiltonian Cycle}} on graphs of bounded treewidth \cite{DBLP:journals/informs/CookS03}. As these equations are very standard, we omit the details here. Since $|\mathcal{S}_t|\leq 2^{O(|\tw(G)|\log (c|\tw(G)|))}\cdot (n+1)^c$ for each node $t$, we can compute all the values of function $\phi$ in a bottom-up manner in $2^{O(|\tw(G)|\log (c|\tw(G)|))}\cdot n^{O(c)}$ time. Hence, the algorithm is compatible with the description.
\end{proof}

Before we proceed to the algorithm for graphs of bounded cliquewidth, let us recall the notation for labeled graphs and clique expressions. A {\emph{labeled graph}} is a pair $(G,\alpha)$, where $G$ is a graph and $\alpha\textrm{ : } V(G)\to [k]$ is a labeling function that associates with each vertex of $G$ one of $k$ different labels. We define three operations on labeled graphs:
\begin{itemize}
\item {\bf{Disjoint union}} $\oplus$ is defined as $(G_1,\alpha_1)\oplus (G_2,\alpha_2)=(G_1\cup G_2, \alpha_1\cup \alpha_2)$. In other words, we take the disjoint union of graphs $G_1$ and $G_2$, and define the labeling as the union of original labelings.
\item {\bf{Join}} $\eta_{i,j}(\cdot)$ for $i,j\in [k]$, $i\neq j$, is defined as $\eta_{i,j}((G,\alpha))=(G',\alpha)$, where $G'$ is $G$ after introducing all possible edges having one endpoint labeled with $i$ and second with $j$.
\item {\bf{Relabel}} $\rho_{i\to j}(\cdot)$ for $i,j\in [k]$, $i\neq j$, is defined as $\rho_{i\to j}((G,\alpha))=(G,\alpha')$, where $\alpha'$ is $\alpha$ with all the values $i$ substituted with $j$.
\end{itemize}
A {\emph{clique expression}} is a term that uses operators $\oplus$, $\eta_{i,j}(\cdot)$, $\rho_{i\to j}(\cdot)$, and constants $\iota_1,\iota_2,\ldots,\iota_k$ that represent one-vertex graphs with the only vertex labeled with $1,2,\ldots,k$, respectively. In this manner a clique expression {\em{constructs}} some labeled graph $(G,\alpha)$. The cliquewidth of a graph $G$ is the minimum number of labels needed in a clique expression that constructs $G$ (with any labeling).

In our dynamic programming routine we use two standard simplification assumptions of clique expressions, cf.~\cite{DBLP:journals/mst/CourcelleMR00}.
\begin{itemize}
\item First, we can assume that whenever join operation $\eta_{i,j}$ is applied, before application there was no edge in the graph with one endpoint labeled with $i$ and second with $j$. Indeed, any such edge would need to be introduced by some previously applied operation $\eta_{i',j'}$ for some labels $i'$ and $j'$ that have been eventually relabeled to $i$ and $j$, respectively. As every edge introduced by this operation will be also introduced by the considered $\eta_{i,j}$ operation, we can safely remove the $\eta_{i',j'}$ operation from the clique expression.
\item Second, by at most doubling the number of labels used we can assume the following: whenever the disjoint union operation $\oplus$ is applied on labeled graphs $(G_1,\alpha_1)$, $(G_2,\alpha_2)$, the sets of labels used by $\alpha_1$ and $\alpha_2$, i.e., the labels with nonempty preimages, are disjoint. The following preprocessing of the clique expression ensures this property: we double the set of labels by creating a copy of each label, and whenever the disjoint union operation is performed, we first relabel all the labels in the second graph to the copies, then perform the disjoint union, and finally relabel the copies to the originals.
\end{itemize}
A clique expression is called {\emph{nice}} if it satisfies these two assumptions. Obviously, one can in polynomial time turn any clique expression into a nice one that uses at most twice as many labels.

\newcommand{\Pp}{\mathcal{P}}
\newcommand{\Ss}{\mathcal{S}}

\begin{ptheorem}\label{thm:cwdp}
There exists an algorithm compatible with the description
\begin{eqnarray*}
\sil{}{\ccn(H),\cw(G)}{\maxdeg(H)\leq 2}
\end{eqnarray*}
\end{ptheorem}
\begin{proof}
Let $C_1,C_2,\ldots,C_c$ be the connected components of $H$, where each $C_q$ is either a path or a cycle, for $q\in [c]$. Before running the dynamic programming routine, we perform the following simplification step. First, assume that every component $C_q$ for $q\in [c]$ has at least two vertices, since isolated vertices in $H$ can be mapped to any vertices of $G$. Second, for each component $C_q$ that is a path, guess the images $w_q^1$ and $w_q^2$ of the endpoints of $C_q$. Third, for each component $C_q$ that is a cycle, take any edge of this cycle and guess the edge $w_q^1w_q^2$ in $G$ it is mapped to. By trying all possible guesses, with $n^{2c}$ overhead in the running time we may assume that the vertices $w_q^t$ for $q\in [c]$ and $t\in \{1,2\}$ are fixed. Thus, we need to check whether one can find paths $P_1,P_2,\ldots,P_c$ in $G$, where each $P_q$ leads from $w_q^1$ to $w_q^2$ and has $|V(C_q)|$ vertices.

Without loss of generality we may assume that the given clique expression $\mathbf{t}$ constructing $G$ is nice. Furthermore, we apply the following simplification step. Add $2c$ additional labels, one for each vertex $w_i^t$ for $i\in [c]$ and $t\in \{1,2\}$. Modify the clique expression by assigning each vertex $w_i^t$ its unique label upon creation, and then additionally performing the same join operations with this label as with the original label of $w_i^t$. In this manner, by using additional $2c$ labels we may assume that in the labeling of the constructed graph $G$, each vertex $w_i^t$ is the only vertex of its label. Without loss of generality assume that the label of vertex $w_i^t$ is $2\cdot (i-1)+t$, i.e., the special labels for vertices $w_i^t$ are exactly the first $2c$ labels. Observe that $k$, the number of used labels in the constructed clique expression, is at most $2r+2c$, where $r$ was the number of labels used by the original clique expression.

We now proceed to the explanation of the dynamic programming routine. Let $K=\{(i,j)\ |\ i,j\in [k],\ i\leq j\}$. By a {\emph{signature}} we mean any pair of functions $(\zeta,\xi)$, where $\zeta$ maps $[c]\times K$ to $\{0,1,\ldots,n\}$ and $\xi$ maps $[c]$ to $\{0,1,\ldots,n\}$. Note that the number of possible signatures is equal to $(n+1)^{c\cdot (1+\binom{k+1}{2})}$. Assume we are given a labeled graph $(J,\alpha)$ and a vector of families of paths $\Ss=(\Pp_1,\Pp_2,\ldots,\Pp_c)$ in $J$, where all the paths in all the families are pairwise vertex-disjoint. We say that vector $\Ss$ is {\emph{compatible}} with signature $(\zeta,\xi)$ if:
\begin{itemize}
\item for each $q\in [c]$ and each $(i,j)\in K$, the number of paths in $\Pp_q$ with one endpoint labeled with $i$ and second with $j$ is equal exactly to $\zeta(q,(i,j))$;
\item for each $q\in [c]$, it holds that $\xi(q)=\sum_{P\in \Pp_q} |V(P)|$.
\end{itemize}
For a subterm $\mathbf{s}$ of the clique expression $\mathbf{t}$ constructing $G$, let $(G_\mathbf{s},\alpha_\mathbf{s})$ be the labeled graph constructed by subterm $\mathbf{s}$. Let now $\phi(\mathbf{s},(\zeta,\xi))$ be the boolean value denoting whether there exists a vector of families of paths $\Ss=(\Pp_1,\Pp_2,\ldots,\Pp_c)$ in $(G_\mathbf{s},\alpha_\mathbf{s})$ that is compatible with signature $(\zeta,\xi)$. Observe that the values $\phi(\mathbf{s},\cdot)$ can be trivially computed when $\mathbf{s}=\iota_\ell$ for some label $\ell$, and the seeken answer to the problem is exactly value $\phi(\mathbf{t},(\zeta_0,\xi_0))$, where $\zeta_0(q,(i,j))=[i=2q-1\ \wedge\ j=2q]$ for $q\in [c]$ and $(i,j)\in K$, and $\xi_0(q)=|V(C_q)|$ for $q\in [c]$. Therefore, it remains to derive recurrential equations using which we can compute values $\phi(\mathbf{s},\cdot)$ using precomputed values of subterms of $\mathbf{s}$, for all three possible operations: disjoint union, join, and relabel. This, however, can be easily done in a standard manner, similarly to the dynamic programming routine for Hamiltonian path on graphs of bounded cliquewidth~\cite{EspelageGW01,Wanke94}. Since these equations are very standard, we omit the details here.

Concluding, for all possible choices of vertices $w_q^t$ for $q\in [c]$ and $t\in \{1,2\}$, the algorithm computes the values $\phi(\mathbf{s},(\zeta,\xi))$ for subterms $\mathbf{s}$ of $\mathbf{t}$ and checks if value $\phi(\mathbf{t},(\zeta_0,\xi_0))$ is true, in which case the algorithm reports a positive answer and terminates. If for no choice a true value has been computed, a negative answer can be reported. Each computation of value $\phi(\mathbf{t},(\zeta_0,\xi_0))$ takes $n^{O(ck^2)}=n^{O(c^3+cr^2)}$ time, which together with $n^{2c}$ overhead for guessing vertices $w_q^t$ gives $n^{O(c^3+cr^2)}$ time. Thus, the algorithm is compatible with the description.
\end{proof}

\subsection{Packing a forest into a forest}
\label{sec:packing-forest-into}

In this section we provide an algorithm that checks whether a forest $H$ can be embedded into a forest $G$, where the number of components of $H$ is allowed to stand in the multiplier. In spite of the fact that the algorithm for embedding a tree into a tree is a classical and relatively straightforward example of dynamic programming (Theorem~\ref{th:matula}), the case when $H$ can be disconnected is significantly more difficult. In particular, the algorithm we give is randomized with false negatives, and the source of randomization is usage of the following subroutine. By a randomized algorithm with false negatives, we mean an algorithm which is always correct given a NO-instance, but given a YES-instance it may provide a negative answer with probability at most $\frac{1}{2}$.

\begin{proposition}[\cite{MulmuleyVV87}]\label{prop:wmatching}
There exists a randomized algorithm with false negatives that, given a bipartite graph $B$ with nonnegative integer weights bounded by $W$ and a target weight $w_0$, checks in time polynomial in $|B|$ and $W$ whether there exists a perfect matching in $B$ of weight exactly $w_0$.
\end{proposition}

The algorithm of Proposition~\ref{prop:wmatching} uses algebraic
techniques: using the Isolation Lemma, which introduces randomness to the framework, it reduces the problem to non-singularity of a certain matrix constructed basing on the graph. 
We remark that even though the algorithm of
Proposition~\ref{prop:wmatching} has been stated
in~\cite{MulmuleyVV87} in terms of simple bipartite graphs, it is easy
to adjust it to the case of multigraphs where two vertices may be
connected by many edges of different weights. Also, by repeating the
algorithm $m$ times we may reduce the probability of giving a false
negative to at most $\frac{1}{2^m}$. We will use
Proposition~\ref{prop:wmatching} to solve efficiently the following
parameterized problem.

\defparproblemu{{\sc{Perfect Matching with Hitting Constraints}} ({\sc{PMwHC}})}{A bipartite graph $B$ and $\ell$ edge subsets $E_1,E_2,\ldots,E_\ell\subseteq E(B)$}{$\ell$}{Is there a perfect matching $M$ in $B$ with the following property: there exist $\ell$ distinct edges $e_1,e_2,\ldots,e_\ell$ of $M$ such that $e_t\in E_t$ for $t\in [\ell]$?}

\begin{lemma}\label{lem:pmwhc}
There exists a randomized algorithm with false negatives that, given an instance $(B,(E_t)_{t\in [\ell]})$ of {\sc{PMwHC}}, solves it in time $2^{O(\ell)}\cdot |B|^{O(1)}$.
\end{lemma}
\begin{proof}
We reduce {\sc{PMwHC}} to the problem of finding a matching of prescribed weight in an auxiliary bipartite graph, which can be solved using Proposition~\ref{prop:wmatching}. Let $(B,(E_t)_{t\in [\ell]})$ be the input instance of {\sc{PMwHC}}. Take $B$, and for each edge $e\in E(B)$, add a copy $e^t$ of this edge for each index $t\in [\ell]$ for which $e\in E_t$. Let $B'$ be the obtained graph; we silently identify edges of $B$ with the corresponding originals in $B'$. For $e'\in E(B')$, we define weight function $w$ as follows:
\begin{equation*}
w(e')=
\begin{cases}
0 & \textrm{if }e'\in E(B),\\
\displaystyle 2^{t-1} + 2^{2\ell-t} & \textrm{if }e'=e^t\textrm{ for some }e\in E(B).
\end{cases}
\end{equation*}
We now claim that the input instance $(B,(E_t)_{t\in [\ell]})$ has some solution $M$ if and only if there exists a perfect matching $M'$ in $B'$ of total weight exactly $2^{2\ell}-1$. 

For one direction, assume that $M$ is a perfect matching in $B$ and let $e_1,e_2,\ldots,e_\ell$ be edges of $M$ witnessing that $M$ is a solution to the {\sc{PMwHC}} problem. We construct a perfect matching $M'$ of $B'$ by taking $M$ in $B'$ and substituting $e_t$ with copy $e^{t}_t$ for each $t\in [\ell]$. It follows that the total weight of $M'$ is $\sum_{t=0}^{2\ell-1} 2^t=2^{2\ell}-1$.

Assume now that we are given a perfect matching $M'$ in $B'$ of total weight exactly $2^{2\ell}-1$. We claim that $M'$ uses exactly one copy of form $e^t$ for each $t\in [\ell]$, and apart from these edges it uses only originals. Note that proving this claim will finish the proof, as then the natural projection of $M'$ onto $B$ will be a perfect matching in $B$, and the edges for which copies were used in $M'$ will witness that it is a solution to the {\sc{PMwHC}} problem. For the sake of contradiction, assume that there is some index $t_0$ such that the number of copies of form $e^{t_0}$ used by $M'$ is not equal to $1$. We can assume that $t_0$ is the minimum such index. We consider now two cases.

Assume first that $M'$ uses no copies of form $e^{t_0}$. By minimality of $t_0$ we infer that
$$w(M')\equiv 2^0+2^1+\ldots+2^{t_0-2}=2^{t_0-1}-1 \mod 2^{t_0}.$$
However, we have that $2^{2\ell}-1\equiv -1 \mod 2^{t_0}$, which is a contradiction.

Assume second that $M'$ uses at least $2$ copies of form $e^{t_0}$. Again, by minimality of $t_0$ we infer that
$$w(M')\geq 2^{2\ell-1}+2^{2\ell-2}+\ldots+2^{2\ell-t_0+1}+2\cdot 2^{2\ell-t_0}=2^{2\ell}>2^{2\ell}-1,$$
which is again a contradiction.

Hence, to solve the {\sc{PMwHC}} problem it suffices to construct $B'$ in polynomial time and run the algorithm of Proposition~\ref{prop:wmatching} for the target weight $2^{2\ell}-1$. Observe that $B'$ has polynomial size in terms of $|B|+\ell$ and the weights in $B'$ are bounded by $2^{2\ell}$, so each application of Proposition~\ref{prop:wmatching} takes $2^{O(\ell)}\cdot |B|^{O(1)}$ time.
\end{proof}

\newcommand{\e}{\circ}

\begin{ptheorem}\label{thm:treesintotrees}
There exists a randomized algorithm with false negatives compatible with the description
\begin{eqnarray*}
\sil{\ccn(H)}{}{\tw(G)\leq 1}
\end{eqnarray*}
\end{ptheorem}
\begin{proof}
Let $n=|V(G)|$. By the description, we have that $G$ is a forest, so we can also assume that $H$ is a forest. Let $H_1,H_2,\ldots,H_k$ be the trees in $H$, where $k=\ccn(H)$. 

We claim that we may also assume that $G$ is in fact a tree. Assume that we have designed a randomized algorithm $\mathbb{A}$ with false negatives that given any forest $H'$ and any tree $G'$, checks if there exists a subgraph isomorphism from $H'$ to $G'$ in $f(\ccn(H'))\cdot |G'|^{O(1)}$ time, for some function $f$. Let $G_1,G_2,\ldots,G_q$ be the connected components of $G$. Run algorithm $\mathbb{A}$ on each connected component $G_i$, for $i\in [q]$, and a forest consisting of each subset of connected components of $H$. Note that thus we call $\mathbb{A}$ at most $2^k\cdot n$ times. Let the results of these calls be stored a function $\phi(\cdot,\cdot)$: $\phi(i,X)$, for $i\in [q]$ and $X\subseteq [k]$, is a boolean value denoting whether there exists a subgraph isomorphism from $\bigcup_{j\in X} H_j$ to $G_i$. Since each connected component $H_j$ must be mapped to one component $G_i$, we may now check whether there exists a subgraph isomorphism from $H$ to $G$ using standard dynamic programming on subsets. That is, we define function $\psi(i,X)$ for $i\in \{0\}\cup [q]$ and $X\subseteq [k]$ with the following meaning: $\psi(i,X)$ is a boolean value denoting whether there exists a subgraph isomorphism from $\bigcup_{j\in X} H_j$ to $\bigcup_{t=1}^t G_t$. Clearly, $\psi(0,X)$ is true if and only if $X=\emptyset$, and the seeken answer to the problem is precisely equal to the value $\psi(q,[k])$. However, it is easy to see that $\psi$ satisfies the following recurrential equation for $t>0$:
$$\psi(t,X)=\bigvee_{Y\subseteq X} \psi(t-1,X\setminus Y)\wedge \phi(t,Y).$$
Note that computation of values $\psi(t,\cdot)$ for some $t$ using precomputed values $\psi(t-1,\cdot)$ takes $O(3^k)$ time. Hence, we can compute values of function $\psi$ in a dynamic programming manner in $O(3^k\cdot n)$ time. For the analysis of the error probability, observe that if we repeat each run of $\mathbb{A}$ at least $1+\log (2^k\cdot n)=n^{O(1)}$ times, then the probability of error when computing each value $\phi(t,Y)$ will be bounded by $\frac{1}{2}\cdot \frac{1}{2^k\cdot n}$. Hence, by the union bound, we infer that the probability that any of the values $\phi(t,Y)$ will be computed incorrectly is bounded by $\frac{1}{2}$.

Assume then that $G$ is a tree, and let us root $G$ in an arbitrary vertex $r$. This imposes a natural parent-child relation in $G$, so for $v\in V(G)$ let $G_v$ be the maximal subtree of $G$ rooted in $v$. Now take any component $H_j$ of $H$, and examine an edge $wu\in E(H_j)$. We define $H_j^{(w,u)}$ to be the connected component of $H_j\setminus wu$ that contains $u$, rooted in $u$. For $u\in V(H_j)$, by $H_j^{(\e,u)}$ we denote simply $H_j$ rooted in $u$. Let $\mathcal{K}_j=\{(w,u),(u,w)\ |\ wu\in E(H_j)\}\cup \{(\e,u)\ |\ u\in V(H_j)\}$. For $u\in V(H)$, by $\iota(u)$ we denote the index of the connected component of $H$ containing $u$, i.e., $u\in V(H_{\iota(u)})$. 

Our goal is to compute functions $\xi(\cdot,\cdot)$ and $\zeta(\cdot,\cdot,\cdot)$ defined as follows:
\begin{itemize}
\item $\xi(v,X)$ for $v\in V(G)$ and $X\subseteq [k]$, is a boolean value denoting whether there exists a subgraph isomorphism from $\bigcup_{j\in X} H_j$ to $G_v$.
\item $\zeta(v,(w,u),X)$ for $v\in V(G)$, $u\in V(H)$, $(w,u)\in \mathcal{K}_{\iota(u)}$ and $X\subseteq [k]\setminus \{\iota(u)\}$, is a boolean value denoting whether there exists a subgraph isomorphism $\eta$ from $H_{\iota(u)}^{(w,u)}\cup \bigcup_{j\in X} H_j$ to $G_v$ such that $\eta(u)=v$.
\end{itemize}
It is clear that the values of functions $\xi(v,\cdot)$ and $\zeta(v,\cdot,\cdot)$ may be computed in constant time whenever $v$ is a leaf of $G$. On the other hand, the seeken answer to the problem is precisely equal to the value $\xi(r,[k])$. Therefore, we would like to compute all the values of functions $\xi$ and $\zeta$ in a bottom-up manner on the tree $G$. Let us concentrate on one vertex $v\in V(G)$; the goal is to compute all the values $\xi(v,\cdot)$ and $\zeta(v,\cdot,\cdot)$ in $f(k)\cdot n^{O(1)}$ time for some function $f$. Let $v_1,v_2,\ldots,v_a$ be the children of $v$ in $G$, ordered arbitrarily.

We start with computing the values $\xi(v,\cdot)$, assuming that values $\zeta(v,\cdot,\cdot)$ have been already computed. For $i\in \{0\}\cup [a]$ and $X\subseteq [k]$, let $\psi_v(i,X)$ be a boolean value denoting whether there exists a subgraph isomorphism from $\bigcup_{j\in X} H_j$ to $\bigcup_{t=1}^i G_{v_t}$. Observe that $\psi_v$ satisfies very similar equations as function $\psi$ from the beginning of the proof. We also have that $\psi_v(0,X)$ is true if and only if $X=\emptyset$, and for $t>0$ we have that $\psi_v$ satisfies equation:
$$\psi_v(t,X)=\bigvee_{Y\subseteq X} \psi_v(t-1,X\setminus Y)\wedge \xi(v_t,Y).$$
Hence, the values of function $\psi_v$ can be again computed using dynamic programming in $O(3^k\cdot n)$ time. Finally, observe that for any $X\subseteq [k]$, any subgraph isomorphism from $\bigcup_{j\in X} H_j$ to $G_v$ either maps some $u\in \bigcup_{j\in X} H_j$ to $v$, or does not and is therefore a subgraph isomorphism from $\bigcup_{j\in X} H_j$ to $G_v\setminus v=\bigcup_{t=1}^a G_{v_t}$. This implies the following recurrence equation:
$$\xi(v,X)=\psi_v(a,X)\vee \bigvee_{u\in \bigcup_{j\in X} V(H_j)} \zeta(v,(\e,u),X\setminus \iota(u)).$$
Observe that thus computation of one value $\xi(v,X)$ takes polynomial time. Therefore, it remains to show how to compute values of function $\zeta(v,\cdot,\cdot)$ using the precomputed values for children of $v$ in $G$.

Let us concentrate on one value $\zeta(v,(w,u),X)$. Let $j_0=\iota(u)$, and let $u_1,u_2,\ldots,u_b$ be the children of $u$ in $H_{j_0}^{(w,u)}$. We can assume that $b\leq a$, since otherwise we trivially have that $\zeta(v,(w,u),X)$ is false. We construct an auxiliary bipartite graph $B=([a],[a],E)$, where the edge set $E$ is defined as follows:
\begin{itemize}
\item Given indices $i$ and $j\leq b$, we put an edge between $i$ on the left side and $j$ on the right side if and only if value $\zeta(v_i,(u,u_{j}),\emptyset)$ is true.
\item For any pair of indices $i$ and $j>b$, we put an edge between $i$ on the left side and $j$ on the right side.
\end{itemize}
Intuitively, indices on the left side corresponds to subtrees (slots) $G_{v_1},\ldots,G_{v_a}$ into which subtrees of $H_{j_0}^{(w,u)}$ must be embedded. Indices $j\leq b$ on the right side correspond to subtrees $H_j$ that must be matched to appropriate slots, while indices $j>b$ correspond to choosing a slot to be empty.

To compute $\zeta(v,(w,u),X)$, we iterate through all the partitions of $X$ into nonempty subsets $X_1,X_2,\ldots,X_\ell$. Note that the number of these partitions is at most $k^k$, thus we obtain a $k^k$ overhead in the running time. For a given partition $X_1,X_2,\ldots,X_\ell$ we define edge sets $E_1,E_2,\ldots,E_\ell$ as follows:
\begin{itemize}
\item Given indices $i$ and $j\leq b$, we add $ij$ to $E_t$ if and only if value $\zeta(v_i,(u,u_{j}),X_t)$ is true.
\item Given indices $i$ and $j>b$, we add $ij$ to $E_t$ if and only if value $\xi(v_i,X_t)$ is true.
\end{itemize}
The following claim expresses the properties of graph $B$ that we wanted to achieve.

\begin{claim}\label{cl:tomatch}
Value $\zeta(v,(w,u),X)$ is true if and only if for at least one partition $X_1,X_2,\ldots,X_\ell$ the corresponding instance $(B,(E_t)_{t\in [\ell]})$ of {\sc{PMwHC}} is a YES-instance.
\end{claim}
\begin{proof}
Assume first that value $\zeta(v,(w,u),X)$ is true, and let $\eta$ be the witnessing subgraph isomorphism. Note that since $\eta(u)=v$, then the children of $u$ in $H_{j_0}^{(w,u)}$ must be mapped to children of $v$ in $G$. Let $X_1,X_2,\ldots,X_\ell$ be partition of $X$ with respect to the following equivalence relation: indices $j$ and $j'$ are equivalent if the images of $H_j$ and $H_{j'}$ in $\eta$ lie in the same subtree $G_{v_i}$. Let then $M$ be a perfect matching of $B$ defined as follows:
\begin{itemize}
\item for each $j\leq b$, let us match $j$ with $i$ if $\eta(u_j)=v_i$;
\item match the remaining indices of the left side with indices $j>b$ arbitrarily.
\end{itemize}
From the definition of graph $B$ it follows immediately that $M$ is indeed a perfect matching in $B$. We now need to distinguish edges $e_1,e_2,\ldots,e_\ell$ that certify that $M$ is a solution to the {\sc{PMwHC}} instance $(B,(E_t)_{t\in [\ell]})$. For $t\in \ell$, let $i_t$ be such an index that components with indices from $X_t$ are mapped by $\eta$ to subtree $G_{v_{i_t}}$, and let $j_t$ be the index matched to $i_t$ in $M$. We observe that $i_tj_t\in E_t$. Indeed, if $j_t\leq b$ then $\eta$ restricted to $H_{j_0}^{(u,u_{j_t})} \cup \bigcup_{j\in X_t} H_j$ is a subgraph isomorphism to $G_{v_{i_t}}$, which certifies that $\zeta(v_{i_t},(u,u_{j_t}),X_t)$ is true. On the other hand, if $j_t>b$ then $\eta$ restricted to $\bigcup_{j\in X_t} H_j$ is a subgraph isomorphism to $G_{v_{i_t}}$, which certifies that $\xi(v_{i_t},X_t)$ is true. Therefore, for partition $X_1,X_2,\ldots,X_\ell$ we have that edges $(i_tj_t)_{t\in [\ell]}$ witness that matching $M$ is a solution to the {\sc{PMwHC}} instance $(B,(E_t)_{t\in [\ell]})$.

Assume now that we are given partition $X_1,X_2,\ldots,X_\ell$ of $X$, a perfect matching $M$ in $B$ and distinct edges $e_1,e_2,\ldots,e_\ell$ of $M$ satisfying the property that $e_t\in E_t$ for $t\in [\ell]$, where sets $E_t$ are defined for partition $X_1,X_2,\ldots,X_\ell$. For $i\in [a]$ we define subgraph isomorphism $\eta_i$ as follows. Let $j$ be the index matched to $i$ in $M$. If $j\leq b$ then by the definition of graph $B$ we may take $\eta_i$ to be a subgraph isomorphism from $H_{j_0}^{(u,u_j)}$ to $G_{v_i}$ such that $\eta_i(u_j)=v_i$. Moreover, if $ij=e_t\in E_t$ for some $t\in [\ell]$, then we may in fact assume that $\eta_i$ is a subgraph isomorphism from $H_{i_0}^{(u,u_j)}\cup \bigcup_{j'\in X_t} H_{j'}$ to $G_{v_i}$ such that $\eta_i(u_j)=v_i$. Similarly, if $j>b$ and $ij=e_t\in E_t$ for some $t\in [\ell]$, then we take $\eta_i$ to be a subgraph isomorphism from $\bigcup_{j'\in X_t} H_{j'}$ to $G_{v_i}$, existing by the definition of $E_t$, and otherwise we take $\eta_i$ to be an empty function. Define $\eta=(u,v)\cup \bigcup_{i=1}^a \eta_i$. Since $M$ is a matching, $X_1,X_2,\ldots,X_\ell$ is a partition of $X$, and vertices $v_i$ for $i\in [a]$ are children of $v$ in $G$, it follows that $\eta$ defined in this manner is a subgraph isomorphism from $H_{j_0}^{(w,u)}\cup \bigcup_{j\in X} H_j$ to $G_v$ such that $\eta(u)=v$.
\cqed\end{proof}

Concluding, to compute every value $\zeta(v,(w,u),X)$, we construct the graph $B$ and iterate through all the partitions of $X$, for each constructing sets $E_1,E_2,\ldots,E_\ell$ and running the algorithm of Lemma~\ref{lem:pmwhc} on the instance $(B,(E_t)_{t\in [\ell]})$. Observe that $\ell\leq k$ and $B$ has polynomial size in terms of $n$, so each application of Lemma~\ref{lem:pmwhc} takes $2^{O(k)}\cdot n^{O(1)}$ time. Thus, computation of each value of $\zeta$ takes $2^{O(k\log k)}\cdot n^{O(1)}$ time, and each computation of values $\xi(v,\cdot)$ for one vertex $v$ takes $2^{O(k)}\cdot n^{O(1)}$ time. Since there are $2^{O(k)}\cdot n^{O(1)}$ values of $\zeta,\xi$ to be computed, the whole algorithm runs in $2^{O(k\log k)}\cdot n^{O(1)}$ time, which is compatible with the description. For the analysis of the error probability, observe that we run the algorithm of Lemma~\ref{lem:pmwhc} at most $2^{O(k\log k)}\cdot n^{O(1)}$ times in total. Therefore, if we repeat each run at least $1+\log (2^{O(k\log k)} \cdot n^{O(1)})=n^{O(1)}$ times, then by the union bound we can infer that the probability that any of the values $\zeta(v,(w,u),X)$ will be computed incorrectly is bounded by $\frac{1}{2}$.
\end{proof}

Whether it is possible to prove a deterministic version Theorem~\ref{thm:treesintotrees} remains a challenging question. For some of the problems that can be solved by the randomized algebraic techniques of \cite{MulmuleyVV87}, no deterministic polynomial-time algorithms were found, despite significant efforts. 
The question is whether the use of such matching algorithms is essential for Theorem~\ref{thm:treesintotrees}, or can be avoided by a different approach.

We finish this section by pointing out that the algorithm of
Lemma~\ref{lem:pmwhc} for {\sc{PMwHC}} can be used to prove the
fixed-parameter tractability of \textsc{Conjoining Bipartite
  Matching}, resolving an open problem of Sorge et
al.~\cite{DBLP:journals/jda/SorgeBNW12}. Given a bipartite graph $G$,
a partition $V_1$, $\dots$, $V_t$ of the vertex set $V(G)$, and a set
$F\subseteq \binom{[t]}{2}$ of pairs, a {\em conjoining matching} is a
perfect matching $M$ of $G$ such that for every pair $(x,y)\in F$, the
matching $M$ has an edge with endpoints in $V_x$ and $V_y$. If $(x_1,y_1)$, $\dots$, $(x_\ell,y_\ell)$ are the pairs in $F$, then we define $E_i$ as the set of edges between $V_{x_i}$ and $V_{y_j}$. Now the existence of a conjoining matching is equivalent to finding a perfect matching containing a distinct edge from each of the sets $E_1$, $\dots$, $E_\ell$. Therefore, Lemma~\ref{lem:pmwhc} implies the fixed-parameter tractability of the problem with parameter $\ell=|F|$.
\begin{theorem}\label{th:conjoining}
  Given a bipartite graph $G$ with a partition $V_1$, $\dots$, $V_t$,
  and a set $F$, the existence of a conjoining perfect matching can be
  tested in time $2^{O(|F|)}\cdot n^{O(1)}$.
\end{theorem}
Sorge et al.~\cite{DBLP:journals/jda/SorgeBNW12} formulated a minimum
cost version of the problem, where a cost function $c$ on the edges is
given, and the task is to find a conjoining perfect matching of total
cost at most $C$. We can generalize Theorem~\ref{th:conjoining} to the
minimum cost version if the costs are polynomially bounded integers
the following way. First, let us briefly sketch how to generalize
Lemma~\ref{lem:pmwhc} to a minimum cost version of \textsc{PMwHC},
where the task is to find a perfect matching of total cost at most $C$
and satisfying the hitting constraints. Let $Z$ be the maximum cost in
$B$. Then we define the weights of an edge $e$ of $B'$  as $\bar w(e)=w(e)(Zn+1)+c(e)$, where
$w(e)$ is the weight in the proof of Lemma~\ref{lem:pmwhc} and $c(e)$
is the cost of edge $e$. Now for any $x\le C$, graph $B$ has a perfect
matching of cost exactly $x$ and satisfying the hitting constraints if
and only if there is a perfect matching of weight exactly
$(2^{2\ell}-1)(Zn+1)+x$ in the constructed graph $B'$; this
correspondence allows us to solve the problem by applications of
Proposition~\ref{prop:wmatching}. With this generalization of
Lemma~\ref{lem:pmwhc} at hand, finding a minimum cost conjoining
perfect matching is immediate.
\begin{theorem}\label{th:conjoining-mincost}
  Given a bipartite graph $G$ with a cost function $c$ on the edges,
  an integer $C$, a partition $V_1$, $\dots$, $V_t$, and a set $F$,
  the existence of a conjoining perfect matching of cost at most $C$
  can be tested in time $2^{O(|F|)}\cdot (nZ)^{O(1)}$, if the costs
  are positive integers not larger than $Z$.
\end{theorem}


\section{Positive results: bounded degree and feedback vertex set}
\label{sec:big-planar-algorithm}
This section contains a sequence of results for the case when $G$ has
bounded degree and feedback vertex set number.  In Section~\ref{sec:structural-result}, we
present a simple structural characterization of such graphs. Roughly
speaking, such graphs can be formed from a bounded set of vertices by
connecting pairs of them with trees. While this structure seems very
restricted, it is still not strong enough for algorithmic purposes:
the only positive results we have that exploits this structure is for
the very restricted case when $\maxdeg(H)\le 2$, i.e., $H$ is a set of
paths and cycles.

The situation becomes significantly more interesting if we impose on
$G$, besides bounded degree and bounded feedback vertex set number,
the additional restriction that it is planar, or more generally, has
bounded genus or excludes a fixed minor. Due to an unexpected
interplay between planarity, bounded-degree, and the fact that $G$ is
composed from trees, we are able to give algorithms parameterized by
maximum degree, feedback set number, and genus of $G$ (for connected
$H$). The main technical engine in the proof can be explained
conveniently using the language of Constraint Satisfaction Problems
(CSPs). In Section~\ref{sec:csp}, we briefly overview the required
background on CSPs, introduce the notion of {\em projection sink}, and
state a result on solving CSPs having projection sinks; this result
might be of independent interest. Then in
Section~\ref{sec:from-subgr-isom-csp}, we show how to solve \subiso by
reducing the problem to a CSP having a projection
sink. Section~\ref{sec:exminors} generalizes the results to the case
when $G$ excludes a fixed minor. This generalization is highly
technical and requires the use of known powerful structure theorems
for the decomposition of such graphs. We have to deal with all the
technical details of such decompositions, including handling vortices
in graphs almost embeddable into some surface.

\subsection{A structural result}\label{sec:structural-result}
We need the following structural characterization of graphs with bounded degree and bounded feedback vertex set size:
\begin{lemma}\label{lem:fvsdegree}
  Given a graph $G$, in polynomial time we can compute a set $Z$ of
  $O(\maxdeg^2(G)\fvs(G))$ vertices such that every component $C$ of
  $G\setminus Z$ is a tree and there are at most two edges between $C$
  and $Z$. Furthermore, every vertex of $Z$ is adjacent to at most one vertex of each component $C$ of $G\setminus Z$.
\end{lemma}
\begin{proof}
  Let us compute a 2-approximate feedback vertex set $X$ using the
  algorithm of Becker and Geiger~\cite{DBLP:journals/ai/BeckerG96},
  that is, $|X|\le 2\fvs(G)$ and every component of $G\setminus X$ is
  a tree. Let $Y=X\cup N(X)$. Consider a component $C$ of $G\setminus
  Y$ and its neighborhood $L_C=N(C)\subseteq Y$. Note that $L_C$ is
  disjoint from $X$, since every neighbor of a vertex of $X$ is in
  $Y$. Therefore, $C\cup L_C$ is fully contained in a component of
  $G\setminus X$ and hence $G[C\cup L_C]$ is a tree. As $G[C]$ is connected,
  this means that the vertices of $L_C$ are leaves of the tree
  $G[C\cup L_C]$. Let $T_C$ be the minimal subtree of $G[C\cup L_C]$
  containing $L_C$. Let $B_C$ be the set of vertices of degree at
  least 3 in $T_C$. We obtain $Z$ from $Y$ by extending, for every
  component $C$ of $G\setminus Y$, the set $Y$ with $B_C$.  We have
  $|X|=O(\fvs(G))$, $|Y|=O(\maxdeg(G)\fvs(G))$, and the total degree
  of the vertices in $Y$ is $O(\maxdeg^2(G)\fvs(G))$. For a given
  component $C$, the size of $B_C$ is at most $|L_C|-2$. The sum of
  $|L_C|$ over all components $C$ is at most the total degree of $|Y|$
  (as $C$ has $|L_C|$ edges to $Y$), hence $|Z|=O(\maxdeg^2(G)\fvs(G))$.

  It is clear that every component of $G\setminus Z$ is a tree.
  Consider a component $C'$ of $G\setminus Z$, which is a subset of a
  component $C$ of $G\setminus Y$. Note that $C\cup N(C)=C\cup L_C$ is
  a subset of $X$, hence $G[C\cup L_c]$ is a tree. We show that there
  are at most two edges between $Z$ and $C'$.  Assume for
  contradiction that there are three edges $e_1$, $e_2$, $e_3$ between
  $C'$ and $Z$. As $G[C']$ is connected, this means that there is a
  vertex $v\in C'$ and there are three internally vertex disjoint
  paths $P_1$, $P_2$, $P_3$ such that they start at $v$, end in
  $N(C')$, and their last edges are $e_1$, $e_2$, $e_3$,
  respectively. The endpoints of these paths are distinct: otherwise,
  there would be a cycle in $C'\cup N(C')\subseteq C\cup L_C$, which
  induces a tree, a contradiction.  It follows that at least three
  components of $G[C\cup L_C]\setminus v$ contain vertices of
  $N(C')\subseteq L_C\cup B_C$. Every vertex of $B_C$ is on a path of
  $G[C\cup L_C]$ connecting two vertices of $L_C$, thus it follows
  that at least three components of $G[C\cup L_C]\setminus v$ contains
  vertices of $L_C$. In other words, $v$ is in $B_C\subseteq Z$, a contradiction.

  Suppose now that there is a vertex $z\in Z$ that is adjacent to two
  vertices $x,y\in C$ for some component $C$ of $G\setminus Z$. Then
  $z$ has to be in the feedback vertex set $X$, as there is a  cycle in
  $C\cup \{z\}$ that contains no other vertex of $Z$. This implies
  $x,y\in N(X)\subseteq Y$, contradicting that $x$ and $y$ are in
  $G\setminus Z$.
\end{proof}

The following positive result uses the structural result to embed paths and cycles into such graphs:
\begin{ptheorem}\label{th:cyclefvs}
There exists an algorithm compatible with the description
\begin{eqnarray*}
\sil{\maxdeg(G),\fvs(G),\ccn(H)}{}{\maxdeg(H)\le 2}
\end{eqnarray*}
\end{ptheorem}
\begin{proof}
  Let us compute the set $Z$ given by Lemma~\ref{lem:fvsdegree}. Note
  that, as $\maxdeg(H)\le 2$, at most two edges incident to each
  vertex of $Z$ is used by the solution. Therefore, we can the guess
  subset of the edges incident to $Z$  used by the solution (at most
  $(\maxdeg(G)^2+1)^{|Z|}$ possibilities) and
  we can remove the unused vertices and edges. Therefore, in the
  following, we can assume that every edge incident to $Z$ is used by
  the solution; in particular, every vertex of $Z$ has degree at most
  $2$. 
\begin{claim}\label{cl:uniquecycle}
Every connected component of $G$ has at most
  one cycle.
\end{claim}
\begin{proof}
  Consider a cycle $C$ in a connected component of $G$. As $Z$ is a
  feedback vertex set, cycle $C$ has to go through a nonempty subset
  $Z_0\subseteq Z$ of vertices. As each vertex of $Z$ has degree at
  most 2, cycle $C$ uses both edges incident to each vertex $z\in Z_0$
  and enters both components of $G\setminus Z$ adjacent to $z$. It
  follows that the connected component of $C$ consists of $Z_0$ and
  these $|Z_0|$ connected components of $G\setminus Z$. For any two
  vertices $z_1,z_2\in Z_0$ and component $K$ of $G\setminus Z$, there
  is a unique path from $z_1$ to $z_2$ via $K$. Thus every cycle of
  the component of $C$ has to go through the vertices of $Z_0$ in the
  same order, connected by paths exactly the same way as in $C$.
\cqed
\end{proof}
First we find a mapping for the components of $H$ that are
cycles. Clearly, each such component should be mapped to a component
of $G$ containing a cycle; as $Z$ is a feedback vertex set, there are
at most $|Z|$ such components. By Claim~\ref{cl:uniquecycle}, each such
component of $G$ has a unique cycle, thus by selecting the component
of $G$, we uniquely determine the image of a cycle of $H$ (if the
lengths are the same). Therefore, by guessing the component of each
cycle in $H$ (at most $|Z|^{\ccn(H)}$ possibilities), we can take care
of the cycles in $H$. After removing the cycles from $H$ and their
images from $G$, graph $H$ consists of only paths and it remains true
that every component of $G$ contains at most one cycle.

Next, we focus on the components of $G$ that still have a cycle; there
are at most $|Z|$ such components. For each such component $K$, we
guess the subset $\cK_H$ of components of $H$ (all of them are paths)
that are mapped to $K$ (at most $2^{\ccn(H)}$ possibilities) and check
whether there is a subgraph isomorphism from $\cK_H$ to $K$. Perhaps the
cleanest way to do this is to argue the following way. As each component of
what we want to map into $K$ is a path, one of the edges of the unique
cycle $C$ of $K$ is not used by the solution. Therefore, for every
edge $e\in C$, we use the algorithm of
Theorem~\ref{thm:strong-MT} to find a subgraph isomorphism from
$\cK_H$ to the tree $K\setminus e$. This takes time
$f(\maxdeg(H),\ccn(H))\cdot |V(G)|^{g(\tw(G))}=f'(\ccn(H))\cdot |V(G)|^{O(1)}$, with an additional factor of $O(|V(G)|)$
for trying all possible edges $e\in C$. Therefore, in time
$(2^{\ccn(H)})^{|Z|}\cdot f(\ccn(H))\cdot |V(G)|^{O(1)}$ we can handle
all components of $G$ containing a cycle. What remains is to find a
subgraph isomorphism from the remaining part of $H$ (which is a set of
paths) to the remaining part of $G$ (which is a forest). Again by
Theorem~\ref{thm:strong-MT}, this can be done in time
$f(\ccn(H))\cdot |V(G)|^{O(1)}$. Therefore, the claimed running time
follows.
\end{proof}

\subsection{Constraint satisfaction problems}
\label{sec:csp}

We briefly recall the most
important notions related to  CSPs.
\begin{definition}\label{def:csp}
An instance of a {\em constraint satisfaction problem} is a triple $(V ,D, C)$,
where:
\begin{itemize}
\item $V$ is a set of variables,
\item $D$ is a domain of values,
\item $C$ is a set of constraints, $\{c_1,c_2,\dots ,c_q\}$.
Each constraint $c_i\in C$ is a pair $\langle
s_i,R_i\rangle$, where:
\begin{itemize}
\item $s_i$ is a tuple of variables of length $m_i$, called the {\em constraint scope,} and
\item $R_i$ is an $m_i$-ary relation over $D$, called the {\em constraint
  relation.}
\end{itemize}
\end{itemize}
\end{definition}
For each constraint $\langle s_i,R_i\rangle$ the tuples of $R_i$
indicate the allowed combinations of simultaneous values for the
variables in $s_i$. The length $m_i$ of the tuple $s_i$ is called the
{\em arity} of the constraint.  A {\em solution} to a constraint
satisfaction problem instance is a function $f$ from the set of
variables $V$ to the domain of values $D$ such that for each
constraint $\langle s_i,R_i\rangle$ with $s_i = \langle
v_{i_1},v_{i_2},\dots,v_{i_m}\rangle$, the tuple $\langle f(v_{i_1}),
f(v_{i_2}),\dots,f(v_{i_m})\rangle$ is a member of $R_i$.  We say that
an instance is {\em binary} if each constraint relation is binary,
i.e., $m_i=2$ for each constraint (hence the term ``binary'' refers to the arity of the constraints and {\em not} to size of the domain).
Note that Definition~\ref{def:csp} allows that a variable appears multiple times in the scope of the constraint. Thus a binary instance can contain a constraint 
of the form $\langle (v,v),R\rangle$, which is essentially a unary constraint.
We will deal only with binary CSPs in this paper. We may assume that there is at most one constraint with the same scope, as two constraint $\langle s,R_1\rangle$ 
and $\langle s,R_2\rangle$ can be merged into a single constraint $\langle s,R_1\cap R_2\rangle$. Therefore, we may assume that the input size $|I|$ of a binary CSP instance is polynomial in $|V|$ and $|D|$, without going into the exact details of how the constraint are exactly represented.

The {\em primal graph} (or {\em Gaifman graph}) of a CSP instance
$I=(V ,D, C)$ is a graph with vertex set $V$ such that distinct vertices $u,v\in V$ are
adjacent if and only if there is a constraint whose scope contains
both $u$ and $v$. 
\begin{theorem}[\cite{Freuder90AA}]\label{th:freuder}
Given a binary CSP instance $I$ whose primal graph has treewidth $w$, a solution can be found in time $|I|^{O(w)}$.
\end{theorem}

The main new definition of the paper regarding CSPs is the notion of
projection graph.  We say that a constraint $\langle
(v_1,v_2),R\rangle$ is a {\em projection from $v_1$ to $v_2$} if for
every $x_1\in D$, there is at most one $x_2\in D$ such that
$(x_1,x_2)\in R$; projection from $v_2$ to $v_1$ is defined
similarly. The {\em projection graph} of a CSP instance is a directed
graph with vertex set $V$ such that there is a directed edge from $u$
to $v$ if and only if there is constraint that is a projection from
$u$ to $v$. Note that it is possible to have edges both from $u$ to
$v$ and from $v$ to $u$ at the same time in the projection
graph.

We say that a variable $v$ is a {\em projection source} if every other
variable can be reached from $v$ in the projection graph. Observe that
if the instance has a projection source $v$, then it can be solved in
polynomial time: setting a value $d\in D$ to $v$ gives at most one
possibility for every other variable, thus we can solve the instance
by trying $|D|$ possibilities for $v$. 

In a similar way, we say that a variable $v$ is a {\em projection
  sink} if it can be reached from every other variable in the
projection graph. Unlike the projection source, it is not clear how
the existence of a projection sink makes the problem any easier. In
fact, any instance can be easily modified (without changing
satisfiability) to have a projection sink: let us introduce a new
variable $v$ and, for every other variable $u$, add a constraint
$\langle (u,v),R\rangle$, where $R=D\times \{d\}$ for some fixed value
$d\in D$. Thus in general, we cannot hope to be able to exploit a
projection sink in any meaningful way. Nevertheless, our main
observation is that in planar graphs (and, more generally, in
bounded-genus graphs), the existence of a projection source
allows us to dramatically reduce treewidth and therefore to solve the
instance efficiently.

\begin{theorem} \label{th:cspprojsink}
Let $I$ be a CSP instance having a projection sink and let $g$ be the genus of the primal graph of $I$. Then $I$ can be solved in time $|I|^{O(g)}$.
\end{theorem}
\begin{proof}
  Let $G$ be the primal graph of $I$. We can determine the genus $g$
  of $G$ and find an embedding of $G$ in an appropriate surface
  $\Sigma$ in time $|V(G)|^{O(g)}$ using the algorithm of Filotti et
  al.~\cite{DBLP:conf/stoc/FilottiMR79}. (Note that a linear-time
  algorithm with double-exponential dependence on $g$ is also known
  \cite{DBLP:journals/siamdm/Mohar99}).

  We can build the projection graph and find a projection sink $v_0$
  in polynomial time. We can also find a subgraph $\overrightarrow{T}$
  of the projection graph such that $v_0$ is reachable from every
  vertex in $\overrightarrow{T}$ and the underlying undirected graph $T$ of
  $\overrightarrow{T}$ is a spanning tree of $G$. We obtain graph $G'$
  from $G$ by {\em cutting open the tree $T$} (see
  Figure~\ref{fig:cutopen}). This operation has been defined and used
  for planar graphs in, e.g.,
  \cite{DBLP:journals/talg/BorradaileKM09,KleSpanner2006,DBLP:conf/icalp/KleinM12}.
  Viewed as a planar embedded graph, walking along the boundary of the
  infinite face is an Euler tour of $T$ that traverses each edge once
  in each direction.  We duplicate every edge of $T$ and replace each
  vertex $v$ of $T$ with $d_T(v)$ copies, where $d_T(v)$ is the degree
  of $v$ in $T$; this transforms the Euler tour described above into a
  cycle enclosing a new face $F$. The edges incident to $v$, but not
  in $T$ can be distributed among the copies of $v$ in a way that
  respects the embedding (see Figure~\ref{fig:cutopen}).

\begin{figure}
\begin{center}
{\small 
                \def\svgwidth{0.6\linewidth}
                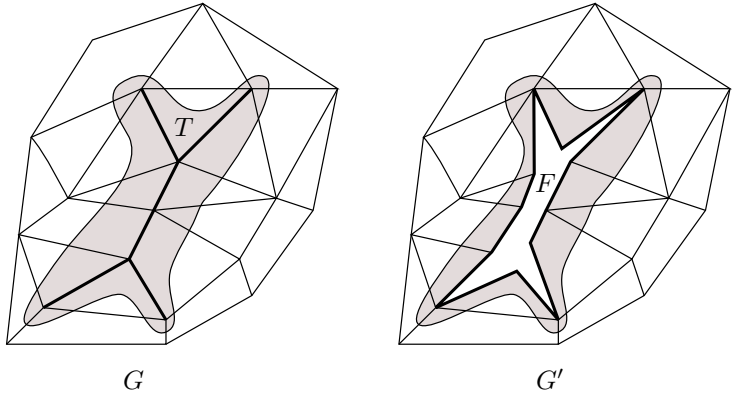
}
\end{center}
\caption{Cutting open a tree $T$ in a planar graph, creating a new face $F$. Note that the operation changes the embedding only inside the gray region.}\label{fig:cutopen}
\end{figure}

  For a graph embedded in a surface $\Sigma$, we can define the
  operation of cutting $T$ open the same way. To see that no further
  technical complications arise compared to the planar case, observe
  that (as $T$ is a tree) there is a region $R$ of $\Sigma$
  homeomorphic to a disc that strictly encloses every vertex and edge
  of $T$ (see the gray area in Figure~\ref{fig:cutopen}). The
  operation of cutting $T$ open can be performed in such a way that
  the embedding is changed only inside $R$, thus it is of no
  consequence that the embedding is not planar outside $R$.

  Let us show that $G'$ has bounded treewidth. As $T$ is a spanning tree,
  every vertex of $G'$ is on the boundary of the face $F$. Therefore,
  if we introduce a new vertex $w$ adjacent to every vertex of $G'$,
  then the resulting graph $G''$ can be still embedded in $\Sigma$ (we
  can embed $w$ in $F$). A result of Eppstein
  \cite[Theorem 2]{eppstein} states that graph with genus $g$ and
  diameter $D$ has treewidth $O(gD)$. Graph $G''$ has diameter at most
  2 (every vertex is a neighbor of $w$), hence a bound of $O(g)$
  follows on the treewidth of $G''$ and its subgraph $G'$.

  We construct a CSP instance $I'=(V',D,C')$ with primal graph $G'$ in
  a fairly straightforward way. Let $V'$ be the set of vertices of
  $G'$. We define a mapping $\phi:V'\to V$ with the meaning that $v\in
  V'$ is a copy of vertex $\phi(v)\in V$. If $u$ and $v$ are adjacent in $G'$,
  then $\phi(u)$ and $\phi(v)$ are adjacent in $G$, hence there is a
  constraint $\langle (\phi(u),\phi(v)),R\rangle$ in $C$. In this case, we
  introduce the constraint $\langle (u,v),R\rangle$ in $C'$.

  We claim that $I$ has a solution if and only if $I'$ has. If this is
  true, then we can solve $I$ by solving $I'$ in time
  $|I|^{O(\tw(G'))}=|I|^{O(g)}$ using the algorithm of
  Theorem~\ref{th:freuder}.

  If $I$ has a solution $f$, then it is easy to see that
  $f'(v)=f(\phi(v))$ is a solution of $I'$: by definition, this
  satisfies every constraint $\langle (u,v),R\rangle$ introduced in
  $C'$. The main complication in the proof of the reverse direction is that we have
  to argue that copies of the same variable receive the same value,
  that is, if $\phi(u_1)=\phi(u_2)$, then $f'(u_1)=f'(u_2)$ in every satisfying assignment $f'$
  of $I'$. The following claim proves this by induction on the distance between $u_1$ and $u_2$ on the boundary of $F$.
\begin{claim}\label{cl:samevalue}
  Let $P$ be subpath of the boundary of $F$ with endpoints $u_1$ and $u_2$
  satisfying $\phi(u_1)=\phi(u_2)$ such that $P$ has no internal vertex
  $w$ with $\phi(w)=v_0$. Then $f'(u_1)=f'(u_2)$ in every satisfying
  assignment $f'$ of $I'$.
\end{claim}
\begin{proof}
  We prove the statement by induction on the length of $P$; if $P$
  consists of a single vertex, then the claim is trivial.  Let
  $u=\phi(u_1)=\phi(u_2)$. If $P$ has an internal vertex $w$ with
  $\phi(w)=u$, then the induction hypothesis implies $f'(u_1)=f'(w)$ and $f'(w)=f'(u_2)$,  and we are done.  Let $u_1^*$ and $u_2^*$ be
  the neighbor of $u_1$ and $u_2$ in $P$, respectively (it is possible
  that $u_1^*=u_2^*$). From the fact that $P$ has no internal vertex
  $w$ with $\phi(w)=u$, it follows that $\phi(u^*_1)=\phi(u^*_2)$: the
  path $P$ describes a part of the Euler walk starting and ending in
  $u$ with no further visits to $u$, hence the second and the
  penultimate vertices of this walk have to be the same neighbor of $u$
  in $T$.  Let $u^*=\phi(u^*_1)=\phi(u^*_2)$. The subtree of
  $T\setminus u$ containing $u^*$ does not contain $v_0$, otherwise
  the walk corresponding to $P$ would visit $v_0$. Therefore, in $\overrightarrow{T}$
  the edge between $u$ and $u^*$ is directed towards $u$. That is,
  there is a constraint $\langle (u^*,u),R\rangle$ or $\langle (u,u^*),R\rangle$ that is a
  projection from $u^*$ to $u$; assume without loss of generality the former case. Therefore, by construction, $C'$
  contains the constraints $\langle (u^*_1,u_1),R\rangle$ and $\langle
  (u^*_2,u_2),R\rangle$. By the induction hypothesis,
  $f'(u^*_1)=f'(u^*_2)$. Thus the projection from $u^*_1$ to $u_1$ and
  the projection from $u^*_2$ to $u_2$ in $C'$ implies
  $f'(u_1)=f'(u_2)$, what we had to show.  \cqed\end{proof}

Claim~\ref{cl:samevalue} immediately implies that variables $u$ with
$\phi(u)=v_0$ have the same value in $f'$: the claim shows this for
the $i$-th and $(i+1)$-st such vertex on the cycle, which implies that
all of them have the same value. Consider now two arbitrary variables
$u_1$ and $u_2$ of $V'$ with $u=\phi(u_1)=\phi(u_2)$ and $u\neq
v_0$. Let $P_1$ and $P_2$ be the two paths on the boundary of $F$
connecting $u_1$ and $u_2$. Observe that it is not possible that both
$P_1$ and $P_2$ contain a vertex $w_1$ and $w_2$, respectively, with
$\phi(w_1)=\phi(w_2)=v_0$: the appearances of different vertices
in the Euler tour have to be nested, that is, vertices $u$, $v_0$,
$u$, $v_0$ cannot appear in this order in the tour.  Therefore, we can
apply the claim on either $P_1$ or $P_2$ to conclude that
$f'(u_1)=f'(u_2)$.

We have shown that every variable $u'$ with $\phi(u')=u$ has the same
value in $f'$; let us define $f(u)$ to be this value. We claim that
$f$ satisfies $I$. Consider a constraint $\langle (u,v),R\rangle$ in
$C$; then there is a corresponding edge $uv$ in $G$. For each edge
$uv$ of $G$ not in $T$, there is a single corresponding edge in $G'$;
and for every $uv$ in $T$, there are two corresponding edges in
$G'$. Let $u'v'$ be an edge of $G'$ corresponding to $uv$, that is,
$\phi(u')=u$ and $\phi(v')=v$. By construction, the constraint
$\langle (u',v'),R\rangle$ appears in $C'$. This means that
$(f'(u'),f'(v'))=(f(u),f(v))$ is in $R$, i.e., $\langle
(u,v),R\rangle$ is satisfied. Thus we have shown that if $I'$ has a
solution, then $I$ has a solution as well.
\end{proof}

\subsection{From subgraph isomorphism to CSP}
\label{sec:from-subgr-isom-csp}

We remove a layer of technical complications by proving the result
first under the assumption that $H$ is connected. Later we present a
clean reduction from the general case to the connected case.
\begin{lemma}\label{lem:bigplanarconnected}
There exists an algorithm compatible with the description
\begin{eqnarray*}
\sil{\maxdeg(G),\fvs(G)}{\genus(G)}{\ccn(H)\le 1}
\end{eqnarray*}
\end{lemma}
\begin{proof}
  Let us compute the set $Z$ given by Lemma~\ref{lem:fvsdegree}. We
  can guess which subset of $Z$ ($2^{|Z|}$ possibilities) and which
  subset of the edges incident to $Z$ (at most $2^{|Z|\maxdeg(G)}$
  possibilities) are used by the solution and we can remove the unused
  vertices and edges. Therefore, in the following, we can assume that
  every edge incident to $Z$ is used by the solution. If $Z$ is empty,
  then graph $G$ is a forest; as $H$ is connected, then the task is to find a tree in a tree and Theorem~\ref{th:matula} can be used. Let
  $z_0$ be an arbitrary vertex of $Z$ and let us guess the vertex
  $v_0=\eta^{-1}(z_0)$ in $H$ (thus we repeat the algorithm described
  below $|V(H)|$ times, for each possible choice of $v_0$).  We may
  also assume that $G$ is connected: as $H$ is connected, vertices of
  $G$ not in the same component as $z_0$ can be safely deleted.

  By Vizing's Theorem, the edges of $H$ can be colored with
  $\maxdeg(H)+1\le \maxdeg(G)+1$ colors; moreover, this coloring may be found in polynomial time \cite{DBLP:journals/ipl/MisraG92}. For every edge $e$ of $G$
  incident to $Z$, we guess the color $c(e)$ of the edge of $H$ mapped
  to $e$ in the solution ($(\maxdeg(G)+1)^{|Z|\maxdeg(G)}$
  possibilities).  Note that we do not define any color for the edges
  of $G$ not incident to $Z$.  We say that a subgraph isomorphism
  $\eta$ {\em respects} the edge coloring $c$ if for every edge $e$ of
  $G$ incident to $Z$, there is an edge $e'\in E(H)$ of color $c(e)$
  such that $\eta$ maps $e'$ to $e$.

  We use the technique of color-coding to further restrict the
  structure of the solutions that we need to find. Let us assign a
  distinct label $\lambda(z)\in \{1,\dots,|Z|\}$ to each vertex $z\in
  Z$ and consider a labeling $\lambda(v)\in \{1,\dots,|Z|\}$ for each
  vertex $v\in V(H)$ (note that $\lambda$ is not necessarily a proper
  coloring of $H$).  We say that a subgraph isomorphism $\eta$ {\em
    respects} the labeling $\lambda$ if $\lambda(\eta(v))=\lambda(v)$
  for every $v\in V(H)$ with $\eta(v)\in Z$.  We will restrict our
  attention to solutions $\eta$ respecting $\lambda$. Of course, it is
  possible that for a particular labeling $\lambda$, there is no such
  solution $\eta$.  However, the standard technique of $k$-perfect
  hash functions allows us to construct a family of $k^{O(k)}\cdot
  \log |V(H)|$ labelings such that every solution $\eta$ respects at
  least one of the mappings in the family
  \cite{DBLP:journals/jacm/AlonYZ95}. Therefore, we repeat the
  algorithm described below for each mapping $\lambda\in {\cal{F}}$;
  if there is a solution, the algorithm will find a solution for at
  least one of the mappings in the family.

We say that a subgraph isomorphism
  $\eta$ from $H$ to $G$ {\em respects the colors of $H$ and $G$} if it respects both the edge coloring $c$ and the mapping $\lambda$, and furthermore, it maps $v_0$ to $z_0$. That is,
\begin{itemize}
\item  for every
  edge $e$ of $G$ incident to $Z$, there is an edge $e'\in E(H)$ of
  color $c(e)$ such that $\eta$ maps $e'$ to $e$, 
\item $\lambda(\eta(v))=\lambda(v)$ for every $v\in V(H)$ with $\eta(v)\in Z$, and
\item $\lambda(v_0)=z_0$.
\end{itemize}
By our exhaustive search of all possibilities for the coloring of the
edges incident to $Z$ and by the properties of the family ${\cal{F}}$ of labelings
$\lambda$ we consider, we may restrict our search to mappings $\eta$
that respect the colors of $H$ and $G$.

For vertices $v_1,v_2\in V(H)$ and colors $c_1,c_2\in \{1,\dots,\maxdeg(H)+1\}$, we say that subgraph $K$ of $H$ is a {\em $(v_1,v_2,c_1,c_2)$-block} if
\begin{enumerate}
\item vertices $v_1$ and $v_2$ have degree 1 in $K$,
\item the edge of $K$ incident to $v_i$ has color $c_i$ for $i=1,2$,
\item for every $u\in V(K)\setminus \{v_1,v_2\}$, every edge of $H$ incident to $u$ is in $K$.
\end{enumerate}
The vertices $v_1,v_2$ are the {\em join} vertices of $K$; every other
vertex is an {\em internal} vertex.  Note that, in particular, if $e$
is an edge of color $c$ connecting $c_1$ and $c_2$, then $e$ forms a
$(v_1,v_2,c,c)$-block. We allow $v_1=v_2$, that is, $K$ has only a
single join vertex (which implies $c_1=c_2$, since there is only a
single edge incident to the join vertex). As $H$ is connected, there
is at most one $(v_1,v_2,c_1,c_2)$-block for a given $v_1,v_2,c_1,c_2$
and we can find it in polynomial time, if it exists. Furthermore, it
is also true that the $(v_1,v_2,c_1,c_2)$-block has at most two
components (as each component of $K$ contains either $v_1$ or $v_2$).
In summary, the $(v_1,v_2,c_1,c_2)$-blocks of $H$ can be classified into three
types: (i) $v_1=v_2$, (ii) $v_1\neq v_2$ and the block is connected,
and (iii) $v_1\neq v_2$ and the block contains two components.

 The definition of $(z_1,z_2,c_1,c_2)$-blocks of
$G$ is the same, but we define it only for $z_1,z_2\in Z$, as only the
edges incident to $Z$ have colors.  Furthermore, we introduce the additional
condition
\begin{enumerate}
\item[4.] $K$ has no vertex in $Z\setminus \{v_1,v_2\}$, that is, no internal vertex of $K$ is in $Z$.
\end{enumerate}
Let $\cK=\{K_G^1,\dots,K_G^{t}\}$ be the set of all {\em connected} $(z_1,z_2,c_1,c_2)$-blocks of $K_G$.
\begin{claim}\label{cl:blocks}
Every block in $\cK$ is a tree and 
the blocks in $\cK$ form a partition of the edge set of $G$.
\end{claim}
\begin{proof}
  For every $(z_1,z_2,c_1,c_2)$-block $K_G\in \cK$, graph
  $K_G\setminus \{z_1,z_2\}$ is connected (since vertices $z_1,z_2$
  have degree 1 in $K_G$) and, since it does not contain any vertex of
  $Z$,  is a tree. It follows that $K_G$ is a tree as well.

  Suppose that two blocks $K_G^{i_1},K_G^{i_2}\in \cK$ share
  edges. Then there is a vertex $v$ and two edges $e_1$ and $e_2$
  incident to $v$ such that $e_1$ appears only in, say, $K_G^{i_1}$,
  but $e_2$ appears in both blocks. Then $v$ has degree at least 2 in
  $K_G^{i_1}$, hence it is an internal vertex of $K_G^{i_1}$, implying
  $v\not\in Z$. It follows that $v$ is an internal vertex of
  $K_G^{i_2}$ as well, thus every edge incident to $v$ is in
  $K_G^{i_2}$, including $e_1$, a contradiction.

  We claim that every vertex $v\not\in Z$ is an internal vertex of at
  least one block in $K_G$. Consider the component $C$ of $G\setminus
  Z$ containing $v$. As $G$ is connected, the construction of $Z$ using Lemma~\ref{lem:fvsdegree}
  implies that there are at most two edges between $C$ and
  $Z$. Component $C$, together with these at most two edges form a
  block in $\cK$ (note that Lemma~\ref{lem:fvsdegree} implies that if
  there are two edges between $C$ and $Z$, then they have distinct
  endpoints in $Z$). Thus every edge incident to a vertex not in $Z$
  is in some block of $\cK$. An edge $e$ of color $c$ between two vertices  $z_1,z_2\in Z$ forms
  a $(z_1,z_2,c,c)$-block that is in $\cK$. Thus every edge of $G$ is in some block of
  $\cK$.\cqed
\end{proof}

We construct a CSP instance $I=(Z,V(H),C)$ (i.e., $Z$ is the set of
variables and the value of each variable is a vertex in $H$). In a
satisfying assignment of $I$, the intuitive meaning of the value of
$z\in Z$ is the vertex of $H$ mapped to $z$ by a solution $\eta$. We
introduce constraint to enforce this interpretation: if these values
imply that a block $K_H$ of $H$ is mapped to a block $K_G$ of $G$,
then our constraints ensure that $K_H$ is a subgraph of $K_G$.

First, we introduce unary constraints that enforce that the solution
respects the colors of $H$ and $G$. For every $z\in Z$, we introduce a
unary constraint $\langle (z),R\rangle$ such that $v\in R$ if
and only if
\begin{itemize}
\item $\lambda(v)=\lambda(z)$, 
\item the same set of colors appears on the edges incident to $v$ in $H$ and on the edges incident to $z$ in $G$, and
\item if $z=z_0$, then $v=v_0$.
\end{itemize}
Next, for
every $(z_1,z_2,c_1,c_2)$-block $K_G$ in $\cK$, we introduce a
constraint $\langle (z_1,z_2),R\rangle$. The relation $R$ is defined in
the following way: for $v_1,v_2\in V(H)$, we have $(v_1,v_2)\in R$ if
and only if
\begin{itemize}
\item the $(v_1,v_2,c_1,c_2)$-block $K_H$ of $H$ exists,
\item  $v_0$ is
  not an internal vertex of $K_H$, that is, $v_0\not\in
  V(K_H)\setminus \{v_1,v_2\}$, and
\item there is a subgraph isomorphism from $K_H$ to $K_G$ that maps $v_1$ to $z_1$ and
maps $v_2$ to $z_2$.
\end{itemize}
The existence of the required subgraph isomorphism from $K_H$ to $K_G$
can be decided as follows. Recall that $K_G$ is a tree and $K_H$ has
at most two components. Therefore, Theorem~\ref{thm:strong-MT} can be
used to decide if there is a subgraph isomorphism from $K_H$ to $K_G$
in time $f(\ccn(K_H),\maxdeg(K_H))\cdot
|V(K_G)|^{\tw(K_G)}=f'(\maxdeg(G))\cdot |V(G)|^{O(1)}$. To enforce
that $v_1$ is mapped to $z_1$, we may add a cycle of length 3 to $v_1$
in $K_H$ and to $z_1$ in $K_G$; similarly, we may add a cycle of
length 4 to both $v_2$ and $z_2$. This increases maximum degree and
treewidth of $K_G$ by at most 1. (Alternatively, we may use
Lemma~\ref{sec:fixing-imag-prescr} to fix the images of $v_1$ and
$v_2$ to $z_1$ and $z_2$ respectively.)  Therefore, we can construct
each constraint relation of the instance in time $f(\maxdeg(G))\cdot
|V(G)|^{O(1)}$.

The following claim shows that we can solve the subgraph isomorphism
problem by solving $I$.  However, at this point it is not clear how to
solve $I$ efficiently: we will need further simplifications to make it
more manageable. In particular, we will modify it in a way that
creates a projection sink.
\begin{claim}\label{cl:tocsp}
$I$ has a solution  if and only if there is subgraph isomorphism $\eta$ from $H$ to $G$ that respects the colors of $H$ and $G$.
\end{claim}
\begin{proof}
Let $\eta$ be a subgraph isomorphism from $H$ to $G$ respecting the
colors of $H$ and $G$. In particular, every vertex of $Z$ is used by
$\eta$. We claim that $f(z)=\eta^{-1}(z)$ is a satisfying assignment
of $I$. Clearly, $f$ satisfies all unary constraints since $\eta$ respects colors of $H$ and $G$.
Consider a $(z_1,z_2,c_1,c_2)$-block $K_G$ of $\cK$. Let $K_H$
be the subgraph of $H$ spanned by the edges of $H$ that are mapped to
$E(K_G)$ by $\eta$. It can be observed that $K_H$ is nonempty an in
fact it is a $(v_1,v_2,c_1,c_2)$-block of $H$ for $v_1=\eta^{-1}(z_1)$
and $v_2=\eta^{-1}(z_2)$. Therefore, $(f(z_1),f(z_2))=(v_1,v_2)$ 
satisfies the constraint $\langle (z_1,z_2),R\rangle$ introduced by
$K_G$.

Suppose now that $f:Z\to V(H)$ is a solution of $I$. We construct a
subgraph isomorphism $\eta$ from $H$ to $G$ in the following way. First,
let $X=f(Z)$ and set $\eta(v)=f^{-1}(v)$ for every $v\in X$ (note that
$f$ is injective on $Z$ as the unary constraints ensure that
$\lambda(f(v))=\lambda(v)$ for every $z\in Z$; this point of the proof
is the reason for using color-coding).  By the definition of the
constraints in $I$, if $K_G^i$ is a $(z_1,z_2,c_1,c_2)$-block of
$\cK$, then $H$ has an $(f(z_1),f(z_2),c_1,c_2)$-block $f(K_G^i)$ and
a subgraph isomorphism $\eta_i$ from $f(K_G^i)$ to $K_G^i$ that maps
$f(z_1)$ to $z_1$ and $f(z_2)$ to $z_2$.

Our goal is to show that the edges of $f(K_G^1)$, $\dots$, $f(K_G^t)$
form a partition of the edge set of $H$.  First we show that every
edge of $H$ is in at least one of these blocks. If edge $e$ is
incident to a vertex $f(z_1)\in X$ and has color $c_1$, then this is
easy to see: the unary constraint on $z_1$ ensures that $z_1$ also has
an edge of color $c_1$, hence there is a $(z_1,z_2,c_1,c_2)$-block
$K_G^i$ in $\cK$ for some $z_2,c_2$ and the corresponding block
$f(K_G^i)$ contains $e$. If $e$ is not incident to $X$, then, since $H$ is connected, there is
a path $P$ whose first edge is $e$, its last edge $e'$ is incident to
$X$, and that has exactly one vertex in $X$. Then path $P$ is fully
contained in the block $f(K_G^i)$ of $H$ that contains $e'$.

To show that every edge is in exactly one of $f(K_G^1)$, $\dots$,
$f(K_G^t)$, let $E'\subseteq E(H)$ be the set of edges that are in
more than one of these blocks.  We claim that if an edge $e_1$ is in
$E'$, then every edge $e_2$ sharing an endpoint $v$ with $e_1$ is also
in $E'$. Let the colors of $e_1$ and $e_2$ be $c_1$ and $c_2$,
respectively.  Suppose that $e_1$ is an edge of both $f(K_G^{i_1})$
and $f(K_G^{i_2})$. If $v$ is an internal vertex of both
$f(K_G^{i_1})$ and $f(K_G^{i_2})$, then every edge of $H$ incident to
$v$ is in both blocks, hence $e_2$ is also in $E'$, and we are
done. Suppose without loss of generality that $v$ is a join vertex
of $f(K_G^{i_1})$, which means that $v=f(z)$ for some $z\in Z$ and
$K_G^{i_1}$ contains the edge of color $c_1$ incident to $z$. This
means that $v$ cannot be a join vertex of $f(K_G^{i_2})$, as this
would imply that $K_G^{i_2}$ also contains the edge of color $c_1$
incident to $z$, but Claim~\ref{cl:blocks} states that the blocks in
$\cK$ are edge disjoint, a contradiction. Therefore, $v$ is an
internal vertex of block $f(K_G^{i_2})$, hence it also contains
$e_2$. Furthermore, there is a $(z,z',c_2,c')$-block $K_G^{i_3}$ of
$G$ that contains the edge of color $c_2$ incident to $z$, hence
$f(K_G^{i_3})$ also contains $e_2$. Note that $f(K_G^{i_2})$ and
$f(K_G^{i_3})$ are distinct blocks since $z$ is an internal vertex of $f(K_G^{i_2})$ and a join vertex of $f(K_G^{i_3})$; hence, $e_2\in E'$
follows. Therefore, if $E'$ is not empty, then, as $H$ is connected by
assumption, every edge of $H$ is in $E'$. This is now the point where
the requirement $v_0\not\in V(K_H)$ in the definition of the binary
constrains becomes important. Let $e$ be an arbitrary edge of $H$
incident to $v_0$, having color $c$. No $f(K_G^{i})$ contains $v_0$ as
an internal vertex (this is ensured by the way the constraints are
defined). Therefore, if $f(K_G^{i})$ contains $e$, then $v_0$ is a
join vertex of $f(K_G^{i})$, which means that $K_G^{i}$ contains the
edge of color $c$ incident to $z_0$. There is a unique such $K_G^{i}$
(as the blocks in $\cK$ are edge disjoint by Claim~\ref{cl:blocks}),
hence $e$ is not in $E'$, implying that $E'$ is empty.

A vertex $v\in V(H)\setminus X$ can be only the internal vertex of a
block $f(K_G^i)$, hence the fact that the blocks $f(K_G^1)$, $\dots$,
$f(K_G^t)$ partition of the edge set of $H$ implies that $v$ is in exactly
one block $f(K_G^i)$. We define $\eta(v)=\eta_i(v_i)$, where $\eta_i$
is the mapping from $f(K_G^i)$ to $K_G^i$. We claim that $\eta$ is a
subgraph isomorphism from $H$ to $G$ respecting the colors. The definition 
of $f$ ensures that $\eta$ respects colors of $H$ and $G$. 
It is also clear that if $v_1,v_2\not\in X$ are adjacent, then they are in the
same block $f(K_G^i)$, implying that
$\eta(v_1)\eta(v_2)=\eta_i(v_1)\eta_i(v_2)$ is an edge of $G$. Suppose
now that $v_1\in X$ (note that vertex $v_2$ can be also in $X$), edge
$v_1v_2$ has color $c_1$, and edge $v_1v_2$ is in $f(K_G^i)$. As
$\eta(v_1)=\eta_i(v_1)$, we have that $\eta$ maps $v_1v_2$ to the edge
$\eta_i(v_1)\eta_i(v_2)$ of $G$. \cqed\end{proof}

Having proved the equivalence of $I$ with finding $\eta$, our goal is
now to modify $I$ such that it has a projection sink.  Let us fix
an arbitrary spanning tree $T$ of $H$ (recall that $H$ is connected by
assumption). Let $K_G$ be a $(z_1,z_2,c_1,c_2)$-block in $\cK$. We say
that subgraph isomorphism $\eta$ of $H$ to $G$ {\em directs $K_G$ from
  $z_1$ to $z_2$} if
\begin{itemize}
\item $\eta^{-1}(z_2)$ is on the path from $\eta^{-1}(z_1)$ to $v_0$
  in $T$, and
\item for the path $P$ from $\eta^{-1}(z_1)$ to $\eta^{-1}(z_2)$ in $T$, we
  have that every edge of $\eta(P)$ is in $K_G$.
\end{itemize}
Observe that if we know that $\eta$ directs $K_G$ from $z_1$ to $z_2$,
then in the corresponding constraint of $I$, the value $f(z_1)$
uniquely determines $f(z_2)$: if the distance of $z_1$ and $z_2$ in
$K_G$ is $d$, then $f(z_2)$ is the vertex at distance $d$ from
$f(z_1)$ in $T$ on the path towards $v_0$. Thus if we guess that
solution $\eta$ directs $K_G$ from $z_1$ to $z_2$, then we can
restrict the corresponding constraint to be a projection from $z_1$ to
$z_2$. The following claim shows that every solution $\eta$ directs a
large number of blocks, which means that the instance can be
simplified significantly.
\begin{claim}\label{cl:projsink}
  Let $\eta$ be a subgraph isomorphism from $H$ to $G$ respecting the
  colors of $H$ and $G$. Then for every $z\in Z$, there is a sequence $z=z_p$,
  $z_{p-1}$, $\dots$, $z_1$, $z_0$ such that for every $1\le i \le p$,
  there is a block in $\cK$ that is directed from $z_i$ to $z_{i-1}$ by
  $\eta$.
\end{claim}
\begin{proof}
  Let $P$ be the path from $\eta^{-1}(z)$ to $v_0$ in $T$. Let
  $z=z_p$, $z_{p-1}$, $\dots$, $z_1$, $z_0=\eta(v_0)$ be the vertices
  of $Z$ as they appear on $\eta(P)$. All the edges of the subpath of
  $\eta(P)$ between $z_i$ and $z_{i-1}$ are contained in a single
  block of $\cK$: no internal vertex of this subpath can be a join
  vertex of a block of $\cK$.  If this subpath is in a
  $(z_i,z_{i-1},c_1,c_2)$-block $K_G\in \cK$, then $\eta$ directs
  $K_G$ from $z_i$ to $z_{i-1}$ by definition.
\end{proof}
Our algorithm branches on which blocks of $\cK$ are directed by the
solution $\eta$ (note that $|\cK|\le |Z|\maxdeg(G)$ as every block has
an edge incident to $Z$). That is, for every $(z_1,z_2,c_1,c_2)$-block
$K_G$ of $G$, we branch on whether $\eta$ directs it from $z_1$ to
$z_2$, directs it from $z_2$ and $z_1$, or does not direct $K_G$. In
the first two cases, we modify the constraint of $I$ corresponding to
$K_G$ accordingly, and it becomes a projection.  After making these
choices for every block $K_G$, we check if the choices satisfy the
property of Claim~\ref{cl:projsink} (this can be checked in polynomial
time). If they don't, we terminate this branch of the
algorithm. Otherwise, let $I'$ be the modified instance obtained
performing all the restrictions corresponding the choices. If the
property of Claim~\ref{cl:projsink} holds, then $z_0$ is a projection
sink of the modified CSP instance $I'$. Therefore, we can use the
algorithm of Theorem~\ref{th:cspprojsink} to solve $I'$.

Note that the primal graph $G'$ of $I'$ is a minor of $G$: for every
$(z_1,z_2,c_1,c_2)$-block $K_G$ in $\cK$, we introduce a constraint on
$z_1$ and $z_2$, and there is a path $P_{K_G}$ between $z_1$ and $z_2$
using only the edges of $K_G$. Let us associate this path with the
constraint on $z_1$ and $z_2$; the paths corresponding to different
constraints are internally disjoint (every vertex not in $Z$ belongs
to only one block). This shows that the primal graph $G'$ of $I'$ is a
minor $G$, hence its genus is not larger than the genus of $G$
(deletions and contractions do not increase the genus). The running
time of the algorithm of Theorem~\ref{th:cspprojsink} is therefore
$|I'|^{O(\genus(G'))}=|I'|^{O(\genus(G))}$. Note that $|I|$ and $|I'|$
are polynomially bounded by the size $n$ of the \subiso. To bound the
total running time, we have to take into account the number of
branches considered by the algorithm: branching on the subset of $Z$
used by the solution ($2^{|Z|}$ choices), choosing the vertex $v_0$
corresponding to $z_0$ ($|V(H)|$ choices), the coloring of the edges
incident to $Z$ (at most $(\maxdeg(G)+1)^{|Z|\maxdeg(G)}$ choices),
going through a family of perfect hash functions ($|Z|^{O(|Z|)}\cdot
\log |V(H)|$ choices), and branching on the choice of which blocks in
$\cK$ are directed by the solution (at most $(3^{|Z|\maxdeg(G)})$
choices). As all the nonpolynomial factors depend only on $\maxdeg(G)$
and $\fvs(G)$ (note that $|Z|=O(\maxdeg(G)^2\fvs(G))$), the total
running time is $f(\maxdeg(G),\fvs(G))\cdot
n^{O(\genus(G))}$. Therefore, the algorithm is compatible with the
specified description.
\end{proof}
With a simple reduction, we can generalize Lemma~\ref{lem:bigplanarconnected} to disconnected graphs:
\begin{ptheorem}\label{th:bigplanar}
There exists an algorithm compatible with the description
\begin{eqnarray*}
\sil{\maxdeg(G),\fvs(G)}{\genus(G),\ccn(H)}{}
\end{eqnarray*}
\end{ptheorem}
\begin{proof}
Let $u_1$, $\dots$, $u_k$ be
  arbitrary vertices of $H$, one from each connected component. We
  guess the images of the set $\{u_1,\dots, u_k\}$: we repeat the
  algorithm described below for every possible choice of distinct
  vertices $v_1$, $\dots$, $v_k$ from $G$. We define $H_c$ by
  introducing new vertices $x$, $x_1$, $\dots$, $x_D$ and making $x$
  adjacent to $\{x_1,\dots, x_D,u_1,\dots, u_k\}$. We define $G_c$
  similarly by introducing new vertices $y$, $y_1$, $\dots$, $y_D$ and
  making $y$ adjacent to $\{y_1,\dots, y_D,v_1,\dots, v_k\}$ and then discarding any component of $G$ that contains no $v_i$. Clearly, $H_c$ and $G_c$ are connected. Then we construct the graphs $H'_c$ and $G'_c$ given by Lemma~\ref{lem:setting-images}. It follows that $H'_c$ is a subgraph of $G'_c$ if and only if there is a
  subgraph isomorphism $\eta$ from $H$ to $G$ such that $\eta$ maps
  $\{u_1,\dots, u_k\}$ to $\{v_1,\dots, v_k\}$.

  For a given choice of $v_1$, $\dots$, $v_k$, we can use the
  algorithm of Lemma~\ref{lem:bigplanarconnected} to check if $H'_c$ is
  a subgraph of $G'_c$. 
Adding the $\ccn(H)$ edges incident to $y$ increases genus by at most $\ccn(H)$ and
attaching the trees given by Lemma~\ref{lem:setting-images} does not increase the genus, hence $\genus(G'_c)\le \genus(G)+\ccn(H)$. As the attached trees have maximum degree 3, we have $\maxdeg(G'_c)=\max\{\maxdeg(G)+1,\ccn(H)\}$. Joining different components and attaching trees does not extend the feedback vertex set number, i.e., $\fvs(G'_c)=\fvs(G)$.
By Lemma~\ref{lem:bigplanarconnected} and Lemma~\ref{lem:multiplicative}, we may assume that the instance $(H'_c,G'_c)$ can be solved in time 
\begin{multline*}
\hat f(\maxdeg(G'_c))\hat f(\fvs(G'_c))\cdot n^{\hat f(\genus(G'_c))}=
\hat f(\max\{\maxdeg(G)+1,\ccn(H)\})\hat f(\fvs(G))\cdot n^{
\hat f(\genus(G)+\ccn(H))}\\
\le 
\hat f(\maxdeg(G)+1)\hat f(\ccn(H))\hat f(\fvs(G))\cdot n^{
\hat f(2\genus(G))\hat f(2\ccn(H))}\\\le
\hat f(\maxdeg(G)+1)\hat f(\fvs(G))\cdot n^{
\hat f(\ccn(H))\hat f(2\genus(G))\hat f(2\ccn(H))}\\
=f_1(\maxdeg(G),\fvs(G))\cdot n^{f_2(\ccn(H),\genus(H))}
\end{multline*}
for 
  some functions $f_1$, $f_2$. Trying all
  possible $v_1$, $\dots$, $v_k$ adds an overhead of $n^{\ccn(H)}$,
  thus the algorithm is compatible with the specified description.
\end{proof}

\subsection{Excluded minors}\label{sec:exminors}

We start by reviewing Robertson and Seymour's structure theorem. We
need first the definition of $(p,q,r,s)$-almost embeddable graphs. For every $n\in\mathbb{N}$, by $P^n$ we denote the path with
vertex set $[n]$ and edges $i(i+1)$ for all $i\in[n-1]$. A
\emph{$p$-ring} is a tuple $(R,v_1,\ldots,v_n)$, where $R$ is a graph
and $v_1,\ldots,v_n\in V(R)$ such that there is a path decomposition
$(P^n,\beta)$ of $R$ of width $p$ with $v_i\in\beta(i)$ for all
$i\in[n]$. A graph $G$ is \emph{$(p,q)$-almost embedded} in a surface
$\Sigma$ if there are graphs $G_0,G_1,\ldots,G_q$ and mutually disjoint
closed disks $\mathbf{D}_1,\ldots,\mathbf{D}_q\subseteq\Sigma$ such that:
\begin{enumerate}
\item $G=\bigcup_{i=0}^qG_i$.
\item $G_0$ is embedded in $\Sigma$ and has empty intersection with the
  interiors of the disks $\mathbf{D}_1,\ldots,\mathbf{D}_q$.
\item The graphs $G_1,\ldots,G_q$ are mutually disjoint.
\item For all $i\in[q]$ we have $E(G_0)\cap E(G_i)=\emptyset$. Moreover, there are $n_i\in\mathbb{N}$ and $v^i_1,\ldots,v^i_{n_i}\in
  V(G)$ such that $V(G_0)\cap V(G_i)=\{v^i_1,\ldots,v^i_{n_i}\}$, and the
  vertices $v^i_1,\ldots,v^i_{n_i}$ appear in cyclic order on the boundary of
  the disk $\mathbb{D}_i$.
\item For all $i\in[q]$ the tuple $(G_i,v^i_1,\ldots,v^i_{n_i})$ is a $p$-ring.
\end{enumerate}
We call the embedding of $G_0$ in $\Sigma$ and the tuple
$(\mathbf{D}_1,\dots,\mathbf{D}_q,G_1,\dots,G_q)$ a {\em
  $(p,q)$-almost embedding} of $G$ in $\Sigma$. A $(p,q)$-almost
embedding is an obstruction for having large clique subgraphs (and
more generally, large clique minors, see
e.g.,~\cite{DBLP:journals/jct/JoretW13}).
\begin{proposition}\label{prop:pq-maxclique}
  For every $p,q,r\in \mathbb{N}$, there is a constant $K\in \mathbb{N}$
  such that every graph $(p,q)$-almost embeddable in a surface of genus $r$ has maximum clique
  size at most $K$.
\end{proposition}

The graphs $G_1$, $\dots$, $G_q$ are the {\em vortices} of the
$(p,q)$-almost embedding.  We say that a $(p,q)$-almost embedding has
{\em hollow vortices} if $V(G_i)\subseteq V(G_0)$ for every $1\le i
\le q$. The following lemma shows that we can achieve this property
with a simple modification of the embedding. The assumption that every
vortex is hollow will be convenient in some of the proofs to follow.
\begin{lemma}\label{lem:hollow}
  Given a $(p,q)$-almost embedding of a graph $G$, one can obtain in
  polynomial time a $(p,q)$-almost embedding of $G$ with hollow vortices.
\end{lemma}
\begin{proof}
  For every $1\le i \le q$, we define an ordering of $V(G_i)$
  extending the ordering $v^i_1$, $\dots$, $v^i_{n_i}$ in the following
  way.  Let $(P^{n_i},\beta^i)$ be the path decomposition of
  $G_i$. For each vertex $u\in V(G_i)$, let $x(u)$ be the smallest
  value $j$ such that $u\in \beta^i(j)$.  Let $\iota: V(G_i)\to \{1,\dots, |V(G_i)|\}$ be an ordering of
  $V(G_i)$ by increasing value of $x(u)$, breaking ties arbitrarily.
  We embed the vertices of $V(G_i)$ on the boundary of $\mathbf{D}_i$
  in this order. Finally, we define
  $\hat\beta^{i}(j)=\beta^i(x(\iota^{-1}(j)))$. We claim that this is indeed a
  path decomposition of width $p$. It remains true that every bag has
  size at most $p+1$ and if two vertices are adjacent in $G_i$, then
  there is a bag containing both of them: if $\beta^i(j)$ is the first
  bag containing both $u$ and $v$, then either $x(u)=j$ or $x(v)=j$
  holds, hence either $u,v\in \hat\beta^i(\iota(u))=\beta^i(x(u))=\beta^i(j)$ or $u,v\in
  \hat\beta^i(\iota(v))=\beta^i(x(v))=\beta^i(j)$ holds.  To see that every vertex
  appears in a consecutive sequence of bags, suppose that $u$ appears
  in bags $\beta^i(j_1)$, $\dots$, $\beta^i(j_2)$. Then $u$ appears in
  $\hat\beta^i(j)$ if $j_1\le x(\iota^{-1}(j)) \le j_2$ holds.  Since the ordering $\iota$ was defined by
  increasing values of $x$, it follows that $\hat\beta^i(j)$ contains
  $u$ for a consecutive sequence of values of $j$. Therefore, ordering $\iota$ and the bags $\hat\beta^i(j)$ indeed form a
  path decomposition.  Performing this modification for every $G_i$ yields a
  $(p,q)$-almost embedding with hollow vortices.
\end{proof}

The following theorem generalizes Theorem~\ref{th:cspprojsink} from
bounded-genus primal graphs to $(p,q)$-almost embeddable primal
graphs.  Note that in this theorem, we assume that the embedding is
given in the input. The main technical difficulty in the proof is how
to interpret the cutting operation in the presence of vortices: we
argue that we get a bounded-treewidth graph even if we cut open only
the components of the tree outside the vortices.
\begin{theorem} \label{th:cspprojsink-genus} Let $I=(V,D,C)$ be a
  binary CSP instance having a projection sink and let $G$ be the
  primal graph of $I$. Suppose that a $(p,q)$-almost embedding of $G$
  into a surface $\Sigma$ of genus $r$ is given. Then $I$ can be
  solved in time $|I|^{O(pqr)}$.
\end{theorem}
\begin{proof}
  Without loss of generality assume that $q>0$, since otherwise we may apply the algorithm of Theorem~\ref{th:cspprojsink}.
  Let $G_0$, $G_1$, $\dots$, $G_q$ be the subgraphs in the
  $(p,q)$-almost embedding of $G$.  By Lemma~\ref{lem:hollow}, we can
  assume that all the vertices are hollow, i.e., $V(G)=V(G_0)$.  As in
  Theorem~\ref{th:cspprojsink}, we build the projection graph and
  obtain the directed tree $\overrightarrow{T}$ and its underlying
  undirected (spanning) tree $T$. The edge set $E(T)\cap E(G_0)$ spans
  a forest $F$ having components $T_1$, $\dots$, $T_c$ in $G_0$ and
  components $\overrightarrow{T_1}$, $\dots$, $\overrightarrow{T_c}$
  in the projection graph. If $s_j$ is the vertex of
  $\overrightarrow{T_j}$ closest to the sink of $\overrightarrow{T}$,
  then $s_j$ is a sink of $\overrightarrow{T_c}$, that is, vertex
 $s_j$ is  reachable from each vertex of $\overrightarrow{T_j}$ using only the
  edges of $\overrightarrow{T_j}$.

\begin{figure}
\begin{center}
                \def\svgwidth{0.90\linewidth}
                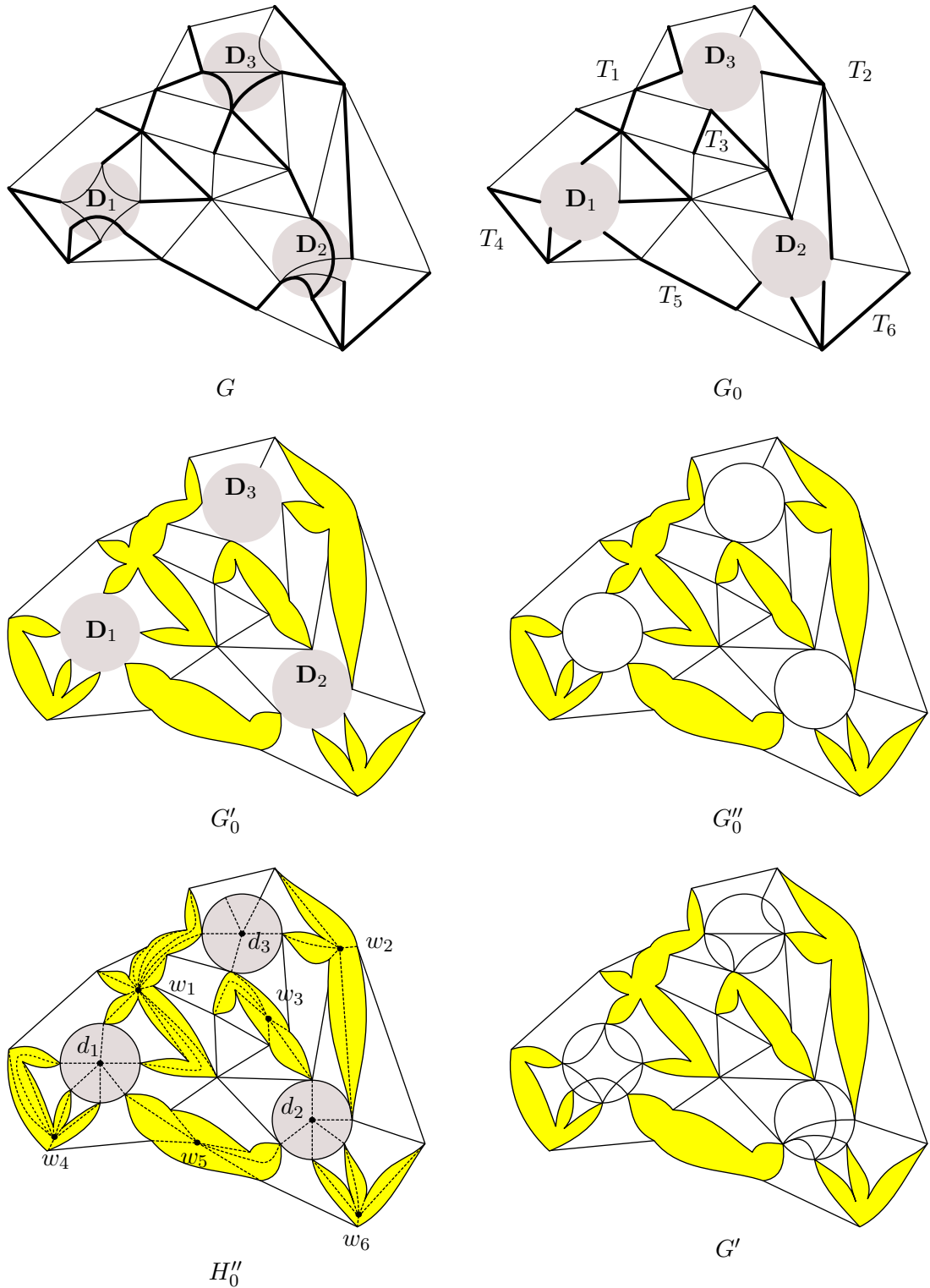
\end{center}
\caption{The graphs appearing in the proof of Theorem~\ref{th:cspprojsink-genus}:
The graph $G$ with the spanning tree $T$; the graph $G_0$ after removing the edges in the vortices; the graph $G'_0$ obtained by  cutting open the trees $T_1$, $\dots$, $T_c$; the graph $G''_0$ having a cycle at the boundary of every vortex; the graph $H''_0$ in the proof of Claim~\ref{cl:vortextw} with the vertices $d_i$ and $w_j$; and the graph $G'$ obtained from $G''_0$ by reintroducing the edges of the vortices.}\label{fig:vortex}

\end{figure}

  Let $G'_0$ be the graph obtained from $G_0$ by cutting open each
  tree $T_j$ (see Figure~\ref{fig:vortex}); let $\phi: V(G'_0)\to V(G_0)$ be a mapping with the
  meaning that $v$ is a copy of $\phi(v)$. Note that, as the trees
  $T_1$, $\dots$, $T_c$ are pairwise vertex-disjoint, cutting them
  open can be done independently and $G'_0$ is still embeddable in
  $\Sigma$. Recall that if a vertex $v$ has degree $d$ in $T_j$, then
  $d$ vertices correspond to $v$ in $G'_0$. For notational
  convenience, we define $G'_0$ such that $V(G'_0)\supseteq V(G_0)$:
  each $v$ of $G_0$ also appears in $G'_0$, together with $d-1$ new
  copies $v^1$, $\dots$, $v^{d-1}$ of $v$. As $G_0$ is disjoint from the interiors of
  $\mathbf{D}_1$, $\dots$, $\mathbf{D}_q$, the operations can be
  performed in such a way that $G'_0$ is also disjoint from the
  interiors of these discs and every vertex $v\in V(G_i)$ is embedded
  at the same place on the boundary of $\mathbf{D}_i$ in the
  embedding of $G'_0$ (but the copies $v^1$, $\dots$, $v^{d-1}$ of $v$ are outside
  $\mathbf{D}_i)$. Note also that every tree $T_j$ contains at least
  one vertex of some $G_i$: the component $T_j$ was created by
  removing edges of the $G_i$'s from the spanning tree $T$.

  Let $I_0=(V,D,C_0)$ be the restriction of $I$  where a constraint $\langle (u,v),R\rangle\in C$ appears in $C_0$ only
  if $uv\in E(G_0)$.  We construct a CSP instance $I'_0=(V'_0,D,C')$
  with primal graph $G'_0$ in the following way. Let $V'_0$ be the set of
  vertices of $G'_0$. If $u$ and $v$ are adjacent in $G'_0$, then
  $\phi(u)$ and $\phi(v)$ are adjacent in $G_0$, hence there is a
  constraint $\langle (\phi(u),\phi(v)),R\rangle$ in $C_0$. In this
  case, we introduce the constraint $\langle (u,v),R\rangle$ in
  $C'_0$.

  As in the proof of Theorem~\ref{th:cspprojsink}, we can argue that
  $I_0$ and $I'_0$ are equivalent and there is a one-to-one
  correspondence between the satisfying assignments of $I_0$ and
  $I'_0$.  Again, we can observe that if $f:V_0\to D$ is a satisfying
  assignment of $I_0$, then $f(\phi(v))$ is a satisfying assignment of
  $I'_0$. For the other direction, we can show that if
  $\phi(u_1)=\phi(u_2)=u$ for variables $u_1,u_2\in V'_0$, then
  $f'(u_1)=f'(u_2)$ for every satisfying assignment $f'$ of
  $I'_0$. This can be argued by considering the tree $T_j$ containing
  $u$; recall that it contains a sink $s_j$. Now we define $f(v)$ to
  be $f'(v')$ for any $v'\in V'_0$ with $\phi(v')=v$ and observe that
  $f$ is a satisfying assignment of $I_0$.

  Finally, we define a CSP instance $I'=(V',D,C')$ that is equivalent
  to $I$. The set $C'$ is $C'_0\cup (C\setminus C_0)$: that is, we
  reintroduce each constraint that was removed from $C$ in the definition of $C_0$ because the
  edge corresponding to it was in a vortex. Recall that if
  $xy\in E(G_i)$, then $x,y\in V(G_i)$ and hence $x,y$ also appear in
  $G'_0$, embedded at the same place as in the embedding of $G_0$.  It
  is clear that $I$ has a satisfying assignment if and only if $I'$
  has. 

  Let $G'$ be the primal graph of $I'$. If we can bound the treewidth
  of $G'$, then we can use Theorem~\ref{th:freuder} to solve $I'$. We
  give a bound first on the treewidth of a graph $G''_0$, which is
  obtained from $G'_0$ as follows: for every $1\le i \le q$, we add a
  cycle on the vertices of $G_i$, in the order they appear on the
  boundary of $\mathbf{D}_i$. Observe that $G''_0$ can be embedded in
  $\Sigma$ as well.
\begin{claim}\label{cl:vortextw}
The treewidth of $G''_0$ is $O(qr)$.
\end{claim}
\begin{proof}
Let $H''_0$ be obtained from $G''_0$ the following way:
\begin{itemize}
\item for every $1\le i \le q$, let us add a new vertex $d_i$ adjacent to the vertices on the boundary of $\mathbf{D}_i$, and
\item for every $1\le j \le c$, let us add a new vertex $w_j$ adjacent to every vertex on the face $F_j$ created when cutting open the tree $T_j$.
\end{itemize}
Note that from the fact that $G''_0$ can be embedded in $\Sigma$, it
follows that $H''_0$ can be embedded in $\Sigma$ as well: $d_i$ and
$w_j$ can be embedded in $\mathbf{D}_i$ and $F_j$,
respectively. We claim that the diameter of $H''_0$ is
$O(q)$. First, we show that every vertex $v$ is at distance at most
$3$ from some vertex $d_i$. Suppose that $\phi(v)$ is in tree
$T_j$. As component $T_j$ was obtained from the spanning tree $T$ by removing the
edges of the $G_i$'s, tree $T_j$ has to contain a vertex $u\in V(G_0)\cap V(G_i)$ for some $1\le i \le q$, that is, lying on the boundary
of $\mathbf{D}_i$. Recall that this vertex $u$ is also on the boundary
of $\mathbf{D}_i$ in the embedding of $G''_0$. Now the path $v w_j u
d_i$ shows that $v$ is at distance at most 3 from $d_i$. Next we bound
the distance between the $d_i$'s. Consider the graph $W$ whose vertex
set is $\{d_1,\dots, d_q\}$ and there is an edge between $d_{i_1}$ and
$d_{i_2}$ if and only if there is a $T_j$ having a vertex from both
$V(G_{i_1})$ and $V(G_{i_2})$. The fact that $T$ is a spanning tree of
$G$ implies that $W$ is connected. Observe that if $d_{i_1}$ and
$d_{i_2}$ are adjacent in $W$, then they are at distance at most $4$
in $H''_0$: if $T_j$ has a vertex both in $V(G_{i_1})$ and
$V(G_{i_2})$, then $w_j$ is adjacent to a vertex $u_{i_1}$ on the
boundary of $\mathbf{D}_{i_1}$ and to a vertex $u_{i_2}$ on the
boundary of $\mathbf{D}_{i_2}$, thus there is a path
$d_{i_1}u_{i_1}w_ju_{i_2}d_{i_2}$ in $H''_0$. As every $d_{i_1}$ and $d_{i_2}$
are at distance at most $q-1$ in $W$, it follows that the distance of
$d_{i_1}$ and $d_{i_2}$ is at most $4(q-1)$ in $H''_0$. We can
conclude that the diameter of $H''_0$ is at most $4(q-1)+6$.  A graph with
genus $r$ and diameter $D$ has treewidth $O(rD)$ \cite[Theorem
2]{eppstein}, thus the treewidth of $H''_0$ is $O(qr)$, which follows
also for its subgraph $G''_0$.  \cqed
\end{proof}

Let $G'''_0$ be obtained from $G''_0$ by replacing each vertex $v$
with a clique $\kappa(v)$ of $p+1$ vertices and, for each edge $uv\in
E(G''_0)$, connecting every vertex of $\kappa(u)$ and $\kappa(v)$. A
tree decomposition of width $w-1$ (that is, maximum bag size $w$) of
$G''_0$ can be easily transformed into a tree decomposition of width
$pw-1$ (that is, maximum bag size of $pw$) of $G'''_0$. Therefore, the
treewidth of $G'''_0$ is $O(pqg)$. We show that $G'$ is a minor of
$G'''_0$. by constructing a minor model of $G'$ in $G'''_0$. If a
vertex $v\in V(G')$ is not in any $V(G_i)$ for $1\le i \le q$, then
let us map $v$ to an arbitrary vertex of $\kappa(v)$. If $v\in
V(G_i)$, then let us consider the bags in the path decomposition
$(P^{n_i},\beta^i)$ of the $p$-ring $(G_i,v^i_1,\dots,v^i_{n_i})$.
We then map $v$ to a set consisting of one arbitrarily chosen vertex from each $\kappa(v^i_j)$ for $1\le j \le n_i$ with $v\in \beta^i(j)$.
It is clear that we map at most $p+1$ vertices each of the sets $\kappa(v^i_j)$ and the image of each vertex
$v\in V(G')$ is connected in $G'''_0$ (because of the cycle we have
added on the vertices of $G_i$ in the construction of $G''_0$, the set
$\{v^i_{x},\dots, v^i_{y}\}$ is connected for every $1\le x \le y \le
n_i$). Furthermore, if $v_1,v_2\in V(G_i)$ are adjacent in $G_i$, then
there is a $1\le j \le n_i$ such that $\beta^i(j)$ contains both of
them, thus their images both intersect the clique $\kappa(v^i_j)$,
implying that there is an edge between the two sets. Therefore, $G'$
is a minor of $G'''_0$, and it follows that the treewidth of $G'$ is
$O(pqr)$.  Consequently, we can solve $I'$ using the algorithm of
Theorem~\ref{th:freuder} in time $|I|^{O(pqr)}$.
\end{proof}

The Graph Structure Theorem states that graphs excluding a fixed minor
have a tree decomposition where the graph induced by each bag has
certain properties, or more precisely, the so-called torso of each bag
has certain properties:
\begin{definition}
Given a graph $G$ and a tree decomposition $(T,\beta)$, we define the {\em torso} at $t$ as the graph
\begin{equation}
\tau(t):=G[\beta(t)]\cup K[\sigma(t)]\cup \bigcup_{\textup{$t'$ is a child of $t$}}K[\sigma(t')],
\end{equation}
where $K[X]$ denotes the clique on a set $X$ of vertices.
\end{definition}

We need the following version of the Graph  Structure Theorem,  which is a weaker restatement of \cite[Theorem 4]{DBLP:journals/jct/DiestelKMW12}:
 \begin{theorem}[\cite{DBLP:journals/jct/DiestelKMW12}]\label{th:gm-diestel}
   For every graph $H$ there are constants $p,q,r,s\in\mathbb{N}$ such
   that every graph $G$ not containing $H$ as a minor has a tree decomposition
   $(T,\beta)$ such that for every $t\in V(T)$, there is a set
   $Z_t\subseteq \beta(t)$ of size at most $s$ such that
   $\tau(t)\setminus Z_t$ is $(p,q)$-almost embeddable in a surface of
   genus at most $r$.  Moreover, there is a $(p,q)$-almost embedding of
   $\tau(t)\setminus Z_t$ such that the following holds: if $t'$ is a
   child of $t$, then either
\begin{itemize}
\item[(i)] 
$\sigma(t')\setminus Z_t$ is contained in a bag of a vortex, or
\item[(ii)] $\sigma(t')\setminus Z_{t}$ induces a clique of size at most 3 in $G_0$, and if this clique has size exactly 3, then it bounds a face in
  the embedding of $G_0$ (here $G_0$ is the subgraph $G_0$ appearing in the definition of the $(p,q)$-almost embeddings).
 \end{itemize}
\end{theorem}

Algorithmic versions of the Graph Structure Theorems do exist
\cite{kawwol11,grokawree13, Decomposition_FOCS2005,MR2046826};
however, Theorem~\ref{th:gm-diestel} is not stated
algorithmically in \cite{DBLP:journals/jct/DiestelKMW12}. While it
should be possible to obtain an algorithmic version of
Theorem~\ref{th:gm-diestel}, this is not an issue in application
(Theorem~\ref{th:cspprojsink-minor} and
Theorem~\ref{lem:bigplanarconnectedminor}), as we can afford to find
the decomposition by brute force.

We need a further refinement of the structure theorem, which involves an extension of the torso.
\begin{definition}
Given a tree decomposition $(T,\beta)$ of $G$, 
the \emph{extended torso} of $G$ at $t\in T$ is the graph $\tau^*(t)$ that is obtained from $\tau(t)$ as follows: for each set $X\subseteq \beta(t)$ such that $t$ has a child $t'$ with $\sigma(t')=X$, we introduce a new vertex $t_X$ adjacent to each vertex of $X$.
\end{definition}
Note that, even if $t$ has two children $t_1$ and $t_2$ with
$\sigma(t_1)=\sigma(t_2)=X$, we introduce only a
single new vertex $t_X$ adjacent to $X$.

The following version shows that we do not need the sets $Z_t$
for bounded-degree graphs and even the extended torso $\tau^*$ has an
embedding.
\begin{lemma}\label{lem:gm-bounded-degree}
  For every graph $H$ and integer $d\in \mathbb{N}$, there are constants
  $p,q,r,d^*\in \mathbb{N}$ such that every graph $G$ not containing $H$ as a minor
  and having maximum degree at most $d$ has a tree decomposition $(T,\beta)$ such
  that for every $t\in V(T)$, the extended torso $\tau^*(t)$ has maximum degree $d^*$ and is $(p,q)$-almost
  embeddable with hollow vortices in a surface of genus at most $r$.
\end{lemma}
\begin{proof}
  Let $(T,\beta)$ be the decomposition given by
  Theorem~\ref{th:gm-diestel}, and let $Z_t$ be the set of size at
  most $s$ such that $\tau(t)\setminus Z_t$ has a $(p,q)$-almost
  embedding in some surface of genus at most $g$ (where the constants
  $p,q,r,s$ depend only on $H$). We prove that $\tau^*(t)$ is
  $(p',q')$-almost embeddable with hollow vortices in a surface of
  genus at most $r'$, where $p',q',r'$ depend only on $H$ and $d$.

  For a child $t'$ of $t$, we may assume that every vertex $v\in
  \sigma(t')$ has a neighbor in $\alpha(t')$: otherwise, omitting $v$ from $\beta(t'')$ for every
  descendant $t''$ of $t'$ would remain a tree decomposition (and it
  is easy to see that the torsos of the bags satisfy the requirements
  even after omitting vertex $v$).  Therefore, we may assume that a
  vertex $v\in \beta(t)$ appears in $\beta(t')$ for at most $d$
  children $t'$ of $t$.  A consequence of this assumption is that we
  may assume that $\tau(t)$ has maximum degree bounded by a constant
  $d_\tau$ depending only on $H$ and $d$: the maximum clique size of a
  graph $(p,q)$-almost embeddable into a surface of genus $r$ can be
  bounded by a function of $p,q,r$ (Proposition~\ref{prop:pq-maxclique}), and the definition of $\tau(t)$
  adds at most $d$ cliques containing $v$ (one for each child $t'$ of
  $t$ with $v\in \sigma(t')$). It also follows that the maximum degree
  of $\tau^*(t)$ can be bounded by a constant $d_{\tau^*}$
  depending only on $H$ and $d$: each new vertex introduced in the
  definition of $\tau^*(t)$ is adjacent to a clique of $\tau(t)$
  (whose size is at most $d_\tau+1$) and each vertex $v$ of $\tau(t)$
  is adjacent to at most $d$ new vertices (as $v$ appears in
  $\sigma(t')$ for at most $d$ children $t'$ of $t$).

  Due to a minor technical detail, we need to handle the following set
  $Y$ of vertices in a special way. Consider the subgraphs $G_0$,
  $G_1$, $\dots$, $G_q$ and the discs $\mathbf{D}_1$, $\dots$,
  $\mathbf{D}_q$ of the $(p,q)$-almost embedding. If three vertices of
  $G_0$ form a triangle that is the boundary of a face containing a
  disc $\mathbf{D}_i$ of the $(p,q)$-almost embedding, then let us
  include these three vertices in the set $Y$; clearly, we have
  $|Y|\le 3q$. First, we extend the embedding of $\tau(t)\setminus
  N[Z_t\cup Y]$ to an embedding of $\tau^*(t)\setminus N[Z_t\cup Y]$,
  and then further extend it to $\tau^*(t)$. Consider a vertex $v$ of
  $(\tau^*(t)\setminus N[Z_t\cup Y])\setminus V(\tau(t))$. This means
  that $t$ has a child $t'$ such that $X=\sigma(t')$ is exactly the
  neighborhood of $v$ in $\tau^*(t)$. As $v$ is in $\tau^*(t)\setminus
  N[Z_t\cup Y]$, it follows that $X\subseteq \beta(t)\setminus
  (Z_t\cup Y)$. Consider now the two possibilities of
  Theorem~\ref{th:gm-diestel}. In case (i), when $X$ is a subset of
  bag of the path decomposition of a vortex, then let us add $v$ to
  this bag. By the definition of $\tau^*(t)$, each vertex we introduce
  has a different neighborhood in the bag, hence we introduce at most
  $2^{p+1}$ new vertices to the bags. It follows that the path
  decomposition has width at most $p':=p+2^{p+1}$ after adding these
  new vertices.  In case (ii), when $X$ induces a clique of size at
  most $3$, then we add $v$ to $G_0$. If $|X|\le 2$, then it is easy
  to extend the embedding of $G_0$ with the vertex $v$. If $|X|=3$,
  then we use the fact given by case (ii) of
  Theorem~\ref{th:gm-diestel} that $X$ bounds a face in the embedding
  of $G_0$. As $X\cap Y=\emptyset$, this face does not contain any of
  the discs $\mathbf{D}_i$. Therefore, we can embed $v$ in this face
  and make it adjacent to all three vertices of $X$. Note that by the
  way we defined $\tau^*(t)$, only at most one new vertex with
  neighborhood $X$ is introduced, thus we need to embed only a single
  new vertex in this face. Therefore, we get a $(p',q)$-almost
  embedding of $\tau^*(t)\setminus N[Z_t\cup Y]$ in a surface of genus
  at most $r$. Lemma~\ref{lem:hollow} allows us to transform this
  embedding in a way that the vortices become hollow.

  Recall that the maximum degree of $\tau^*(t)$ can be bounded by a
  constant $d_{\tau^*}$ depending only on $H$ and $d$. Therefore,
  $N[Z_t\cup Y]\le (d_{\tau^*}+1)(s+3q)$ and $N[Z_t\cup Y]$ is
  incident to at most $d_{\tau^*}(d_{\tau^*}+1)(s+3q)$ edges of
  $\tau^*(t)$. We extend the $(p',q)$-almost embedding of
  $\tau^*(t)\setminus N[Z_t\cup Y]$ by adding the vertices of
  $N[Z_t\cup Y]$ to $G_0$, embedding them arbitrarily on the surface,
  and then for each edge $xy$ incident to $N[Z_t\cup Y]$, we extend
  the surface with a new handle and connect $x$ and $y$ using this
  handle. Note that the $(p',q)$-almost embedding of
  $\tau^*(t)\setminus N[Z_t\cup Y]$ has hollow vertices, which means
  that every vertex is actually embedded in the surface. Therefore,
  the operation of connecting two vertices with a handle is well
  defined. Repeating this step for every edge incident to $N[N_t\cup
  Y]$ results in a $(p',q)$-almost embedding of $\tau^*(t)$ in a surface
  of genus at most $r':=r+d_{\tau^*}(d_{\tau^*}+1)(s+3q)$, which
  depends only on $H$ and $d$.
\end{proof}

We are now ready to state a generalization of
Theorem~\ref{th:cspprojsink-genus} to bounded-degree graphs excluding
a fixed minor:
\begin{theorem} \label{th:cspprojsink-minor} Let $I=(V,D,C)$ be a
  binary CSP instance having a projection sink and let $G$ be the
  primal graph of $I$. Suppose that $G$ has maximum degree $d$ and
  excludes a graph $H$ as a minor. Then $I$ can be solved in time
  $f_1(|V|)\cdot |I|^{f_2(H,d)}$, for some computable functions $f_1$ and $f_2$.
\end{theorem}
\begin{proof}
  Let $(T,\beta)$ be a decomposition of $G$ as in
  Lemma~\ref{lem:gm-bounded-degree}; we may find such a decomposition
  by brute force in time depending only on $|V|$. For every $t\in
  V(T)$, we define the CSP instance $I_t=(\gamma(t),D,C_t)$, where
  $C_t$ contains only those constraints of $C$ whose scope is fully
  contained in $\gamma(t)$. For every $t\in V(T)$ and every mapping
  $\phi:\sigma(t)\to D$, we define a subproblem $(t,\phi)$, whose
  value is ``true'' if $\phi$ can be extended to a satisfying
  assignment of $I_t$. We show how to solve these subproblems by
  bottom-up dynamic programming.

  We show how to solve a subproblem $(t,\phi)$ assuming that we have
  already solved all the subproblems corresponding to the children of
  $t$. We create a new CSP instance $I^*_t=(V^*_t,D^*,C^*_t)$, whose
  primal graph is a subgraph of $\tau^*(t)$. Let $d^*$ be the maximum
  degree of $\tau^*(t)$, which is a constant depending only on $H$ and
  $d$ by Lemma~\ref{lem:gm-bounded-degree}. The domain $D^*$ is $\bigcup_{i=1}^{d^*}D^i$, i.e., tuples of
  length at most $d^*$ over $D$. The set $C^*_t$ contains the following constraints:
\begin{itemize}
\item For each vertex $v\in \beta(t)$, we add the
  unary constraint $\langle (v),D\rangle$ restricting its value to
  $D$.
\item For each vertex $v\in \sigma(t)$, we add the
  unary constraint $\langle (v),\{\phi(v)\}\rangle$, ensuring that the assignment respects $\phi$ on $\sigma(t)$.
\item  For each constraint $\langle (u,v)\rangle,R\rangle\in C_t$ with
  $u,v\in \beta(t)$, we add the same constraint to $C^*_t$.
\item  For each
  child $t'$ of $t$ with $X=\sigma(t')$, there is a
  corresponding vertex $t_X$ of $\tau^*(t)$ whose neighborhood is $X$.
  Let $x_1$, $\dots$, $x_{|X|}$ be the neighbors of $t_X$ ordered in
  an arbitrary way. Given a tuple $\mathbf{d}=(d_1,\dots, d_{|X|})\in
  D^{|X|}$, let $\phi_{\mathbf{d}}:X\to D$ be the assignment with
  $\phi_{\mathbf{d}}(x_i)=d_i$ for every $1\le i \le |X|$.  We add the unary constraint
\[
\langle (t_X), \{ \mathbf{d}\in D^{|X|}\mid \text{subproblem $(t',\phi_{\mathbf{d}})$ is true}\}\rangle.
\]
That is, the $|X|$-tuple appearing on $t_X$ should describe an
assignment of $X$ that can be extended to a satisfying assignment of
$I_{t'}$.  Note that there can be more than one child $t'$ with
$X=\sigma(t')$; we add a unary constraint of this form for each such
child (but these constraints can be merged by taking the intersection
of the constraint relations). Finally, for every $1\le i \le |X|$, we
add the constraint
\[
\langle (t_X,x_i), \{((d_1,\dots, d_{|X|}), d)\mid d_i=d\}\rangle,
\]
which ensures that the $i$-th component of the value of $t_X$ is the
same as the value of the $i$-th neighbor of $t_X$.  In other words,
the $|X|$-tuple appearing on $t_X$ should describe the assignment
appearing on $X$.  Note that this constraint is a projection from
$t_X$ to $x_i$.
\end{itemize}
This completes the
description of $I^*_t$; observe that the primal graph of $I^*_t$ is a
subgraph of $\tau^*(t)$. 
\begin{claim}\label{cl:extending}
$I^*_t$ has a satisfying assignment if and only if subproblem $(t,\phi)$ is true.
\end{claim}
\begin{proof}
  Suppose that $I_t$ has a satisfying assignment $\phi_t$ extending
  $\phi$. We construct as satisfying assignment $\phi^*_t$ of $I^*_t$ as
  follows. For every $v\in \beta(t)$, we set
  $\phi^*_t(v)=\phi_t(v)$. Let $t_X$ be a variable of $I^*_t$ not in
  $\beta(t)$ with neighborhood $X=(x_1,\dots,x_{|X|})$. We set
  $\phi^*_t(t_X)=(\phi_t(x_1),\dots, \phi_t(x_{|X|}))$. It is easy to
  verify that all the constraints of $I^*_t$ are satisfied. In
  particular, $\phi_t$ restricted to $\gamma(t')$ shows that the
  assignment $\mathbf{d}=(\phi_t(x_1),\dots, \phi_t(x_{|X|}))$ on
  $(x_1,\dots, x_{|X|})$ can be extended to a satisfying assignment of
  $\gamma(t)$, that is, subproblem $(t',\phi_{\mathbf{d}})$ is true,
  implying that $\mathbf{d}$ satisfies the unary constraints on $t_X$.

  For the other direction, suppose that $\phi^*_t$ is a satisfying
  assignment of $I^*_t$. Consider a variable $t_X\not\in \beta(t)$ with
  neighborhood $X=(x_1,\dots,x_{|X|})$.  The binary constraints
  between $t_X$ and the $x_i$'s ensure that
  $\phi^*_t(t_X)=(\phi^*_t(x_1),\dots, \phi^*_t(x_{|X|}))$. Therefore,
  the unary constraints on $t_X$ ensure that the restriction of
  $\phi^*_t$ to $X$ can be extended to a satisfying assignment of
  $I_{t'}$ for every child $t'$ of $t$ with $\sigma(t')=X$. This way, we can extend $\phi^*_t$ to $\alpha(t')$ for every
  every child $t'$ of $t$ and obtain a satisfying assignment $\phi_t$
  of $I_t$. Note that the unary constraints on $\sigma(t)$ ensure that
  $\phi^*_t$ (and hence $\phi_t$) extends $\phi$. It follows that
  subproblem $(t,\phi)$ is true.  \cqed\end{proof}

In order to solve
$I^*_t$ using the algorithm of Theorem~\ref{th:cspprojsink-genus}, we have to ensure that the instance has a projection sink. We add further constraints to obtain instances having this property. By assumption, $I$ has a projection sink $v_0$.
We define a set $Y$ of vertices as follows:
\begin{itemize}
\item If $v_0\in \beta(t)$, then let $Y=\{v_0\}$.
\item  If $v_0$ is in $\alpha(t')$ for some child $t'$ of $t$, then let $Y=\sigma(t')$.
\item If $v_0$ is not in $\gamma(t)$, then let $Y=\sigma(t)$.
\end{itemize}
Observe that, in all three cases, set $Y$ induces a clique in
$\tau(t)$ and every path from a vertex of $\beta(t)$ to $v_0$ in $G$
intersects $Y$.  For every mapping $\psi:Y\to D$, we define an
instance $I^*_{t,\psi}$ that has satisfying assignment if and only if
$I^*_t$ has a satisfying assignment extending $\psi$. The instance
$I^*_{t,\psi}$ is obtained from $I^*_t$ as follows:
\begin{itemize}
\item For every $u,v\in Y$, we introduce a constraint $\langle (u,v),
  \{(\psi(u),\psi(v)\}\rangle$, forcing that any satisfying assignment
  agrees with $\psi$ on $Y$. Clearly, these constraint create an edge
  between any two vertices of $Y$ in the projection graph.
\item For every $u,v\in \sigma(t)$, we introduce a constraint $\langle
  (u,v), \{(\phi(u),\phi(v)\}\rangle$, forcing that any satisfying
  assignment agrees with $\phi$ on $\sigma(t)$. Note that $I^*_t$
  already has unary constraints forcing the values of these variable,
  but we need these binary constraints to ensure that there is an edge
  between any two vertex of $\sigma(t)$ in the projection graph.
\item For every child $t'$ of $t$ and every pair $u,v\in \sigma(t')$
  such that the projection graph of $I$ contains a directed path from
  $u$ to $v$ fully contained in $\gamma(t')$, we introduce a
  constraint $\langle (u,v), R\rangle$ defined as follows. Pick an
  arbitrary such path from $u$ to $v$ and suppose that $\langle (x_0,x_1),
  R_1\rangle$, $\dots$, $\langle (x_{n-1},x_n), R_{n}\rangle$ with
  $u=x_0$ and $v=x_n$ is a sequence of projection constraints corresponding to
  this directed path. Then we add the constraint $\langle (u,v),
  R_1\circ \dots \circ R_{n}\rangle$, where the product $R_a \circ
  R_b$ of two binary relations over $D$ is defined as $\{(x,y)\in
  D^2\mid \exists z: \text{ $(x,z)\in R_a$ and $(z,y)\in
    R_b$}\}$. Observe that this new constraint is a projection from
  $u$ to $v$.
\end{itemize}
First, let us show that these additional constraints are compatible
with any solution extending $\psi$:
\begin{claim}
$I^*_t$ has a satisfying assignment extending $\psi:Y\to D$ if and only if $I^*_{t,\psi}$ has a satisfying assignment.
\end{claim}
\begin{proof}
  The if direction is clear: $I^*_{t,\psi}$ is more restrictive than
  $I^*_t$, and in particular, the variables in $Y$ are forced to have
  the values given by $\psi$. For the only if direction, consider a
  satisfying  assignment $\psi^*_t$ of $I^*_t$ extending $\psi$. It clear
  that every constraint in the first two groups are satisfied by
  $\psi^*_t$. To see that constraints in the last group are also
  satisfied, recall that, as in the proof of Claim~\ref{cl:extending}, $\psi^*_t$ restricted to $\beta(t)$ can be
  extended to a satisfying assignment $\psi_t$ of $I_t$. Given a
  constraint $\langle (u,v), R_1\circ \dots \circ R_{n}\rangle$,
  assignment $\psi_t$ satisfies each of the constraints $\langle
  (x_0,x_1), R_1\rangle$, $\dots$, $\langle (x_{n-1},x_n),
  R_{n}\rangle$ of $I$, hence it follows that
  $(\psi_t(u),\psi_t(v))=(\psi^*_t(u),\psi^*_t(v))\in R_1\circ \dots
  \circ R_{n}$.
\cqed\end{proof}
Next we show that these additional constraint indeed create a projection sink:
\begin{claim}
The primal graph of $I^*_{t,\psi}$ is a subgraph of $\tau^*(t)$ and $I^*_{t,\psi}$ has a projection sink.
\end{claim}
\begin{proof}
  Observe that every new constraint introduced in the construction of
  $I^*_{t,\psi}$ corresponds to an edge inside a clique of
  $\tau^*(t)$, thus it remains true that the primal graph is a
  subgraph of $\tau^*(t)$. We claim that an arbitrary vertex $y\in Y$
  is a projection sink of $I^*_{t,\psi}$. If $v$ is a variable of
  $I^*_{t,\psi}$ that is not in $\beta(t)$ (i.e., it corresponds to a
  vertex $t_X$ introduced in the definition of $\tau^*(t)$), then all
  the binary constraints on $v$ are projections from $v$ to a neighbor
  of $v$. Therefore, it is sufficient to prove that for every $v\in
  \beta(t)$, the projection graph of $I^*_{t,\psi}$ contains a
  directed path from $v$ to $y$.  By assumption, the projection graph
  of $I$ contains a directed path $P$ from $v$ to $v_0$. We show that
  $P$ can be transformed into a path from $v$ to $y$ in the projection
  graph of $I^*_{i,\psi}$. Every directed edge of $P$ that is
  contained in $\beta(t)$ is also a directed edge of the projection
  graph of $I^*_{t,\psi}$ (since the constrains of $I$ with scope in
  $\beta(t)$ are also contained in $I^*_t$, and $I^*_{t,\psi}$ is more
  restrictive than $I^*_t$). Suppose that $P$ has a subpath $P'$ from
  $a\in \beta(t)$ to $b\in \beta(t)$ with every internal vertex in
  $\alpha(t')$ for some child $t'$ of $t$. Then $a,b\in \sigma(t')$
  and the last group of constraint introduced in the definition of
  $I^*_{t,\psi}$ contains a constraint that is a projection from $a$
  to $b$. Thus subpath $P'$ can be replaced by the edge from $a$ to
  $b$. Suppose that there is a subpath $P'$ of $P$ from $a\in
  \beta(t)$ to $b\in \beta(t)$ with every internal vertex outside
  $\gamma(t)$. Then $a,b\in \sigma(t)$ and there is an edge from $a$
  to $b$ in the projection graph of $I^*_{t,\psi}$ (because of the
  binary constraint on $a$ and $b$ that forces them to have values
  $\phi(a)$ and $\phi(b)$, respectively). Eventually, if $P$ reaches a
  vertex $y'$ of $Y$, then we can terminate the path with an edge from
  $y'$ to $y$: we have introduced a binary constraint on $y$ and $y'$,
  which forces them to have values $\psi(y')$ and $\psi(y)$,
  respectively.  \cqed\end{proof} Therefore, we can solve $I^*_t$ by
solving each instance $I^*_{t,\psi}$ using the algorithm of
Theorem~\ref{th:cspprojsink-genus}. Note that the maximum degree (and
hence the maximum clique size) of $\tau^*(t)$ is bounded by constant
$d^*$ depending only on $H$ and $d$. Therefore, we solve
$|D|^{|Y|}=|I|^{O(d^*)}$ instances and the size of each instance is
polynomial in $|V|$ and $|D|^{O(d^*)}$ (recall that the domain of
$I^*_{t,\psi}$ contains tuples of length $d^*$ over $D$). The time
required to solve each instance is $|I|^{O(pqr)}$, where $p$, $q$, $r$
are all constants depending only on $H$ and $d$. Taking into account
the time required to find the decomposition (which depends only on
$|V|$), the claimed running time $f_1(|V|)\cdot |I|^{f_2(H,d)}$ follows.
\end{proof}

Equipped with Theorem~\ref{th:cspprojsink-minor}, we can prove the
following variant of Lemma~\ref{lem:bigplanarconnectedminor} simply by
replacing Theorem~\ref{th:cspprojsink} with
Theorem~\ref{th:cspprojsink-minor} in the proof:
\begin{lemma}\label{lem:bigplanarconnectedminor}
There exists an algorithm compatible with the description
\begin{eqnarray*}
\sil{\maxdeg(G),\fvs(G)}{\minor(G),\maxdeg(H)}{\ccn(H)\le 1}
\end{eqnarray*}
\end{lemma}
\begin{proof}
  The proof is the same as in Lemma~\ref{lem:bigplanarconnected}, but
  we use Theorem~\ref{th:cspprojsink-minor} instead of
  Theorem~\ref{th:cspprojsink} to solve the CSP instance $I'$. As the
  primal graph $G'$ of $I'$ is a minor of $G$, it is true that $\minor(G')\le \minor(G)$. The number of vertices of $G'$
  is $O(\maxdeg(G)^2\fvs(G))$ and the maximum degree of $G'$ is at
  most $\maxdeg(H)+1$ (recall that in the proof of
  Lemma~\ref{lem:bigplanarconnected}, we assume that the edges
  incident to a vertex of $Z$ are colored with at most $\maxdeg(H)+1$
  colors). Therefore, invoking the algorithm of
  Theorem~\ref{th:cspprojsink-minor} has running time
  $f_1(\maxdeg(G),\fvs(G))\cdot n^{f_2(\minor(G),\maxdeg(H))}$, which
  is compatible with the required specification.
\end{proof}

With the same reduction as in the proof of Theorem~\ref{th:bigplanarminor}, we can handle disconnected graphs:
\begin{ptheorem}\label{th:bigplanarminor}
There exists an algorithm compatible with the description
\begin{eqnarray*}
\sil{\maxdeg(G),\fvs(G)}{\minor(G),\maxdeg(H),\ccn(H)}{}
\end{eqnarray*}
\end{ptheorem}
\begin{proof}
  The reduction to the connected case is the same as in the proof of
  Theorem~\ref{th:bigplanar}. The running time can be analysed in a
  similar way, note that $\minor(G'_c)=\minor(G)$, as joining
  different components and attaching trees do not increase the size of
  the larger clique minor.
\begin{multline*}
\hat f(\maxdeg(G'_c))\hat f(\fvs(G'_c))\cdot n^{\hat f(\minor(G'_c))\hat f(\maxdeg(H'_c))}\\=
\hat f(\max\{\maxdeg(G)+1,\ccn(H)\})\hat f(\fvs(G))\cdot n^{
\hat f(\minor(G))\hat f(\max\{\maxdeg(H)+1,\ccn(H)\})}\\
 \le 
 \hat f(\maxdeg(G)+1)\hat f(\ccn(H))\hat f(\fvs(G))\cdot n^{
 \hat f(\minor(G))\hat f(\maxdeg(H)+1)\hat f(\ccn(H))}\\\le
 \hat f(\maxdeg(G)+1)\hat f(\fvs(G))\cdot n^{\hat f(\ccn(H))\hat f(\minor(G))\hat f(\maxdeg(H)+1)\hat f(\ccn(H))}\\
 =f_1(\maxdeg(G),\fvs(G))\cdot n^{f_2(\ccn(H),\minor(G,\maxdeg(H)))}
\end{multline*}
for some functions $f_1$, $f_2$. Trying all possible $v_1$, $\dots$,
$v_k$ adds an overhead of $n^{\ccn(H)}$, thus the algorithm is
compatible with the specified description.
\end{proof}


\section{Easy and classical negative results}
\label{sec:easy-class-negat}
In this section we survey the negative results that are either known, or follow from very simple reductions. 

\subsection{Bin packing reductions}

We start with a group of simple reductions that follow from hardness of bin packing. As the starting point of all our reductions we take the following {\sc{Unary Bin Packing}} problem.

\defproblemu{{\sc{Unary Bin Packing}}}{Positive integers $s_1,s_2,\ldots,s_p$ denoting the sizes of items, number of bins $k$, and bin capacity $B$, all encoded in unary}{Is there an assignment of all the items to $k$ bins, so that the total size of items assigned to each bin does not exceed capacity $B$?}

{\sc{Unary Bin Packing}} is a classical NP-hard problem~\cite{garey-johnson}. Jansen et al.~\cite{JansenKMS13} observe that the problem is W[1]-hard when parameterized by the number of bins $k$, even though it admits a simple $O(n^{O(k)})$ time algorithm, where $n$ is the length of input.

\begin{ntheorem}\label{thm:bin-packing}
Unless $FPT=W[1]$, there is no algorithm compatible with the description
\begin{eqnarray*}
\sil{\ccn(G)}{}{\maxdeg(G)\leq 2,\tw(G)\leq 1}
\end{eqnarray*}
\end{ntheorem}
\begin{proof}
We provide a polynomial-time parameterized reduction from {\sc{Unary Bin Packing}} parameterized by the number of bins $k$. Given an instance $(\{s_i\}_{1\leq i\leq p},k,B)$ of {\sc{Unary Bin Packing}}, construct an instance $(H,G)$ of \subiso as follows. As $G$ take a disjoint union of $k$ paths of length $B-1$ (thus having $B$ vertices). As $H$ take a disjoint union of $p$ paths of lengths $s_1-1,s_2-1,\ldots,s_p-1$ (thus having $s_1,s_2,\ldots,s_p$ vertices, respectively). Observe that each path of $H$ must be mapped to one of the paths of $G$, and that a subset of paths of $H$ can be simultaneously mapped into one of the paths of $G$ if and only if their total number of vertices does not exceed $B$. Thus, subgraph isomorphisms from $H$ to $G$ correspond to assignments of items to the bins such that the capacities are not exceeded.
\end{proof}

\begin{ntheorem}\label{thm:bin-packing-univ}
Unless $P=NP$, there is no algorithm compatible with the description
\begin{eqnarray*}
\sil{}{\pw(G),\fvs(G)}{\ccn(H)\leq 1,\tw(H)\leq 1,\genus(G)\leq 0}
\end{eqnarray*}
\end{ntheorem}
\begin{proof}
We modify the reduction of Theorem~\ref{thm:bin-packing} to show an NP-hardness reduction from the unparameterized version of {\sc{Unary Bin Packing}}. Assume without loss of generality that $k\geq 4$, since otherwise we can solve the problem in polynomial time. Given an instance $(\{s_i\}_{1\leq i\leq p},k,B)$ of {\sc{Unary Bin Packing}}, perform the same construction as in the proof of Theorem~\ref{thm:bin-packing}, that is, construct $k$ paths of length $B-1$ in $G$, and $p$ paths of lengths $s_1-1,s_2-1,\ldots,s_p-1$ in $H$. Now add a universal vertex $v^*$ in $G$ and a vertex $u^*$ in $H$ that is adjacent to one vertex of each connected component of $H$, chosen arbitrarily. This concludes the construction.

It is easy to see that $G$ and $H$ are connected and planar, and $H$ is moreover a tree. Moreover observe that after removing $v^*$ from $G$, $G$ becomes a disjoint union of paths. It follows that $\pw(G)\leq 2$ and $\fvs(G)\leq 1$.

We now prove that the input and the output instance are equivalent. Note that by the assumption that $k>4$, we have that $u^*$ and $v^*$ are the only vertices of $H$ and $G$, respectively, that have degrees at least $4$. It follows that any subgraph isomorphism $\eta$ from $H$ to $G$ must map $u^*$ to $v^*$. Hence, $\eta$ restricted to $H\setminus u^*$ must be a subgraph isomorphism from $H\setminus u^*$ to $G\setminus v^*$. The rest of the argumentation is the same as in the proof of Theorem~\ref{thm:bin-packing}; note here that any subgraph isomorphism $\eta$ from $H\setminus u^*$ to $G\setminus v^*$ may be extended to a subgraph isomorphism from $H$ to $G$ by putting $\eta(u^*)=v^*$.
\end{proof}

\begin{ntheorem}\label{thm:bin-packing-path}
Unless $P=NP$, there is no algorithm compatible with the description
\begin{eqnarray*}
\sil{}{\pw(G)}{\maxdeg(H)\leq 2, \maxdeg(G)\leq 3, \ccn(G)\leq 1, \tw(G)\leq 1}
\end{eqnarray*}
\end{ntheorem}
\begin{proof}
We again modify the reduction of Theorem~\ref{thm:bin-packing} to show an NP-hardness reduction from the unparameterized version of {\sc{Unary Bin Packing}}. Given an instance $(\{s_i\}_{1\leq i\leq p},k,B)$ of {\sc{Unary Bin Packing}}, perform the same construction as in the proof of Theorem~\ref{thm:bin-packing}, that is, construct $k$ paths of length $B-1$ in $G$, and $p$ paths of lengths $s_1-1,s_2-1,\ldots,s_p-1$ in $H$. Now add paths $P_H$ and $P_G$ both of length $2B+k+1$ to $H$ and $G$, respectively. Thus, $P_G$ has $2B+k+2$ vertices. Let $X$ be the set of $k$ middle vertices on $P_G$, that is, located in distance at least $B+1$ from both endpoints of $P_G$. For each path $P$ of length $B-1$ in $G$, attach one endpoint of $P$ to a vertex of $X$ so that every vertex of $X$ has exactly one path attached. This concludes the construction.

By Lemma~\ref{lem:pathw} we infer that $G$ is a tree of constant pathwidth, and it has maximum degree $3$. Moreover, $H$ is still a disjoint union of paths.

We now prove that the input and the output instance are equivalent. Observe that since $G$ is a tree, the only two vertices in $G$ that can be connected by a path of length exactly $2B+k+1$, are the endpoints of path $P_G$. Since endpoints of $P_H$ can be connected by a path of length $2B+k+1$ in $H$, it follows that any subgraph isomorphism $\eta$ from $H$ to $G$ must necessarily map the endpoints of $P_H$ to the endpoints of $P_G$, and hence also the whole path $P_H$ to path $P_G$, since $G$ is a tree. The rest of the argumentation is the same as in the proof of Theorem~\ref{thm:bin-packing}; note here that any subgraph isomorphism $\eta$ from $H\setminus P_H$ to $G\setminus P_G$ may be extended to a subgraph isomorphism from $H$ to $G$ by mapping $P_H$ to $P_G$.
\end{proof}

\subsection{Known results}

In this section we survey negative results that can be found in the literature. The following theorem follows directly from the result of Garey et al.~\cite{GareyJT76} that the Hamiltonian path problem is NP-hard in planar cubic graphs.

\begin{ntheorem}\label{thm:planar-cubic}
Unless $P=NP$, there is no algorithm compatible with the description
\begin{eqnarray*}
\sil{}{}{\maxdeg(H)\leq 2, \ccn(H)\leq 1, \tw(H)\leq 1, \maxdeg(G)\leq 3, \genus(G)\leq 0}
\end{eqnarray*}
\end{ntheorem}

The following case represents simply W[1]-hardness of the {\sc{Clique}} problem.

\begin{ntheorem}\label{thm:clique}
Unless $FPT=W[1]$, there is no algorithm compatible with the description
\begin{eqnarray*}
\sil{|V(H)|}{\cw(H)}{\ccn(H)\leq 1}
\end{eqnarray*}
\end{ntheorem}

Finally, Fomin et al.~\cite{FominGLS10} proved that the {\sc{Hamiltonian Cycle}} problem is W[1]-hard when parameterized by cliquewidth of the graph. By applying a standard Turing reduction from {\sc{Hamiltonian Cycle}} to {\sc{Hamiltonian Path}}, we can infer that also {\sc{Hamiltonian Path}} does not admit an FPT algorithm unless $FPT=W[1]$. Thus, we obtain the following result.

\begin{ntheorem}\label{thm:hampath-cliquewidth}
Unless $FPT=W[1]$, there is no algorithm compatible with the description
\begin{eqnarray*}
\sil{\cw(G)}{}{\ccn(H)\leq 1,\maxdeg(H)\leq 2, \tw(H)\leq 1}
\end{eqnarray*}
\end{ntheorem}


\section{Hardness results}
\label{sec:hardness-results}

We present the nontrivial hardness proofs of the paper in this
section. With the exception of the reductions in
Section~\ref{sec:embedding-paths-into} (where we introduce 
the intermediary problem \expas), all the reductions in this section are directly from the
following problem:

\defparproblemu{\gridtiling}{For every $1\leq i,j\leq k$, a subset $S_{i,j}\subseteq [n]\times [n]$.}{$k$}{Is there a tiling function $\tiling: [k]\times [k]\to [n]\times [n]$ such that for all $1\leq i,i',j,j'\leq k$ the following holds: (i) $\tiling(i,j)\in S_{i,j}$; (ii) the first coordinates of $\tiling(i,j)$ and $\tiling(i,j')$ are equal; (iii) the second coordinates of $\tiling(i,j)$ and $\tiling(i',j)$ are equal.}

We often denote the first and the second coordinate of the tiling function $\tiling$ as $\tiling_1,\tiling_2$, i.e., $\tiling(i,j)=(\tiling_1(i,j),\tiling_2(i,j))$. If $\tiling$ is a solution to an instance of \gridtiling, we also denote $\tiling_1(i)=\tiling_1(i,1)$ and $\tiling_2(j)=\tiling_2(1,j)$. Thus $\tiling(i,j)=(\tiling_1(i),\tiling_2(j))$ for all $(i,j)\in [k]\times [k]$. Note that properties (ii) and (iii) can be checked only for pairs of indices $i,i'$ and $j,j'$ differing by exactly one. We will also refer to conditions (ii) and (iii) as to row and column conditions, respectively.

The W[1]-hardness of \gridtiling has been proved in
\cite{DBLP:conf/icalp/Marx12}. This problem can be served as a
convenient starting point for proving hardness results for planar
problems: we need to represent each cell $(i,j)$ of the grid with a
gadget whose states can represent the pairs in $S_{i,j}$ and interacts
only with the 4 adjacent gadgets. 

\subsection{Preliminary gadget constructions}
We introduce different basic ways of constructing gadgets, which will be used repeatedly in the reductions of Sections~\ref{sec:embedding-tree-into}--\ref{sec:embedding-paths-into}.

\subsubsection{$1$-in-$n$ choice gadget}

We first introduce a very basic construction that will be used the further constructions. If we have a graph $G$ and a vertex $v\in G$, then we define the operation of {\emph{attaching a key $K_i$}} to $v$ as follows: we attach a path of length $i$ to $v$ with additional two pendant edges (i.e., two new degree-1 neighbors) attached to its other end. The attached subtree will be called the key gadget $K_i$, while the vertex to which it is attached is called its root. Observe that if we have a key gadget $K_i$ in $H$ and we know that it should embed into some key gadget $K_{i'}$ in $G$ while preserving roots, then this is possible if and only if $i=i'$.

We define also the universal key gadget $U_i$ as follows: if we attach a universal key gadget $U_i$ to a vertex, we create a path of length $i+1$ attached to it, and add a pendant edge to every internal vertex of this path. Note that all key gadgets $K_{i'}$ for $1\leq i'\leq i$ can be embedded into a universal key gadget $U_i$.

We now proceed to the main construction of this section. Assume we are given a set of rooted trees $T_1,T_2,\ldots,T_p$ and we would like to implement a choice of exactly one of this trees. More precisely, we would like to design a gadget $H_0$ in $H$ and a its counterpart $G_0$ in $G$, such that the gadget $H_0$ in $H$ may be 'almost' fitted into the counterpart $G_0$ in $G$. The only part that is not fitting is one subtree isomorphic to some $T_i$ that protrudes from $G_0$ via a prescribed interface vertex $\interface$. However, we have a full control of the choice of this subtree: for every $i\in \{1,2,\ldots,p\}$ we can set the subgraph isomorphism inside the gadgets so that exactly $T_i$ will protrude.

\begin{lemma}\label{lem:constpw-choice}
Assume that we are given a set of rooted trees $T_1,T_2,\ldots,T_p$, each of pathwidth at most $c$, size larger than $2p+2$, and maximum degree at most $3$. Let the roots of $T_1,T_2,\ldots,T_p$ be $r_1,r_2,\ldots,r_p$, respectively; assume furthermore that $r_1,r_2,\ldots,r_p$ have degrees at most $2$. Then in polynomial time one can construct two rooted trees $H_0$ and $G_0$, such that the following hold:
\begin{itemize}
\item[(i)] Trees $H_0$ and $G_0$ are rooted in $r_H$ and $r_G$, respectively, have pathwidth at most $c+c'$ for some constant $c'$, and maximum degree $3$. Moreover, $r_H$ and $r_G$ are of degree $1$.
\item[(ii)] There are subtrees $T_1',\ldots,T_p'$ of $H_0$, rooted in some vertices $r_1',\ldots,r_p'$, such that $T_i'$ is isomorphic to $T_i$ with $r_i'$ corresponding to $r_i$.
\item[(iii)] Tree $G_0$ has one prespecified vertex $\interface$, called the {\emph{interface vertex}}.
\item[(iv)] For every $i=1,2,\ldots,p$ there exists a partial subgraph isomorphism $\hm_i$ from $H_0$ to $G_0$ that maps all the vertices of $V(H_0)\setminus V(T_i')\cup \{r_i'\}$ into $V(G_0)$ in such a manner that $\hm(r_H)=r_G$, $\hm(r_i')=\interface$.
\item[(v)] For every partial subgraph isomorphism $\hm$ from $H_0$ to $G_0$ with the following properties: 
\begin{itemize}
\item[(a)] $\hm(r_H)=r_G$, 
\item[(b)] $\hm$ respects boundary $\{\interface\}$,
\end{itemize}
there exists an index $i_0\in \{1,2,\ldots,p\}$ such that $\hm(r_{i_0}')=\interface$ and $V(T_{i_0}')\setminus \{r_{i_0}'\}$ have undefined images in $\hm$.
\end{itemize}
\end{lemma}
\begin{proof}
Let us pick an integer $M>\max_{i=1,2,\ldots,p} |V(T_i)|$. Constructions of $H_0$ and $G_0$ are depicted in Figure~\ref{fig:constpw-choice}.

We begin with construction of tree $H_0$. We construct the root $r_H$ and create a path of length $(p-1)M$ attached to it. Let us denote the vertices of this path by $v_0,v_1,\ldots,v_{(p-1)M}$ in this order, where $v_0=r_H$. For every $i=1,2,\ldots,p-1$, we perform the following steps. Attach a path $s^i_0 - s^i_1 - \ldots - s^i_{(p-i)M}$ of length $(p-i)M$ to the vertex $v_{iM}=s^i_0$, and attach a tree $T_i'$ isomorphic to $T_i$ to $s^i_{(p-i)M}$ so that $s^i_{(p-i)M}=r_i'$. In the same manner, attach a tree isomorphic to $T_p$ to vertex $v_{(p-1)M}$, hanging on a path $s^p_0 - s^p_1 - \ldots - s^p_{M}$ of length $M$, where $v_{(p-1)M)}=s^p_0$. Observe that in this manner, there is exactly one tree attached via a long path to every vertex of form $v_{iM}$ for $i=1,\ldots,p-2$, while for $v_{(p-1)M}$ there are two such trees.

For every tree $T_i'$ and its root $r_i'$, take the parent of $r_i'$ (the vertex one edge closer to the root $r_H$) and denote it by $r_i''$. Note that $r_i''$ has degree $2$ so far. Attach a key gadget $K_i$ to $r_i''$. Intuitively, the role of the key gadgets is to ensure that every tree is embedded into an appropriate slot in $G_0$, as different key gadgets cannot be embedded into one another. This concludes the construction of $H_0$; we now proceed to $G_0$.

We construct the root $r_G$ and create a path of length $pM$ attached to it. Let us denote the vertices of this path by $w_0,w_1,\ldots,w_{pM}$ in this order, where $w_0=r_G$ and $\interface=w_{pM}$. For every $i=1,2,\ldots,p-1$, we attach a tree (called further the {\emph{remainder tree}}) to vertex $w_{iM}$; this tree is the whole subtree of $H_0$ below vertex $v_{iM}$ that contains $T_p'$. For $i=1,2,\ldots,p-1$, let $d_i$ be the child of $w_{iM}$ in this subtree; note that $d_i$ has degree $2$ so far. Attach a path of length $(p-i)M-1$ to $d_i$ and attach a copy of tree $T_i$ (called further {\emph{supply tree}}) at the end of this path; moreover, add the key gadget $K_i$ to the parent of the root of this copy.

Finally, we add a universal key gadget $U_p$ to the parent of $\interface$, that is, vertex $w_{pM-1}$. Note that thus every key gadget of any tree can be embedded into $U_p$ such that the root of the key gadget maps to the root of $U_p$.

Properties (ii) and (iii) follow directly from the construction. For property (i), the only nontrivial claim is the upper bound on pathwidth. To bound the pathwidth of $G_0$, we can first apply Lemma~\ref{lem:pathw} to trees $T_i'$ attached on long paths with keys $K_i$ attached. Then we can again apply Lemma~\ref{lem:pathw} to the whole remainder trees, which are constructed by attaching the graphs considered in the previous sentence to a long path. Finally, we can again apply Lemma~\ref{lem:pathw} to the whole graph $G_0$ which is constructed by attaching the remainder trees and the universal gadget to the path $w_0 - w_1 - \ldots - w_{pM}$. Bounding the pathwidth of $H_0$ can be done in the same manner, but we need to apply Lemma~\ref{lem:pathw} only twice.

\begin{figure}[htbp!]
        \centering
        \subfloat[Graph $H_0$]{
                \centering
                \def\svgwidth{0.55\columnwidth}
                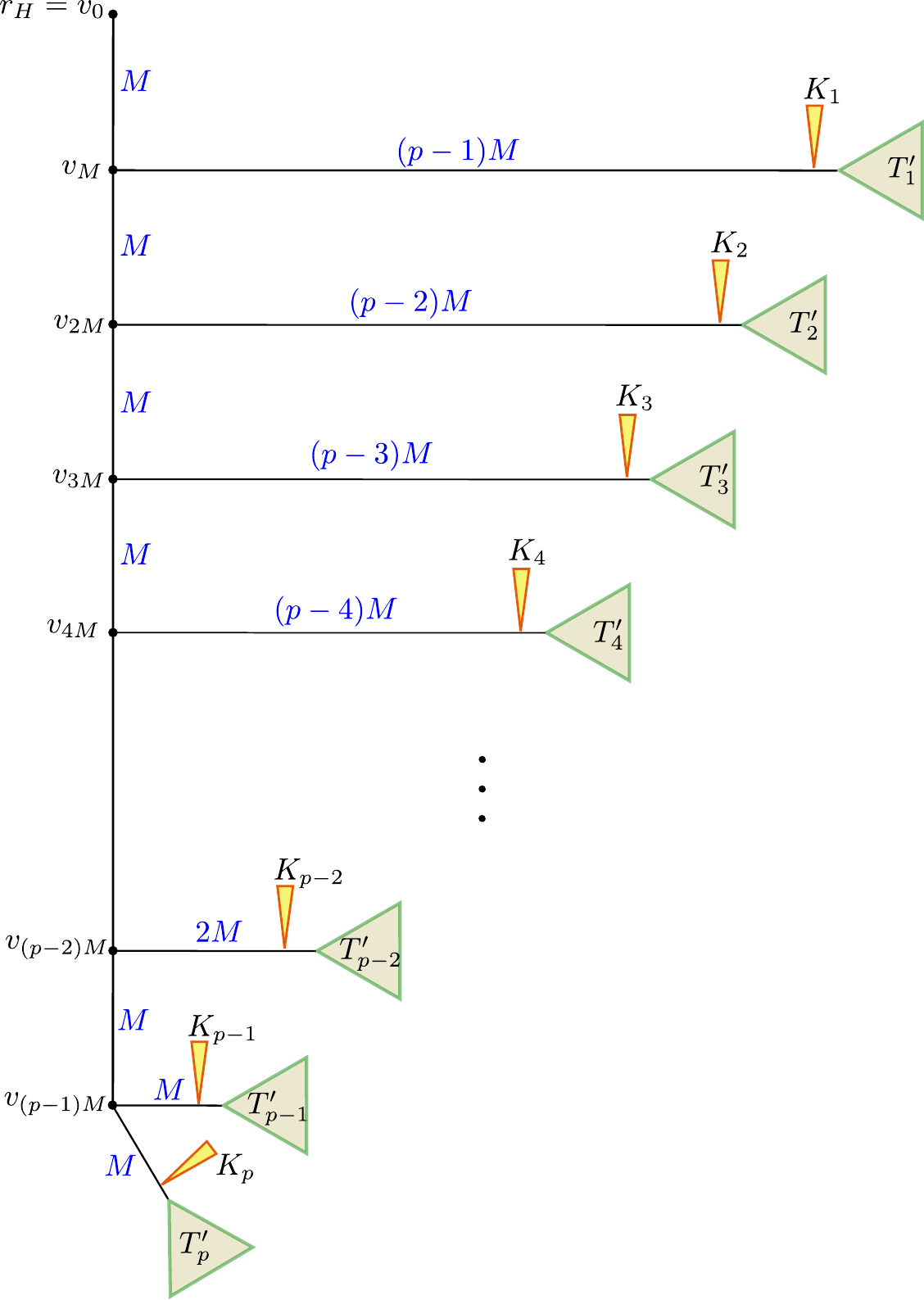
        }
        \qquad
        \subfloat[Graph $G_0$]{
                \centering
                \def\svgwidth{0.33\columnwidth}
                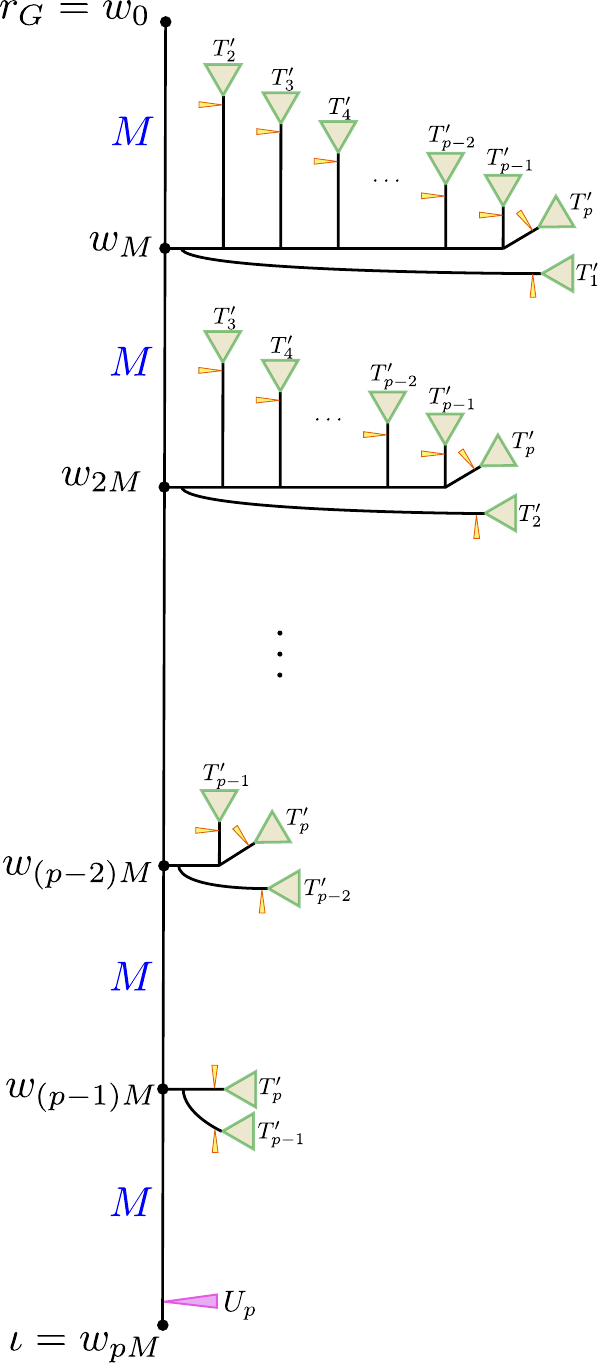
        }
\caption{Construction of Lemma~\ref{lem:constpw-choice}.}\label{fig:constpw-choice}
\end{figure}

To ensure that property (iv) is satisfied, consider the following partial subgraph isomorphism $\hm_i$. We map $r_H$ to $r_G$ and a prefix of path $v_0 - v_1 - \ldots - v_{(p-1)M}$ of length $iM$ to the prefix of path $w_0 - w_1 - \ldots - w_{pM}$ of length $iM$. We map all the trees attached to internal vertices of this prefix in $H_0$ to corresponding supply trees in $G_0$. We are left with mapping two subtrees hanging below $v_{iM}$: one containing a copy of $T_i$ on a path of length $(p-i)M$ (together with a key), and the second containing all the trees $T_j'$ for $j'>i$. We map this second subtree to the remainder tree in $G_0$ which, according to the construction, is isomorphic. We are left with tree $T_i$ on a path of length $(p-i)M$; we map this path to the path connecting $v_{iM}$ with $\interface$ so that $\hm_i(r_i')=\interface$, and leave $V(T_i')\setminus \{r_i'\}$ unmapped. Note that the key gadget above tree $T_i'$ can be mapped into the universal key gadget above $\interface$. The mapping constructed in this manner satisfies all the properties requested from $\hm_i$.

We are left with proving property (v). Assume that $\hm$ is the assumed partial subgraph isomorphism, and let $j$ be defined as the largest index among $1,2,\ldots,p-1$ such that $\hm(v_{jM})=w_{jM}$. Assume first that $j=p-1$. Then path $v_0 - v_1 - \ldots - v_{(p-1)M}$ must be mapped to $w_0 - w_1 - \ldots - w_{(p-1)M}$ and the two subtrees hanging below $v_{(p-1)M}$ in $H_0$ must be mapped to the two subtrees hanging below $w_{(p-1)M}$ in $G_0$. It follows that if the subtree containing $T_{i_0}$ for $i_0\in \{p-1,p\}$ is mapped into subtree containing $\interface$, then $\hm(r_{i_0}')=\interface$ and $V(T_{i_0}')\setminus \{r_{i_0}'\}$ remains unmapped; note that here we use the fact that every tree $T_i$ has size more than $2p+2$, thus tree $T_{i_0}$ cannot be embedded into the universal key $U_p$.

Consider now the case when $j<p-1$. As before, $v_0 - v_1 - \ldots - v_{jM}$ must be mapped to $w_0 - w_1 - \ldots - w_{jM}$ and the two subtrees hanging below $v_{jM}$ in $H_0$ must be mapped to the two subtrees hanging below $w_{jM}$ in $G_0$. However, the subtree containing $T_p'$ in $H_0$ cannot be mapped into the subtree containing $\interface$ in $G_0$, as then we would have that $\hm(v_{(j+1)M})=w_{(j+1)M}$, a contradiction with maximality of $j$. Hence, it is the subtree consisting of a copy $T_j$ attached on a path of length $(p-j)M$ with a corresponding key added above it that is mapped into the subtree containing $\interface$. Consider now the set of vertices that are in distance $(p-j)M-1$ from $w_{jM}$ in this subtree. Each of these vertices has a key gadget $K_{j'}$ for $j'>j$ attached, apart from the vertex $w_{pM-1}$, parent of $\interface$, which has the universal key $U_p$. As no tree $T_i$ can be embedded into any key gadget, and the key gadget $K_j$ cannot be embedded into any other key gadget, we infer that subgraph isomorphism $\hm$ must map the key gadget of $T_j'$ into the universal key $U_p$ attached to $w_{pM-1}$, vertex $r_j'$ into $\interface$, and leave the whole $V(T_j')\setminus \{r_j'\}$ unmapped.
\end{proof}

\subsubsection{Moustache gadgets}

We now proceed to the next type of gadget that we use later. We give two different constructions of a gadget having the same behaviour, but the constructions differ in structural properties. More precisely, the first construction ensures that the gadget has a feedback vertex set constant size at the cost of allowing vertices of unbounded degree, while the second construction ensures that the degrees are bounded by $3$ but does not bound feedback vertex set.

\begin{lemma}\label{lem:gadget-moustachefvs}
Assume we are given a set $S\subseteq [n]\times[n]$ and an integer $M$ larger than $n$. Then in time polynomial in $|S|$ and $M$ one can construct a graph $G_0$ and a set of graphs $\{H_a\}_{a\in [n]}$ such that the following properties hold:
\begin{itemize}
\item[(i)] Each $H_a$ is a tree of constant pathwidth and maximum degree $3$, and has a prescribed root $r_H^a$ and sink $s_H^a$, both being leaves of the tree.
\item[(ii)] $G_0$ is a planar graph with a prescribed root $r_G$, sink $s_G$, and two interface vertices $\interface_1,\interface_2$, all lying on the outer face and of degree $1$.
\item[(iii)] For every $(a,b)\in S$ there exists a partial subgraph isomorphism $\hm_{(a,b)}$ from $H_a$ to $G_0$, such that $\hm_{(a,b)}(r_H^a)=r_G$, $\hm_{(a,b)}(s_H^a)=s_G$, and the only part of $H_a$ not mapped is a path of length $2M-b$ attached to the rest of $H_a$ at a single root which is mapped to vertex $\interface_1$, and path of length $M+b$ attached to the rest of $H_a$ at a single root which is mapped to vertex $\interface_2$.
\item[(iv)] For every partial subgraph isomorphism $\hm$ from $H_a$ to $G_0$, having the following properties: 
\begin{itemize}
\item[(a)] $\hm(r_H)=r_G$,
\item[(b)] $\hm$ respects boundary $\{\interface_1,\interface_2,s_G\}$, and 
\item[(c)] image of $s_H^a$ is undefined or belongs to $\{\interface_1,\interface_2,s_G\}$,
\end{itemize}
we have that $\hm(s_H^a)=s_G$, there exists an index $b$ such that $(a,b)\in S$, there is a path of length $2M-b$ that is unmapped by $\hm$ but its root is mapped to $\interface_1$, and a path of length $M+b$ that is unmapped by $\hm$ but its root is mapped to $\interface_2$.
\item[(v)] There exist two vertices of $G_0$, whose removal make $G_0$ into a forest of constant pathwidth.
\end{itemize}
\end{lemma}
\begin{proof}
Let $M$ be any constant satisfying $M\geq n+10$. Constructions of $H_a$ and $G_0$ are depicted in Figure~\ref{fig:moustache-gadget}.

We begin with the construction of $H_a$. Construct an path $Y$ of length $M$ and denote its vertices by $y_0,y_1,\ldots,y_M$, in this order. Attach two paths of length $5M$ (denoted $X$) and $3M+1$, respectively, at vertex $y_0$. The other end of the path of length $3M+1$ will be the root vertex $r_H^a$. On this path add a single pendant edge to the vertex that is in distance $M+a+1$ from $r_H^a$. Now, perform a symmetric construction at vertex $y_M$: attach two paths of length $5M$ (denoted $Z$) and $3M+1$, respectively, at $y_M$, where the other end of the path of length $3M+1$ is the sink vertex $s_H^a$. Also, add a single pendant edge to the vertex of this path that is in distance $M+a+1$ from $s_H^a$.
This completes the construction of $H_a$.

We now proceed to the construction of $G_0$. Begin with constructing two paths $P$ and $Q$ of length $M$ each, sharing one endpoint $c$. Let $p_0,p_1,\ldots,p_{M-1},c$ be the vertices of $P$ in this order, and $c,q_1,q_2,\ldots,q_M$ be the vertices of $Q$ in this order. Attach a path $P'$ of length $3M$ to $p_0$ and a path $Q'$ of length $3M$ to $q_M$; the endpoints of $P',Q'$, other than $p_0$ and $q_M$, will be the interface vertices $\interface_1,\interface_2$, respectively. Now construct the root $r_G$ and sink $s_G$ and create a pseudo-root $r_G'$ and pseudo-sink $s_G'$ being neighbors of $r_G$ and $s_G$, respectively. The root and the sink will be pendant edges attached to the pseudo-root and the pseudo-sink. For every $b\in [n]$, create two paths of length $3M$: one connecting $r_G'$ with $p_b$ and the second connecting $s_G'$ with $q_b$. For every $(a,b)\in S$ add a single pendant edge to the vertex in distance $M+a$ from $s_G'$ on the path connecting it with $p_b$, and do the symmetric operation for the sink. This concludes the construction of $G_0$.

\begin{figure}[htbp!]
        \centering
        \subfloat[Graph $H_a$]{
                \centering
                \def\svgwidth{0.75\columnwidth}
                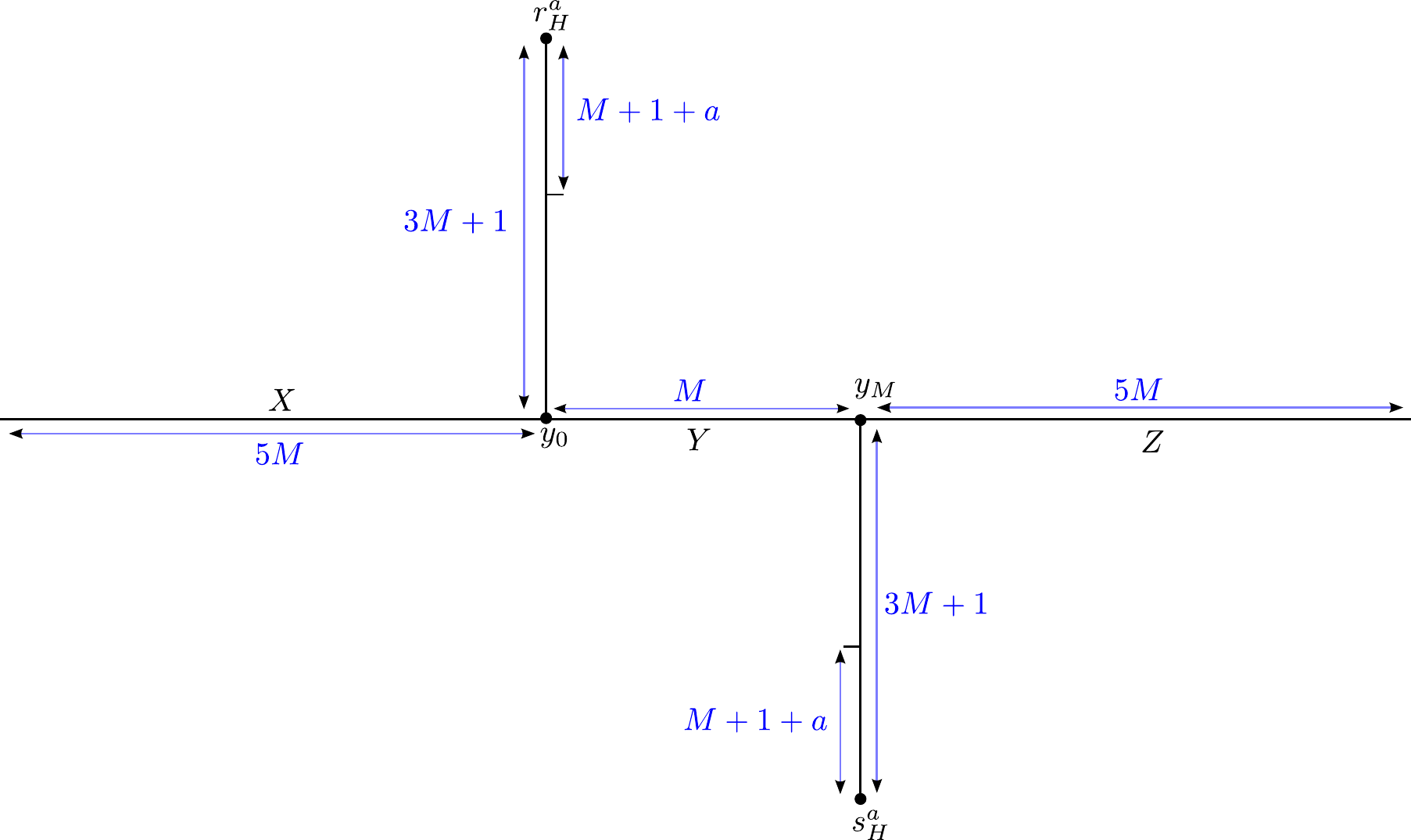
        }
        \qquad
        \subfloat[Graph $G_0$]{
                \centering
                \def\svgwidth{0.75\columnwidth}
                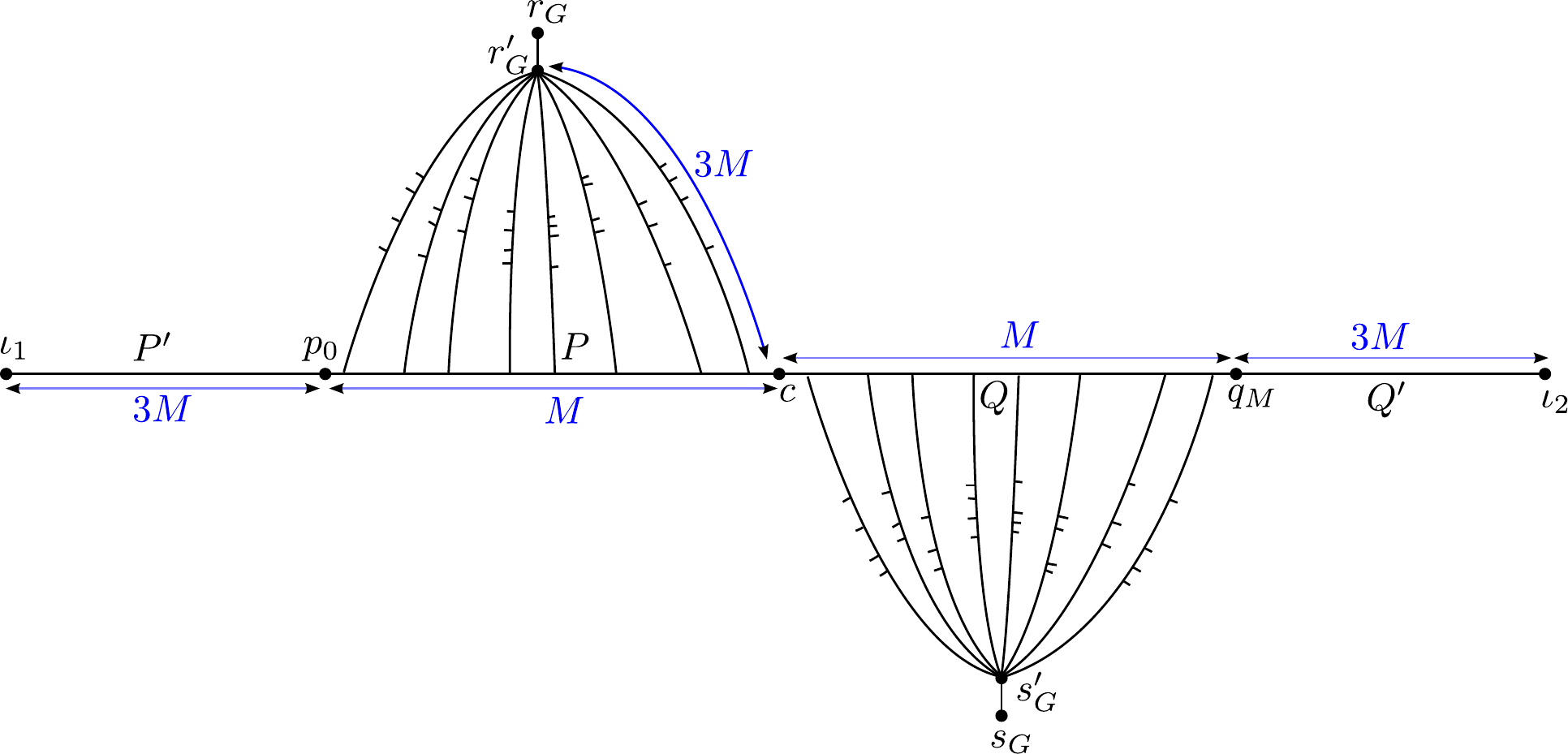
        }
\caption{Construction of Lemma~\ref{lem:gadget-moustachefvs}.}\label{fig:moustache-gadget}
\end{figure}

Observe that each $H_a$ is a tree of constant pathwidth and maximum degree $3$, and $r_H^a$, $s_H^a$ have both degrees one. Graph $G_0$ is planar, with all the special vertices lying on the outer face and having degree $1$. Moreover, removal of $\{r_G',s_G'\}$ from $G_0$ makes it a forest with constant pathwidth. Thus, properties (i), (ii), and (v) are satisfied.

To prove that property (iii) is satisfied, we construct a partial subgraph isomorphism explicitely. Fix $(a,b)\in S$; we aim to construct a partial subgraph isomorphism $\hm_{(a,b)}$. Map the path connecting $r_H^a$ with $y_0$ to the path connecting $r_G$ with $p_b$, and map the path connecting $s_H^a$ with $y_M$ to the path connecting $s_G$ with $q_b$. Note that $(a,b)\in S$ ensures also that the vertices in distance $M+a+1$ from $r_H^a$ and $s_H^a$ on the corresponding paths, which have additional pendant edges, are also mapped to vertices with additional pendant edges in $G_0$; hence, we can extend the mapping also to these pendant edges. Now map the path $Y$ to the fragment of path $P-Q$ between $p_b$ and $q_b$. We are left with mapping the paths $X$ and $Z$, each of length $5M$. We partially map them into the remaining parts of the paths $P'-P$ and $Q'-Q$, leaving a path of length $2M-b$ unmapped on the side of $\interface_1$, and a path of length $M+b$ unmapped on the side of $\interface_2$.

We now prove property (iv). Assume that, for some fixed $a$, we are given a partial subgraph isomorphism $\hm$ that satisfies the assumptions. Consider first the path that connects $r_H^a$ with $y_0$. It follows that this path must be mapped onto one of the paths connecting $r_G'$ with one of the vertices $p_b$ prolonged by the edge $r_Gr_G'$, in such a manner that $y_0$ is mapped onto $p_b$. Moreover, as the pendant edge of the vertex in distance $a+1$ from $r_H^a$ must be mapped somewhere, we infer that the image of this vertex must have an additional neighbor and, by the construction, we infer that $(a,b)\in S$.

Vertex $y_0$ has two disjoint parts of the tree $H_a$ attached, while its image, $p_b$, together with the path from $r_G$ to $p_b$ separates $G_0$ into two components: one containing $\interface_1$, and one containing $\interface_2$ and $s_G$. It follows that each of the considered parts of $H_a$ must be mapped into one of these two components of $G_0$. In the part of $H_a$ that contains $y_M$, we see that there exists a vertex of degree $3$ in distance $M$ from $y_0$, namely $c_M$. Observe that in the part of $G_0$ containing $\interface_1$, none of the vertices of degree $3$ can be reached from $p_b$ by a path of length $M$: the set of vertices of degree $3$ in distance at most $M$ is a subset of $\{p_1,p_2,\ldots,p_{b-1}\}$, and for all of them there is exactly one path connecting them to $p_b$, of length smaller than $M$. Note also that $\interface_1$ is in distance larger than $M$ from $p_b$. We infer that the part of $H_a$ that contains $y_M$ cannot be mapped into the part of $G_0$ that contains $\interface_1$, hence it must be mapped into the second part. As a result, the second part of $H_a$, that is, path $X$ of length $5M$, must be mapped into the part of $G_0$ that contains $\interface_1$.

Let us concentrate on path $X$. As the part of $G_0$ it is mapped into is a tree with all the vertices in distance at most $4M$ from $p_b$, the only way to partially embed $X$ into this component is to map it to the path connecting $p_b$ with $\interface_1$, and leave a subpath of length $2M-b$ unmapped.

Let us now concentrate on the second part of $H_a$. Observe that all the vertices of degree $3$ that are in distance at most $M$ from $p_b$ lie on the path $P-Q$, and exactly one of them, that is $q_b$, can be accessed from $p_b$ by a path of length $M$. We infer that $\hm(y_M)=q_b$ and the path $Y$ is mapped to the subpath of $P-Q$ between $p_b$ and $q_b$. We are left with considering embedding of path $Z$ and path connecting $y_M$ with $s^H_a$. Recall that the image of $s_a^H$ is either undefined, or equal to $\interface_1$, $\interface_2$ or $s_G$; by the reasoning so far we see that $\interface_1$ is not an option. Observe that $s_a^H$ is in distance $3M+1$ from $y_M$. On the other hand, in $G_0$ vertex $s_G$ is in distance $3M+1$ from $q_b$, while $\interface_2$ is in distance larger than $3M+1$ from $q_b$ (recall that $b\leq n\leq M-10$). If the image of $s_a^H$ was undefined, then the interior of the path from $y_M$ to $s_a^H$ would need to contain a vertex mapped to $\interface_2$ or $s_G$; this vertex would be in distance smaller than $3M+1$ from $y_M$, which is a contradiction. On the other hand, $s_a^H$ cannot be mapped to $\interface_2$, as the distance between $q_b$ and $\interface_2$ is larger than the distance between $y_M$ and $s^H_a$. We infer that $s_a^H$ must be mapped to $s_G$ and, consequently, the path between $y_M$ and $s_a^H$ must be mapped to the path between $q_b$ and $s_G'$, prolonged by the edge $s_Gs_G'$.

Finally, observe that the part of $G_0$ into which path $Z$ must be mapped, is a tree with all the vertices in distance at most $4M$ from $q_b$. Hence, the only way to embed path $Z$, which is of length $5M$, is to embed it into the remainder of path $Q-Q'$, leaving $M+b$ vertices unmapped.
\end{proof}

We now modify the construction of Lemma~\ref{lem:gadget-moustachefvs} to obtain bounded maximum degree in graph $G_0$ for the price of possibly unbounded feedback vertex set number. More precisely, we substitute property (v) for property (v'): $G_0$ has maximum degree $3$ and has constant pathwidth.

\begin{lemma}\label{lem:gadget-moustachedegree}
Assume we are given a set $S\subseteq [n]\times[n]$ and an integer $M$ larger than $n$. Then in time polynomial in $|S|$ and $M$ one can construct a graph $G_0$ and a set of graphs $\{H_a\}_{a\in [n]}$ such that properties (i)-(iv) of Lemma~\ref{lem:gadget-moustachefvs} hold, and in addition:
\begin{itemize}
\item[(v')] $G_0$ has maximum degree $3$ and $\pw(G_0)\leq c$ for some constant $c$.
\end{itemize}
\end{lemma}
\begin{proof}
Gadgets $\{H_a\}_{a\in [n]}$ are exactly the same as in the proof of Lemma~\ref{lem:gadget-moustachefvs}. We only change the construction of $G_0$ as follows.

Consider the subgraph of $G_0$ induced by vertices in distance at most $M$ from $r_G$; that is, the root, the pseudoroot and prefixes of length $M-1$ of the paths connecting the pseudoroot with $P$. Note that this part of the $G_0$ does not include any vertices to which pendant edges were attached. Let $s_1,s_2,\ldots,s_\ell$ be the last vertices of these prefixes, for $\ell\leq n$. Remove all the vertices of this subgraph apart from $r_G$, $r'_G$ and $s_i$ for $i=1,2,\ldots,\ell$, and substitute them with the following gadget. Introduce a path $r'_G=x_1 - x_2 - x_3 - \ldots - x_{\ell}$. For every $i=1,2,\ldots,\ell$, connect $x_i$ with $s_i$ by a path of length $M-i$. This concludes the construction of the substitution gadget; note that in this manner the substitution gadget is a tree and all the vertices $s_i$ are in distance $M-1$ from $r'_G$. Perform a symmetric construction for the sink. See Figure~\ref{fig:moustache-substitution} for reference.

\begin{figure}[htbp!]
        \centering
        \subfloat[Before substitution]{
                \centering
                \def\svgwidth{0.25\columnwidth}
                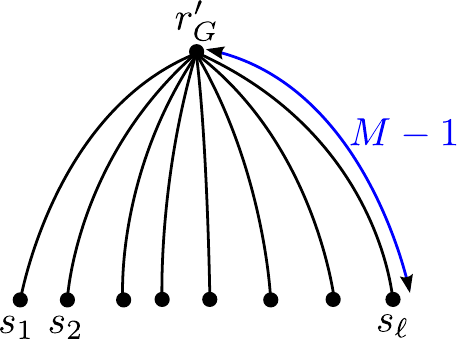
        }
        \qquad $\Longrightarrow$ \qquad
        \subfloat[After substitution]{
                \centering
                \def\svgwidth{0.25\columnwidth}
                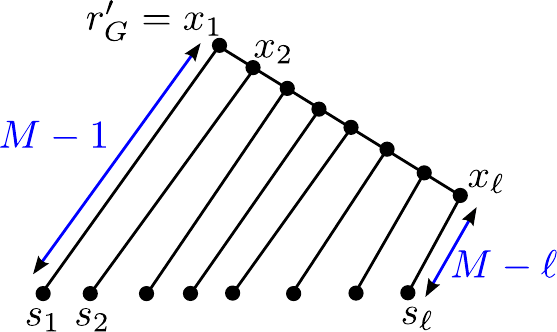
        }
\caption{Substitution in the proof of Lemma~\ref{lem:gadget-moustachedegree}.}\label{fig:moustache-substitution}
\end{figure}

It is easy to observe that $G_0$ constructed in this manner still satisfies property (ii). For property (v'), the only nontrivial part is the bound on pathwidth. Observe, however, that $G_0$ can be constructed from two subdivisions of a grid $2\times \ell$ grid, connected by a path and with some paths attached to different vertices. Since a $2\times \ell$ grid has bounded pathwidth and taking a subdivision or attaching paths can increase the pathwidth by at most $1$, it follows that $G_0$ has bounded pathwidth. The proofs of properties (iii) and (iv) follow the same lines as in Lemma~\ref{lem:gadget-moustachefvs}: choice of a path outgoing from the pseudoroot $r'_G$ is substituted with the choice of a path connecting it with an appropriate vertex $s_i$.
\end{proof}

\subsubsection{Biclique gadget}\label{sec:biclique-gadget}

Assume that we are given vertices $x_1,x_2,\ldots,x_p$ and $y_1,y_2,\ldots,y_p$ and some large constant $M$. Let $M_j=5^j\cdot M$. By introducing a {\emph{biclique gadget}} between $(x_1,x_2,\ldots,x_p)$ and $(y_1,y_2,\ldots,y_p)$ we mean the following construction. For every $j=1,2,\ldots,p$ introduce a vertex $x_j'$ and a path of length $3M_j$ between $x_j'$ and $y_j$. Then, let us introduce a complete bipartite graph with $\{x_1,x_2\ldots,x_p\}$ as one partite set, and $\{x_1',x_2',\ldots,x_p'\}$ as the second. This concludes the construction of the biclique gadget; the construction is depicted in Figure~\ref{fig:biclique}.

\begin{figure}[htbp!]
                \centering
                \def\svgwidth{0.4\columnwidth}
                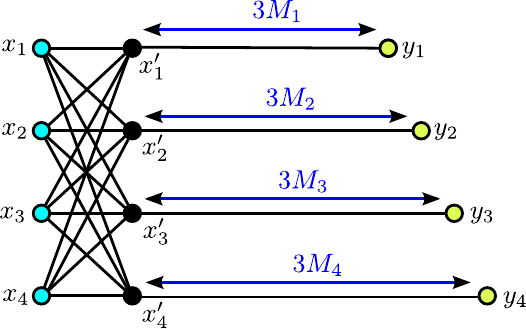
\caption{The biclique gadget for $p=4$}\label{fig:biclique}
\end{figure}

The biclique gadget will be an essential tool for us when we will prove a lower bound requiring cliquewidth of the graph to be constant. Therefore, we unfortunately need to always remember how the construction is performed, as its exact shape will be used in the reasonings that the cliquewidth of the {\emph{whole}} graph is constant. However, we can encapsulate the intended behavior of the gadget with respect to embeddings in the following lemma.

\begin{lemma}\label{lem:biclique-behaviour}
Let $G_0$ be a biclique gadget between $(x_1,x_2,\ldots,x_p)$ and $(y_1,y_2,\ldots,y_p)$ for some integer $M$. Moreover, let $a_1,a_2,\ldots,a_p$ and $b_1,b_2,\ldots,b_p$ be integers such that $1\leq a_j,b_j<M_j$ for every $j=1,2,\ldots,p$. Let $H_0$ be a graph consisting of $2p$ paths $Q_1,Q_2\ldots,Q_p$ and $R_1,R_2,\ldots,R_p$, where the length of each $Q_j$ is equal to $M_j+a_j$ and the length of each $R_j$ is equal to $2M_j-b_j$. Let $q_j,r_j$ be one endpoints of paths $Q_j,R_j$, respectively. Then the following conditions are equivalent:
\begin{itemize}
\item[(i)] There exists a subgraph isomorphism $\hm$ from $H_0$ to $G_0$ such that $\hm(q_j)=x_j$ and $\hm(r_j)=y_j$ for each $j=1,2,\ldots,p$.
\item[(ii)] $a_j\leq b_j$ for each $j=1,2,\ldots,p$.
\end{itemize}
\end{lemma}
\begin{proof}
For $j=1,2,\ldots,p$, let $P_j$ be the unique shortest path of length $3M_j+1$ from $x_j$ to $y_j$, which first accesses vertex $x_j'$ and then continues along a path of length $3M_j$ to $y_j$. Note that paths $\{P_j\}$ are pairwise vertex disjoint.

Assume first that condition (ii) is satisfied; we are going to construct the subgraph isomorphism $\hm$. Paths $Q_j$ and $R_j$ are mapped into path $P_j$ in such a manner that $q_j$ is mapped to $x_j$ (first end of $P_j$) and $r_j$ is mapped to $y_j$ (second end of $P_j$). Note that since $a_j\leq b_j$, we have that $(M_j+a_j)+(2M_j-b_j)\leq 3M_j$, so both paths $Q_j$ and $R_j$ can simultaneously fit into $P_j$. Since paths $P_j$ are pairwise vertex-disjoint $\hm$ constructed in this manner is a valid subgraph isomorphism.

Assume now that condition (i) is satisfied. Consider first the second vertices on paths $\{Q_j\}$, i.e., neighbors of $\{q_j\}$ on these paths. These neighbors must be mapped to neighbors of $\{x_j\}$, that is, to vertices $\{x_j'\}$, yet there is only $p$ of them. It follows that neighbors of $\{q_j\}$ on paths $Q_j$ must be mapped to vertices $\{x_j'\}$ via some bijection; let us assume that the neighbor of $q_j$ is mapped to $x_{h(j)}$, where $h$ is a permutation of $[p]$. We infer that then the images of paths $Q_j$ must continue to be mapped into the paths from the corresponding vertex $x_{h(j)}'$ to the vertex $y_{h(j)}$. 

We are now going to prove that the permutation $h$ is identity. Assume otherwise that $h$ is not identity and let $t$ be the largest index such that $h(t)\neq t$. Since $h$ is a permutation of $[p]$, we infer that $h(t)<t$. Hence, the suffix of path $Q_t$ after the first edge, which is of length at least $M_t$, is to be mapped into path between $x_{h(t)}'$ and $y_{h(t)}'$, which is of length at $3M_{h(t)}<M_t$. This is a contradiction.

Now that we know that $h$ is identity, we infer that for every $j=1,2,\ldots,p$, paths $Q_j$ and $R_j$ must be simultaneously fit into $P_j$. Since $P_j$, $Q_j$ and $R_j$ are of length $3M_j+1$, $M_j+a_j$, and $2M_j-b_j$, respectively, this is only possible if $(M_j+a_j)+(2M_j-b_j)\leq 3M_j$, which is equivalent to $a_j\leq b_j$.
\end{proof}

\subsection{Embedding a tree into a planar graph}
\label{sec:embedding-tree-into}

We first provide a family of reductions that prove the hardness of various special cases of \subiso where a tree is to be embedded into a planar graph or into a graph of small genus. It turns out that we can observe a delicate interplay between various parameters: the number of connected components of $H$, the maximum degrees of $G$ and $H$, the feedback vertex set number of $G$, and the genus of $G$. We first begin with a group of reductions that intuitively show an interaction between the number of connected components of $H$ and topological complexity of $G$. We proceed further by looking closer at the case when both $H$ and $G$ are required to be connected and planar, and observe that then the crucial parameters are the feedback vertex set number of $G$ and the maximum degrees of $G$ and $H$.

\subsubsection{Connectedness of $H$ versus topological complexity of $G$}

We start with the simplest of the reductions, which is also a base for the later ones.

\begin{lemma}\label{lem:grid-manycomp}
There exists a polynomial-time reduction that, given an instance of \gridtiling with parameter $k$, outputs an equivalent instance $(H,G)$ of \subiso with the following properties:
\begin{multicols}{2}
\begin{itemize}
\item $\ccn(H)=k^2$,
\item $H$ is a forest of constant pathwidth, and 
\item $\maxdeg(H)\leq 3$;
\end{itemize}
\vfill
\columnbreak
\begin{itemize}
\item $\ccn(G)=1$,
\item $G$ is planar and $\maxdeg(G)\leq 3$, and
\item $\pw(G),\fvs(G)=O(k^2)$.
\end{itemize}
\end{multicols}
\end{lemma}

The immediate corollary of Lemma~\ref{lem:grid-manycomp} is the following:

\begin{ntheorem}\label{thm:grid-manycomp}
Unless $FPT=W[1]$, there is no algorithm compatible with the description
\begin{eqnarray*}
\sil{\ccn(H),\pw(G),\fvs(G)}{\pw(H)}{\tw(H)\leq 1, \ccn(G)\leq 1, \maxdeg(G)\leq 3, \genus(G)\leq 0}
\end{eqnarray*}
\end{ntheorem}

We proceed to the proof of Lemma~\ref{lem:grid-manycomp}.

\newcommand{\ii}{\alpha}
\newcommand{\jj}{\beta}

\begin{proof}[Proof of Lemma~\ref{lem:grid-manycomp}]
Let $(\{S_{i,j}\}_{1\leq i,j\leq k})$ be the given \gridtiling instance, where $S_{i,j}\subseteq [n]\times [n]$. Choose an integer $M>\max(2n,10)$.

We first create a family of rooted trees $T_{\ii,\jj}$ for $1\leq \ii,\jj \leq n$, which encode the choice of an element $(\ii,\jj)$ from $S_{i,j}$. To construct $T_{\ii,\jj}$, start with creating a binary tree of depth $3$, thus having $8$ leaves $l_1,l_2,\ldots,l_8$ (ordered as in the prefix traversal of the tree). Add a key $K_t$ to leaf $l_t$, for every $1\leq t\leq 8$. Moreover, add paths of length $\ii$ to $l_1$ and $l_5$, paths of length $M-\ii$ to $l_2$ and $l_6$, paths of length $\jj$ to $l_3$ and $l_7$ and paths of length $M-\jj$ to $l_4$ and $l_8$.

We now create the graphs $H$ and $G$ at the same time. For every pair of indices $(i,j)\in [k]\times [k]$, create trees $H^{(i,j)}$ (added to $H$ as disjoint components), and $G^{(i,j)}$ (added to $G$ as disjoint components) given by Lemma~\ref{lem:constpw-choice} for the set of trees $\{T_{\ii,\jj}\ |\ (\ii,\jj)\in S_{i,j}\}$. Denote the roots of $H^{(i,j)}$ and $G^{(i,j)}$ as $r^{(i,j)}_H$ and $r^{(i,j)}_G$, respectively, and the interface vertex of $G^{(i,j)}$ as $\interface^{(i,j)}$. 

\begin{figure}[htbp!]
                \centering
                \def\svgwidth{0.7\columnwidth}
                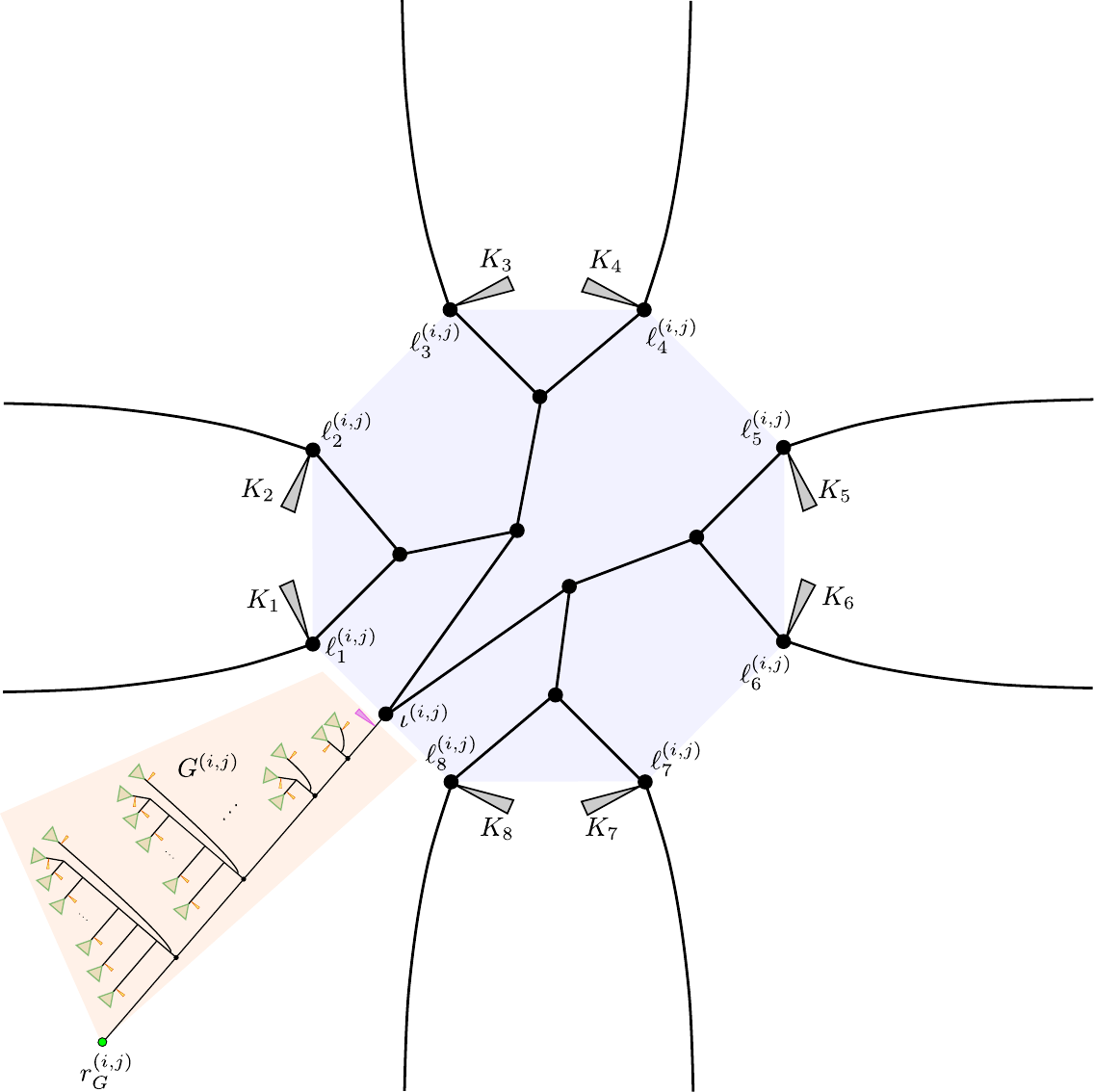
\caption{A closer look on the construction of one grid node $(i,j)$ and the corresponding gadget $G^{(i,j)}$.}\label{fig:octagon}
\end{figure}

We now continue the construction of $G$ as follows. For every pair of indices $i,j$, $1\leq i,j\leq k$, construct a binary tree of depth $3$ rooted in $\interface^{(i,j)}$, and denote its leaves by $l^{(i,j)}_t$ for $1\leq t\leq 8$ (ordered as in the prefix traversal of the tree). Add a key $K_t$ to leaf $l_t^{(i,j)}$, for every $1\leq t\leq 8$. The part constructed so far will be referred to as {\emph{node}} $(i,j)$; we now add some paths which connect neighboring nodes, i.e., nodes with exactly one coordinate differing by exactly one. For every pair of indices $1\leq i,j\leq k$ such that $j>1$, connect the following vertices via paths of length $3M+1$: $l_1^{(i,j)}$ with $l_6^{(i,j-1)}$, and $l_2^{(i,j)}$ with $l_5^{(i,j-1)}$. For every pair of indices $1\leq i,j\leq k$ such that $i>1$, connect the following vertices via paths of length $3M+1$: $l_3^{(i,j)}$ with $l_8^{(i-1,j)}$, and $l_4^{(i,j)}$ with $l_7^{(i-1,j)}$. For every index $i$, $1\leq i\leq k$, connect the following vertices via paths of length $3M+1$: $l_1^{(i,1)}$ with $l_2^{(i,1)}$, and $l_5^{(i,k)}$ with $l_6^{(i-1,k)}$. Finally, for every index $j$, $1\leq j\leq k$, connect the following vertices via paths of length $3M+1$: $l_3^{(1,j)}$ with $l_4^{(1,j)}$, and $l_7^{(k,j)}$ with $l_8^{(k,j)}$. See Figures~\ref{fig:octagon} and~\ref{fig:1-in-n-grid} for reference.

Let $G'$ be the part of $G$ with gadgets $G^{(i,j)}$ removed, but vertices $\interface^{(i,j)}$ left; in particular, $G'$ contains the keys attached to vertices $l_t^{(i,j)}$ for $(i,j)\in [k]\times [k]$ and $t\in [8]$. We refer to $G'$ as to the {\emph{node grid}}: it consists of nodes connected via long paths in a grid manner.

The core of the construction is ready; now, using Lemma~\ref{lem:setting-images} add trees of constant pathwidth to vertices $r^{(i,j)}_H$ and $r^{(i,j)}_G$ to ensure that for each $i,j\in [k]\times[k]$, $r^{(i,j)}_H$ is mapped to $r^{(i,j)}_G$. Note that this application does not introduce vertices of degrees larger than $3$, as vertices $r^{(i,j)}_H$ and $r^{(i,j)}_G$ had degrees $1$.

\begin{figure}[htbp!]
        \centering
        \subfloat[Graph $H$]{
                \centering
                \def\svgwidth{0.45\columnwidth}
                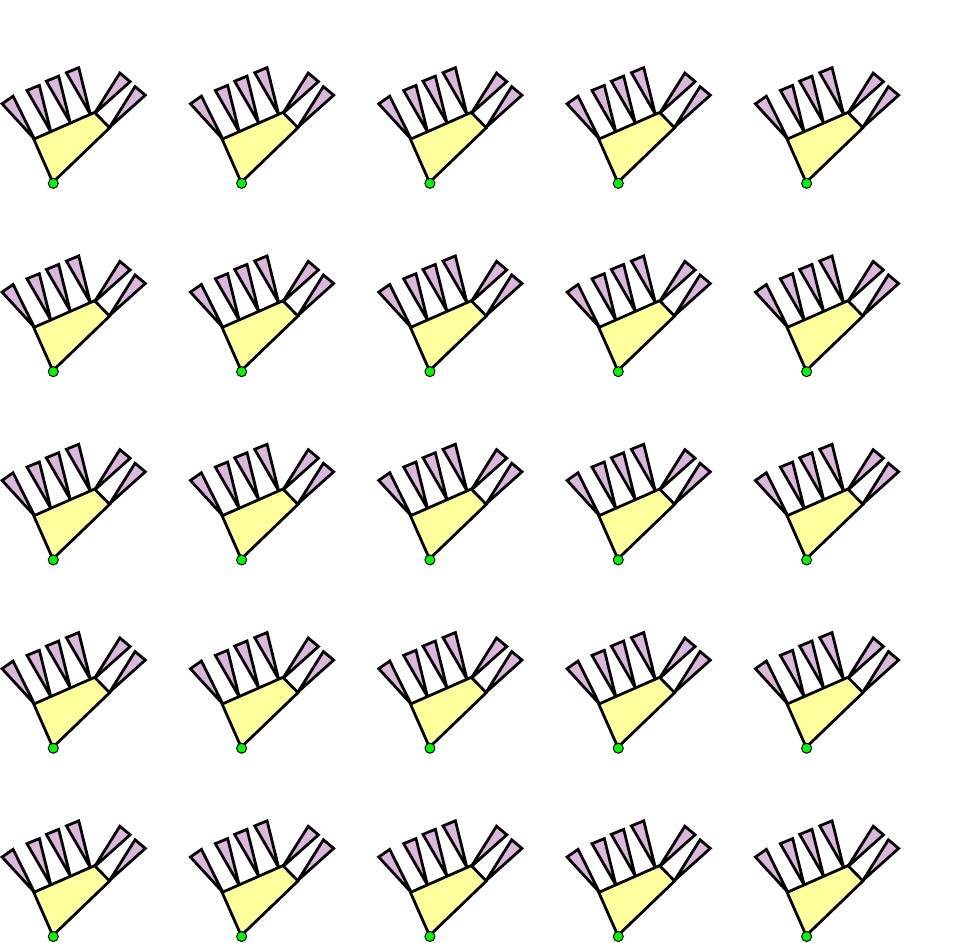
        }
        \qquad
        \subfloat[Graph $G$]{
                \centering
                \def\svgwidth{0.45\columnwidth}
                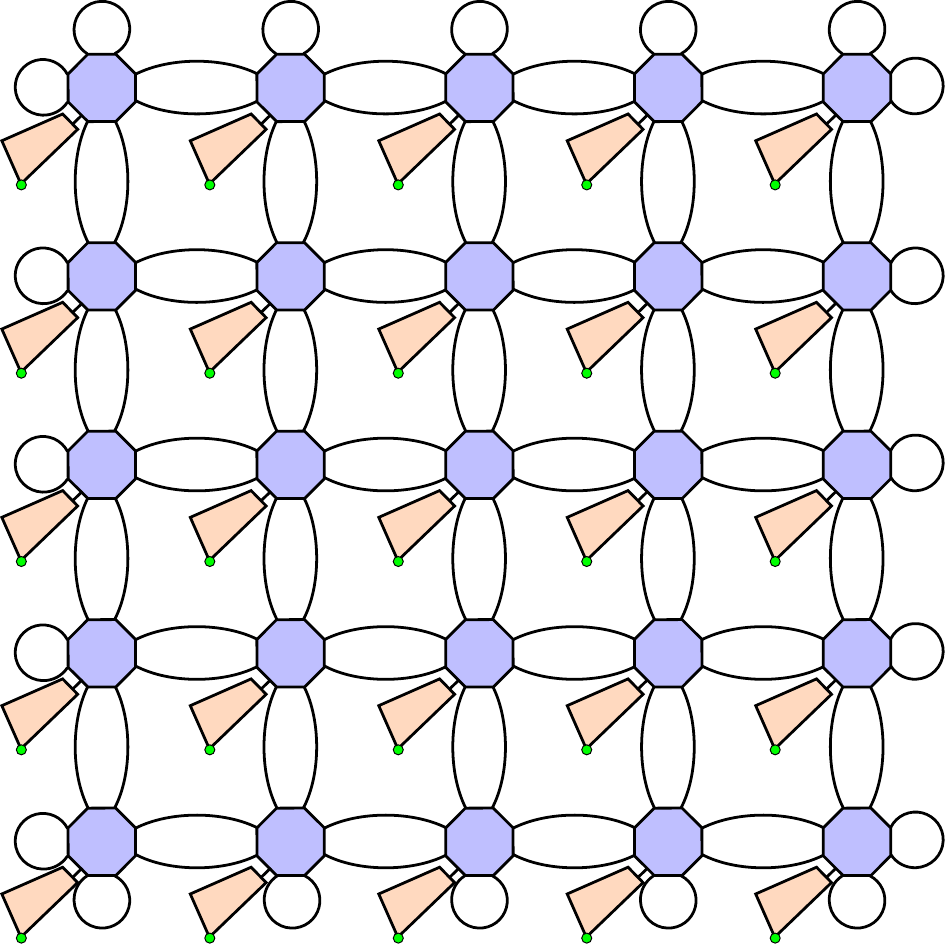
        }
\caption{Schematic overview of the construction of Lemma~\ref{lem:grid-manycomp}. The quadrilaterals on the left picture depict gadgets $H^{(i,j)}$ with trees $T_{\ii,\jj}$ protruding. One can embed each of these gadgets into the corresponding gadget $G^{(i,j)}$ in $G$ (quadrilaterals on the right side), apart from exactly one protruding tree, chosen as one likes, that needs to be embedded into the node grid $G'$.}\label{fig:1-in-n-grid}
\end{figure}

This concludes the construction of the instance. By Lemma~\ref{lem:constpw-choice}, $H$ is a disjoint union of graphs of constant pathwidth, so it has constant pathwidth. Moreover, it has $k^2$ connected components, is a forest, and has maximum degree $3$. On the other hand $G$ has also maximum degree $3$, and is planar. To see that $G$ has pathwidth and optimum feedback vertex set of size $O(k^2)$, observe that after deleting $O(k^2)$ vertices in the binary trees of depth $8$ attached to vertices $\interface^{(i,j)}$, $G$ becomes a forest of constant pathwidth. We now formally prove the equivalence.

Assume that we are given a solution $\tiling$ to the given \gridtiling instance. We construct a subgraph isomorphism for the graph before application of Lemma~\ref{lem:setting-images} that satisfies the property that $r^{(i,j)}_H$ is mapped to $r^{(i,j)}_G$ for each $(i,j)\in [k]\times [k]$; by Lemma~\ref{lem:setting-images} we may then extend this subgraph isomorphism also on the additional gadgeteering introduced by this application.

We set the image of each $r^{(i,j)}_H$ to $r^{(i,j)}_G$. Using Lemma~\ref{lem:constpw-choice}, we extend the subgraph isomorphism to gadgets $H^{(i,j)}$ so that
a tree $T_{\tiling(i,j)}$ protrudes out from the gadget  $G^{(i,j)}$, that is, tree
 $H^{(i,j)}$ is partially mapped to $G^{(i,j)}$ so that the root of a tree isomorphic to $T_{\tiling(i,j)}$ is mapped to $\interface^{(i,j)}$, and this tree consists of the vertices of $H^{(i,j)}$ not mapped so far. To map trees $T_{\tiling(i,j)}$ into the remaining part of $G$, map the binary trees of depth $3$ that appear in the first four levels of $T_{\tiling(i,j)}$ to the binary trees rooted in $\interface^{(i,j)}$ (i.e., to the node $(i,j)$) so that the indices of leaves are preserved. Moreover, map the keys in $H$ to the corresponding keys in $G$. Extend the subgraph isomorphism so that the long paths attached to these leaves in $H$ are mapped into the paths connecting node $(i,j)$ with neighboring nodes (or with itself, if it is on the edge of the node grid). Observe that the assumption about $\tiling$ being a solution ensure that the two paths mapped into a connection between neighboring nodes will always fit: for example, the connection (of length $3M+1$) between $l_1^{(i,j)}$ with $l_6^{(i,j-1)}$ will accommodate paths of lengths $M+\tiling_1(i,j)$ and $2M-\tiling_1(i,j-1)$, which both fit due to $\tiling_1(i,j)=\tiling_1(i,j-1)$. The same argument can be applied to the other $3$ types of connections, and to the loops at the sides of the node grid.

Now assume that we are given a subgraph isomorphism $\hm:V(H)\to V(G)$. By Lemma~\ref{lem:setting-images}, we have that $\hm(r^{(i,j)}_H)=r^{(i,j)}_G$ for every $(i,j)\in [k]\times [k]$, and $\hm$ can be restricted to the graphs before application of Lemma~\ref{lem:setting-images} so that it is still a subgraph isomorphism. From now on, we work with this restriction keeping in mind that $\hm(r^{(i,j)}_H)=r^{(i,j)}_G$ for every $(i,j)\in [k]\times [k]$.

Using Lemma~\ref{lem:constpw-choice} we infer, that for every pair of indices $(i,j)$ there is a pair of indices $\tiling(i,j)=(\tiling_1(i,j),\tiling_2(i,j))\in S_{i,j}$ such that a tree isomorphic to $T_{\tiling(i,j)}$ is embedded into the node grid $G'$ in such a manner that the root of $T_{\tiling(i,j)}$ is mapped into $\interface^{(i,j)}$. We are to prove that $\tiling$ is a solution to the input \gridtiling instance. As we already know that $\tiling(i,j)\in S_{i,j}$, it remains to prove that for all the relevant indices $i,j$ it holds that $\tiling_1(i,j)=\tiling_1(i,j-1)$ and $\tiling_2(i,j)=\tiling_2(i-1,j)$.

First consider the binary trees of depth $3$ in the nodes of the node grid $G'$. As the key gadgets must be fit appropriately and long paths between nodes cannot accommodate any key gadget, it follows that every such binary tree in $H^{(i,j)}$ must be embedded into the corresponding tree in the node $(i,j)$ in such a manner that leaves are mapped to leaves with the same indices. Moreover, the key gadgets must be mapped onto the corresponding ones, so the long paths attached to the leaves in $H$ must be mapped into the long paths connecting two neighboring nodes.

Consider two indices $i,j$, such that $1\leq i\leq k$ and $1<j\leq k$. We are to prove that $\tiling_1(i,j)=\tiling_1(i,j-1)$. Consider path between $l_1^{(i,j)}$ and $l_6^{(i,j-1)}$. This path is of length $3M+1$, and accommodates two paths protruding from $l_1^{(i,j)}$ and $l_6^{(i,j-1)}$ of lengths $M+\tiling_1(i,j)$ and $2M-\tiling_1(i,j-1)$. It follows that $M\geq \tiling_1(i,j)+M-\tiling_1(i,j-1)$, hence $\tiling_1(i,j)\leq \tiling_1(i,j-1)$. As path between $l_2^{(i,j)}$ and $l_5^{(i,j-1)}$ must accommodate paths of lengths $2M-\tiling_1(i,j)$ and $M+\tiling_1(i,j-1)$ protruding from different endpoints, we similarly infer that $\tiling_1(i,j)\geq \tiling_1(i,j-1)$ and, consequently, $\tiling_1(i,j)=\tiling_1(i,j-1)$. This proves the row condition of the \gridtiling problem; the column condition can be proved by an analogous reasoning on the second coordinate.
\end{proof}

We now make two simple modifications of the construction of Lemma~\ref{lem:grid-manycomp}, that show that we may require connectedness of $H$ at a cost of allowing more complicated topological structure of $G$. In both cases the genus of $G$ needs to be allowed as a parameter, yet we can assume that $G$ does not admit a constant-size clique as a minor only if we allow the maximum degrees of $G$ and $H$ as parameters.

\begin{lemma}\label{lem:grid-manycomp-conn-1}
There exists a polynomial-time reduction that, given an instance of \gridtiling with parameter $k$, outputs an equivalent instance $(H,G)$ of \subiso with the following properties:
\begin{multicols}{2}
\begin{itemize}
\item $\ccn(H)=1$,
\item $H$ is a tree of constant pathwidth, and
\item $\maxdeg(H)\leq \max(4,k^2+1)$;
\end{itemize}
\vfill
\columnbreak
\begin{itemize}
\item $\ccn(G)=1$,
\item $\maxdeg(G)\leq \max(4,k^2+1)$,
\item $\genus(G)=O(k^2)$, $\minor(G)\leq 5$, and
\item $\pw(G),\fvs(G)=O(k^2)$.
\end{itemize}
\end{multicols}
\end{lemma}
\begin{proof}
We only describe the difference in the construction. Add vertices to $r^*_H$ and $r^*_G$ to $H$ and $G$, respectively, and make them adjacent to all the vertices $r^{(i,j)}_H$ and $r^{(i,j)}_G$ for $1\leq i,j\leq k$, respectively. Moreover, when applying Lemma~\ref{lem:setting-images} ensure additionally that $r^*_H$ is forced to be mapped onto $r^*_G$. Note that in this manner degrees of $r^*_H$ and $r^*_G$ are $k^2+1$, while the degrees of vertices $r^{(i,j)}_H$ and $r^{(i,j)}_G$ for $1\leq i,j\leq k$ are increased by $1$, so the maximum degrees of $G$ and $H$ are at most $\max(4,k^2+1)$. Graph $H$ becomes a tree, and it is clear that its pathwidth is still constant since removing $r^*_H$ breaks $H$ into components of constant pathwidth. It is also easy to observe that $\pw(G),\fvs(G)\leq O(k^2)$ by the same argument as in the proof of Lemma~\ref{lem:grid-manycomp}, i.e., after removing all the grid nodes $G$ becomes a forest of constant pathwidth. Moreover, the $k^2$ edges between $r^*_G$ and $r^{(i,j)}_G$ for $1\leq i,j\leq k$ can be realized using $k^2$ additional handles attached to the surface into which $G$ is embedded, which shows that $\genus(G)=O(k^2)$.  Finally, after removing $r^*_G$, graph $G$ becomes planar; hence $G$ is an apex graph and $\minor(G)\leq 5$ follows. The proof of equivalence of instances follows the same lines, with the exception that vertex $r^*_H$ must be mapped to $r^*_G$.
\end{proof}

\begin{lemma}\label{lem:grid-manycomp-conn-many}
There exists a polynomial-time reduction that, given an instance of \gridtiling with parameter $k$, outputs an equivalent instance $(H,G)$ of \subiso with the following properties:
\begin{multicols}{2}
\begin{itemize}
\item $\ccn(H)=1$,
\item $H$ is a tree of constant pathwidth, and
\item $\maxdeg(H)\leq 3$;
\end{itemize}
\vfill
\columnbreak
\begin{itemize}
\item $\ccn(G)=1$,
\item $\maxdeg(G)\leq 3$,
\item $\genus(G)=O(k^2)$, and
\item $\pw(G),\fvs(G)=O(k^2)$.
\end{itemize}
\end{multicols}
\end{lemma}
\begin{proof}
We only describe the difference in the construction. For every gadget $H^{(i,j)}$ in $H$ and $G^{(i,j)}$ in $G$, consider two vertices of degree $2$ closest to the root of the gadget (that is, closest to $r^{(i,j)}_H$ in $H^{(i,j)}$ and to $r^{(i,j)}_G$ in $G^{(i,j)}$). Using these two vertices in each of the gadgets, arbitrarily connect all the gadgets $H_{(i,j)}$ in $H$ in a path-like manner, using for connections paths of length much larger then the total number of vertices used in the construction so far. That is, order the gadgets in any manner and for every two consecutive connect the first chosen vertex in the first gadget with the second chosen vertex in the second gadget, using a path an appropriate length. Perform the same construction both in $G$ as well, where the chosen order of gadgets is the same. The augmentation of the construction is depicted in Figure \ref{fig:1-in-n-grid-conn}.

\begin{figure}[htbp!]
        \centering
        \subfloat[Graph $H$]{
                \centering
                \def\svgwidth{0.45\columnwidth}
                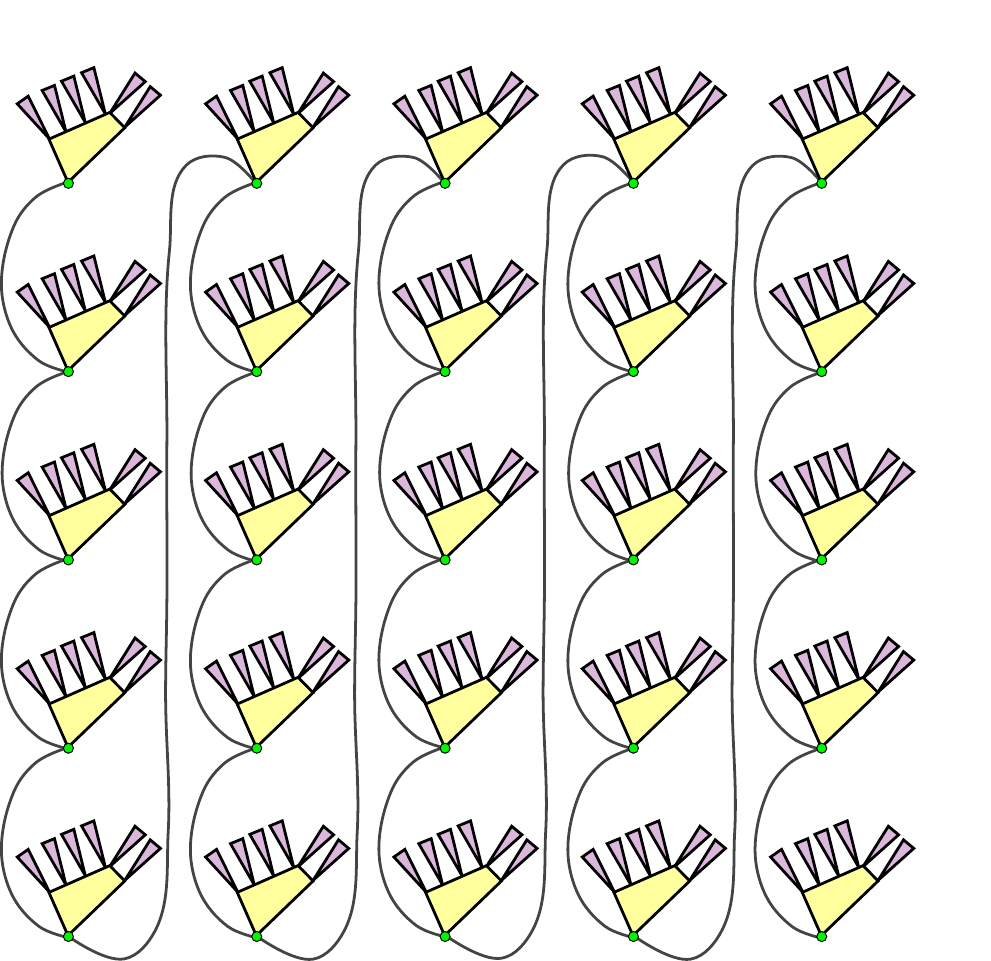
        }
        \qquad
        \subfloat[Graph $G$]{
                \centering
                \def\svgwidth{0.45\columnwidth}
                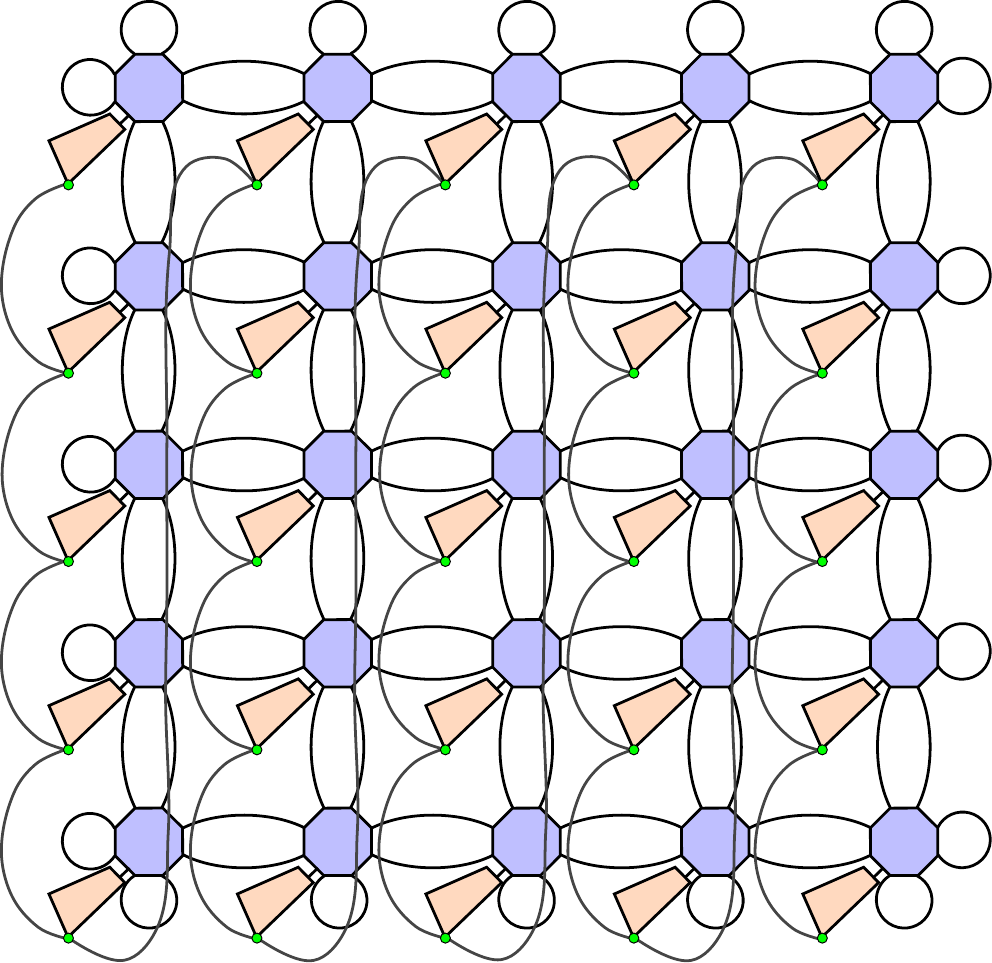
        }
\caption{Schematic overview of the modification in the proof of Lemma~\ref{lem:grid-manycomp-conn-many}. Although we can choose an arbitrary ordering of gadgets $H^{(i,j)}$ and $G^{(i,j)}$ for connections, we choose the one used in the proof of Lemma~\ref{lem:grid-manycomp-conn-many-biclique}.}\label{fig:1-in-n-grid-conn}
\end{figure}

A close examination of the proof of Lemma~\ref{lem:grid-manycomp} shows that the obtained instance is also equivalent to the input \gridtiling instance. One needs just to observe that parts of the gadgets $H_{(i,j)}$ cannot be embedded into the long connections between gadgets $G^{(i,j)}$, as the connections are too long and therefore cannot accommodate part of the tree $H^{(i,j)}$ below $r^{(i,j)}_H$, which contains vertices of degree $3$. Observe also that the augmentation of the construction did not introduce any vertices of degree larger than $3$. Clearly, $H$ is a tree as it resulted from a forest connected in a path-like manner, and it is easy to see that its pathwidth is still constant using Lemma~\ref{lem:pathw}. To see that $\genus(G)=O(k^2)$ one just needs to observe that the $k^2-1$ additional connections introduced in $G$ can be realized in $k^2-1$ additional handles attached to the surface into which $G$ is embedded.
\end{proof}

Lemmas~\ref{lem:grid-manycomp-conn-1} and~\ref{lem:grid-manycomp-conn-many} justify the following claims.

\begin{ntheorem}\label{thm:many-comp-minor}
Unless $FPT=W[1]$, there is no algorithm compatible with the description
\begin{eqnarray*}
\sil{\maxdeg(G),\pw(G),\fvs(G),\genus(G)}{\pw(H), \minor(G)}{\ccn(H)\leq 1, \tw(H)\leq 1}
\end{eqnarray*}
\end{ntheorem}

\begin{ntheorem}\label{thm:many-comp-genus}
Unless $FPT=W[1]$, there is no algorithm compatible with the description
\begin{eqnarray*}
\sil{\pw(G),\fvs(G),\genus(G)}{\pw(H)}{\ccn(H)\leq 1, \tw(H)\leq 1, \maxdeg(G)\leq 3}
\end{eqnarray*}
\end{ntheorem}

Finally, we show how to modify the construction of Lemma~\ref{lem:grid-manycomp-conn-many} to guarantee that $G$ has constant cliquewidth, for the cost of losing planarity of $G$. To this end, we use of the biclique gadget introduced in Section~\ref{sec:biclique-gadget}.

\begin{lemma}\label{lem:grid-manycomp-conn-many-biclique}
There exists an FPT reduction that, given an instance of \gridtiling with parameter $k$, outputs an equivalent instance $(H,G)$ of \subiso with the following properties:
\begin{multicols}{2}
\begin{itemize}
\item $\ccn(H)=1$,
\item $H$ is a tree of constant pathwidth, and
\item $\maxdeg(H)\leq 3$;
\end{itemize}
\vfill
\columnbreak
\begin{itemize}
\item $\ccn(G)=1$,
\item $\maxdeg(G)\leq O(k)$,
\item $\genus(G)=O(k^3)$,
\item $\pw(G),\fvs(G)=O(k^2)$ and $\cw(G)\leq c$ for some constant $c$.
\end{itemize}
\end{multicols}
\end{lemma}
\begin{proof}
We show only how the reduction of Lemma~\ref{lem:grid-manycomp-conn-many} need to be further modified.

Firstly, we need to restrict the way of choosing the order in which gadgets $H_{(i,j)}$ (and thus also $G^{(i,j)}$) are connected in a path-like manner. Recall that in the construction of Lemma~\ref{lem:grid-manycomp-conn-many} we ordered them arbitrarily, and connected any two consecutive ones using a long path attached to two carefully chosen vertices. We perform the same construction, but we explicitely order the gadgets lexicographically with respect to coordinates. Thus, gadget $H_{(i,j)}$ will be connected to $H_{(i+1,j)}$ for $(i,j)\in [k-1]\times [k]$, while gadget $H_{(k,j)}$ will be connected to $H_{(1,j+1)}$ for $j\in [k-1]$. Of course, while constructing $G$ we build corresponding connections between gadgets $G_{(i,j)}$, exactly as described in the proof of Lemma~\ref{lem:grid-manycomp-conn-many}.

Secondly, observe that in the reduction of Lemma~\ref{lem:grid-manycomp-conn-many}, we used the same constant $M$ for constructing every connection between two neighboring nodes of the grid. However, if the constant $M$ used would differ between various connections, i.e., if we would change the length of each connection $c$ from $M$ to some other $M_c$ and use $M_c$ for lengths of all the paths intended to be embedded into this connection (in the constructions of Lemma~\ref{lem:constpw-choice}), then the proof would work in exactly the same manner: the only property used was that paths of length $M_c+a$ and $2M_c-b$ can simultaneously fit into connection of length $3M_c+1$ if and only if $a\leq b$. Hence, we use the following constants $M_c$ for the connections:
\begin{itemize}
\item for every connection between $l_1^{(i,j)}$ and $l_6^{(i,j-1)}$ for $1\leq i\leq k$ and $1< j\leq k$, we use $M\cdot 5^{2j}$;
\item for every connection between $l_2^{(i,j)}$ and $l_5^{(i,j-1)}$ for $1\leq i\leq k$ and $1< j\leq k$, we use $M\cdot 5^{2j-1}$;
\item for all the other connections we use constant $M$ as in the original construction.
\end{itemize}
Now that we have modified the construction, we can make use the biclique gadget. That is, for every $j$, $1<j\leq k$, in $G$ instead of paths of lengths $M\cdot 5^1,M\cdot 5^2,\ldots,M\cdot 5^{2k}$ between vertices $(l_5^{(1,j-1)},l_6^{(1,j-1)},l_5^{(2,j-1)},l_6^{(2,j-1)},\ldots,l_5^{(k,j-1)},l_6^{(k,j-1)})$ and $(l_2^{(1,j)},l_1^{(1,j)},l_2^{(2,j)},l_1^{(2,j)},\ldots,l_2^{(k,j)},l_1^{(k,j)})$, we introduce a biclique gadget between these vertices. Note that in this manner we simply introduce $O(k^3)$ edges to the construction, $O(k^2)$ between every pair of consecutive columns. The modification is depicted in Figure~\ref{fig:biclique-application}.

\begin{figure}[htbp!]
                \centering
                \def\svgwidth{0.15\columnwidth}
                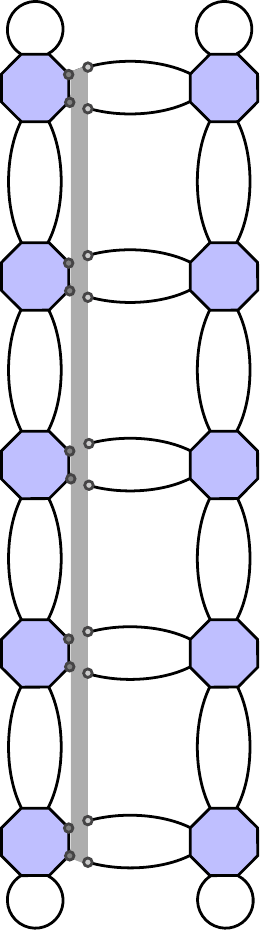
\caption{Introduction of the biclique gadget between two consecutive columns in the proof of Lemma~\ref{lem:grid-manycomp-conn-many-biclique}. The gray area represents the introduced biclique.}\label{fig:biclique-application}
\end{figure}

As we only added some edges in the construction, a solution to the \gridtiling instance can be translated to the solution of the \subiso instance in the same manner as in the proof of Lemma~\ref{lem:grid-manycomp-conn-many}. However, the reasoning of the second implication holds in the same manner by Lemma~\ref{lem:biclique-behaviour}: the only part that is changed, that is, the argument about fitting paths into horizontal connections, is ensured by Lemma~\ref{lem:biclique-behaviour}.

The graph $H$ has still all the properties that were ensured in Lemma~\ref{lem:grid-manycomp-conn-many}, as we did not modify its construction. Pathwidth and feedback vertex set of $G$ are of size $O(k^2)$ by a similar argument: removing the same vertices as before plus all the vertices of the introduced complete bipartite graphs (whose total number is $O(k^2)$) makes $G$ into a forest of constant pathwidth. Genus of $G$ is at most $O(k^3)$, since each of the newly introduced $O(k^3)$ edges can be realized by a private handle. Finally, it is easy to verify that the maximum degree of $G$ is $O(k)$, as the only vertices with degree higher than $3$ are vertices in the newly introduced bicliques. 

We are left with arguing that $G$ has constant cliquewidth. Let us sketch how to construct $G$ using a constant number of labels. We construct $G$ column-by-column, where a column $j_0$ consists of: all the nodes $(i,j_0)$ and connections between them, and all gadgets $G_{(i,j_0)}$ with all connections between them. It is easy to see that each such a column might be constructed using a constant number of labels, note that this holds also for the first and the last columns that have loops attached to every grid node. Moreover, we can also build the parts of the connections to the previous column, up to the point when a full bipartite graph is introduced. Furthermore, using a constant number of labels we can distinguish: (i) the set of all the ends of these connections that need to be made adjacent to the previous column; (ii) the set of all the vertices $l_5^{(i,j_0)}$ and $l_6^{(i,j_0)}$ that need to be made adjacent to the next column; (iii) the vertex in gadget $G_{(1,j_0)}$ that need to be connected to a vertex in gadget $G_{(k,j_0)}$ (assuming $j_0>1$); (iv) the vertex in gadget $G_{(k,j_0)}$ that need to be connected to a vertex in gadget $G_{(1,j_0+1)}$ (assuming $j_0<k$). The whole construction is then performed as follows. We build consecutive columns using a separate, constant-size set of labels, and then make connections with the previous column by one join operation that creates the biclique, and by constructing the long path connecting appropriate vertices in $G_{(1,i_0)}$ and $G_{(k,j_0-1)}$. After making the connections, we rename all the labels of vertices that will not participate in further connections to one special forgotten label.
\end{proof}

Lemma~\ref{lem:grid-manycomp-conn-many-biclique} justifies the following claims.

\begin{ntheorem}\label{thm:grid-manycomp-conn-many-biclique}
Unless $FPT=W[1]$, there is no algorithm compatible with the description
\begin{eqnarray*}
\sil{\maxdeg(G),\pw(G),\fvs(G),\genus(G)}{\pw(H),\cw(G)}{\ccn(H)\leq 1, \maxdeg(H)\leq 3, \tw(H)\leq 1}
\end{eqnarray*}
\end{ntheorem}

\subsubsection{Feedback vertex set number of $G$ versus maximum degree of $G$}

We now proceed to the next result, that will use the following reduction. Intuitively, it says that we may ask for planarity of $G$ and connectedness of $H$ for the price of allowing unbounded degree of vertices in $G$.

\begin{lemma}\label{lem:grid-connfvs}
There exists a polynomial-time reduction that, given an instance of \gridtiling with parameter $k$, outputs an equivalent instance $(H,G)$ of \subiso with the following properties:
\begin{multicols}{2}
\begin{itemize}
\item $\ccn(H)=1$,
\item $H$ is a tree of constant pathwidth, and
\item $\maxdeg(H)\leq 3$;
\end{itemize}
\vfill
\columnbreak
\begin{itemize}
\item $\ccn(G)=1$,
\item $G$ is planar, and
\item $\pw(G),\fvs(G)=O(k^2)$.
\end{itemize}
\end{multicols}
\end{lemma}

The immediate corollary of Lemma~\ref{lem:grid-connfvs} is the following:

\begin{ntheorem}\label{thm:grid-connfvs}
Unless $FPT=W[1]$, there is no algorithm compatible with the description
\begin{eqnarray*}
\sil{\pw(G),\fvs(G)}{\pw(H)}{\ccn(H)\leq 1, \maxdeg(H)\leq 3, \tw(H)\leq 1, \genus(G)\leq 0}
\end{eqnarray*}
\end{ntheorem}

We proceed to the proof of Lemma~\ref{lem:grid-connfvs}.

\begin{proof}[Proof of Lemma~\ref{lem:grid-connfvs}]
Let $(\{S_{i,j}\}_{1\leq i,j\leq k})$ be the given \gridtiling instance, where $S_{i,j}\subseteq [n]\times [n]$. 

\begin{figure}[htbp!]
        \centering
        \subfloat[Graph $H$]{
                \centering
                \def\svgwidth{0.45\columnwidth}
                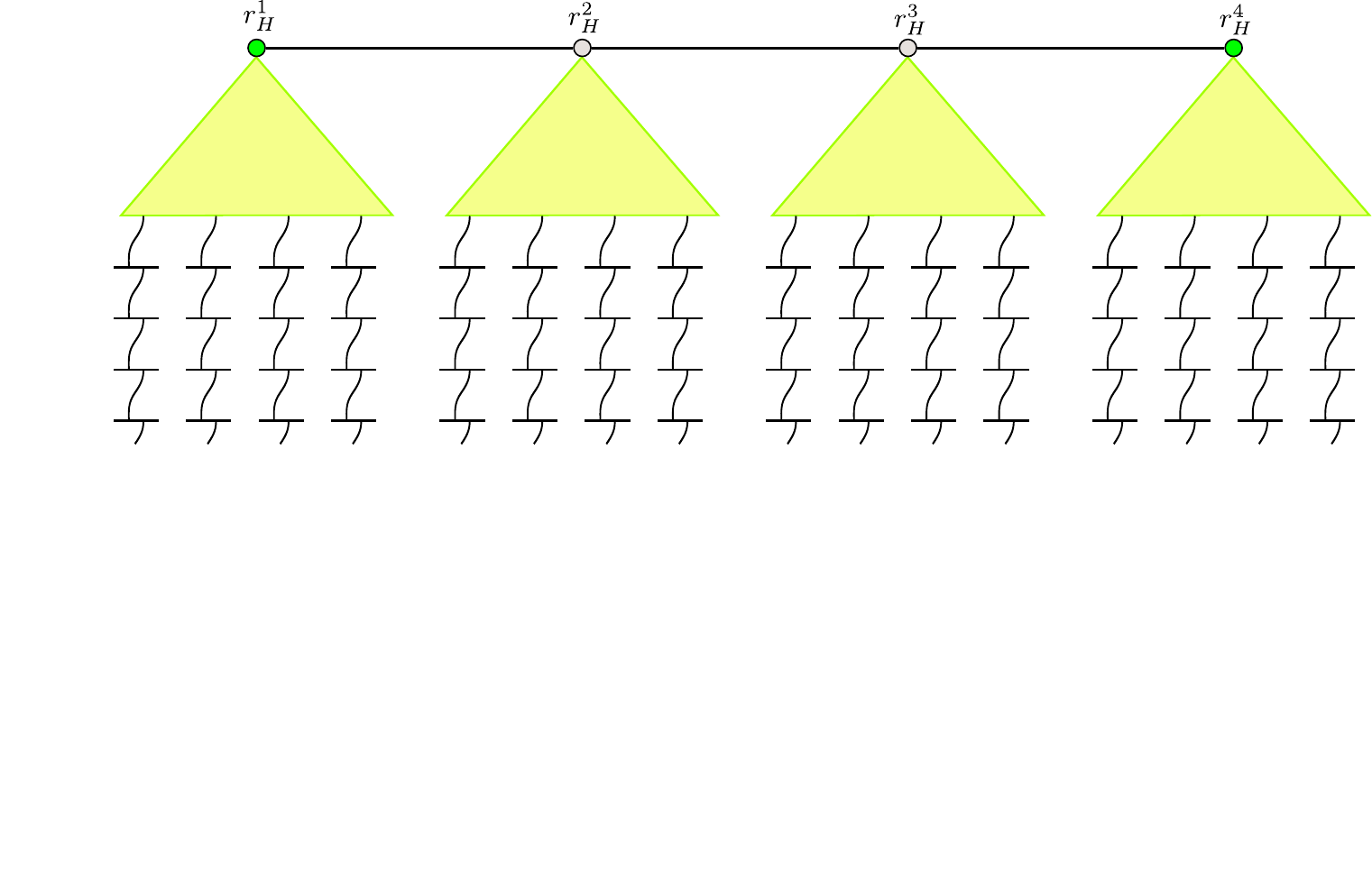
        }
        \qquad
        \subfloat[Graph $G$]{
                \centering
                \def\svgwidth{0.45\columnwidth}
                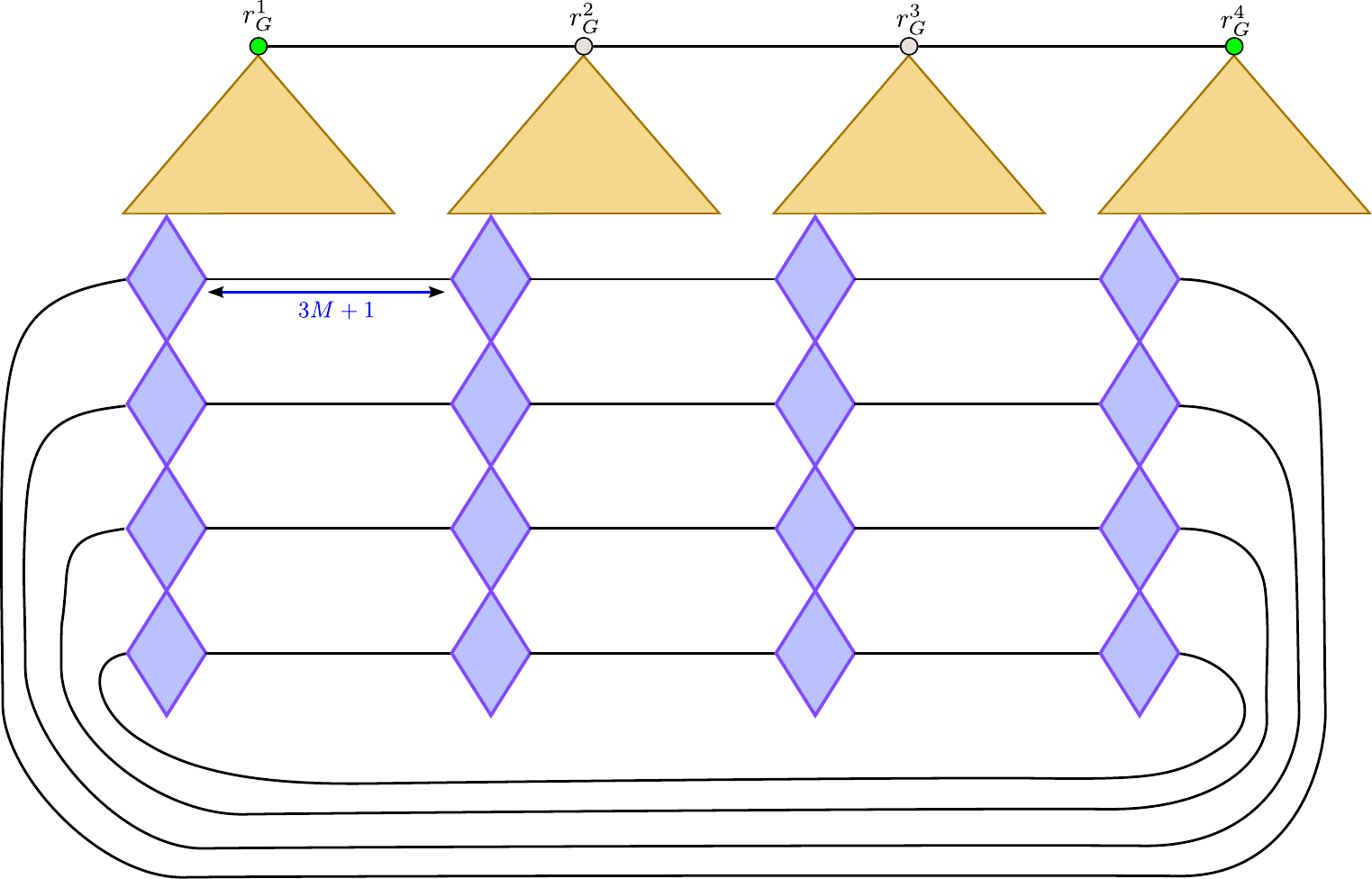
        }
\caption{Construction of Lemmas~\ref{lem:grid-connfvs} and~\ref{lem:grid-conndegree}. The rhombs depict gadgets $G_0^{i,j}$ given by Lemma~\ref{lem:gadget-moustachefvs} or Lemma~\ref{lem:gadget-moustachedegree}.}\label{fig:grid-moustache-down}
\end{figure}

We construct graphs $G$ and $H$ at the same time; the construction is depicted in Figure~\ref{fig:grid-moustache-down}. First, for every $a\in [n]$ and $i,j\in [k]\times [k]$, construct graphs $H_a^{i,j}$ and $G_0^{i,j}$ given by Lemma~\ref{lem:gadget-moustachedegree} for the set $S_{i,j}$, choosing a constant $M$ larger than $\max(n,k,10)$. For every $a\in [n]$, $i\in [k-1]$ and $j=[k]$, identify the sink of $H_a^{i,j}$ and the root of $H_a^{i+1,j}$, thus arranging of gadgets for fixed $a$ and $j$ into a graph $P_{a,j}$ that has path-like structure. Note that $P_{a,j}$ is a tree of constant pathwidth. Let the root $r_{a,j}$ of this graph be equal to the root of $H_a^{1,j}$.

For every $j\in [k]$ create gadgets $H^j_0$ and $G^j_0$ given by Lemma~\ref{lem:constpw-choice} for the set of trees $P_{1,j},P_{2,j},\ldots,P_{n,j}$. Let $r^j_H$ be the root of $H^j_0$, $r^j_G$ be the root $G^j_0$, and $\interface^j$ be the interface vertex of $G^j_0$. For every $j\in [k-1]$, introduce an edge between $r^j_H$ and $r^{j+1}_H$, and between $r^j_G$ and $r^{j+1}_G$.

We now continue the construction of $G$. For every $i\in [k-1]$ and $j\in [k]$, identify the sink of $G_0^{i,j}$ and the source of $G_0^{i+1,j}$. Moreover, for every $j\in [k]$, identify $G_0^{1,j}$ with $\interface^j$. To conclude, for every $i,j\in [k]$ connect the second interface vertex of $G_0^{i,j}$ with the first interface vertex of $G_0^{i,j+1}$ via a path of length $3M+1$, where $G_0^{i,k+1}=G_0^{i,1}$.

To ensure that roots $r^j_H$ are mapped to roots $r^j_G$, we make use of Lemma~\ref{lem:setting-images}. We apply Lemma~\ref{lem:setting-images} to vertices $r^1_H$ and $r^k_H$, setting their images to $r^1_G$ and $r^k_G$, respectively. This concludes the whole construction. 

It follows from the construction and from Lemmas~\ref{lem:constpw-choice} and \ref{lem:gadget-moustachefvs} that $H$ is a tree of constant pathwidth and maximum degree $3$: Lemma~\ref{lem:constpw-choice} ensures us that gadgets $H^j_0$ have constant pathwidth, and then we attach them to a common long path and we can use Lemma~\ref{lem:pathw}. Moreover, $G$ is connected and planar; its embedding into the plane is presented in Figure~\ref{fig:grid-moustache-down}. Note also that removing from $G$ all the roots, sinks, interface vertices of gadgets $G_0^{i,j}$, as well as pairs of vertices whose existence is guaranteed by property (v) of the construction of Lemma~\ref{lem:gadget-moustachefvs}, makes $G$ into a forest of constant pathwidth. As in this manner we remove $O(k^2)$ vertices, this means that $\fvs(G),\pw(G)\leq O(k^2)$. It remains to show that the output instance is equivalent to the input one.

Assume that we are given a solution $\tiling$ to the given \gridtiling instance. We create a subgraph isomorphism of the graphs before application of Lemma~\ref{lem:setting-images} such that $r^1_H$ is mapped to $r^1_G$ and $r^k_H$ is mapped to $r^k_G$. Lemma~\ref{lem:setting-images} ensures us that this subgraph isomorphism may be extended on the whole $H$ and $G$.

We first map roots $r^j_H$ to corresponding roots $r^j_G$. Extend this mapping to gadgets $H^j_0$ by choosing a partial subgraph isomorphism that maps $H^j_0$ into $G^j_0$ leaving $P_{j,\tiling_2(j)}$ unmapped, where the root of $P_{j,\tiling_2(j)}$ has been mapped to the interface vertex of $G^j_0$ (identified with the root of $G_0^{j,1}$). We now need to map the trees $P_{j,\tiling_2(j)}$ into the grid-like structure created by gadgets $G_0^{i,j}$. Recall that $\tiling_2(j)=\tiling_2(i,j)$ for every $j\in [k]$.

We consecutively map subgraphs $H_{\tiling_2(j)}^{1,j},H_{\tiling_2(j)}^{2,j},\ldots,H_{\tiling_2(j)}^{k,j}$ that create $P_{j,\tiling_2(j)}$ into subgraphs $G_0^{1,j},G_0^{2,j},\ldots,G_0^{k,j}$, using property (iv) ensured by Lemma~\ref{lem:gadget-moustachefvs} applied in subgraph $G_0^{i,j}$ for $a=\tiling_2(j)=\tiling_2(i,j)$ and $b=\tiling_1(i)=\tiling_1(i,j)$. In this manner, $P_{j,\tiling_2(j)}$ gets mapped into the subgraph induced by $G_0^{1,j},G_0^{2,j},\ldots,G_0^{k,j}$, apart from, for each pair of indices $(i,j)\in [n]\times[n]$, a path of length $M+\tiling_1(i,j)$ that needs to be rooted at the second interface vertex of $G_0^{i,j}$, and a path of length $2M-\tiling_1(i,j+1)$ that needs to be rooted at the first interface vertex of $G_0^{i,j+1}$. However, as $\tiling_1(i,j)=\tiling_1(i+1,j)=\tiling_1(i)$, both of these paths can simultaneously fit into the connection of length $3M+1$ between these interface vertices.

Assume now that we are given a subgraph isomorphism $\hm$ from $H$ into $G$. By Lemma~\ref{lem:setting-images}, we have that $\hm(r^1_H)=r^1_G$, $\hm(r^k_H)=r^k_G$, and $\hm$ can be restricted to the graphs before application of Lemma~\ref{lem:setting-images} so that it is still a subgraph isomorphism. From now on, we work with this restriction keeping in mind that $\hm(r^1_H)=r^1_G$ and $\hm(r^k_H)=r^k_G$.

Now observe that as $M>k$, path $r^1_G-r^2_G-\ldots-r^k_G$ is the unique shortest path between $r^1_G$ and $r^k_G$ in $G$. As path $r^1_H-r^2_H-\ldots-r^k_H$ in tree $H$ is of the same length, it follows that $r^j_H$ must be mapped to $r^j_G$ for every $j\in [k]$.

Using property (v) ensured by Lemma~\ref{lem:constpw-choice} we infer that for every $j\in [k]$ there exists an index $\tiling_2^j$ such that $H_0^j$ is mapped into $G_0^j$ apart from a tree isomorphic to $P_{j,\tiling_2^j}$; moreover, this tree must be mapped into the remaining part of the graph in such a manner that its root is mapped to the interface vertex of $G_0^j$ identified with the root of $G_0^{1,j}$. By inductive use of Lemma~\ref{lem:gadget-moustachedegree} for every $j\in [k]$ and consecutive $i=1,2,\ldots,k$, we infer that subgraphs $H^{i,j}_{\tiling_2^j}$ need to be mapped to subgraphs $G^{i,j}_0$ so that for every $(i,j)\in [k]\times [k]$ there is an index $\tiling_1^{i,j}$ satisfying conditions: (i) $(\tiling_1^{i,j},\tiling_2^{j})\in S_{i,j}$, (ii) the only unmapped part of $H^{i,j}_{\tiling_2^j}$ is a path of length $2M-\tiling_1^{i,j}$ rooted in the first interface vertex of $G_0^{i,j}$, and a path of length $M+\tiling_1^{i,j}$ rooted in the second interface vertex. All these paths need to fit simultaneously into connections between interface vertices of neighboring gadgets $G_0^{i,j}$. As two paths fit if and only if their sum of length is at most $3M$, we infer that for every $i\in [k]$ we have $\tiling_1^{i,1}\leq \tiling_1^{i,2} \leq \tiling_1^{i,3} \leq \ldots \leq \tiling_1^{i,k} \leq \tiling_1^{i,1}$. Hence, all these numbers must be equal; denote this common value by $\tiling_1^{i}$. As $(\tiling_1^{i},\tiling_2^{j})\in S_{i,j}$ for every $(i,j)\in [k]\times [k]$, it follows that $\tiling(i,j)=(\tiling_1^{j},\tiling_2^{i})$ is a solution to the input \gridtiling instance.

\end{proof}

If we substitute the construction of Lemma~\ref{lem:gadget-moustachefvs} in the proof of Lemma~\ref{lem:grid-connfvs} with the construction of Lemma~\ref{lem:gadget-moustachedegree}, we may trade allowing unbounded degree of $G$ for allowing unbounded feedback vertex set number of $G$. Hence, in the same manner we obtain the following result:

\begin{lemma}\label{lem:grid-conndegree}
There exists a polynomial-time reduction that, given an instance of \gridtiling with parameter $k$, outputs an equivalent instance $(H,G)$ of \subiso with the following properties:
\begin{multicols}{2}
\begin{itemize}
\item $\ccn(H)=1$,
\item $H$ is a tree of constant pathwidth, and
\item $\maxdeg(H)\leq 3$;
\end{itemize}
\vfill
\columnbreak
\begin{itemize}
\item $\ccn(G)=1$,
\item $G$ is planar and $\maxdeg(G)\leq 3$, and
\item $\pw(G)=O(k^2)$.
\end{itemize}
\end{multicols}
\end{lemma}

The immediate corollary of Lemma~\ref{lem:grid-conndegree} is the following:

\begin{ntheorem}\label{thm:grid-conndegree}
Unless $FPT=W[1]$, there is no algorithm compatible with the description
\begin{eqnarray*}
\sil{\pw(G)}{\pw(H)}{\ccn(H)\leq 1, \tw(H)\leq 1, \maxdeg(G)\leq 3, \genus(G)\leq 0}
\end{eqnarray*}
\end{ntheorem}

We would like here to remark that in Lemmas~\ref{lem:grid-connfvs} and \ref{lem:grid-conndegree} one could again insert the biclique gadget introduced in Section~\ref{sec:biclique-gadget} between every two consecutive columns, thus ensuring that $G$ has constant cliquewidth at the cost of increasing its genus. Unfortunately, in case of Lemma~\ref{lem:grid-conndegree} we would then also need to increase the maximum degree of $G$, since the introduced bicliques would contain $O(k)$ vertices each. As a result, in both cases the obtained lower bounds would be weaker than Theorem~\ref{thm:grid-manycomp-conn-many-biclique}.

\subsection{Embedding a small planar graph}

The classical {\sc{Clique}} problem shows that parameterization only by the size of $H$ gives an intractable problem. In the following, we show that adding the constraint that $H$ is  planar and of bounded-degree does not help for the complexity.

\begin{lemma}\label{lem:multicolored-grid}
There exists a polynomial-time reduction that, given an instance of \gridtiling with parameter $k$, outputs an equivalent instance $(H,G)$ of \subiso with the following properties:
\begin{multicols}{2}
\begin{itemize}
\item $\ccn(H)=1$,
\item $H$ is planar and $\maxdeg(H)\leq 3$,
\item $|V(H)|=O(k^4)$;
\end{itemize}
\vfill
\columnbreak
\begin{itemize}
\item $\ccn(G)=1$.
\end{itemize}
\end{multicols}
\end{lemma}

The immediate corollary of Lemma~\ref{lem:multicolored-grid} is the following:

\begin{ntheorem}\label{thm:multicolored-grid}
Unless $FPT=W[1]$, there is no algorithm compatible with the description
\begin{eqnarray*}
\sil{|V(H)|}{}{\ccn(H)\leq 1, \maxdeg(H)\leq 3, \genus(H)\leq 0}
\end{eqnarray*}
\end{ntheorem}

We proceed to the proof of Lemma~\ref{lem:multicolored-grid}.

\begin{proof}[Proof of Lemma~\ref{lem:multicolored-grid}]
Let $(\{S_{i,j}\}_{1\leq i,j\leq k})$ be the given \gridtiling instance, where $S_{i,j}\subseteq [n]\times [n]$. The construction of the instance of \subiso is depicted in Figure~\ref{fig:multicolored-grid}.

\begin{figure}[htbp!]
        \centering
        \subfloat[Graph $H$]{
                \centering
                \def\svgwidth{0.45\columnwidth}
                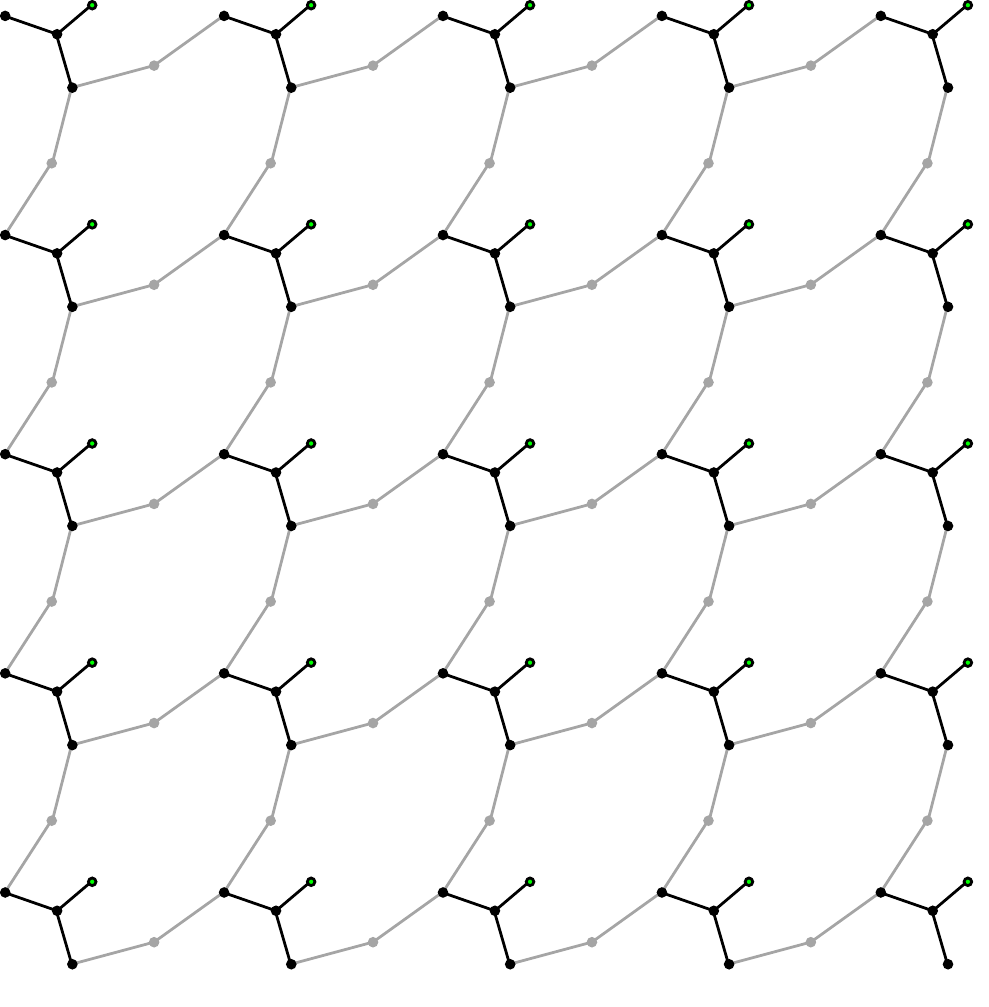
        }
        \qquad
        \subfloat[Graph $G$]{
                \centering
                \def\svgwidth{0.45\columnwidth}
                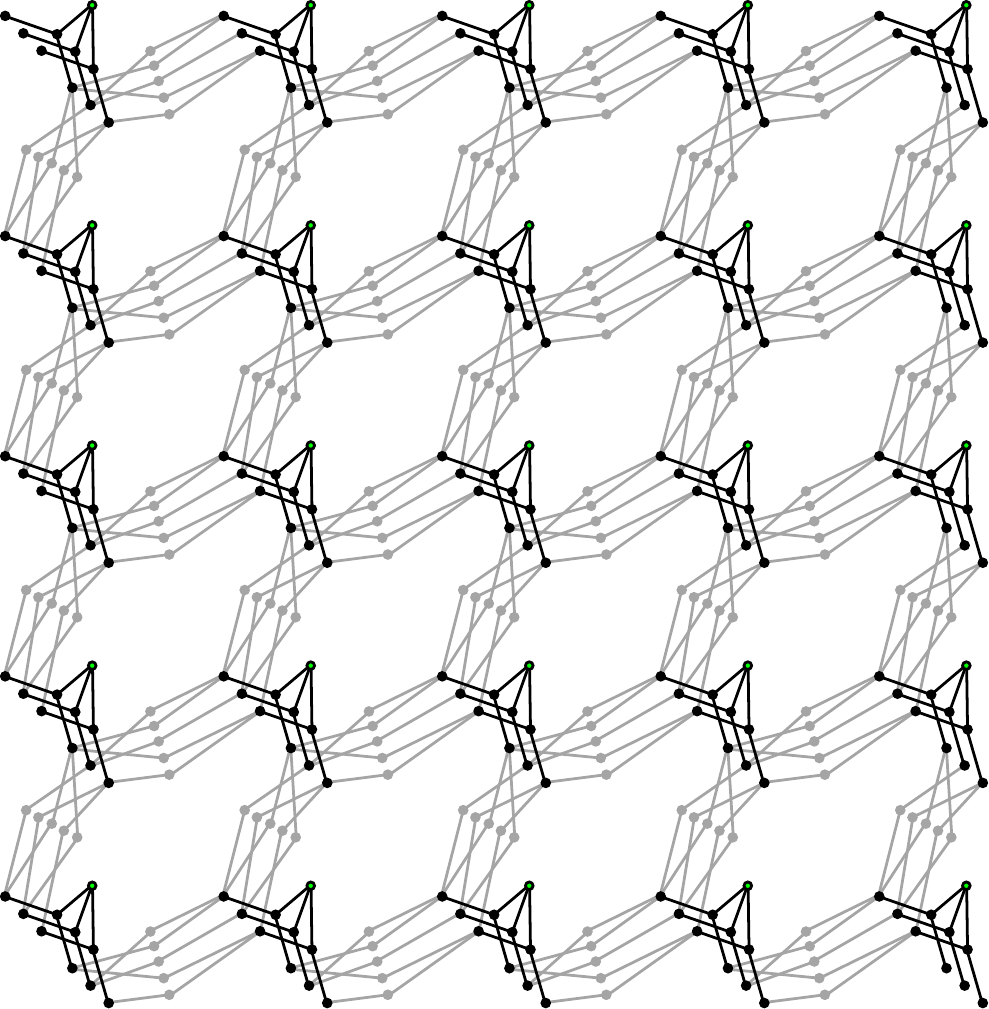
        }
\caption{Construction of Lemma~\ref{lem:multicolored-grid}. Black nodes and edges denote subgraphs induced by vertices $l$, $r$, $c$, $u$ that form the grid structure, while connections between the nodes of the grid are colored gray.}\label{fig:multicolored-grid}
\end{figure}

We start with defining the graph $H$. For each $(i,j)\in [k]\times [k]$, introduce four vertices: $l_{i,j}$,  $r_{i,j}$, $u_{i,j}$, and $c_{i,j}$, and edges $l_{i,j}c_{i,j}$, $r_{i,j}c_{i,j}$, and $u_{i,j}c_{i,j}$. For every $i\in [k]$ and $j\in [k-1]$, connect $l_{i,j+1}$ with $r_{i,j}$ via a path of length $2$. For every $i\in [k-1]$ and $j\in [k]$, connect $l_{i+1,j}$ with $r_{i,j}$ via a path of length $2$. Note that $H$ created in this manner is planar; its embedding can be seen in Figure~\ref{fig:multicolored-grid}.

We now proceed to the graph $G$. For each $(i,j)\in [k]\times[k]$, create a vertex $u^G_{i,j}$. Moreover, for every $(a,b)\in S_{i,j}$, create vertices $l^{(a,b)}_{i,j}$, $r^{(a,b)}_{i,j}$, and $c^{(a,b)}_{i,j}$, and introduce edges $l^{(a,b)}_{i,j}c^{(a,b)}_{i,j}$, $r^{(a,b)}_{i,j}c^{(a,b)}_{i,j}$, and $u^G_{i,j}c^{(a,b)}_{i,j}$. Now, for every $i\in [k]$ and $j\in [k-1]$ and every triple $a,b,b'$ such that $(a,b)\in S_{i,j}$ and $(a,b')\in S_{i,j+1}$, connect $l^{(a,b')}_{i,j+1}$ with $r^{(a,b)}_{i,j}$ via a path of length $2$. Finally, for every $i\in [k-1]$ and $j\in [k]$ and every triple $a,a',b$ such that $(a,b)\in S_{i,j}$ and $(a',b)\in S_{i+1,j}$, connect $l^{(a',b)}_{i+1,j}$ with $r^{(a,b)}_{i,j}$ via a path of length $2$.

To finalize the construction, apply Lemma~\ref{lem:setting-images} to ensure that vertices $u_{i,j}$ are mapped to $u^G_{i,j}$. Note that in $G$ constructed so far, vertices of degree at least $3$, including vertices $u^G_{i,j}$, induce a forest with every component having diameter $4$, so in $G$ there is no path of length larger than $4$ that passes only through vertices $u^G_{i,j}$ or through other vertices of degree at least $3$. Hence, in the application of Lemma~\ref{lem:setting-images} we can use $L=5$ and thus we introduce only $O(k^4)$ additional vertices to $H$ and $G$.

Clearly, after the application of Lemma~\ref{lem:setting-images}, graph $H$ is still planar, both $G$ and $H$ are connected, and the maximum degree of $H$ is equal to $3$. Moreover, $|V(H)|=O(k^4)$. It remains to show formally that the constructed instance $(H,G)$ of \subiso is equivalent to the given \gridtiling instance.

Assume that we are given a solution $\tiling$ to the given \gridtiling instance. We create a subgraph isomorphism of the graphs before application of Lemma~\ref{lem:setting-images} such that $u_{i,j}$ is mapped to $u_{i,j}^G$ for each $(i,j)\in [k]\times [k]$. Lemma~\ref{lem:setting-images} ensures us that this subgraph isomorphism may be extended on the whole $H$ and $G$.

For every $(i,j)\in [k]\times [k]$, we map $l_{i,j}$, $c_{i,j}$ and $r_{i,j}$ to $l^{\tiling(i,j)}_{i,j}$, $c^{\tiling(i,j)}_{i,j}$ and $r^{\tiling(i,j)}_{i,j}$, respectively. By the construction of $G$ and from the fact that $\tiling$ is a solution to the \gridtiling problem, we have that for every $i\in [k]$ and $j\in [k-1]$, vertices $l^{\tiling(i,j+1)}_{i,j+1}$ and $r^{\tiling(i,j)}_{i,j}$ are connected via a path of length $2$.  Hence, we can map the path of length $2$ between $l_{i,j+1}$ and $r_{i,j}$ to this path of length $2$ between $l^{\tiling(i,j+1)}_{i,j+1}$ and $r^{\tiling(i,j)}_{i,j}$. We perform a symmetric reasoning for paths of length $2$ between $l^{\tiling(i+1,j)}_{i+1,j}$ and $r^{\tiling(i,j)}_{i,j}$ for $i\in [k-1]$ and $j\in [k]$, thus concluding the construction of a subgraph isomorphism from $H$ to $G$.

Assume now that we are given a subgraph isomorphism $\hm$ from $H$ into $G$. By Lemma~\ref{lem:setting-images}, we have that $\hm(u_{i,j})=u_{i,j}^G$ for each $(i,j)\in [k]\times [k]$, and $\hm$ can be restricted to the graphs before application of Lemma~\ref{lem:setting-images} so that it is still a subgraph isomorphism. From now on, we work with this restriction keeping in mind that $\hm(u_{i,j})=u_{i,j}^G$ for each $(i,j)\in [k]\times [k]$.

For each $(i,j)\in [k]\times [k]$, vertex $c_{i,j}$ has to be mapped to one of the neighbors of $u_{i,j}$, namely to one of vertices $c_{i,j}^{(a,b)}$ for $(a,b)\in S_{i,j}$. Let this vertex be $c_{i,j}^{\tiling(i,j)}$. We are going to prove the $\tiling$ is a solution to the given \gridtiling instance; note that we already know that $\tiling(i,j)\in S_{i,j}$ for all $(i,j)\in [k]\times [k]$. For $i\in [k]$ and $j\in [k-1]$, consider vertices $c_{i,j+1}^{\tiling(i,j+1)}$ and $c_{i,j}^{\tiling(i,j)}$. Observe that these two vertices are in distance $4$ if $\tiling_1(i,j+1)=\tiling_1(i,j)$, as then they can be connected via a path of length $4$ passing through $l_{i,j+1}^{\tiling(i,j+1)}$ and $r_{i,j}^{\tiling(i,j)}$, and otherwise they are in distance larger than $4$. As their preimages, $c_{i,j+1}$ and $c_{i,j}$, are in distance $4$ in $H$, it follows that $\tiling_1(i,j+1)=\tiling_1(i,j)$ for every $i\in [k]$ and $j\in [k-1]$. Analogously, we can prove that $\tiling_2(i+1,j)=\tiling_2(i,j)$ for every $i\in [k-1]$ and $j\in [k]$, hence $\tiling$ is indeed a solution to the given \gridtiling instance.
\end{proof}

\subsection{Embedding paths into a planar graph}
\label{sec:embedding-paths-into}

In this section we present the last group of reductions that show hardness of embedding a number of arbitrarily long paths into a planar graph. We find it more convenient to introduce a new problem as a source of our reductions, called \expas, which can be viewed as a version of \planararcsupply \cite{DBLP:conf/iwpec/BodlaenderLP09} where the instance is much more constrained. We  prove first that \expas is $W[1]$-hard by a reduction from \gridtiling, and then use \expas to show hardness of the remaining parameterizations of \subiso.

\subsubsection{\expas}

The \expas problem is defined as follows:

\defparproblemu{\expas}{A planar, weakly connected multidigraph $D$, supply sets $S_a\subseteq [M_a-1]\times [M_a-1]$ for each $a\in A(D)$ with a property that $x+y=M_a$ for each $(x,y)\in S_a$, where $M_a$ is an integer associated with arc $a$, and demands $r_v$ for each $v\in V(D)$}{$|D|$}{Is there a supply function $\supply$ defined on the arcs of $D$, $\supply=(\supply_1,\supply_2)$, such that for all $a\in A(D)$ we have that $\supply(a)\in S_a$, and for each $v\in V(D)$ it holds that $\sum_{(v,w)\in A(D)} \supply_1((v,w))+\sum_{(w,v)\in A(D)} \supply_2((w,v))=r_v$?}

The difference between \expas and the \planararcsupply problem, defined by Bodlaender et al.~\cite{DBLP:conf/iwpec/BodlaenderLP09}, is that instead of asking for the demand of $v$ to be satisfied with {\emph{at least}} $r_v$ supply from the arcs incident to $v$, here we ask it to be satisfied with {\emph{exactly}} $r_v$ supply. Moreover, we require that for every supply set, the integers in the supply pairs sum up to a constant depending on the arc only. Clearly, an instance of \expas is a NO-instance if the sum of requirements is not equal to the sum of numbers $M_a$ through the whole arc set. Hence, we will consider only instance where this equality holds. Observe that if we add this requirement to the problem statement, then we can relax the condition that the demand for vertex $v$ must be satisfied with exactly $r_v$ supply to with at least $r_v$ supply: if for at least one vertex there is surplus of supply, then there must be another vertex with deficit. Thus, the \expas problem is in fact a more constrained version of \planararcsupply.

We now prove that the \expas problem is $W[1]$-hard, thus providing at the same time an alternative proof of hardness of \planararcsupply.

\begin{theorem}\label{thm:expas-hardness}
\expas is $W[1]$-hard.
\end{theorem}
\begin{proof}
We provide a parameterized reduction from the \gridtiling problem. Let $(\{Q_{i,j}\}_{1\leq i,j\leq k})$ be the given \gridtiling instance, where $Q_{i,j}\subseteq [n]\times [n]$ (note that we use the letter $Q$ to denote the sets of available tiles in the \gridtiling instance in order to avoid confusion with the constructed \expas instance). By doubling some row and column, if necessary, we may assume that $k$ is even, that is, $k=2k'$ for some integer $k'$. The reduction is depicted in Figure~\ref{fig:expas}.

\begin{figure}[htbp!]
                \centering
                \def\svgwidth{0.65\columnwidth}
                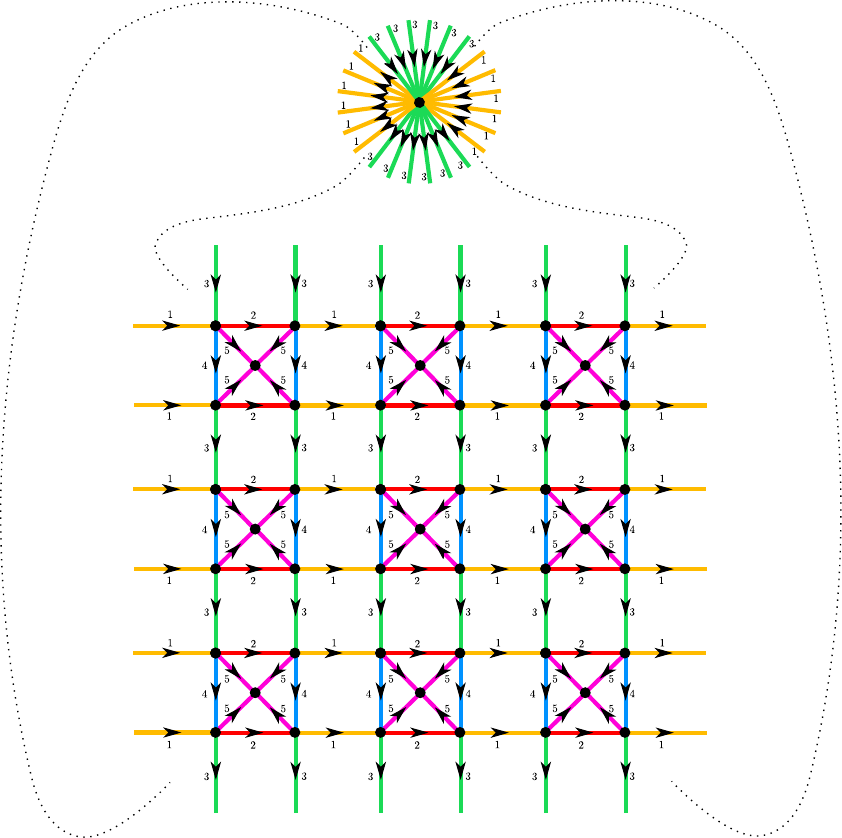
\caption{The reduction of Theorem~\ref{thm:expas-hardness}. Arcs of types $1,2,3,4,5$ are numbered accordingly and colored yellow, red, green, blue, and pink, respectively. In order not to create confusion, the arcs incident to the synchronizer $s^*$ (on the top) are not drawn in full details; recall that $s^*$ has incidents arcs connecting it to all the vertices on the boundary of the grid.}\label{fig:expas}
\end{figure}

For each $(i,j)\in [k]\times [k]$, we create one vertex $v_{i,j}$. For each $(\alpha,\beta)\in [k']\times[k']$ we create one {\emph{checker}} vertex $c_{\alpha,\beta}$. Finally, we create one extra vertex $s^*$ that we will call {\emph{synchronizer}}. We will use $5$ types of arcs, denoted $1,2,3,4,5$. First we describe how the arcs are placed, then we define the type of each arc.
\begin{itemize} 
\item For each $\beta\in [k'-1]$ and $i\in [k]$, let us introduce an arc of type $1$ from $v_{i,2\beta}$ to $v_{i,2\beta+1}$. 
\item For each $\beta\in [k']$ and $i\in [k]$, let us introduce an arc of type $2$ from $v_{i,2\beta-1}$ to $v_{i,2\beta}$. 
\item For each $\alpha\in [k'-1]$ and $j\in [k]$, let us introduce an arc of type $3$ from $v_{2\alpha,j}$ to $v_{2\alpha+1,j}$. 
\item For each $\alpha\in [k']$ and $j\in [k]$, let us introduce an arc of type $4$ from $v_{2\alpha-1,j}$ to $v_{2\alpha,j}$.
\item For each $i\in [k]$, let us introduce an arc of type $1$ from the synchronizer $s^*$ to $v_{i,1}$ and from $v_{i,k}$ to the synchronizer $s^*$. 
\item For each $j\in [k]$, let us introduce an arc of type $3$ from the synchronizer $s^*$ to $v_{1,j}$ and from $v_{k,j}$ to the synchronizer $s^*$. 
\item Finally, for every $(\alpha,\beta)\in [k']\times[k']$, let us introduce arcs of type $5$ from $v_{2\alpha-1,2\beta-1}$, $v_{2\alpha,2\beta-1}$, $v_{2\alpha-1,2\beta}$, and $v_{2\alpha,2\beta}$ to the checker vertex $c_{\alpha,\beta}$.
\end{itemize}
Note that in this manner each vertex $v_{i,j}$ is adjacent to one arc of each type.

Arcs of type $1$ and $2$ will be also called {\emph{horizontal}} arcs, and we say that arc $a$ is in $i$-th row if it is incident to at least one vertex of form $v_{i,j}$ for some $j\in [k]$. Similarly, arcs of type $3$ and $4$  will be called {\emph{vertical}} arcs, and we say that arc $a$ is in $j$-th column if it is incident to at least one vertex of form $v_{i,j}$ for some $i\in [k]$. Note that this definition applies also to arcs incident with $s^*$.

Clearly, we have that the size of the constructed multidigraph $D$ is $O(k^2)$, and it is  weakly connected and planar: its planar embedding is depicted in Figure~\ref{fig:expas}. It remains to present how the supply sets $S_a$ are defined for all the five types of arcs, and how the demands $r_v$ are defined for each vertex.

\newcommand{\expc}{\phi}

Let $M$ be any integer larger than $n$. 
\begin{itemize}
\item If $a$ is an arc of type $1$, we set $S_a=\{(x,M-x)\ |\ x\in [n]\}$, thus we have $M_a=M$.
\item If $a$ is an arc of type $2$, we set $S_a=\{(xM,M^2-xM)\ |\ x\in [n]\}$, thus we have $M_a=M^2$.
\item If $a$ is an arc of type $3$, we set $S_a=\{(yM^2,M^3-yM^2)\ |\ y\in [n]\}$, thus we have $M_a=M^3$.
\item If $a$ is an arc of type $4$, we set $S_a=\{(yM^3,M^4-yM^3)\ |\ y\in [n]\}$, thus we have $M_a=M^4$.
\end{itemize}
We now proceed to arcs of type $5$. Assume that $a$ is an arc of type $5$ from the vertex $v_{2\alpha-1,2\beta-1}$ to $c_{\alpha,\beta}$. For each $(x,y)\in Q_{2\alpha-1,2\beta-1}$, let $\expc_1(x,y)=(M-x)+xM+(M^3-yM^2)+yM^3$; note that values of $\expc_1$ are pairwise different for different pairs $(x,y)$. For each such $(x,y)\in Q_{2\alpha-1,2\beta-1}$, we construct one supply pair in $S_a$, namely pair $(M+M^2+M^3+M^4-\expc_1(x,y),\expc_1(x,y))$, thus setting $M_a=M+M^2+M^3+M^4$. For arcs of type $5$ from vertices $v_{2\alpha-1,2\beta}$, $v_{2\alpha,2\beta-1}$, $v_{2\alpha,2\beta}$ to $c_{\alpha,\beta}$, we perform the same construction, but we use functions $\expc_2$, $\expc_3$, $\expc_4$ defined as follows:
\begin{eqnarray*}
\expc_2(x,y) & = & x+(M^2-xM)+(M^3-yM^2)+yM^3, \\
\expc_3(x,y) & = & (M-x)+xM+yM^2+(M^4-yM^3), \\
\expc_4(x,y) & = & x+(M^2-xM)+yM^2+(M^4-yM^3).
\end{eqnarray*}
We say that function $\expc_c$ for $c=1,2,3,4$ {\emph{applies}} to arc $a$ of type $5$, if $\expc_c$ was used when constructing $S_a$. Note that for any $(x,x',y,y')\in [n]^4$ it holds that 
\begin{equation}\label{eq:checker}
\expc_1(x,y)+\expc_2(x,y')+\expc_3(x',y)+\expc_4(x',y')=2\cdot (M+M^2+M^3+M^4)
\end{equation}

We now set the demands. For each $v_{i,j}$ for $(i,j)\in [k]\times[k]$, we set $r_{v_{i,j}}=M+M^2+M^3+M^4$. For each $c_{\alpha,\beta}$ for $(\alpha,\beta)\in [k']\times[k']$, we set $r_{c_{\alpha,\beta}}=2\cdot(M+M^2+M^3+M^4)$. Finally, we set $r_{s^*}=k\cdot(M+M^3)$. This concludes the construction. We are left with proving that the obtained instance of \expas is equivalent to the input instance $(\{Q_{i,j}\}_{1\leq i,j\leq k})$ of \gridtiling.

Assume that we are given a solution $\tiling$ to the input \gridtiling instance; we construct a solution $\supply$ to the constructed \expas instance as follows. 
\begin{itemize}
\item For every arc $a$ of type $1$ in the $i$-th row, let $\supply(a)=(\tiling_1(i),M-\tiling_1(i))$. 
\item For every arc $a$ of type $2$ in the $i$-th row, let $\supply(a)=(\tiling_1(i)M,M^2-\tiling_1(i)M)$. 
\item For every arc $a$ of type $3$ in the $j$-th column, let $\supply(a)=(\tiling_2(j)M^2,M^3-\tiling_2(j)M^2)$. 
\item For every arc $a$ of type $4$ in the $j$-th column, let $\supply(a)=(\tiling_2(j)M^3,M^4-\tiling_2(j)M^3)$. 
\item For every arc $a$ of type $5$ with vertex $v_{i,j}$ being the tail, let $\supply(a)=(M+M^2+M^3+M^4-\expc_c(\tiling(i,j)),\expc_c(\tiling(i,j)))$ where $\expc_c$ is the function that applies to arc $a$.
\end{itemize}
Using the fact that $\tiling$ is a solution to the input instance of \gridtiling, it is easy to verify that $\supply$ is indeed a solution to the output \expas instance:
\begin{itemize}
\item Condition $\supply(a)\in S_a$ for every arc $a$ of type $1,2,3,4$ follows directly from the definition of $\supply$, while for type $5$ it follows from the fact that $\tiling(i,j)\in S_{i,j}$ for each $(i,j)\in [k]\times[k]$.
\item Demands of vertices $v_{i,j}$ are satisfied because the sum of contributions of arcs of types $1,2,3,4$ incident to $v_{i,j}$ is exactly compensated by the subtracted $\expc_c(\tiling(i,j))$ in the contribution of arc of type $5$ incident to $v_{i,j}$, where $\expc_c$ is the function that applies to it.
\item Demands of vertices $c_{\alpha,\beta}$ are satisfied by (\ref{eq:checker}) applied to $(\tiling_1(2\alpha-1),\tiling_1(2\alpha),\tiling_2(2\beta-1),\tiling_2(2\beta))$.
\item Demand of the synchronizer $s^*$ is satisfied because the contribution of every pair of horizontal arcs from the same row incident to $s^*$ is equal to $M$, while the contribution of every pair of vertical arcs from the same column incident to $s^*$ is equal to $M^3$.
\end{itemize}

Assume now that we are given a solution $\supply$ to the constructed instance of \expas. For every arc $a$ of type $1$, let $\supply'(a)=i$ iff $\supply(a)=(i,M-i)$. Similarly if $a$ is of type $2$, then $\supply'(a)=i$ if $\supply(a)=(iM,M^2-iM)$; if $a$ is of type $3$, then $\supply'(a)=i$ if $\supply(a)=(iM^2,M^3-iM^2)$; and if $a$ is of type $4$, then $\supply'(a)=i$ iff $\supply(a)=(iM^3,M^4-iM^3)$.

Consider one vertex $v_{i,j}$, and let $a_1,a_2,a_3,a_4,a_5$ be the incident arcs of types $1,2,3,4,5$, respectively. We have that $\supply(a_5)=(M+M^2+M^3+M^4-\expc_c(x,y),\expc_c(x,y))$ for some $(x,y)\in S_{i,j}$, where $\expc_c$ is the function that applies to $a_5$. Let $\tiling(i,j)=(x,y)$; we would like to prove that $\tiling$ is a solution to the input \gridtiling instance. Note that we already know that $\tiling(i,j)\in S_{i,j}$ for every $(i,j)\in [k]\times [k]$.

Assume that $(i,j)=(2\alpha-1,2\beta-1)$ for some $(\alpha,\beta)\in [k']\times[k']$, so $\expc_c=\expc_1$; the other cases are symmetric. The total contribution of arcs $a_1,a_2,a_3,a_4$ to the supply of this vertex is equal to $(M-\supply'(a_1))+\supply'(a_2)M+(M^3-\supply'(a_3)M^2)+\supply'(a_4)M^3$. Since the demand $r_{v_{i,j}}=M+M^2+M^3+M^4$ is satisfied exactly, we infer that
$$\expc_1(x,y)=(M-\supply'(a_1))+\supply'(a_2)M+(M^3-\supply'(a_3)M^2)+\supply'(a_4)M^3.$$
Now observe that function $(t_1,t_2,t_3,t_4)\to (M-t_1)+t_2M+(M^3-t_3M^2)+t_4M^3$ is injective for $(t_1,t_2,t_3,t_4)\in [n]^4$, hence by the definition of $\expc_1$ we infer that 
\begin{itemize}
\item $\supply'(a_1)=\supply'(a_2)=\tiling_1(i,j)$, and
\item $\supply'(a_3)=\supply'(a_4)=\tiling_2(i,j)$. 
\end{itemize}
A symmetric reasoning for the other three types of vertices $v_{i,j}$, depending on the parity of coordinates, shows that this conclusion holds for every vertex $v_{i,j}$ for $(i,j)\in [k]\times[k]$. For every horizontal arc $a=(v_{i,j}, v_{i,j+1})$ we have that $\tiling_1(i,j)=\supply'(a)=\tiling_1(i,j+1)$, and for every vertical arc $a=(v_{i,j}, v_{i+1,j})$ we have that $\tiling_2(i,j)=\supply'(a)=\tiling_2(i+1,j)$. This proves that the row and the column conditions hold for solution $\tiling$.
\end{proof}

\newcommand{\inta}{\kappa}
\newcommand{\intb}{\lambda}
\newcommand{\intc}{\rho}

\subsubsection{Gadgets for arcs}

We now provide two different constructions that implement the behaviour of arcs in the \expas problem in the language of \subiso. The first of the reductions ensures that the gadget has small feedback vertex set, while the second that the maximum degree is $3$.

\begin{lemma}\label{lem:arc-gadget-fvs}
Assume we are given a positive integer $M$ and a set $S\subseteq [M-1]\times [M-1]$ such that $x+y=M$ for all $(x,y)\in S$. Then in time polynomial in $|S|$ and $M$, it is possible to construct a graph $H$ that has the following properties:
\begin{itemize}
\item[(i)] $H$ is planar, connected, and it has $4$ prespecified interface vertices $\inta_1$, $\inta_2$, $\intb_1$, $\intb_2$, all lying on the boundary of the outer-face, in this counter-clockwise order.
\item[(ii)] $|V(H)|=4+2M$.
\item[(iii)] After removing $\inta_1$, $\inta_2$, $\intb_1$, $\intb_2$, the gadget $H$ becomes a single path.
\item[(iv)] For every $(x,y)\in S$, one can find two vertex-disjoint paths $P_1$ and $P_2$ in $H$, where $P_1$ has length $2x+1$ and leads from $\inta_1$ to $\inta_2$, while $P_2$ has length $2y+1$ and leads from $\intb_1$ to $\intb_2$.
\item[(v)] Assume we are given two vertex-disjoint paths $P_1$ and $P_2$ in $H$, such that $P_1$ leads from $\inta_1$ to $\inta_2$, $P_2$ leads from $\intb_1$ to $\intb_2$, and every vertex of $H$ lies either on $P_1$ or $P_2$. Then the lengths of the two paths satisfy  $(|P_1|,|P_2|)=(2x+1,2y+1)$ for some $(x,y)\in S$.
\end{itemize}
\end{lemma}
\begin{proof}
Create the four interface vertices $\inta_1$, $\inta_2$, $\intb_1$, $\intb_2$ and a path $Q=q_1-q_2-\ldots-q_{2M}$ of length $2M-1$. Create an edge between $\inta_1$ and $q_1$, and between $\intb_1$ and $q_{2M}$. Moreover, for each $(x,y)\in S$ (recall that $y=M-x$) create an edge between $q_{2x}$ and $\inta_2$, and between $q_{2x+1}$ and $\intb_2$. This concludes the construction; properties (ii), and (iii) follow immediately from the construction, while the planar embedding satisfying property (i) is depicted in Figure~\ref{fig:arc-gadgets}.

\begin{figure}[htbp!]
        \centering
        \subfloat[Construction of Lemma~\ref{lem:arc-gadget-fvs}]{
                \centering
                \def\svgwidth{0.33\columnwidth}
                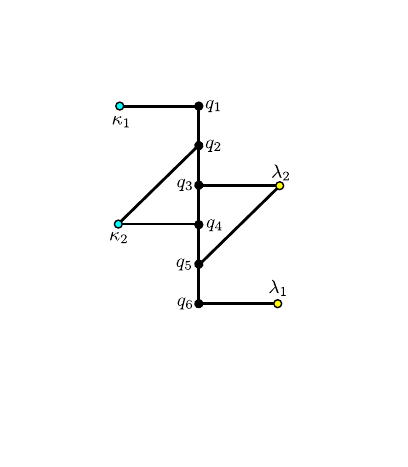
        }
        \qquad\qquad\qquad
        \subfloat[Construction of Lemma~\ref{lem:arc-gadget-degree}]{
                \centering
                \def\svgwidth{0.33\columnwidth}
                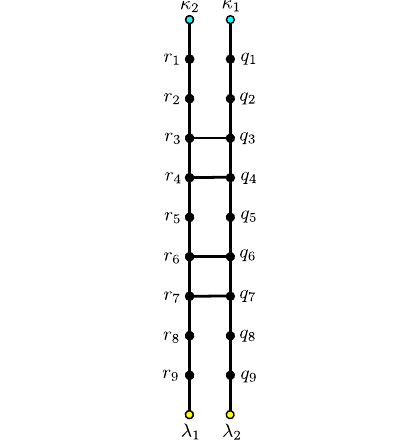
        }
\caption{Constructions of Lemmas~\ref{lem:arc-gadget-fvs} and~\ref{lem:arc-gadget-degree} for $M=3$ and $S=\{(1,2),(2,1)\}$.}\label{fig:arc-gadgets}
\end{figure}

To see that property (iv) is satisfied, construct paths $P_1$ and $P_2$ as follows. As $P_1$, take the prefix of $Q$ traversed from $q_1$ to $q_{2x}$ with $\inta_1$ appended in the beginning and $\inta_2$ appended in the end. As $P_2$, take the suffix of $Q$ from $q_{2M}$ to $q_{2x+1}$ with $\intb_1$ appended in the beginning and $\intb_2$ appended in the end. The construction of $H$ shows that these are indeed paths, while from the fact that $y=M-x$ we infer that $|P_1|=2x+1$ and $|P_2|=2y+1$.

We proceed to property (v). Assume paths $P_1$ and $P_2$ satisfy the conditions of property (v). Observe that $P_1$ must start in $\inta_1$, then proceed to $q_1$ which is the only neighbor of $\inta_1$, and then traverse a prefix of $Q$ up to some vertex $q_s$, from which it proceeds to $\inta_2$ and ends there. Similarly, $P_2$ must start in $\intb_1$, then proceed to $q_{2M}$ which is the only neighbor of $\intb_1$, and then traverse a suffix of $Q$ up to some vertex $q_t$, from which it proceeds $\intb_2$ and ends there. Since $P_1$ and $P_2$ are vertex-disjoint and every vertex of $H$ lies either on $P_1$ or on $P_2$, we infer that $t=s+1$. Moreover, since both $q_s$ and $q_{s+1}$ have edges connecting them to $\inta_2$ and to $\intb_2$, respectively, by the construction of $H$ we infer that $s=2x$ for some $(x,M-x)\in S$. Thus $|P_1|=2x+1$ and $|P_2|=2y+1$ for some $(x,y)\in S$.
\end{proof}

Now we make a second construction that has very similar properties.

\begin{lemma}\label{lem:arc-gadget-degree}
Assume we are given a positive integer $M$ and a set $S\subseteq [M-1]\times [M-1]$ such that $x+y=M$ for all $(x,y)\in S$. Then in time polynomial in $|S|$ and $M$ it is possible to construct a graph $H$ that has the following properties:
\begin{itemize}
\item[(i)] $H$ is planar, connected, and it has $4$ prespecified interface vertices $\inta_1$, $\inta_2$, $\intb_1$, $\intb_2$, all of degree $1$ and lying on the outer-face, in this counter-clockwise order.
\item[(ii)] $|V(H)|=4+6M$.
\item[(iii)] $H$ has maximum degree $3$ and has constant pathwidth.
\item[(iv)] For every $(x,y)\in S$, one can find two vertex disjoint paths $P_1$ and $P_2$ in $H$, where $P_1$ leads from $\inta_1$ to $\inta_2$, $P_2$ leads from $\intb_1$ to $\intb_2$, $|P_1|=6x+1$, and $|P_2|=6y+1$.
\item[(v)] Assume we are given two vertex-disjoint paths $P_1$ and $P_2$ in $H$, such that $P_1$ leads from $\inta_1$ to $\inta_2$, $P_2$ leads from $\inta_2$ to $\intb_2$, and every vertex of $H$ lies either on $P_1$ or $P_2$. Then $(|P_1|,|P_2|)=(6x+1,6y+1)$ for some $(x,y)\in S$.
\end{itemize}
\end{lemma}
\begin{proof}
Create the four interface vertices $\inta_1$, $\inta_2$, $\intb_1$, $\intb_2$ and two paths $Q=\inta_1-q_1-q_2-\ldots-q_{3M}-\intb_2$, $R=\inta_2-r_1-r_2-\ldots-r_{3M}-\intb_1$, each of length $3M+1$. For each $(x,y)\in S$ (recall that $y=M-x$) create an edge between $q_{3x}$ and $r_{3x}$, and between $q_{3x+1}$ and $r_{3x+1}$. This concludes the construction; properties (ii), and (iii) follow immediately from the construction, while the planar embedding satisfying property (i) is depicted in Figure~\ref{fig:arc-gadgets}.

To see that property (iv) is satisfied, construct paths $P_1$ and $P_2$ as follows. As $P_1$, take the prefix of $Q$ traversed from $\inta_1$ to $q_{3x}$, concatenated with the edge $q_{3x}r_{3x}$ and a prefix of $R$ traversed from $r_{3x}$ to $\inta_2$. As $P_2$, take the suffix of $R$ traversed from $\intb_1$ to $r_{3x+1}$, concatenated with the edge $r_{3x+1}q_{3x+1}$ and a suffix of $Q$ traversed from $q_{3x+1}$ to $\intb_2$. The construction of $H$ shows that these are indeed paths, while from the fact that $y=M-x$ we infer that $|P_1|=6x+1$ and $|P_2|=6y+1$.

We proceed to property (v). Assume paths $P_1$ and $P_2$ satisfy the conditions of property (v). Observe that path $P_1$ must start by traversing a prefix of $Q$ from $\inta_1$ to some $q_s$, then use the edge $q_sr_s$, and in order to end in $\inta_2$, it needs then to traverse the prefix of $R$ from $r_s$ to $\inta_2$. Similarly, $P_2$ must start by traversing a suffix of $R$ from $\intb_1$ to some $r_t$, then use the edge $r_tq_t$, and in order to end in $\intb_2$, it needs then to traverse the suffix of $Q$ from $q_t$ to $\intb_2$. Since $P_1$ and $P_2$ are vertex-disjoint and every vertex of $H$ lies either on $P_1$ or on $P_2$, we infer that $t=s+1$. Moreover, since both $q_s$ and $q_{s+1}$ have edges connecting them to $r_s$ and $r_{s+1}$, respectively, by the construction of $H$ we infer that $s=3x$ for some $(x,M-x)\in S$. Thus $|P_1|=3x+1$ and $|P_2|=3y+1$ for some $(x,y)\in S$.
\end{proof}

\subsubsection{The reductions}

We are now in a position to present the reductions showing the hardness of embedding long paths of specified lengths into a planar graph.

\begin{lemma}\label{lem:expas-fvs}
There exists an FPT reduction that, given an instance of \expas with $|D|=k$, outputs an equivalent instance $(H,G)$ of \subiso with the following properties:
\begin{multicols}{2}
\begin{itemize}
\item $\ccn(H)=O(k)$, and
\item $H$ is a forest of paths;
\end{itemize}
\vfill
\columnbreak
\begin{itemize}
\item $\ccn(G)=1$,
\item $G$ is planar, and
\item $\fvs(G),\pw(G)\leq O(k)$.
\end{itemize}
\end{multicols}
\end{lemma}

The immediate corollary of Lemma~\ref{lem:expas-fvs} is the following:

\begin{ntheorem}\label{thm:expas-fvs}
Unless $FPT=W[1]$, there is no algorithm compatible with the description
\begin{eqnarray*}
\sil{\ccn(H),\fvs(G),\pw(G)}{}{\maxdeg(H)\leq 2, \tw(H)\leq 1, \ccn(G)\leq 1, \genus(G)\leq 0}
\end{eqnarray*}
\end{ntheorem}

We proceed to the proof of Lemma~\ref{lem:expas-fvs}.

\newcommand{\dsc}{O}

\begin{proof}[Proof of Lemma~\ref{lem:expas-fvs}]
Let $(D,\{S_a\}_{a\in A(D)},\{r_v\}_{v\in V(D)})$ be the input instance of \expas. We start with the construction of $G$ by modifying the planar embedding of $D$; the construction is depicted in Figure~\ref{fig:packing-paths}. For every vertex $v\in V(D)$, introduce a small disc $\dsc_v$ that covers $v$. Let $d=d(v)$ be the number of arcs incident to $v$ and let us fix some ordering $a_1,a_2,\ldots,a_d$ of these arcs in the clockwise order around $v$ in the embedding of $D$ (note here that we consider both out- and in-arcs of $v$). Whenever we mention the clockwise ordering of arcs around a vertex, we mean this ordering. For every $a_i$ that is an out-arc of $v$, construct vertices $\inta^{a_i}_1,\inta^{a_i}_2$ and place them on the boundary of $\dsc_v$. For every $a_i$ that is an in-arc of $v$, construct vertices $\intb^{a_i}_1,\intb^{a_i}_2$ and place them on the boundary of $\dsc_v$. Order the constructed vertices on the boundary of $\dsc_v$ in the clockwise order as follows: $\intc^{a_1}_1,\intc^{a_1}_2,\intc^{a_2}_1,\intc^{a_2}_2,\ldots,\intc^{a_d}_1,\intc^{a_d}_2$, where $\intc$ stands for $\inta$ or $\intb$ depending whether corresponding $a_i$ is an out- or an in-arc, respectively.

Now, for every arc $a\in A(D)$, introduce a gadget $G_a$ given by Lemma~\ref{lem:arc-gadget-fvs} between $\inta^a_1$, $\inta^a_2$, $\intb^a_1$, $\intb^a_2$ for set $S=S_a$ and $M=M_a$. Note that this gadget can be realized in the plane embedding in the place previously occupied by a small neighborhood of arc $a$, and outside discs $\dsc_v$ for every $v\in V(D)$.

\begin{figure}[htbp!]
        \centering
        \subfloat[Input instance]{
                \centering
                \def\svgwidth{0.40\columnwidth}
                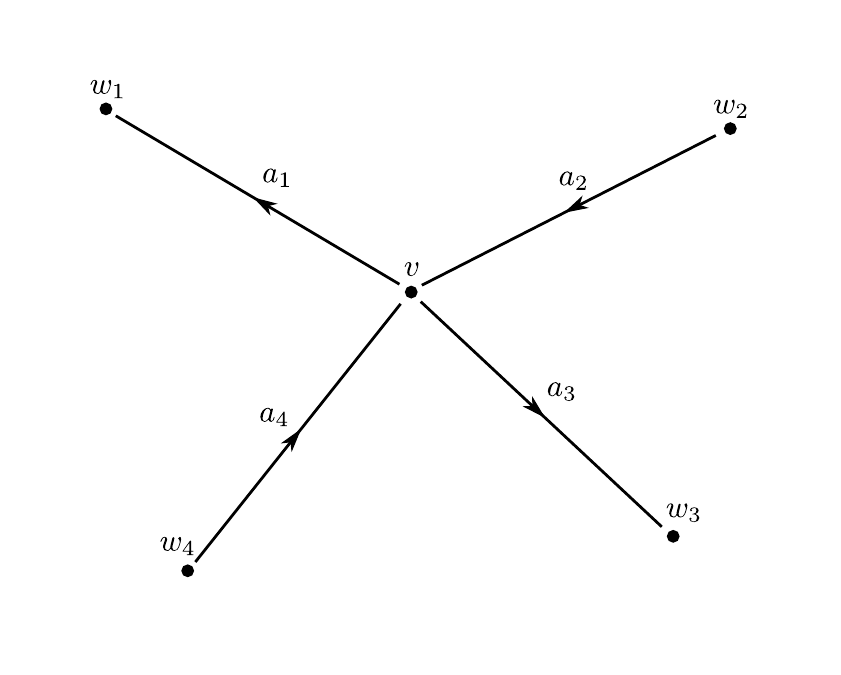
        }
        \qquad $\Longrightarrow$ \qquad
        \subfloat[Output instance]{
                \centering
                \def\svgwidth{0.40\columnwidth}
                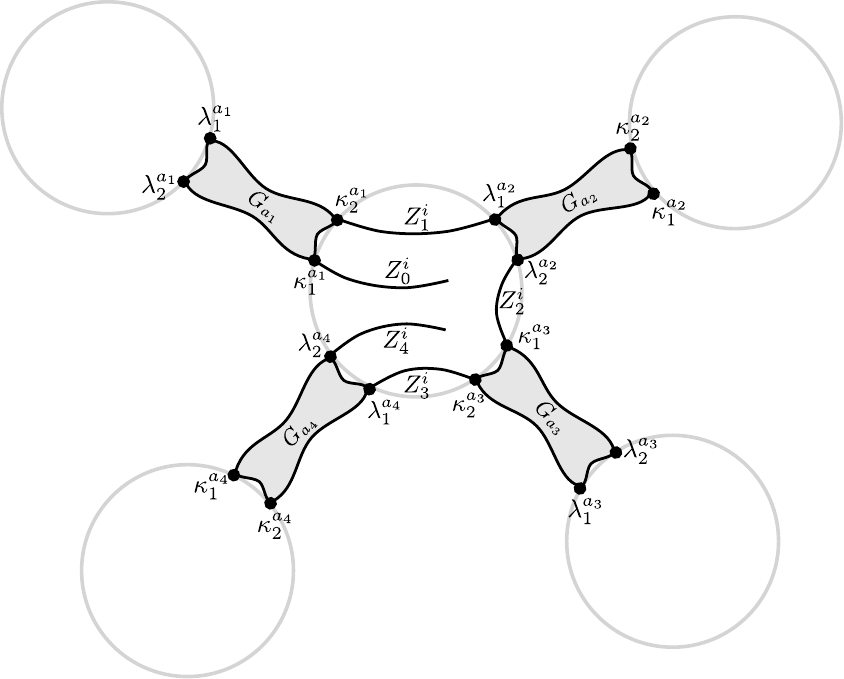
        }
\caption{Construction of Lemma~\ref{lem:expas-fvs} for a vertex $v\in V(D)$ incident to arcs $a_1=(v,w_1)$, $a_2=(w_2,v)$, $a_3=(v,w_3)$, and $a_4=(w_4,v)$.}\label{fig:packing-paths}
\end{figure}

Let $p=|V(D)|$, let $v_1,v_2,\ldots,v_p$ be an arbitrary ordering of vertices of $V(D)$, and let $N=1+\sum_{a\in A(D)} (2M_a+4)=1+\sum_{v\in V(D)} (2r_v+2d(v))$ be the total number of vertices introduced so far in the construction plus one (recall that $d(v)$ is the number of arcs incident to $v$). Let $N_i=(k+1)^{2i-1}\cdot N$. We examine one vertex $v_i$; let $a_1,a_2,\ldots,a_{d(v_i)}$ be arcs incident to $v_i$ in the clockwise order around $v_i$, and let $\intc^{a_1}_1,\intc^{a_1}_2,\intc^{a_2}_1,\intc^{a_2}_2,\ldots,\intc^{a_{d(v_i)}}_1,\intc^{a_{d(v_i)}}_2$ be the introduced vertices on the boundary of $\dsc_{v_i}$ in this order. Introduce now:
\begin{itemize}
\item one path $Z^i_0$ of length $N_i$ ending in $\intc^{a_1}_1$;
\item one path $Z^i_{d(v_i)}$ of length $N_i$ ending in $\intc^{a_{d(v_i)}}_2$;
\item for each $j=1,2,\ldots,d(v_i)-1$, one path $Z^i_j$ of length $N_i+1$ connecting $\intc^{a_j}_2$ and $\intc^{a_{j+1}}_{1}$.
\end{itemize}
By {\em{internal vertices}} of paths $Z^i_j$ we mean the vertices that are not interface vertices of form $\intc^{a_j}_t$. Thus, every path $Z^i_j$ for $j=0,1,\ldots,d(v_i)$ has exactly $N_i$ internal vertices.

For the graph $H$, create $p$ paths $L_1,L_2,\ldots,L_p$. The path $L_i$ will have length exactly $(d(v_i)+1)\cdot (N_i+1)+d(v_i)+2r_{v_i}-2$. Note that $p\leq k$, so $\ccn(H)\leq O(k)$. This concludes the construction.

It is easy to verify that the total number of vertices in graphs $G$ and $H$ is equal, using the assumption that $\sum_{a\in A(D)} M_a=\sum_{v\in V(D)} r_v$ that we can always make about an instance of \expas. Hence, in any subgraph isomorphism from $H$ to $G$, it must necessary hold that every vertex of $G$ is an image of some vertex of $H$. We will exploit this property heavily.

Clearly, $G$ is connected since $D$ was weakly connected, and along with $G$ we have also constructed its planar embedding. Moreover, after removing all the interface vertices, whose number is bounded by $4|A(D)|\leq 4k$, $G$ becomes a forest of paths. Thus $\fvs(G),\pw(G)\leq O(k)$. We are left with proving the equivalence of instances.

Assume that $\supply$ is a solution to the input instance $(D,\{S_a\}_{a\in A(D)},\{r_v\}_{v\in V(D)})$ of \expas. We construct the embedding as follows. For every arc $a=(v_i,v_j)\in A(D)$, using Lemma~\ref{lem:arc-gadget-fvs} we create two paths inside the gadget $G_a$: one between $\inta^{a}_1$ and $\inta^{a}_2$ of length $2x+1$ and one between $\intb^{a}_1$ and $\intb^{a}_2$ of length $2y+1$, where $(x,y)=\supply(a)$. The first path will be a subpath $L_i$, and the second path will be a subpath of $L_j$. For every $i=1,2,\ldots,p$, let $a_1,a_2,\ldots,a_{d(v_i)}$ be arcs incident to $v_i$ in the clockwise around $v_i$, and let $\intc^{a_1}_1,\intc^{a_1}_2,\intc^{a_2}_1,\intc^{a_2}_2,\ldots,\intc^{a_{d(v_i)}}_1,\intc^{a_{d(v_i)}}_2$ be the introduced vertices on the boundary of $\dsc_{v_i}$ in this order. We construct image of $L_i$ by 
\begin{itemize}
\item taking all the subpaths of $L_i$ constructed in gadgets $G_{a_1},G_{a_2},\ldots,G_{a_{d(v_i)}}$,
\item joining every two consecutive ones using the paths $Z_i^j$ for $j=1,2,\ldots,d(v_i)-1$,
\item and appending the paths $Z^i_0$ and $Z^i_{d(v_i)}$ at the ends.
\end{itemize}
Using the fact that $\supply$ is a solution to the input instance of \expas, it is easy to verify that the path constructed in this manner has exactly length $(d(v_i)+1)\cdot (N_i+1)+d(v_i)+2r_v$. Since constructed paths are pairwise vertex-disjoint, they can serve as images of paths $L_i$ in a subgraph isomorphism from $H$ to $G$.

Assume now that we are given a subgraph isomorphism $\eta$ from $H$ to $G$. Observe first that $G$ has at least $2p$ vertices of degree $1$, that is, ends of paths $Z^i_0$ and $Z^i_{d(v_i)}$ for $i=1,2,\ldots,p$. Since every vertex of $G$ is an image of some vertex of $H$, we infer that the ends of paths in $H$ must bijectively map to these $2p$ vertices in $G$.

Examine first the image of path $L_p$. We claim that the total number of vertices of $H$ that are contained on paths $L_1,L_2,\ldots,L_{p-1}$ is smaller than $N_p$. Since path $L_i$ has exactly $(d(v_i)+1)\cdot (N_i+1)+d(v_i)+2r_{v_i}-1=(d(v_i)+1)\cdot N_i+2d(v_i)+2r_{v_i}$ vertices, and $\sum_{v\in V(D)} (2r_v+2d(v))<N$, we have that the total number of vertices on paths $L_1,L_2,\ldots,L_{p-1}$ is smaller $N+\sum_{i=1}^{p-1} (d(v_i)+1)\cdot N_i$. However,
\begin{eqnarray*}
N+\sum_{i=1}^{p-1} (d(v_i)+1)\cdot N_i & = & N+\sum_{i=1}^{p-1} (d(v_i)+1)\cdot (k+1)^{2i-1}\cdot N \\
 & \leq & N\cdot \sum_{i=0}^{p-1} (k+1)^{2i} = N\cdot \frac{(k+1)^{2p}-1}{(k+1)^2-1}\leq N\cdot (k+1)^{2p-1}=N_p.
\end{eqnarray*}
We infer that for every path $Z^p_j$, for $j=1,2,\ldots,d(v_p)-1$, it holds that at least one internal vertex of $Z^p_j$ is an image in $\eta$ of some vertex of $L_p$. Since every vertex of $G$ is an image of some vertex of $H$, we infer that:
\begin{itemize}
\item[(i)] the image of $L_p$ must begin in the degree-1 end of $Z^p_0$ and must end in the degree-1 end of $Z^p_{d(v_p)}$;
\item[(ii)] all the vertices of paths $Z^p_j$ for $j=0,1,2,\ldots,d(v_p)$ must be images of vertices of $L_p$.
\end{itemize}

We now can proceed with the same reasoning for the vertex $v_{p-1}$: similar computations show that the number of vertices on paths $L_1,L_2,\ldots,L_{p-2}$ plus the vertices of $L_p$ with images not on paths $Z^p_j$ for $j=0,1,\ldots,d(v_p)$ (their number is exactly $2r_{v_p}$) is smaller than $N_{p-1}$. Hence, every path $Z^{p-1}_j$ for $j=0,1,\ldots,d(v_{p-1})$ contains an internal vertex that is an image of a vertex of $L_{p-1}$, and we can conclude the analogues of corollaries (i) and (ii) for $L_{p-1}$. Performing the same reasoning for $p-2,p-3,\ldots,1$, we obtain that for each $i=1,2,\ldots,p$:
\begin{itemize}
\item[(i)] the image of $L_i$ must begin in the degree-1 end of $Z^i_0$ and must end in the degree-1 end of $Z^i_{d(v_i)}$;
\item[(ii)] all the vertices of paths $Z^i_j$ for $j=0,1,2,\ldots,d(v_i)$ must be images of vertices of $L_i$.
\end{itemize}

In particular, for every arc $a=(v_i,v_j)\in A(D)$, we have that $\inta^a_1$ and $\inta^a_2$ are images of vertices of $L_i$, and $\intb^a_1$ and $\intb^a_2$ are images of vertices of $L_j$. Since $\inta^a_1$, $\inta^a_2$, $\intb^a_1$, $\intb^a_2$ have only one incident arcs not belonging to the gadget $G_a$, and they are not ends neither of $L_i$ nor of $L_j$, it follows that the gadget constructed for $a$ must contain two paths: $P_1$ connecting $\inta^a_1$, $\inta^a_2$ and being a subpath of the image of $L_i$, and $P_2$ connecting $\intb^a_1$, $\intb^a_2$ and being a subpath of the image of $L_j$. Since every vertex of $G$ is an image of some vertex of $H$, we infer that every vertex of the gadget $G_a$ must belong either to $P_1$ or to $P_2$. By Lemma~\ref{lem:arc-gadget-fvs} we infer that $(|P_1|,|P_2|)=(2x+1,2y+1)$ for some $(x,y)\in S_a$. We define a supply function $\supply$ by setting $\supply(a)=(x,y)$, and claim that $\supply$ is a solution to the input \expas instance.

We already know that $\supply(a)\in S_a$ for every $a\in A(D)$, so let us proceed to checking the demands. For a path $L_i$, let us count how many vertices are images of vertices of $L_i$. We know that $(d(v_i)+1)\cdot (N_i+2)-2=(d(v_i)+1)\cdot (N_i+1)+d(v_i)-1$ of these images lie on paths $Z^i_j$ for $j=0,1,\ldots,d(v_i)$. For every out-arc $a$ of $v_i$ there is $2\supply_1(a)$ images in the interior of gadget $G_a$ (i.e., not counting the interface vertices), and for every in-arc $a$ of $v_i$ there is $2\supply_2(a)$ images in the interior of gadget $G^a$. Since $|V(L_i)|=(d(v_i)+1)\cdot (N_i+1)+d(v_i)+2r_{v_i}-1$, we have that $\sum_{(v_i,v')\in A(D)} \supply_1((v_i,v'))+\sum_{(v',v_i)\in A(D)} \supply_2((v',v_i))=r_{v_i}$ and we are done.
\end{proof}

If we substitute usage of Lemma~\ref{lem:arc-gadget-fvs} with Lemma~\ref{lem:arc-gadget-degree} and adjust the sizes in the construction, since the gadgets of Lemma~\ref{lem:arc-gadget-degree} are of size $6M+4$ instead of $2M+4$ and accommodate also $6$ times longer paths, we obtain the following reduction. Note here that the argument for bounding the pathwidth of $G$ is as follows: after removing all the interface vertices ($O(k)$ of them), $G$ breaks into components of bounded pathwidth.

\begin{lemma}\label{lem:expas-degree}
There exists an FPT reduction that, given an instance of \expas with $|D|=k$, outputs an equivalent instance $(H,G)$ of \subiso with the following properties:
\begin{multicols}{2}
\begin{itemize}
\item $\ccn(H)=O(k)$, and
\item $H$ is a forest of paths;
\end{itemize}
\vfill
\columnbreak
\begin{itemize}
\item $\ccn(G)=1$,
\item $G$ is planar and $\maxdeg(G)\leq 3$, and
\item $\pw(G)\leq O(k)$.
\end{itemize}
\end{multicols}
\end{lemma}

The immediate corollary of Lemma~\ref{lem:expas-degree} is the following:

\begin{ntheorem}\label{thm:expas-degree}
Unless $FPT=W[1]$, there is no algorithm compatible with the description
\begin{eqnarray*}
\sil{\ccn(H),\pw(G)}{}{\maxdeg(H)\leq 2, \tw(H)\leq 1, \ccn(G)\leq 1, \maxdeg(G)\leq 3, \genus(G)\leq 0}
\end{eqnarray*}
\end{ntheorem}

\newcommand{\ill}{\alpha}
\newcommand{\irr}{\beta}

\subsubsection{Bounding cliquewidth}

We now present how to adjust the reductions of Lemmas~\ref{lem:expas-fvs} and~\ref{lem:expas-degree} so that at the cost of increasing genus of the graph we can ensure that it has constant cliquewidth. The approach will be similar in spirit to the biclique gadget introduced in Section~\ref{sec:biclique-gadget}, but because our assumptions about $H$ are very strong, we cannot use this construction directly. Fortunately, we are able to provide an even simpler argument. Let us concentrate on the case when we want to control the size of feedback vertex set, i.e., use Lemma~\ref{lem:arc-gadget-fvs} for the construction of arc gadgets. 

\begin{lemma}\label{lem:expas-fvs-biclique}
There exists an FPT reduction that, given an instance of \expas with $|D|=k$, outputs an equivalent instance $(H,G)$ of \subiso with the following properties:
\begin{multicols}{2}
\begin{itemize}
\item $\ccn(H)=O(k)$, and
\item $H$ is a forest of paths;
\end{itemize}
\vfill
\columnbreak
\begin{itemize}
\item $\ccn(G)=1$,
\item $\genus(G)\leq O(k^2)$,
\item $\fvs(G),\pw(G)\leq O(k)$ and $\cw(G)\leq c$ for some constant $c$.
\end{itemize}
\end{multicols}
\end{lemma}
\begin{proof}
We present only differences with respect to the proof of Lemma~\ref{lem:expas-fvs}.

The construction of gadgets for the arcs of the input instance $(D,\{S_a\}_{a\in A(D)},\{r_v\}_{v\in V(D)})$ of \expas stays the same. When constructing paths $Z^i_j$ we use constants $N_i=(2k)^{2i-1}\cdot N$ instead of original $(k+1)^{2i-1}\cdot N$, and the paths $Z^i_j$ for $j=1,2,\ldots,d(v_i)-1$ are of length $2N_i+3$ instead of $N_i+1$ (paths $Z^i_0$ and $Z^i_{d(v_i)}$ are still of length $N_i$). We also adjust the lengths of paths in $H$: path $L_i$ will be of length $(d(v_i)-1)\cdot (N_i+2)+(d(v_i)+1)\cdot (N_i+1)+d(v_i)+2r_{v_i}-2=2d(v_i)\cdot (N_i+1)+2d(v_i)+2r_{v_i}-3$ instead of original $(d(v_i)+1)\cdot (N_i+1)+d(v_i)+2r_{v_i}-2$, thus accounting for increased lengths of paths $Z^i_j$.

Divide every path $Z^i_j$ for $j=1,2,\ldots,d(v_i)-1$ into two paths $Z^i_{j,1}$ and $Z^i_{j,2}$ of length $N_i+1$ each, where vertices of $Z^i_{j,1}$ consists of the first $N_i+2$ vertices which are closest to $\intc^{a_j}_2$ (including $\intc^{a_j}_2$), vertices of $Z^{i}_{j,2}$ consist of the last $N_i+2$ vertices which are closest to $\intc^{a_{j+1}}_{1}$ (including $\intc^{a_{j+1}}_{1}$), and exactly one edge of $Z^i_j$ between an end of $Z^i_{j,1}$ other than $\intc^{a_j}_2$ and an end of $Z^i_{j,2}$ other than $\intc^{a_{j+1}}_1$ belongs to neither of them. Let this edge be $\ill^i_j\irr^i_j$, where $\ill^i_j\in V(Z^i_{j,1})$ and $\irr^i_j\in V(Z^i_{j,2})$. For paths $Z^i_{j,1},Z^i_{j,2}$, we also exclude vertices $\ill^i_j$ and $\irr^i_j$ from the set of internal vertices. Thus, every path $Z^i_{j,1},Z^i_{j,2}$ has also $N_i$ internal vertices, and the total number of paths $Z^i_0,Z^i_{1,1},Z^i_{1,2},Z^i_{2,1},\ldots,Z^i_{d(v)-1,1},Z^i_{d(v)-1,2},Z^i_{d(v)}$ is $2d(v)$.

We now introduce the crucial modification. Introduce a complete bipartite graph $K$ with all the vertices $\ill^i_j$ in one partite set and all the vertices $\irr^i_j$ in the second partite set. This concludes the construction.

Before we proceed to the proof that the output instance is still equivalent to the input one, let us check the structural properties of $H$ and $G$. Clearly, $H$ is still a forest of at most $k$ paths. $G$ is still connected, and since while introducing the biclique we added at most $O(k^2)$ new edges to a planar graph, each of these edges may be realized using a private handle. Thus $\genus(G)\leq O(k^2)$. After removing the same $O(k)$ vertices from $G$ as before, plus all the vertices of the introduced biclique (at most $O(k)$ of them), $G$ becomes a forest of paths; hence, $\fvs(G),\pw(G)\leq O(k)$. We are left with sketching that $\cw(G)\leq c$ for some constant $c$. Since each arc gadget has constant pathwidth, it can be constructed together with adjacent paths $Z^i_{j,t}$ (or $Z^i_{0},Z^i_{d(v_i)}$) using constant number of labels. Moreover, we can use additional $2$ labels so that when constructing every gadget we assign constructed vertices $\ill^i_j$, $\irr^i_j$ a label expressing whether they are $\ill$ or $\irr$. Therefore, we can construct all the arc gadgets separately, take their disjoint union, and at the end introduce the biclique $K$ using one join operation.

We now prove equivalence of the input and output instance. Since in the construction we just modified some lengths of paths and added edges in the biclique, the solution to the input instance of \expas translates to the solution of output instance of \subiso in the same manner as in the proof of Lemma~\ref{lem:expas-fvs}. We are left with the second direction, where intuitively the difficult part is to prove that the introduced complete bipartite graph $K$ could not create new, unexpected solutions of \subiso. Let $\hm$ be a subgraph isomorphism from $H$ to $G$. Note that we still have the property that $|V(H)|=|V(G)|$, so every vertex of $G$ is an image of some vertex of $H$ in $\hm$.

Firstly, we claim that the crucial claim from the proof of Lemma~\ref{lem:expas-fvs}, i.e., that paths $Z^i_j$ are images of subpaths of $L_i$, still holds. More precisely, we claim that for every $i=1,2,\ldots,|V(D)|$:
\begin{itemize}
\item[(i)] the image of $L_i$ must begin in the degree-1 end of $Z^i_0$ and must end in the degree-1 end of $Z^i_{d(v_i)}$;
\item[(ii)] all the vertices of paths $Z^i_0,Z^i_{1,1},Z^i_{1,2},Z^i_{2,1},\ldots,Z^i_{d(v)-1,1},Z^i_{d(v)-1,2},Z^i_{d(v)}$ for $j=0,1,2,\ldots,d(v_i)$ must be images of vertices of $L_i$.
\end{itemize}
Mimicking the arguments from the proof of Lemma~\ref{lem:expas-fvs}, the only thing we need to verify is the bookkeeping argument that ensures that after paths $L_{i'}$ for $i'>i$ have been already considered, the image of path $L_i$ has a nonempty intersection with each set of $N_i$ vertices, so in particular with internal vertices of each path $Z^i_0$, $Z^i_{d(v_i)}$, and $Z^i_{j,t}$ for $j=1,2,\ldots,d(v_i)-1$ and $t=1,2$. This boils down to proving inequality 
$$N+\sum_{i=1}^{i_0-1} 2d(v_i)\cdot N_i\leq N_{i_0}$$
for every $i_0=1,2,\ldots,p$; note that under the sum the have the total number of internal vertices on paths $Z^i_0,Z^i_{1,1},Z^i_{1,2},Z^i_{2,1},\ldots,Z^i_{d(v)-1,1},Z^i_{d(v)-1,2},Z^i_{d(v)}$. Repeating the arguments from the proof of Lemma~\ref{lem:expas-fvs},
\begin{eqnarray*}
N+\sum_{i=1}^{i_0-1} 2d(v_i)\cdot N_i & = & N+\sum_{i=1}^{i_0-1} 2d(v_i)\cdot (2k)^{2i-1}\cdot N \\
 & \leq & N\cdot \sum_{i=0}^{i_0-1} (2k)^{2i} = N\cdot \frac{(2k)^{2i_0}-1}{(2k)^2-1}\leq N\cdot (2k)^{2i_0-1}=N_{i_0}.
\end{eqnarray*}
Hence we infer that properties (i) and (ii) are indeed satisfied. 

Note now that the rest of the argumentation from the proof of Lemma~\ref{lem:expas-fvs} is a bookkeeping argument that counts the number of vertices in the image of each path $L_i$. The crucial argument is that every gadget introduced for an arc $a=(v_i,v_j)$ has four interface vertices: $\inta^a_1$ and $\inta^a_2$ that are images of vertices of $L_i$, and $\intb^a_1$ and $\intb^a_2$ are images of vertices of $L_j$, and that this means that the gadget must accommodate the image of a subpath of $L_i$ and of a subpath of $L_j$. Yet this argument holds in the same manner also in the current situation, because we already identified the images of ends of $L_i$, and they do not lie inside any arc gadget. The rest of the argumentation holds therefore in the same manner. Note that in particular we do not need to check that the solution does not use the new edges of the biclique $K$ (as it in fact can), since the argumentation only counts the number of vertices in the image of each $L_i$, and does not consider the actual placement of the image of $L_i$ in $G$.

\end{proof}

The immediate corollary of Lemma~\ref{lem:expas-fvs-biclique} is the following:

\begin{ntheorem}\label{thm:expas-fvs-biclique}
Unless $FPT=W[1]$, there is no algorithm compatible with the description
\begin{eqnarray*}
\sil{\ccn(H), \fvs(G), \pw(G), \genus(G)}{\cw(G)}{\maxdeg(H)\leq 2, \tw(H)\leq 1, \ccn(G)\leq 1}
\end{eqnarray*}
\end{ntheorem}

If we now perform the same operation in the proof of Lemma~\ref{lem:expas-degree}, which differs from the proof of Lemma~\ref{lem:expas-fvs} essentially only by usage of Lemma~\ref{lem:arc-gadget-degree} instead of Lemma~\ref{lem:arc-gadget-fvs} for the constructions of arc gadgets, we obtain the following reduction.

\begin{lemma}\label{lem:expas-degree-biclique}
There exists an FPT reduction that, given an instance of \expas with $|D|=k$, outputs an equivalent instance $(H,G)$ of \subiso with the following properties:
\begin{multicols}{2}
\begin{itemize}
\item $\ccn(H)=O(k)$, and
\item $H$ is a forest of paths;
\end{itemize}
\vfill
\columnbreak
\begin{itemize}
\item $\ccn(G)=1$,
\item $\genus(G)\leq O(k^2)$ and $\maxdeg(G)\leq O(k)$,
\item $\pw(G)\leq O(k)$ and $\cw(G)\leq c$ for some constant $c$.
\end{itemize}
\end{multicols}

\end{lemma}

Note here that the construction of the complete bipartite graph $K$ introduces vertices of degrees at most $k+1$. The immediate corollary of Lemma~\ref{lem:expas-degree-biclique} is the following:

\begin{ntheorem}\label{thm:expas-degree-biclique}
Unless $FPT=W[1]$, there is no algorithm compatible with the description
\begin{eqnarray*}
\sil{\ccn(H),\maxdeg(G),\pw(G),\genus(G)}{\cw(G)}{\maxdeg(H)\leq 2,\tw(H)\leq 1,\ccn(G)\leq 1}
\end{eqnarray*}
\end{ntheorem}


\section{Conclusions}
\label{sec:conclusions}The main contribution of the paper is
developing a framework for studying the different parameterizations of
\subiso and completely answering every question arising in this
framework. Systematic studies of parameterizations have been performed
before for various problems
\cite{DBLP:journals/jcss/SamerS10,DBLP:conf/kr/LacknerP12,DBLP:conf/ecai/LacknerP12,planning-ijcai13},
but never on such a massive scale as in the present paper. We have
demonstrated that even if the number of questions is on the order of
billions, finding the maximal set of positive results and the maximal
set of negative results that explain every specific question of the
framework is a doable project and might involve only a few dozen
concrete results. At such a large scale, even verifying that a set of
results explains every possible question is a daunting task. We have
resorted to the help of a computer program that is able to check this
efficiently; the program can be helpful for similar investigations in
the future.

While developing the framework and showing that it can be completely
explained by a small set of results is the conceptually most novel
part of the paper, we would like to emphasize that some of the
concrete positive and negative results are highly nontrivial and
technically novel. On the algorithmic side, we have discovered a
simple, but unexpectedly challenging case: packing a forest $H$ into a
forest $G$, parameterized by the number of connected components of
$H$. We presented a nontrivial randomized dynamic programming
algorithm for this problem using algebraic matching algorithms. Our
investigations turned up an unlikely combination of parameters that
can result in tractable problems: maximum degree, feedback vertex set
number, and genus of $G$. In a somewhat surprising manner,
tractability relies on the fact that a certain property, the existence
of a projection sink, allows us to dramatically reduce treewidth in
bounded-genus CSP instances. This new result on CSPs can be of
independent interest. We have generalized the result to some extent
 to graphs excluding a
fixed minor (with a slightly different parameterization). The generalization is not just a straightforward
application of known structure theorems: we had to use a fairly
complicated dynamic programming scheme on tree decompositions to exploit the
existence of a projection sink and we had to handle almost embeddable graphs instead of bounded-genus graphs,
including all the gory details of vortices.

On the hardness side, many of our W[1]-hardness proofs involve planar
(or bounded-genus) graphs. W[1]-hardness proofs are typically very
involved, as they require complicated gadget constructions. Reducing
from the \gridtiling problem helps streamlining these reductions, but
the actual gadgets have to be constructed in a problem-specific
way. In our case, the construction of the gadgets is particularly
challenging, since we have to satisfy extreme restrictions, such
simultaneously satisfying bounds on, say, cliquewidth, pathwidth, and
maximum degree.

 It might not be apparent from the paper,
but the authors did exercise some restraint when defining the
framework. Only those graph parameters were included in the framework
that already had some interesting nontrivial connection to the \subiso
problem. One could extend the framework with further parameters, such
as chromatic number, girth, or (edge) connectivity, but it is not
clear whether these parameters would influence the complexity of the
problem in an interesting way and whether these parameters would add
anything to the message of the results besides further
complications. Moreover, recall that we have introduced 5 particularly
interesting constraints corresponding to small fixed values of certain
parameters. We did not investigate all possible such constraints
(e.g., small fixed values of cliquewidth or graphs excluding a
$K_6$-minor), as it is unlikely that results involving these specific
constraints would be of as much interest as results on planar graphs
or forests.

The reader might wonder: do the authors advocate this kind of massive
investigation for each and every problem? It seems that the \subiso
problem is particularly suited for such treatment. First, previous
results suggest that a wide range of parameters influence the
complexity of the problem in nontrivial ways. Second, the \subiso
problem involves two graphs $H$ and $G$ and the same parameter for $H$
or $G$ can play very different role. This effectively doubles the
number of parameters that need to be considered. Therefore, the
problem has a very complicated ecology of parameters that can be
understood only with a large-scale formal investigation. For other
problems, say, \textsc{Vertex Coloring}, the complexity landscape is
expected to be much simpler, and probably fewer new results (if any)
need to be invented to explain every combination of
parameters. Therefore, we suggest exploring problems using a detailed
framework similar to ours only if there is evidence for complex
interaction of parameters. Variants of \subiso might be natural
candidates for such investigations: for example, (i) the homomorphism
problem for graphs, (ii) colored versions of \subiso, (iii) extension
versions of \subiso (where we have to extend a partial subgraph
isomorphism given in the input), or (iv) the counting version of \subiso
(this problem was suggested by Petteri Kaski).

\bibliographystyle{abbrv}
\bibliography{subgraph-isomorphism}

\end{document}